\documentclass{article}
\usepackage[margin=1in,footskip=0.25in]{geometry}

\usepackage{cite}
\usepackage{bm}
\usepackage{amsmath,amsthm}
\numberwithin{equation}{section}
\usepackage{extarrows}
\usepackage{amssymb}
\usepackage{graphicx}
\usepackage{xcolor}
\usepackage{subfig}
\usepackage{hyperref}
\usepackage{algorithmic,algorithm}

\newcommand{\mr}{\mathrm}

\newcommand{\BE}{\begin{equation}}
\newcommand{\EE}{\end{equation}}
\newcommand{\BS}{\begin{subequations}}
\newcommand{\ES}{\end{subequations}}
\newcommand{\UH}{\mr{H}}   
\newcommand{\UT}{\mr{T}}   
\newcommand{\ampa}{\rm AMP.A}
\newcommand{\amps}{\rm AMP.S}

\newcommand{\alphah}{\hat{\alpha}}

\newcommand{\tauh}{\hat{\tau}}

\newcommand{\deltaGlobal}{\delta_{\mr{global}}}   
\newcommand{\deltaAMP}{\delta_{\mr{AMP}}}   

\newcommand{\Mydef}{\overset{  \scriptscriptstyle \Delta  }{=}}

\newtheorem{theorem}{Theorem}

\newtheorem{definition}{Definition}
\newtheorem{remark}{Remark}
\newtheorem{lemma}{Lemma}

\newtheorem{finding}{Finding}

\def\qand{\ \text{and} \ } 

\graphicspath{{figures/}}

\begin{document}
\title{Optimization-based AMP for Phase Retrieval:
 The Impact of Initialization and $\ell_2$-regularization}

\author{Junjie Ma\thanks{Department of Statistics, Columbia University.  \texttt{jm4520@columbia.edu}\;\;} \and Ji Xu\thanks{Department of Computer Science, Columbia University. \texttt{jixu@cs.columbia.edu}\;\;} \and Arian Maleki\thanks{Department of Statistics, Columbia University. \texttt{arian@stat.columbia.edu}}}

\maketitle

\begin{abstract}
We consider an $\ell_2$-regularized non-convex optimization problem for recovering signals from their noisy phaseless observations. We design and study the performance of a message passing algorithm that aims to solve this optimization problem. We consider the asymptotic setting $m,n \rightarrow \infty$, $m/n \rightarrow \delta$ and obtain sharp performance bounds, where $m$ is the number of measurements and $n$ is the signal dimension. We show that for complex signals the algorithm can perform accurate recovery with only $m=  \left(\frac{64}{\pi^2}-4\right)n \approx 2.5n$ measurements. Also, we provide sharp analysis on the sensitivity of the algorithm to noise. We highlight the following facts about our message passing algorithm: (i) Adding $\ell_2$ regularization to the non-convex loss function can be beneficial. (ii) Spectral initialization has marginal impact on the performance of the algorithm. The sharp analyses in this paper, not only enable us to compare the performance of our method with other phase recovery schemes, but also shed light on designing better iterative algorithms for other non-convex optimization problems. 
\end{abstract}


\tableofcontents
\newpage 

\section{Introduction}\label{Sec:introduction}

\subsection{Notations}
$\bar{a}$ denotes the conjugate of a complex number $a$. $\angle a$ denotes the phase of $a$. We use bold lower-case and upper case letters for vectors and matrices respectively. For a matrix $\bm{A}$, $\bm{A}^{\UT}$ and $\bm{A}^{\UH}$ denote the transpose of a matrix and its Hermitian respectively. Throughout the paper, we also use the following two notations: $\mathbf{1}\Mydef[1,\ldots,1]^{\UT}$ and $\mathbf{0}\Mydef[0,\ldots,0]^{\UT}$. $\phi(x)$ and $\Phi(x)$ are used for the probability density function and cumulative distribution function of the standard Gaussian random variable. A random variable $a$ said to be circularly-symmetric Gaussian, denoted as $a\sim\mathcal{CN}(0,\sigma^2)$, if $a=a_R+\mr{i}a_I$ and $a_R$ and $a_I$ are two independent real Gaussian random variables with mean zero and variance $\sigma^2/2$. Finally, we define $\langle \bm{a},\bm{b} \rangle \Mydef \sum_{i=1}\bar{a}_i b_i$ for $\bm{a},\bm{b}\in\mathbb{C}^{d}$ .

\subsection{Informal statement of our results}

Phase retrieval refers to the task of recovering a signal $\bm{x}_*\in\mathbb{C}^{n\times 1}$ from its $m$ phaseless linear measurements:
\BE
y_a = \bigg| \sum_{i=1}^n A_{ai}x_{*,i} \bigg| +w_a ,\quad a = 1,2,\ldots,m,
\EE
where $x_{*,i}$ is the $i$th component of $\bm{x}_*$ and $w_a\sim\mathcal{CN}(0,\sigma^2_w)$ a Gaussian noise. The recent surge of interest \cite{Candes2013,Eetrapalli2013,Eldar2014,CaLiSo15,ChenCandes17,Wang2016,Zhang2016reshaped,Goldstein2016phasemax,Aahmani2016,Cai2016,sun2016geometric,Soltanolkotabi2017,Duchi2017,Lu17,Davis2017,Soltanolkotabi2017,Tan2017phase,Jeong2017,Zeng2017,Mondelli2017,Dhifallah17,Dhifallah2017phase,Abbasi2017,Qu2017} has led to a better understanding of the theoretical aspects of this problem. Thanks to such research we now have access to several algorithms, inspired by different ideas, that are theoretically guaranteed to recover $\bm{x}_*$ exactly in the noiseless setting. Despite all this progress, there is still a gap between the theoretical understanding of the recovery algorithms and what practitioners would like to know. For instance, for many algorithms, including Wirtinger flow \cite{CaLiSo15,ChenCandes17} and amplitude flow \cite{Wang2016,Zhang2016reshaped}, the exact recovery is guaranteed with either $c n \log n$ or $cn$ measurements, where $c$ is often a fixed but large constant that does not depend on $n$. In both cases, it is often claimed that the large value of $c$ or the existence of $\log n$ is an artifact of the proving technique and the algorithm is expected to work with $cn$ for a reasonably small value of $c$. Such claims have left many users wondering 
\begin{enumerate}
\item [Q.1] Which algorithm should we use? Since the theoretical analyses are not sharp, they do not shed any light on the relative performance of different algorithms. Answering this question through simulations is very challenging too, since many factors including the distribution of the noise, the true signal  $\bm{x}_*$, and the number of measurements may have impact on the answer.

\item[Q.2] When can we trust the performance of these algorithms in the presence of noise? Suppose for a moment that we know the minimum number of measurements that is required for the exact recovery through simulations. Should we collect the same number of measurements in the noisy settings too? 

\item[Q.3] What is the impact of initialization schemes, such as spectral initialization? Can we trust these initialization schemes in the presence of noise? How should we compare different initialization schemes? 
\end{enumerate}

Researchers have developed certain intuition based on a combination of theoretical and empirical results,  to give heuristic answers to these questions. However, as demonstrated in a series of papers in the context of compressed sensing, such folklores are sometimes inaccurate \cite{Zheng17}. To address Question Q.1, several researchers have adopted the asymptotic framework $m,n \rightarrow \infty$, $m/n \rightarrow \delta$, and provided sharp analyses for the performance of several algorithms \cite{Dhifallah17, Dhifallah2017phase,Abbasi2017}. This line of work studies recovery algorithms that are based on convex optimization. In this paper, we adopt the same asymptotic framework and study the following popular non-convex problem, known as amplitude-based optimization \cite{Zhang2016reshaped,Wang2016,wang2017solving}: 
\BE\label{Eqn:amplitude}
\underset{\bm{x}}{\min}\quad  \sum_{a=1}^m \left(y_a-|(\bm{Ax})_a|\right)^2 + \frac{\mu_k}{2} \|\bm{x}\|_2^2.
\EE
where $(\bm{Ax})_a$ denotes the $a$-th entry of $\bm{Ax}$. Note that compared to the optimization problem discussed in \cite{Zhang2016reshaped,Wang2016}, \eqref{Eqn:amplitude} has an extra $\ell_2$-regularizer. Regularization is known to reduce the variance of an estimator and hence is expected to be useful when $\bm{w} \neq \bm{0}$. However, as we will try to clarify later in this section, since the loss function  $ \sum_{a=1}^m \left(y_a-|(\bm{Ax})_a|\right)^2$ is non-convex, regularization can help the iterative algorithm that aims to solve \eqref{Eqn:amplitude} even in the noiseless settings. \\

Since \eqref{Eqn:amplitude} is a non-convex problem, the algorithm to solve it matters. In this paper, we study a message passing algorithm that aims to solve \eqref{Eqn:amplitude}. As a result of our studies we 
\begin{enumerate}
\item present sharp characterization of the mean square error (even the constants are sharp) in both noiseless and noisy settings.
\item present a quantitative characterization of the gain initialization and regularization can offer to our algorithms. 
\end{enumerate}
Furthermore, the sharpness of our results enables us to present a quantitative and accurate comparison with convex optimization based recovery algorithms \cite{Dhifallah17, Dhifallah2017phase,Abbasi2017} and give partial answers to Question Q.1 mentioned above. Below we introduce our message passing algorithm and informally state some of our main results. The careful and accurate statements of our results are postponed to Section \ref{Sec:contributions}. \\

 Following the steps proposed in \cite{Rangan11}, we obtain the following algorithm called, \textit{Approximate Message Passing for Amplitude-based optimization} (AMP.A). Starting from an initial estimate $\bm{x}^0\in\mathbb{C}^{n\times1}$, AMP.A proceeds as follows for $t\ge0$:

\BS 
\begin{align*}
\bm{p}^t &= \bm{Ax}^t - \frac{\lambda_{t-1}}{\delta}\cdot\frac{g(\bm{p}^{t-1},\bm{y})}{-\mr{div}_p(g_{t-1})  },\\
\bm{x}^{t+1}&=\lambda_t\cdot \left(\bm{x}^t +\bm{A}^{\UH} \frac{g(\bm{p}^t,\bm{y})}{- \mr{div}_p(g_{t})  } \right). 
\end{align*}
In these iterations
\[
g({p},{y})=y\cdot \frac{p}{|p|}-p,
\]
and
\[
\begin{split}
\lambda_t&=\frac{-\mr{div}_p(g_t) }{-  \mr{div}_p(g_t) +\mu_k \left(\tau_t +\frac{1}{2}\right)},\\
\tau^t &=\frac{1}{\delta}\frac{\tau^{t-1} + \frac{1}{2}}{ - \mr{div}_p(g_{t-1}) }\cdot \lambda_{t-1}.
\end{split}
\]
\ES
In the above, $p/|p|$ at $p=0$ can be any fixed number and does not affect the performance of AMP.A. Further, the ``divergence'' term $\mr{div}_p(g_{t}) $ is defined as
\BE \label{Eqn:partial_p_complex}
\begin{split}
\mr{div}_p(g_{t})  &\Mydef  \frac{1}{m}\sum_{a=1}^m  \frac{1}{2}\left( \frac{\partial g(p_a^t,y_a)}{\partial p_a^R} -\mr{i} \frac{\partial g(p_a^t,y_a)}{\partial p_a^I}\right)\\
&= \frac{1}{m}\sum_{a=1}^m\frac{y_a}{2|p_a^t|}-1,
\end{split}
\EE
where $p_a^R$ and $p_a^I$ denote the real and imaginary parts of $p_a^t$ respectively (i.e., $p_a^t=p_a^R+\mr{i} p_a^I$). 
For readers' convenience, we include the derivations of AMP.A in Appendix \ref{Sec:AMP_derivations}. 

The first point that we would like to discuss here is the effect of the regularizer on AMP.A. For the moment suppose that the noise $\bm{w}$ is zero. Does including the regularizer in \eqref{Eqn:amplitude} benefit AMP.A? Clearly, any regularization may introduce unnecessary bias to the solution. Hence, if the final goal is to obtain $\bm{x}_*$ exactly we should set $\mu_k =0$. However, the optimization problem in \eqref{Eqn:amplitude} is non-convex and iterative algorithms intended to solve it can get stuck at bad local minima. In this regard, regularization can still help AMP.A to escape bad local minima through continuation. Continuation is popular in convex optimization for improving the convergence rate of iterative algorithms \cite{hale2008fixed}, and has been applied to the phase retrieval problem in \cite{Balan2012reconstruction}. In continuation we start with a value of $\mu_k$ for which AMP.A is capable of finding the global minimizer of \eqref{Eqn:amplitude}. Then, once $\ampa$ converges we will either decrease or increase $\mu_k$ a little bit (depending on the final value of $\mu$ for which we want to solve the problem) and use the previous fixed point of AMP.A as the initialization for the new AMP.A. We continue this process until we reach the value of $\mu_k$ we are interested in.  For instance, if we would like to solve the noiseless phase retrieval problem then $\mu_k$ should eventually go to zero so that we do not introduce unnecessary bias. The rationale behind continuation is the following. Let $\mu_k$ and $\mu'_k$ be two different values of  the regularization parameter, and they are close to each other. Suppose that the global minimizer of \eqref{Eqn:amplitude} with regularization parameter $\mu'_k$ is $\bm{x}(\mu'_k)$ and is given to the user. Suppose further that the user would like to find the global minimizer of \eqref{Eqn:amplitude} with $\mu_k$. Then, it is conceivable that the global minimizer of the new problem is close to $\bm{x}(\mu'_k)$.\footnote{Given the sometimes complex geometry of non-convex problems, this might not always be the case. } Hence, the user can initialize AMP.A with  $\bm{x}(\mu'_k)$ and hope that the algorithm may converge to the global minimizer of \eqref{Eqn:amplitude} for $\mu_k$.

A more general version of the continuation idea we discussed above is to let $\mu_k$ change at every iteration (denoted as $\mu^t_k$), and set $\lambda_t$ according to $\mu^t_k$:
\BE \label{eq:continuationstra}
\lambda_t=\frac{- \mr{div}_p(g_t) }{- \mr{div}_p(g_t) +\mu_k^t \left(\tau_t +\frac{1}{2}\right)},
\EE
This way we can not only automate the continuation process, but also let AMP.A decide which choice of $\mu_k$ is appropriate at a given stage of the algorithm. 
Our discussion so far has been heuristic. It is not clear whether and how much the generalized continuation can benefit the algorithm. To give a partial answer to this question we focus on the following particular continuation strategy: $\mu_k^t=\frac{1+2\mr{div}_p(g_t)}{1+2\tau_t}$ and obtain the following version of AMP.A:
\BS \label{Eqn:AMP_complex}
\begin{align}
\bm{p}^t &= \bm{Ax}^t -\frac{2}{\delta}g(\bm{p}^{t-1},\bm{y}), \\
\bm{x}^{t+1} &= 2\left[-\mr{div}_p(g_t)\cdot \bm{x}^t +\bm{A}^{\UH} g(\bm{p}^t,\bm{y})\right].
\label{Eqn:AMP_complex_b}
\end{align}
\ES

Below we informally discuss some of the results we will prove in this paper. \\

\noindent \textbf{Informal result 1.} Consider the AMP.A algorithm for complex-valued signals with $\mu_k^t=\frac{1+2\mr{div}_p(g_t)}{1+2\tau_t}$. Under the noiseless setting, if $\delta > \frac{64}{\pi^2} - 4 \approx 2.5$, then $\bm{x}^t$ ``converges to'' $\bm{x}_{\ast}$ as long as the initial estimate $\bm{x}^0$ is not orthogonal to $\bm{x}_{\ast}$ and $\|\bm{x}^0\|=\|\bm{x}_*\|$. When $2<\delta < \frac{64}{\pi^2} - 4$,  $\ampa$ has a fixed point at $\bm{x}_{\ast}$. However, it has to be initialized very carefully to reach $\bm{x}_*$. \\

Before we discuss and explain the implications of this result, let us expand the scope of our results. This extension enables us to compare our results with existing work \cite{Dhifallah17, Dhifallah2017phase,Abbasi2017}. So far, we have discussed the case $\bm{x}_* \in \mathbb{C}^n$. However, in some applications, such as  astronomical imaging, we are interested in real-valued signals $\bm{x}_* \in \mathbb{R}^n$. In Section \ref{sec:real_valued}, we will introduce a real-valued version of $\ampa$. The following informal result summarizes the performance of this algorithm. \\

\noindent \textbf{Informal result 2.} Consider the AMP.A algorithm for real-valued signals with $\mu_k^t=\frac{2+2\mr{div}_p(g_t)}{1+2\tau_t}$. Under the noiseless setting, if $\delta > \frac{\pi^2}{4} -1 \approx 1.5$, then $\bm{x}^t$  ``converges to'' $\bm{x}_{\ast}$ as long as the initialization is not orthogonal to $\bm{x}_{\ast}$. When $1+ \frac{4}{\pi^2}<\delta < \frac{\pi^2}{4} - 1$,  $\ampa$ has a fixed point at $\bm{x}_{\ast}$.  However, it has to be initialized very carefully to reach $\bm{x}_*$. \\

We would like to make the following remarks about these two results:

\begin{enumerate}
\item As is clear from our second informal result, when $\delta< 1+ \frac{4}{\pi^2}$, $\ampa$ cannot converge to $\bm{x}_*$. This value of $\delta$ is different from the information theoretic lower bound $\delta =1$. This discrepancy is in fact due to the type of continuation we used in this paper. Note that this issue does not happen in the complex-valued $\ampa$. The search for a better continuation strategy for the real-valued $\ampa$ is left as future research.

\item Simulation results presented in our forthcoming paper \cite{Plan} show that for real-valued signals, AMP.A with $\mu_k=0$ can only recover when $\delta> 2.5$. As mentioned in our second informal result, continuation has improved the threshold of correct recovery to $ \delta \approx 1.5$.    

\item  How much does spectral initialization improve the performance of AMP.A? To answer this question, let us focus on the real-valued signals. As discussed in our second Informal result, two values of $\delta$ are important for AMP.A: $\delta =\frac{\pi^2}{4}-1\approx1.5$ and $\delta = 1+ \frac{4}{\pi^2} \approx 1.4$. If $\delta > 1.5$, then AMP.A recovers $\bm{x}_*$ exactly as long as the initialization is not orthogonal to $\bm{x}_{\ast}$. In this case spectral method helps, since it offers an initialization that is not orthogonal to $\bm{x}_{\ast}$. However, if the mean of $\bm{x}_*$ is not zero, a simple initial estimate $\mathbf{1} = [1,1, \ldots, 1]^\UT$ can work as well as the spectral initialization. Hence, in this case spectral initialization does not offer a major improvement. A more important question is whether spectral initialization can help AMP.A to perform exact recovery for $\delta < 1.5$. Our forthcoming paper \cite{Plan} shows that the answer to this question is negative. Hence, as long as the final estimate of AMP.A is concerned, the impact of spectral initialization seems to be marginal. 

\end{enumerate}

Now let us discuss the performance of AMP.A under noisy settings. We assume that the measurement noise is Gaussian and small. Clearly, in this setting exact recovery is impossible, hence we study the asymptotic mean square error defined as the following almost sure limit ($\theta_t \Mydef \angle \frac{1}{n}\langle \bm{x}_*,\bm{x}^t \rangle$)
\BE\label{Eqn:AMSE_def}
{\rm AMSE}(\delta, \sigma^2_w) \triangleq \lim_{t \rightarrow \infty} \frac{\|\bm{x}^t -e^{\mr{i}\theta_t} \bm{x}_*\|_2^2}{n},
\EE

\noindent \textbf{Informal result 3.} Consider the AMP.A algorithm for complex-valued signals with $\mu_k^t=\frac{1+2\mr{div}_p(g_t)}{1+2\tau_t}$. Let $\delta > \frac{64}{\pi^2} - 4 \approx 2.5$, then 
\begin{equation} \label{eq:PRNoiseSens1COMPLEX}
\lim_{\sigma^2_w \rightarrow 0}  \frac{{\rm AMSE}(\delta, \sigma^2_w)}{\sigma^2_w} = \frac{4}{1 - \frac{2}{\delta}}.
\end{equation}
Notice that the above result was derived based under the assumption $\mathbb{E}[|A_{ai}|^2]=1/m$. To interpret the above result correctly, we should discuss the signal to noise ratio of each measurement. Suppose that $\frac{1}{n} \|\bm{x}_*\|^2 =1$. Then the signal to noise ratio of each measurement is $\mathbb{E} \big[\left|\sum_i A_{ai} x_{*,i} \right|^2\big]/\sigma^2_w = \frac{1}{\delta \sigma_w^2}$. In other words, as we increase the number of measurements or equivalently $\delta$, then we reduce the signal to noise ratio of each measurement too. This causes some issues when we compare the ${\rm AMSE}(\delta, \sigma_w^2)$ for different values of $\delta$. One easy fix is to assume that the variance of the noise is $\sigma_w^2 = \frac{\tilde{\sigma}_w^2}{\delta}$, where $\tilde{\sigma}_w^2$ is a fixed number. Then we can define the noise sensitivity as
\[
{\rm NS} (\tilde{\sigma}_w^2, \delta) =  \frac{{\rm AMSE}(\delta, \sigma^2_w)}{\tilde{\sigma}^2_w}.
\]
It is straightforward to use \eqref{eq:PRNoiseSens1COMPLEX} to show that ${\rm NS} (\tilde{\sigma}_w, \delta) =  \frac{4}{\delta-2}$.  Note that if we use $\ampa$ with $\delta \approx \deltaAMP$, then the noise sensitivity is approximately $8$. If this level of noise sensitivity is not acceptable for an application, then the user should collect more measurements to reduce the noise sensitivity.  Noise sensitivity can also be calculated for real-valued AMP.A: \\

\noindent \textbf{Informal result 4.} Consider the AMP.A algorithm for real-valued signals with $\mu_k^t=\frac{2+2\mr{div}_p(g_t)}{1+2\tau_t}$. Let $\delta > \frac{\pi^2}{4} -1 \approx 1.5$, then 
\[
\lim_{\sigma^2_w \rightarrow 0}  \frac{{\rm AMSE}(\delta, \sigma^2_w)}{\sigma^2_w} = \frac{1}{\left(1+\frac{4}{\pi^2}\right)^{-1} - \frac{1}{\delta} }.
\]

\subsection{Related work}\label{ssec:relwork}

\subsubsection{Existing theoretical work}

Early theoretical results on phase retrieval, such as PhaseLift \cite{Candes2013} and PhaseCut \cite{Waldspurger2015}, are based on semidefinite relaxations. For random Gaussian measurements, a variant of PhaseLift can recover the signal exactly (up to global phase) in the noiseless setting using $O(n)$ measurements \cite{Candes2014solving}. However, PhaseLift (or PhaseCut) involves solving a semidefinite programming (SDP) and is computationally prohibitive for large-scale applications. A different convex optimization approach for phase retrieval, which has the same $O(n)$ sample complexity, was independently proposed in \cite{Goldstein2016phasemax} and \cite{Aahmani2016}. This method is formulated in the natural signal space and does not involve lifting, and is therefore computationally more attractive than SDP-based counterparts. However, both methods require an anchor vector that has non-zero correlation with the true signal, and the quality of the recovery highly depends on the quality of the anchor.

Apart from convex relaxation approaches, non-convex optimization approaches attract considerable recent interests. These algorithms typically consist of a carefully designed initialization step (usually accomplished via a spectral method \cite{Eetrapalli2013}) followed by iterations that refine the estimate. An early work in this direction is the alternating minimization algorithm proposed in \cite{Eetrapalli2013}, which has sub-optimal sample complexity. Another line of work includes the Wirtinger flow algorithm \cite{CaLiSo15,Ma2017implicit}, truncated Wirtinger flow algorithm \cite{ChenCandes17}, and other variants\cite{Cai2016,Zhang2016reshaped,Wang2016,wang2017solving,Soltanolkotabi2017}. Other approaches include Kaczmarz method \cite{Wei2015,Chi2016,Tan2017phase,Jeong2017}, trust region method \cite{sun2016geometric}, coordinate decent \cite{Zeng2017},  prox-linear algorithm \cite{Duchi2017} and Polyak subgradient method \cite{Davis2017}.

All the above theoretical results guarantee successful recovery with $m=\delta n$ measurements (or more) where $\delta$ is a fixed often large constant. However, such theories are not capable of providing fair comparison among different algorithms. To resolve this issue researchers have started studying the performance of different algorithms under the asymptotic setting $m/n\to\delta$ and $n\to\infty$. An interesting iterative projection method was proposed in\cite{Li2015phase}, whose dynamics can be characterized exactly under this asymptotic setting. However, \cite{Li2015phase} does not analyze the number of measurements required for this algorithm to work. The work in \cite{Lu17} provides sharp characterization of the spectral initialization step (which is a key ingredient to many of the above algorithms). The analysis in \cite{Lu17} reveals a phase transition phenomenon: spectral method produces an estimate not orthogonal to the signal if and only if $\delta$ is larger than a threshold (called ``weak threshold'' in \cite{Mondelli2017}). Later, \cite{Mondelli2017} derived the information-theoretically optimal weak threshold (which is $0.5$ for the real-valued model and $1$ for the complex-valued model) and proved that the optimal weak threshold can be achieved by an optimally-tuned spectral method. Using the non-rigorous replica method from statistical physics, \cite{Dhifallah17} analyzes the exact threshold of $\delta$ (for the real-value setting) above which the PhaseMax method in \cite{Goldstein2016phasemax} and \cite{Aahmani2016} achieves perfect recovery. The analysis in \cite{Dhifallah17} shows that the performance of PhaseMax highly depends on initialization (see Fig. 1 of \cite{Dhifallah17}), and the required $\delta$ is lower bounded by $2$ for real-valued models. The analysis in \cite{Dhifallah17} was later rigorously proved in \cite{Dhifallah2017phase} via the Gaussian min-max framework \cite{thrampoulidis2016precise,thrampoulidis2015regularized}, and a new algorithm called PhaseLamp was proposed. The PhaseLamp method has superior
recovery performance over PhaseMax, but again it does not work when $\delta<2$ for real-valued models. A recent paper \cite{salehi2018precise} extends the asymptotic analysis of \cite{Dhifallah2017phase} to the complex-valued setting, and it was shown that PhaseMax cannot work for $\delta<4$. On the other hand, AMP.A proposed in this paper achieves perfect recovery when $\delta>1.5$ and $\delta>2.5$, for the real and complex-valued models respectively. Further, \cite{Dhifallah17,Dhifallah2017phase} focus on the noiseless scenario, while in this paper we also analyze the noise sensitivity of AMP.A. Finally, a recent paper \cite{Abbasi2017} derived an upper bound of $\delta$ such that PhaseLift achieves perfect recovery. The exact value of this upper bound can be derived by solving a three-variable convex optimization problem and empirically \cite{Abbasi2017} shows that $\delta\approx3$ for real-valued models.

\subsubsection{Existing work based on AMP}
Our work in this paper is based on the approximate message passing (AMP) framework \cite{DoMaMo09,Bayati&Montanari11}, in particular the generalized approximate message passing (GAMP) algorithm developed and analyzed in \cite{Rangan11,Javanmard2013}. A key property of AMP (including GAMP) is that its asymptotic behavior can be characterized exactly via the state evolution platform \cite{DoMaMo09,Bayati&Montanari11,Rangan11,Javanmard2013}. 

For phase retrieval, a Bayesian GAMP algorithm has been proposed in \cite{Schniter2015}. However, \cite{Schniter2015} did not provide rigorous performance analysis, partly due to the heuristic treatments used in the algorithm (such as damping and restart). Another work related to ours is the recent paper \cite{barbier2017phase} (appeared on Arxiv while we are preparing this paper), which analyzed the phase transitions of the Bayesian GAMP algorithms for a class of nonlinear acquisition models. For the phase retrieval problem, a phase transition diagram was shown in \cite[Fig.~1]{barbier2017phase} under a Bernoulli-Gaussian signal prior. The numerical results in \cite{barbier2017phase} indeed achieve state-of-the-art reconstruction results for real-valued models. However, \cite{barbier2017phase} did not provide the analysis of their results and in particular did not mention how they handle a difficulty related to initialization. Further, the algorithm in \cite{barbier2017phase} is based on the Bayesian framework which assumes that the signal and the measurements are generated according to some known distributions. 
Contrary to \cite{Schniter2015} and \cite{barbier2017phase}, this paper considers a version of GAMP derived from solving the popular optimization problem \eqref{Eqn:amplitude}. We provide rigorous performance analysis of our algorithm for both real and complex-valued models. Note that the advantages and disadvantages of Bayesian and optimization-based techniques have been a long debate in the field of Statistics. Hence, we do not repeat those debates here. Given our experience in the fields of compressed sensing and phase retrieval, it seems that the performance of Bayesian algorithms are more sensitive to their assumptions than the optimization-based schemes. Furthermore, performance analyses of Bayesian algorithms are often very challenging under ``non-ideal'' situations which the algorithms are not designed for.

Here, we emphasize another advantage of our approach. Given the fact that the most popular schemes in practice are iterative algorithms derived for solving non-convex optimization problems, the detailed analyses of $\ampa$ presented in our paper may also shed light on the performance of these algorithms and suggest new ideas to improve their performances.

\subsubsection{Fundamental limits}
It the literature of phase retrieval, it is well known that to make the signal-to-observation mapping injective one needs at least $m=4n$ measurements  \cite{Bandeira2014} (or $m=2n$ \cite{Balan2006signal} in the case of real-valued models).
On the other hand, the measurement thresholds obtained in this paper are $\delta=\frac{64}{\pi^2}-4\approx2.5$ and $\delta=\frac{\pi^2}{4}-1\approx1.5$ respectively. In fact, our algorithm can in principal recover the signal when $\delta>2$ and $\delta>1+\frac{4}{\pi^2}$ (or $\delta>1$ if continuation is not applied) for complex and real-valued models, provided that the algorithm is initialized close enough to the signal (though no known initialization strategy can accomplish this goal). Hence, our threshold are even smaller than the injectivity bounds. We emphasize that this is possible since the injectivity bounds derived in \cite{Balan2006signal,Bandeira2014} are defined for \textit{all} $\bm{x}_{*}$ (which can depend on $\bm{A}$ in the worst case scenario). This is different from our assumption that $\bm{x}_{*}$ is independent of $\bm{A}$, which is more relevant in applications where one has some freedom to randomize the sampling mechanism. In fact, several papers have observed that their algorithm can operate at the injectivity thresholds $\delta=2$ for real-valued models \cite{Wang2016,Duchi2017}. These two different notions of thresholds were discussed in \cite{Jalali2016}. In the context of phase retrieval, the reader is referred to the recent paper \cite{bakhshizadeh2017}, which showed that by solving a compression-based optimization problem, the required number of observations for recovery is essentially the information dimension of the signal (see \cite{bakhshizadeh2017} for the precise definition). For instance, if the signal is $k$-sparse and complex-valued, then $2k$ measurements suffice.

\subsection{Organization of the paper}
The structure of the rest of the paper is as follows: Section \ref{Sec:contributions} mentions the asymptotic framework of the paper, and summarizes our main results on the asymptotic analysis of $\ampa$. Section \ref{sec:real_valued} discusses the real-valued $\ampa$ algorithm and its analysis. Section \ref{sec:proofmainresults} presents the proofs of our main results. 

\section{Asymptotic analysis of $\ampa$}\label{Sec:contributions}

In this section, we present the asymptotic platform under which $\ampa$ is studied, and we derive a set of equations, known as state evolution (SE), that capture the performance of $\ampa$ under the asymptotic analysis.

\subsection{Asymptotic framework and state evolution}\label{ssec:rigorousSE}

Our analysis of $\ampa$ is carried out based on a standard asymptotic framework developed in \cite{Bayati&Montanari11,Bayati&Montanari12}. In this framework, we let $m,n \rightarrow \infty$, while $m/n \rightarrow \delta$.  Within this section, we will write $\bm{x}_*$, $\bm{x}^t$, $\bm{w}$ and $\bm{A}$ as $\bm{x}_*(n)$, $\bm{x}^t(n)$, $\bm{w}(n)$ and $\bm{A}(n)$ to make explicit their dependency on the signal dimension $n$. In this section we focus on the complex-valued AMP. We postpone the discussion of the real-valued AMP until Section \ref{sec:real_valued}. Following \cite{Mousavi2015}, we introduce the following definition of converging sequences.
\begin{definition}\label{Def:converge_seq}
The sequence of instances $\{\bm{x}_{*}(n),\bm{A}(n),\bm{w}(n)\}$ is said to be a converging sequence if the following hold:
\begin{itemize}
\item[--] $\frac{m}{n}\to\delta\in(0,\infty)$, as $n\to\infty$.
\item[--] $\bm{A}(n)$ has i.i.d. Gaussian entries where $A_{ij}\sim\mathcal{CN}(0,1/m)$.
\item[--] The empirical distribution of $\bm{x}_*(n)\in\mathbb{C}^n$ converges weakly to a probability measure $p_X$ with bounded second
moment. Further, $\frac{1}{n}\|\bm{x}_*(n)\|^2\to\kappa^2$ where $\kappa^2\in(0,\infty)$ is the second moment of $p_X$. For convenience and without loss of generality, we assume $\kappa=1$.\footnote{Otherwise, we can introduce the following normalized variables: $\tilde{\bm{y}}= \bm{y}/\kappa$, $\tilde{\bm{x}}=\bm{x}/\kappa$, $\tilde{\bm{w}}=\bm{w}/\kappa$, $\tilde{\bm{x}}^t=\bm{x}^t/\kappa$ and $\tilde{\bm{p}}^t=\bm{p}^t/\kappa$. One can verify that the AMP.A algorithm defined in \eqref{Eqn:AMP_complex} for these normalized variables remains unchanged. Therefore, we can view that our analyses are carried out for these normalized variables; we don't need to actually change the algorithm though.}
\item[--] The empirical distribution of $\bm{w}(n)\in\mathbb{C}^n$ converges weakly to $\mathcal{CN}(0,\sigma^2_w)$. 
\end{itemize}
\end{definition}

Under the asymptotic framework introduced above, the behavior of $\ampa$ can be characterized exactly. Roughly speaking, the estimate produced by $\ampa$ in each iteration is approximately distributed as the (scaled) true signal $+$ additive Gaussian noise; in other words, ${\bm x}^t$ can be modeled as $\alpha_t \bm{x}_* + \sigma_t \bm{h}$, where $\bm{h}$ behaves like an iid standard complex normal noise. We will clarify this claim in Theorem \ref{The:SE_complex} below. The scaling constant $\alpha_t$ and the noise standard deviation $\sigma_t$ evolve according to a known deterministic rule, called the state evolution (SE), defined below. 
\begin{definition}\label{Def:SE_map_complex}
Starting from fixed $(\alpha_0,\sigma^2_0)\in\mathbb{C}\times\mathbb{R}_+\backslash(0,0)$, the sequences $\{\alpha_t\}_{t\ge1}$ and $\{\sigma^2_t\}_{t\ge1}$ are generated via the following
 recursion:
 \BE \label{Eqn:SE_complex}
\begin{split}
\alpha_{t+1} &= \psi_{1}(\alpha_t,\sigma_t^2),\\
\sigma^2_{t+1} &= \psi_{2}(\alpha_t,\sigma_t^2;\delta,\sigma^2_w),
\end{split}
\EE
where $\psi_{1}:\mathbb{C}\times\mathbb{R}_+\mapsto\mathbb{C}$ and $\psi_{2}:\mathbb{C}\times\mathbb{R}_+\mapsto\mathbb{R}_+$ are respectively  given by
\[
\begin{split}
\psi_1(\alpha,\sigma^2)&= 2\cdot\mathbb{E} \left[ \partial_z g (P, Y) \right] =\mathbb{E}\left[ \frac{\bar{Z} P}{|Z|\,|P|} \right],\\
\psi_2(\alpha,\sigma^2;\delta,\sigma^2_w)&= 4\cdot\mathbb{E} \left[| g (P, Y)|^2 \right] =4\cdot\mathbb{E}\left[ \left( |P| - |Z|-W\right)^2\right].
\end{split}
\]
In the above equations, the expectations are over all random variables involved: $Z\sim\mathcal{CN}(0,1/\delta)$, $P=\alpha Z + \sigma B$ where $B\sim\mathcal{CN}(0,1/\delta)$ is independent of $Z$, and $Y=|Z|+W$ where $W\sim\mathcal{CN}(0,\sigma^2_w)$ is independent of both $Z$ and $B$.
Further, the partial Wirtinger derivative $\partial_z g(p, |z|+w)$ is defined as: 
\[
\begin{split}
\partial_z g(p, |z|+w) &\Mydef  \frac{1}{2}\left[ \frac{\partial }{\partial z_R}g(p, |z|+w) -\mr{i}  \frac{\partial }{\partial z_I}g(p, |z|+w)\right],
\end{split}
\]
where $z_R$ and $z_I$ are the real and imaginary parts of $z$ (i.e., $z=z_R+\mr{i}z_I$).
\end{definition}

\begin{remark}
 The functions $\psi_1$ and $\psi_2$ are well defined except when both $\alpha$ and $\sigma^2$ are zero. 
 \end{remark}
 \begin{remark}
Most of the analysis in this paper is concerned with the noiseless case. For brevity, we will often write $\psi_2(\alpha,\sigma;\delta,0)$ (where $\sigma^2_w=0$) as $\psi_2(\alpha,\sigma;\delta)$. Further, when our focus is on $\alpha$ and $\sigma^2$ rather than $\delta$, we will simply write $\psi_2(\alpha,\sigma^2;\delta)$ as $\psi_{2}(\alpha,\sigma^2)$. 
\end{remark}

In Appendix \ref{Sec:SE_derivations_complex}, we simplify the functions $\psi_1(\cdot)$ and $\psi_2(\cdot)$ into the following expressions (with $\theta_{\alpha}$ being the phase of $\alpha$):
\BS \label{Eqn:map_expression_complex}
\begin{align}
\psi_{1}(\alpha,\sigma^2)  &=e^{\mr{i}\theta_{\alpha}}\cdot\int_0^{\frac{\pi}{2}}\frac{|\alpha|\sin^2\theta}{\left(|\alpha|^2\sin^2\theta +\sigma^2 \right)^{\frac{1}{2}}}\mr{d}\theta \label{Eqn:map_expression_complex_a},\\
\psi_{2}(\alpha,\sigma^2;\delta,\sigma^2_w) &= \frac{4}{\delta}\left(|\alpha|^2+\sigma^2+1- \int_0^{\frac{\pi}{2}} \frac{ 2|\alpha|^2\sin^2\theta +\sigma^2}{ \left(| \alpha|^2\sin^2\theta +\sigma^2\right)^{\frac{1}{2}} }\mr{d}\theta\right)+4\sigma^2_w.\label{Eqn:map_expression_complex_b}
\end{align}
\ES 
The above expressions for $\psi_1$ and $\psi_2$ are more convenient for our analysis.

The state evolution framework for generalized AMP (GAMP) algorithms \cite{Rangan11} was first introduced and analyzed in \cite{Rangan11} and later formally proved in \cite{Javanmard2013}. As we will show later in Theorem \ref{The:SE_complex}, SE characterizes the macroscopic behavior of $\ampa$. To apply the results in \cite{Rangan11,Javanmard2013} to AMP.A, however, we need two generalizations. First, we need to extend the results in \cite{Rangan11,Javanmard2013} to complex-valued models. This is straightforward by applying a complex-valued version of the conditioning lemma introduced in \cite{Rangan11,Javanmard2013}. Second, existing results in \cite{Rangan11,Javanmard2013} require the function $g$ to be smooth. Our simulation results in case of complex-valued $\ampa$ show that SE predicts the performance of $\ampa$ despite the fact that $g$ is not smooth. Since our paper is long, we postpone the proof of this claim to another paper. Instead we use the smoothing idea discussed in \cite{Zheng17} to connect the SE equations presented in \eqref{Eqn:SE_complex} with the iterations of $\ampa$ in \eqref{Eqn:AMP_complex}. Let $\epsilon>0$ be a small fixed number. Consider the following smoothed version of $\ampa$:
\BS
\begin{align}
\bm{p}^t &= \bm{Ax_\epsilon}^t -\frac{2}{\delta}g_{ \epsilon}(\bm{p}^{t-1},\bm{y}), \nonumber\\
\bm{x_\epsilon}^{t+1} &= 2\left[-\mr{div}_p(g_{ t,\epsilon})\cdot \bm{x_\epsilon}^t +\bm{A}^{\UH} g_{\epsilon} (\bm{p}^t,\bm{y})\right], \nonumber
\end{align}
\ES
where $g_{\epsilon}(\bm{p}^{t-1},\bm{y})$ refers to a vector produced by applying $g_\epsilon:\mathbb{C}\times\mathbb{R}_+\mapsto\mathbb{C}$ below component-wise:
\[
g_{\epsilon} (p,y)\Mydef y\cdot h_\epsilon(p)-p,
\]
where for $p= p_1+ \mr{i} p_2$, $h_\epsilon (p)$ is defined as
\[
h_\epsilon(p) \Mydef
\frac{p_1+ \mr{i}p_2}{\sqrt{p_1^2+p_2^2+ \epsilon}} .
\]
Note that as $\epsilon \rightarrow 0$, $g_{t, \epsilon} \rightarrow g_t$ and hence we expect the iterations of smoothed-$\ampa$ converge to the iterations of $\ampa$. 

\begin{theorem}[asymptotic characterization]\label{The:SE_complex}
Let $\{\bm{x}_*(n),\bm{A}(n),\bm{w}(n)\}$ be a converging sequence of instances. For each instance, let $\bm{x}^0(n)$ be an initial estimate independent of $\bm{A}(n)$. Assume that the following hold almost surely
\BE\label{Eqn:def_alpha_sigma}
\lim_{n\to\infty} \frac{1}{n}\langle \bm{x}_*, \bm{x}^0 \rangle=\alpha_0\quad\text{and}\quad \lim_{n\to\infty} \frac{1}{n}\|\bm{x}^0 \|^2=\sigma^2_0+|\alpha_0|^2.
\EE
Let $\bm{x}_\epsilon^t(n)$ be the estimate produced by the smoothed $\ampa$ initialized by $\bm{x}^0(n)$ (which is independent of $\bm{A}(n)$) and $\bm{p}^{-1}(n)=\mathbf{0}$. Let $\epsilon_1, \epsilon_2, \ldots$ denote a sequence of smoothing parameters for which $\epsilon_i \rightarrow 0$ as $i \rightarrow \infty$ Then, for any iteration $t\ge1$, the following holds almost surely
\BE \label{Eqn:AWGN_property}
\lim_{j \rightarrow \infty} \lim_{n\to\infty}\frac{1}{n}\sum_{i=1}^n | x_{\epsilon_j, i}^{t}(n)-e^{\mr{i}\theta_t}\,x_{*,i}|^2=\mathbb{E}\left[| X^{t}-e^{\mr{i}\theta_t}X_*|^2\right]=\big|1-|\alpha_t|\big|^2+\sigma^2_t,
\EE
where $\theta_{t}=\angle \alpha_t$, $X^{t}=\alpha_t X_* +\sigma_t H$ and $X_*\sim p_{X}$ is independent of $H\sim\mathcal{CN}(0,1)$. Further, $\{\alpha\}_{t\ge1}$ and $\{\sigma_t^2\}_{t\ge1}$ are determined by \eqref{Eqn:SE_complex} with initialization $\alpha_0$ and $\sigma_0^2$. 
\end{theorem}

The proof of Theorem~\ref{The:SE_complex} is given in Section \ref{Sec:The:SE_complex_proof}.

\subsection{Convergence of the SE for noiseless model}\label{Sec:asym_framework}
We now analyze the dynamical behavior of the SE. Before we proceed, we point out that in phase retrieval, one can only hope to recover the signal up to global phase ambiguity \cite{Eetrapalli2013,Candes2013,CaLiSo15}, for generic signals without any structure. In light of \eqref{Eqn:AWGN_property}, AMP.A is successful if $|\alpha_t|\to 1$ and $\sigma_0^2\to0$ as $t\to\infty$. 

Let us start with the following interesting feature of the state evolution, which can be seen from \eqref{Eqn:map_expression_complex}.

\begin{lemma}\label{Lem:psi_phase_invariance_0}
For any $(\alpha_0,\sigma^2_0)\in\mathbb{C}\times\mathbb{R}_+\backslash(0,0)$, $\psi_1$ and $\psi_2$ satisfy the following properties:
\begin{enumerate}
\item [(i)] $\psi_1(\alpha,\sigma^2)=\psi_1(|\alpha|,\sigma^2)\cdot e^{\mr{i}\theta_{\alpha}}$, with $e^{\mr{i}\theta_{\alpha}}$ being the phase of $\alpha$;
\item[(ii)] $\psi_2(\alpha,\sigma^2)=\psi_2(|\alpha|,\sigma^2)$.
\end{enumerate}
Hence, if $\theta_{t}$ denotes the phase of $\alpha_t$, then $\theta_t = \theta_0$.  
\end{lemma}
In light of this lemma, we can focus on real and nonnegative values of $\alpha_t$. In particular, we assume that $\alpha_0\ge0$ and we are interested in whether and under what conditions can the SE converge to the fixed point $(\alpha,\sigma^2)=(1,0)$. The following two values of $\delta$ will play critical roles in the analysis of SE:
\BE
\begin{split}
\deltaAMP&\Mydef\frac{64}{\pi^2}-4 \approx 2.5, \\
\deltaGlobal &\Mydef 2.
\end{split}\nonumber
\EE

Our next theorem reveals the importance of $\deltaAMP$. The proof of this theorem detailed in Section \ref{ssec:proofTheo:PhaseTransition_complex}.

\begin{theorem}[convergence of SE]\label{Theo:PhaseTransition_complex}
Consider the noiseless model where $\sigma^2_w=0$. If $\delta>\deltaAMP$, then for any $0<|\alpha_0|\le1$ and $\sigma_0^2\le1$, the sequences $\{\alpha_t\}_{t\ge1}$ and $\{\sigma^2_t\}_{t\ge1}$ defined in \eqref{Eqn:SE_complex} converge to
\[
\lim_{t\to\infty} |\alpha_t|=1\quad\text{and}\quad\lim_{t\to\infty} \sigma_t^2=0.
\] 
\end{theorem}

Notice that $\alpha_0\neq0$ is essential for the success of AMP.A. This can be seen from the fact that $\alpha=0$ is always a fixed point of $\psi_1(\alpha,\sigma^2)$ for any $\sigma^2>0$.
From our definition of $\alpha_0$ in Theorem \ref{The:SE_complex}, $\alpha_0=0$ is equivalent to $\frac{1}{n}\langle \bm{x}_*,\bm{x}^0 \rangle =0$. This means that the initial estimate $\bm{x}^0$ cannot be orthogonal to the true signal vector $\bm{x}_*$, otherwise there is no hope to recover the signal no matter how large $\delta$ is.

The following theorem describes the importance of $\deltaGlobal$ and its proof can be found in Section \ref{Sec:proof_global_complex}.


\begin{theorem}[local convergence of SE]\label{Lem:fixed_point}
When $\sigma^2_w=0$, then $(\alpha,\sigma^2)=(1,0)$ is a fixed point of the SE in \eqref{Eqn:map_expression_complex}. Furthermore, if $\delta> \deltaGlobal$, then there exist two constants $\epsilon_1>0$ and $\epsilon_2>0$ such that the SE converges to this fixed point for any $\alpha_0\in(1-\epsilon_1,1)$ and $\sigma^2_0\in(0,\epsilon_2)$. On the other hand if $\delta< \deltaGlobal$, then the SE cannot converge to $(1,0)$ except when initialized there.
\end{theorem}

\begin{figure}
\begin{center}
\includegraphics[width=.32\textwidth]{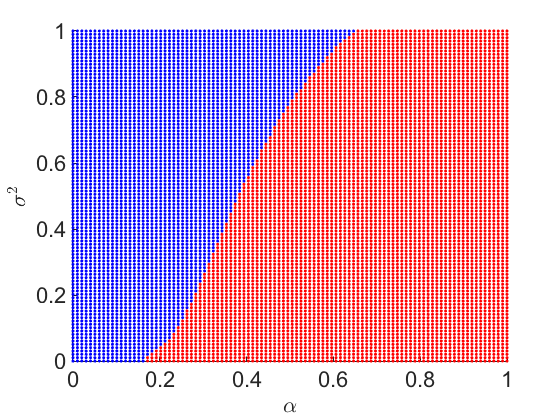}
\includegraphics[width=.32\textwidth]{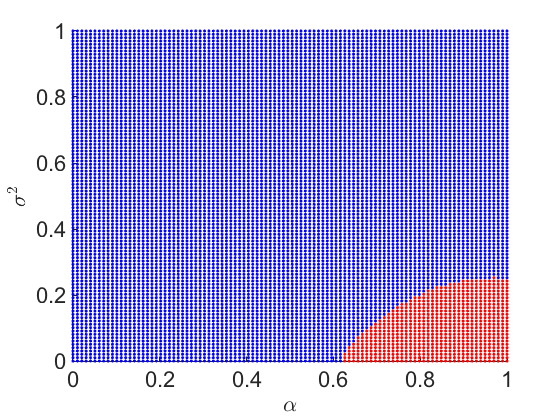}
\includegraphics[width=.32\textwidth]{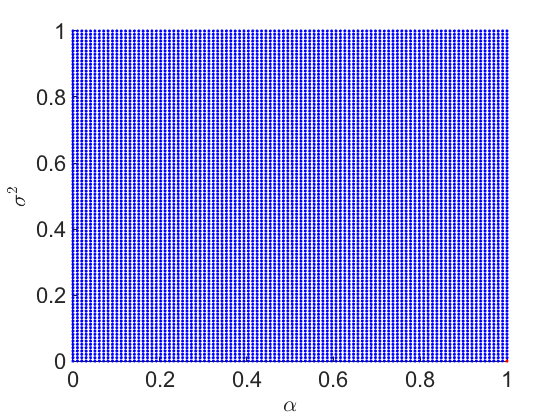}
\caption{The red region exhibits the basin of attraction of $(\alpha, \sigma^2) =(1,0)$. From left to right $\delta= 2.45$, $\delta= 2.3$, $\delta= 2.1$. Note that the basin of attraction of $(1,0)$ in the case of $\delta = 2.1$ is a really small region in the bottom-right corner of the graph. The results are obtained by running the state evolution (SE) of AMP.A (complex-valued version) with $\alpha_0$ and $\sigma^2_0$ chosen from $100\times 100$ values equispaced in $[0,1]\times [0,1]$. W}
\label{fig:basinattract}
\end{center}
\end{figure}

According to Theorem \ref{Lem:fixed_point}, with proper initialization, SE can potentially converge to $(\alpha, \sigma^2)$ even if $\deltaGlobal < \delta < \deltaAMP$. However, there are two points we should emphasize here: (i) we find that when $\delta< \deltaAMP$, standard initialization techniques, such as the spectral method, do not help $\ampa$ converge to $\bm{x}_*$. Hence, the question of finding initialization in the basin of attraction of $(\alpha, \sigma^2)= (1,0)$ (when $\delta< \deltaAMP$) remains open for future research. (In Appendix \ref{sec:spectral}, we briefly discuss how we combine spectral initialization with AMP.A. More details will be reported in our forthcoming paper \cite{Plan}.) (ii) As $\delta$ decreases from $\deltaAMP$ to $\deltaGlobal$ the basin of attraction of $(\alpha, \sigma^2)=  (1,0)$ shrinks. Check the numerical results in Figure \ref{fig:basinattract}.

\subsection{Noise sensitivity}
So far we have only discussed the performance of $\ampa$ in the ideal setting where the noise is not present in the measurements. In general, one can use \eqref{Eqn:SE_complex} to calculate the asymptotic MSE (AMSE) of $\ampa$ as a function of the variance of the noise and $\delta$. However, as our next theorem demonstrates it is possible to obtain an explicit and informative expression for AMSE of $\ampa$ in the high signal-to-noise ratio (SNR) regime.

\begin{theorem}[noise sensitivity]\label{thm:noisesens_comp}
Suppose that $\delta>\deltaAMP=\frac{64}{\pi^2}-4$ and $0<|\alpha_0|\le1$ and $\sigma_0^2<1$. Then, in the high SNR regime the asymptotic MSE defined in \eqref{Eqn:AMSE_def} behaves as
\BE
\lim_{\sigma^2_w\to0}\frac{\mr{AMSE}(\sigma^2_w,\delta)}{\sigma^2_w} = \frac{4}{1 -\frac{2}{\delta}}.\nonumber
\EE
\end{theorem}

The proof of this theorem can be found in Appendix \ref{Sec:proof_noise}.

\section{Extension to real-valued signals}\label{sec:real_valued}
Until now our focus is on complex-valued signals. In this section, our goal is to extend our results to real-valued signals. Since most of the results are similar to the complex-valued case, we will skip the details and only emphasize on the main differences. 

\subsection{$\ampa$ Algorithm}\label{Sec:algorithm_real}
In the real-valued case, $\ampa$ uses the following iterations:
\BS \label{Eqn:AMP_real}
\begin{align}
\bm{x}^{t+1} &=-\mr{div}_p(g_t)\cdot \bm{x}^t +\bm{A}^{\UT} g(\bm{p}^t,\bm{y}), \nonumber\\ 
\bm{p}^t &= \bm{Ax}^t -\frac{1}{\delta}g(\bm{p}^{t-1},\bm{y}), \nonumber\label{Eqn:AMP_real_b}
\end{align}
where $g(p,y):\mathbb{R}\times\mathbb{R}_{+}\mapsto \mathbb{R}$ is given by
\BE \label{Eqn:g_def_real}
g(p,y)\Mydef y\cdot \mr{sign}(p)-p, \nonumber
\EE
\ES
where $\mr{sign}(p)$ denotes the sign of $p$. We emphasize that the divergence term $\mr{div}_p(g_t)$ contains a Dirac delta at $0$ due to the discontinuity of the sign function. This makes the calculation of the divergence in the $\ampa$ algorithm tricky. One can use the smoothing idea we discussed in Section \ref{ssec:rigorousSE}. Alternatively, there are several possible approaches to estimate the divergence term. These practical issues will be discussed in details in our follow-up paper \cite{Plan}.  

\subsection{Asymptotic Analysis}\label{Sec:algorithm_real}
Our analysis is based on the same asymptotic framework detailed in Section \ref{Sec:asym_framework}. The only difference is that the measurement matrix is now real Gaussian with $A_{ij}\sim\mathcal{N}(0,1/m)$ and $w_a\sim\mathcal{N}(0,\sigma^2_w)$. In the real-valued setting, the state evolution (SE) recursion of $\ampa$ in \eqref{Eqn:AMP_real} becomes the following. 
\begin{definition}\label{Def:SE_map_real}
Starting from fixed $(\alpha_0,\sigma_0^2)\in\mathbb{R}\times\mathbb{R}_+\backslash(0,0)$ the sequences $\{\alpha_t\}_{t\ge1}$ and $\{\sigma^2_t\}_{t\ge1}$ are generated via the following iterations: 
\BE \label{Eqn:SE_real}
\begin{split}
\alpha_{t+1} &= \psi_{1}(\alpha_t,\sigma_t^2),\\
\sigma^2_{t+1} &= \psi_{2}(\alpha_t,\sigma_t^2;\delta,\sigma^2_w),
\end{split}
\EE
where, with some abuse of notations, $\psi_{1}:\mathbb{R}\times\mathbb{R}_+\mapsto\mathbb{R}$ and $\psi_{2}:\mathbb{R}\times\mathbb{R}_+\mapsto\mathbb{R}_+$ are now defined as
\BE
\begin{split}
\psi_1(\alpha,\sigma^2)  &= \mathbb{E}[\partial_z g(P,|Y|)]=\mathbb{E}[\mr{sign}(Z\,P)],\\
\psi_2(\alpha,\sigma^2;\delta,\sigma^2_w) &= \mathbb{E}[g^2(P,|Y|)]=\mathbb{E}\left[(|Z|-|P|+W)^2\right].
\end{split}\nonumber
\EE
The expectations are over the following random variables: $Z\sim\mathcal{N}(0,1/\delta)$, $P=\alpha Z + \sigma B$ where $B\sim\mathcal{N}(0,1/\delta)$ is independent of $Z$, and $Y=|Z|+W$ where $W\sim\mathcal{N}(0,\sigma^2_w)$ independent of both $Z$ and $B$.
\end{definition}

In Appendix \ref{Sec:SE_derivations_real}, we derived the following closed-form expressions of $\psi_1$ and $\psi_2$:
\BS \label{Eqn:map_expression_real}
\begin{align}
\psi_{1}(\alpha,\sigma^2)  &= \frac{2}{\pi}\mr{arctan}\left( \frac{\alpha}{\sigma} \right),\label{Eqn:map_expression_real_a}\\
\psi_{2}(\alpha,\sigma^2;\delta,\sigma^2_w) &= \frac{1}{\delta}\left[\alpha^2+\sigma^2+1-\frac{4\sigma}{\pi} -\frac{4\alpha}{\pi}\mr{arctan}\left( \frac{\alpha}{\sigma} \right)\right]+\sigma^2_w. \label{Eqn:map_expression_real_b}
\end{align}
\ES

As in the complex-valued case, we would like to study the dynamics of these two equations. The following lemma simplifies the analysis.

\begin{lemma}\label{lem:psi1_real} 
$\psi_1\left(\alpha,\sigma^2\right)$ and $\psi_2 (\alpha,\sigma^2)$ in \eqref{Eqn:map_expression_real}  and \eqref{Eqn:map_expression_real_b} have the following properties: 
\begin{enumerate}
\item[(i)] $\psi_1 (\alpha, \sigma^2) = \psi_1(|\alpha|, \sigma^2)\cdot {\rm sign} (\alpha)$. 
\item[(ii)] $\psi_2(\alpha, \sigma^2) = \psi_2(|\alpha|, \sigma^2)$. 
\end{enumerate}
\end{lemma}

Again the following two values of $\delta$ play a critical role in the performance of AMP:
\BE
\begin{split}
\delta_{\rm AMP} &= \frac{\pi^2}{4}-1\approx1.47 , \\
\delta_{\rm global} &= 1+\frac{4}{\pi^2}\approx 1.40.
\end{split}\nonumber
\EE

The following two theorems correspond to Theorems  \ref{Theo:PhaseTransition_complex} and \ref{Lem:fixed_point} that explain the dynamics of SE for complex-valued signals. The proofs can be found in Section \ref{Sec:SE_conver_proof} and Section \ref{ssec:prooflemmaglobalminimum} respectively.

\begin{theorem}[convergence of SE]\label{The:PT_real}
Suppose that $\delta>\deltaAMP=\frac{\pi^2}{4}-1$ and $\sigma^2_w=0$. For any $\alpha_0\in\mathbb{R}\backslash0$ and $\sigma^2_0<\infty$, the sequences $\{\alpha_t\}_{t\ge1}$ and $\{\sigma^2_t\}_{t\ge1}$ defined in \eqref{Eqn:SE_real} converge:
\[
\lim_{t\to\infty} | \alpha_t | =1\quad\text{and}\quad\lim_{t\to\infty}  \sigma^2_t =0.
\]
\end{theorem}
Note that in Theorem \ref{The:PT_real} the sequences converge for any $\sigma^2_0<\infty$. This result is stronger than the complex-valued counterpart, which requires $0<|\alpha_0|\le1$ and $\sigma^2_0\le1$ (see Theorem \ref{Theo:PhaseTransition_complex}).

\begin{theorem}[local convergence of SE]\label{lem:global_real}
For the noiseless setting where $\sigma^2_w=0$, $(\alpha,\sigma^2)=(1,0)$ is a fixed point of the SE in \eqref{Eqn:map_expression_complex}. Furthermore, if $\delta> \deltaGlobal$, then there exist two constants $\epsilon_1>0$ and $\epsilon_2>0$ such that the SE converges to this fixed point for any $\alpha_0\in(1-\epsilon_1,1)$ and $\sigma^2_0\in(0,\epsilon_2)$. On the other hand if $\delta< \deltaGlobal$, then the SE cannot converge to $(1,0)$ except when initialized there.
\end{theorem}

Note that $\deltaGlobal$ here is different from the information theoretic limit $\delta=1$. We should emphasize that if we had not used the continuation discussed in \eqref{eq:continuationstra}, then the basin of attraction of $(\alpha, \sigma) = (1,0)$ would be non-empty as long as $\delta>1$.  

Finally, we discuss the performance of $\ampa$ in the high SNR regime. See Section \ref{Sec:proof_noise} for its proof.

\begin{theorem}[noise sensitivity] \label{thm:noisesens_real}
Suppose that $\delta>\deltaAMP=\frac{\pi^2}{4}-1$ and $\alpha_0\in\mathbb{R}\backslash0$ and $\sigma^2_0<\infty$. Then, in the high SNR regime we have
\[
\lim_{\sigma^2_w\to0}\frac{\mr{AMSE}(\sigma^2_w,\delta)}{\sigma^2_w} = \frac{1}{\left(1+\frac{4}{\pi^2}\right)^{-1}-\frac{1}{\delta}}.
\]
\end{theorem}


\section{Proofs of our main results}\label{sec:proofmainresults}


\subsection{Background on Elliptic Integrals}\label{ssec:ellipticintegrals}

The functions that we have in \eqref{Eqn:SE_complex} are related to the first and second kinds of elliptic integrals. Below we review some of the properties of these functions that will be used throughout our paper. Elliptic integrals (elliptic integral of the second kind) were originally proposed for the study of the arc length of ellipsoids. Since their appearance, elliptic integrals have appeared in many problems in physics and chemistry, such as characterization of planetary orbits. Three types of elliptic integrals are of particular importance, since a large class of elliptic integrals can be reduced to these three. We introduce two of them that are of particular interest in our work. 

\begin{definition}
The first and second kinds of complete elliptic integrals, denoted by $K(m)$ and $E(m)$ (for $-\infty<m<1$) respectively, are defined as \cite{Byrd1971}
\BS\label{Eqn:def_elliptic}
\begin{align}
K(m) &=\int_0^{\frac{\pi}{2}}\frac{1}{(1-m\sin^2\theta)^{\frac{1}{2}}}\mr{d}\theta,\\
E(m) &=\int_0^{\frac{\pi}{2}}(1-m\sin^2\theta)^{\frac{1}{2}}\mr{d}\theta.
\end{align}
For convenience, we also introduce the following definition:
\BE
T(m) = E(m)-(1-m)K(m).
\EE
\ES
\end{definition}
In the above definitions, we continued to use $m$, to follow the convention in the literature of elliptic integrals. Previously, $m$ was defined to be the number of measurements, but such abuse of notation should not cause confusion as the exact meaning of $m$ is usually clear from the context. 

Below, we list some properties of elliptic integrals that will be used in this paper. The proofs of these properties can be found in standard references for elliptic integrals and thus omitted (e.g., \cite{Byrd1971}).  
\begin{lemma}\label{Lem:elliptic}
The following hold for $K(m)$ and $E(m)$ defined in \eqref{Eqn:def_elliptic}:
\begin{enumerate}
\item[(i)] $K(0)=E(0)=\frac{\pi}{2}$. Further, for $\epsilon\to0$, $E(1-\epsilon)$ and $K(1-\epsilon)$ behave as
\begin{eqnarray}
E(1- \epsilon) &=& 1+ \frac{\epsilon}{2} \left(\log  \frac{4}{\sqrt{\epsilon}} - 0.5\right)  + O(\epsilon^2\log (1/\epsilon)) \nonumber \\
K(1-\epsilon) &=& \log\left(\frac{4}{\sqrt{\epsilon}}\right) + O(\epsilon\log(1/\epsilon)  ).\nonumber
\end{eqnarray}
\item[(ii)] On $m\in(0,1)$, $K(m)$ is strictly increasing, $E(m)$ is strictly decreasing, and $T(m)$ is strictly increasing. 
\item[(iii)] For $m>-1$,
\BS
\begin{align}
K(-m)&=\frac{1}{\sqrt{1+m}}K\left(\frac{m}{1+m}\right),\nonumber\\
E(-m)&=\sqrt{1+m}E\left(\frac{m}{1+m}\right).\nonumber
\end{align}
\ES
\item[(iv)] The derivatives of $K(m)$, $E(m)$ and $T(m)$ are given by (for $m<1$)
\BE\label{Eqn:elliptic_dif}
\begin{split}
K'(m)&=\frac{E(m)-(1-m)K(m)}{2m(1-m)},\\
E'(m)&=\frac{E(m)-K(m)}{2m},\\
T'(m)&=\frac{1}{2}K(m).
\end{split}
\EE

\end{enumerate}
\end{lemma}

Furthermore, we will use a few more elliptic integrals in our work. Next lemma and its proof connects these elliptic integrals to Type I and Type II elliptic integrals. 

\begin{lemma}\label{Lem:Aux_2}
The following equalities hold for any $m\ge0$:
\BS
\begin{align}
&\int_0^{\frac{\pi}{2}}\frac{\cos^2\theta}{\left(1+m\sin^2\theta\right)^{\frac{3}{2}}}\mr{d}\theta=\int_0^{\frac{\pi}{2}}\frac{\sin^2\theta}{\left( 1+m\sin^2\theta \right)^{\frac{1}{2}}}\mr{d}\theta,\label{Lem:Aux_2_a}\\
&\int_0^{\frac{\pi}{2}}\frac{3m\cos^2\theta }{ (1+m\sin^2\theta)^{\frac{5}{2}} }\mr{d}\theta +\int_0^{\frac{\pi}{2}}\frac{1 }{ (1+m\sin^2\theta)^{\frac{3}{2}} }\mr{d}\theta=\int_0^{\frac{\pi}{2}}\frac{1+2m\sin^2\theta}{\left(1+m\sin^2\theta\right)^{\frac{1}{2}}}\mr{d}\theta. \label{Lem:Aux_2_b}
\end{align}
\ES
\end{lemma}
\begin{proof}
We will only prove \eqref{Lem:Aux_2_b}. \eqref{Lem:Aux_2_a} can be proved in the same way. The idea is to express the integrals using elliptic integrals defined in \eqref{Eqn:def_elliptic}, and then apply known properties of elliptic integrals (Lemma~\ref{Lem:elliptic}) to simplify the results. The same tricks in proving \eqref{Lem:Aux_2_b} are used to derive other related integrals in this paper. Below, we will provide the full details for the proof of \eqref{Lem:Aux_2_b}, and will not repeat such calculations elsewhere. The LHS of \eqref{Lem:Aux_2_b} can be rewritten as:
\BE\label{Lem:Aux_2_c}
\int_0^{\frac{\pi}{2}} \frac{ 3m}{(1+m\sin^{2}\theta)^{\frac{5}{2}} }\mr{d}\theta-\int_0^{\frac{\pi}{2}} \frac{ 3m\sin^2\theta}{(1+m\sin^{2}\theta)^{\frac{5}{2}} }\mr{d}\theta+\int_0^{\frac{\pi}{2}} \frac{ 1}{(1+m\sin^{2}\theta)^{\frac{3}{2}} }\mr{d}\theta=\int_0^{\frac{\pi}{2}}\frac{1+2m\sin^2\theta}{\left(1+m\sin^2\theta\right)^{\frac{1}{2}}}\mr{d}\theta. 
\EE
The equality in \eqref{Lem:Aux_2_c} can be proved by combining the following identities together with straightfroward manipulations:
\BS\label{Eqn:integral_identities}
\begin{align}
\text{(i):}\quad \int_0^{\frac{\pi}{2}} \frac{\sin^2\theta}{(1 + m\sin^2\theta)^{\frac{1}{2}}}\mr{d}\theta  &= \frac{(m+1)E\left(\frac{m}{1+m}\right)-K\left(\frac{m}{1+m}\right)}{m\sqrt{1+m}},\\
\text{(ii):}\quad\int_0^{\frac{\pi}{2}} \frac{\sin^2\theta}{(1 + m\sin^2\theta)^{\frac{3}{2}}}\mr{d}\theta&=\frac{K\left(\frac{m}{1+m}\right)-E\left(\frac{m}{1+m}\right)}{m\sqrt{1+m}},\\
\text{(iii):}\quad\int_0^{\frac{\pi}{2}} \frac{1}{(1 + m\sin^2\theta)^{\frac{3}{2}}}\mr{d}\theta&=\frac{1}{\sqrt{1+m}}E\left(\frac{m}{1+m}\right),\\
 \text{(iv):}\quad\int_0^{\frac{\pi}{2}} \frac{\sin^2\theta}{(1 + m\sin^2\theta)^{\frac{5}{2}}}\mr{d}\theta  &= \frac{-(1-m)E\left(\frac{m}{1+m}\right)+K\left(\frac{m}{1+m}\right)}{3m(1+m)^{\frac{3}{2}}},\\
\text{(v):}\quad\int_0^{\frac{\pi}{2}} \frac{1}{(1 + m\sin^2\theta)^{\frac{5}{2}}}\mr{d}\theta &= \frac{2(m+2)E\left(\frac{m}{1+m}\right)-K\left(\frac{m}{1+m}\right)}{3(1+m)^{\frac{3}{2}}},
\end{align}
\ES
where $K(m)$ and $E(m)$ denote the complete elliptic integrals of the first and second kinds (see~\eqref{Eqn:def_elliptic}). First, consider the identity (i) in \eqref{Eqn:integral_identities}:
\BE
\begin{split}
\int_0^{\frac{\pi}{2}} \frac{\sin^2\theta}{(1 + m\sin^2\theta)^{\frac{1}{2}}}\mr{d}\theta &= \frac{1}{m}\int_0^{\frac{\pi}{2}} (1 + m\sin^2\theta)^{\frac{1}{2}}\mr{d}\theta -\frac{1}{m}\int_0^{\frac{\pi}{2}} \frac{1}{(1 + m\sin^2\theta)^{\frac{1}{2}}}\mr{d}\theta \\
&\overset{(a)}{=}\frac{1}{m}\left[ E(-m)-K(-m) \right]\\
&\overset{(b)}{=}\frac{1}{m}\left[\sqrt{1+m}E\left(\frac{m}{1+m}\right)-\frac{1}{\sqrt{1+m}}K\left(\frac{m}{1+m}\right)\right],
\end{split}\nonumber
\EE
where (a) is from the definition of $K(m)$ and $E(m)$ in \eqref{Eqn:def_elliptic}, and (b) is from Lemma~\ref{Lem:elliptic} (iii).

Identity (ii) can be proved as follows:
\BE\label{Eqn:Ellip_2_proof}
\begin{split}
\int_0^{\frac{\pi}{2}} \frac{\sin^2\theta}{(1 + m\sin^2\theta)^{\frac{3}{2}}}\mr{d}\theta & = -2\frac{\mr{d}}{\mr{d}m}\int_0^{\frac{\pi}{2}} \frac{1}{(1 + m\sin^2\theta)^{\frac{1}{2}}}\mr{d}\theta \\
&=-2\frac{\mr{d}}{\mr{d}m}K(-m)\\
&\overset{(a)}{=}\frac{ (1+m)K(-m) -E(-m)}{m(1+m)}\\
&\overset{(b)}{=}\frac{K\left(\frac{m}{1+m}\right)-E\left(\frac{m}{1+m}\right)}{m\sqrt{1+m}},
\end{split}
\EE
where (a) is due to Lemma~\ref{Lem:elliptic} (iv) and (b) is from Lemma~\ref{Lem:elliptic} (iii).

For identity (iii), we have
\BE\label{Eqn:Ellip_3_proof}
\begin{split}
\int_0^{\frac{\pi}{2}} \frac{1}{(1 + m\sin^2\theta)^{\frac{3}{2}}}\mr{d}\theta &= \int_0^{\frac{\pi}{2}} \frac{1}{(1 + m\sin^2\theta)^{\frac{1}{2}}}\mr{d}\theta - m \cdot \int_0^{\frac{\pi}{2}} \frac{\sin^2\theta}{(1 + m\sin^2\theta)^{\frac{3}{2}}}\mr{d}\theta \\
&\overset{(a)}{=}K(-m) - m\cdot \frac{ (1+m)K(-m) -E(-m)}{m(1+m)}\\
&=\frac{E(-m)}{1+m}\\
&\overset{(b)}{=}\frac{1}{\sqrt{1+m}}E\left(\frac{m}{1+m}\right),
\end{split}
\EE
where step (a) follows from the third step of \eqref{Eqn:Ellip_2_proof}, and step (b) follows from Lemma~\ref{Lem:elliptic} (iii).

Identity (iv) can be proved in a similar way:
\BE
\begin{split}
\int_0^{\frac{\pi}{2}} \frac{\sin^2\theta}{(1 + m\sin^2\theta)^{\frac{5}{2}}}\mr{d}\theta & = -\frac{2}{3}\cdot\frac{\mr{d}}{\mr{d}m}\int_0^{\frac{\pi}{2}} \frac{1}{(1 + m\sin^2\theta)^{\frac{3}{2}}}\mr{d}\theta\\
&\overset{(a)}{=}-\frac{2}{3}\cdot\frac{\mr{d}}{\mr{d}m}\frac{E(-m)}{1+m}\\
&\overset{(b)}{=} \frac{ (1+m)K(-m)-(1-m)E(-m) }{3m(1+m)^2}\\
&\overset{(c)}{=}\frac{-(1-m)E\left(\frac{m}{1+m}\right) -K\left(\frac{m}{1+m}\right) }{3m(1+m)^{\frac{3}{2}}},
\end{split}\nonumber
\EE
where (a) is from the third step of \eqref{Eqn:Ellip_3_proof}, step (b) is from Lemma~\ref{Lem:elliptic} (iv) and (c) is from Lemma~\ref{Lem:elliptic} (iii). 

Lastly, identity (v) can be proved as follows:
\BE
\begin{split}
\int_0^{\frac{\pi}{2}} \frac{1}{(1 + m\sin^2\theta)^{\frac{5}{2}}}\mr{d}\theta &=\int_0^{\frac{\pi}{2}} \frac{1}{(1 + m\sin^2\theta)^{\frac{3}{2}}}\mr{d}\theta - m\cdot \int_0^{\frac{\pi}{2}} \frac{\sin^2\theta}{(1 + m\sin^2\theta)^{\frac{5}{2}}}\mr{d}\theta\\
&\overset{(a)}{=}  \frac{E(-m)}{1+m}- m \cdot \frac{ (1+m)K(-m)-(1-m)E(-m) }{3m(1+m)^2}\\
&\overset{(b)}{=}\frac{2(m+2)E\left(\frac{m}{1+m}\right)-K\left(\frac{m}{1+m}\right)}{3(1+m)^{\frac{3}{2}}},
\end{split}\nonumber
\EE
where step (a) follows from the derivations of the previous two identities and (b) is again due to Lemma~\ref{Lem:elliptic} (iii).
\end{proof}

\subsection{Proof of Theorem~\ref{The:SE_complex}}\label{Sec:The:SE_complex_proof}

Since the proof of the real-valued and complex valued signals look similar, for the sake of notational simplicity we present the proof for the real-valued signals. First note that according to \cite[Lemma 13]{Mondelli2017}\footnote{The proof for a more general result was first presented in \cite{Javanmard2013}. However, we found \cite{Mondelli2017} easier to follow. The reader may also find \cite[Claim 1]{Rangan11} and related discussions useful, although no formal proof was provided.} for the smoothed $\ampa$ algorithm we know that almost surely
\[
\lim_{n\to\infty}\frac{1}{n}\sum_{i=1}^n \left(x_{\epsilon_j, i}^{t+1}(n)-\mr{sign}( \alpha_t )\cdot x_{*,i}\right)^2 = \mathbb{E}(X_{\epsilon_j}^{t+1}-\mr{sign}(\alpha_t)\cdot X_*)^2,
\]
where $X_\epsilon^{t}=\alpha_{\epsilon,t} X_* +\sigma_{\epsilon, t} H$ and $X_*\sim p_{X}$ is independent of $H\sim\mathcal{N}(0,1)$, and $\alpha_{\epsilon,t}$ and $\sigma_{\epsilon, t}$ satisfy the following iterations:
\begin{eqnarray}
\alpha_{\epsilon,t+1} &=&  \mathbb{E} \left[ \partial_z g_\epsilon (P^t, Y) \right],  \nonumber \\
\sigma_{\epsilon, t+1}^2 &=&  \mathbb{E}  [g_\epsilon^2 (P^t, Y)],\nonumber
\end{eqnarray}
where $Y=|Z|+W$, $P^t  =  \alpha_{\epsilon,t} Z + \sigma_{\epsilon,t}B$, where $B \sim \mathcal{N}(0,1/\delta)$ is independent of $Z \sim \mathcal{N}(0,1/\delta)$ and $W\sim\mathcal{N}(0,1/\delta)$. It is also straightforward to use an induction step similar to the one presented in the proof of Theorem 1 of \cite{Zheng17} and show that $(\alpha_{\epsilon,t}, \sigma^2_{\epsilon, t}) \rightarrow (\alpha_t, \sigma^2_{ t})$ as $i \rightarrow \infty$, where $(\alpha_t, \sigma^2_{ t})$ satisfy 
\begin{eqnarray}
\alpha_{t+1} &=& \mathbb{E} \left[ \partial_z g (P^t, Y) \right],  \nonumber \\
\sigma_{ t+1}^2 &=& \mathbb{E} [ g^2 (P^t, Y)].\nonumber
\end{eqnarray}


\subsection{Proof of Theorem \ref{Theo:PhaseTransition_complex}} \label{ssec:proofTheo:PhaseTransition_complex}

The goal of this section is to prove Theorem  \ref{Theo:PhaseTransition_complex}. However, since the proof is very long we start with the proof sketch to help the reader navigate through the complete proof.

\subsubsection{Roadmap of the proof}\label{Sec:roadmap_complex}
Our main goal is to study the dynamics of the iterations:
\BE\label{eq:sedynamics}
\begin{split}
\alpha_{t+1} &= \psi_{1}(\alpha_t,\sigma_t^2),\\
\sigma^2_{t+1} &= \psi_{2}(\alpha_t,\sigma_t^2;\delta),
\end{split}
\EE
Notice that according to the assumptions of the theorem, we assume that we initialized the dynamical system with $\alpha_0>0$.
Our first hope is that this dynamical system will not oscillate and will converge to the solutions of the following system of nonlinear equations: 
\BE\label{eq:FPequation}
\begin{split}
\alpha &= \psi_{1}(\alpha,\sigma^2),\\
\sigma^2 &= \psi_{2}(\alpha,\sigma^2;\delta),
\end{split}
\EE
Hence, the first step is to characterize and understand the fixed points of the solutions of \eqref{eq:FPequation}. Toward this goal we should study the properties of $\psi_{1}(\alpha,\sigma^2)$ and $\psi_{2}(\alpha,\sigma^2;\delta)$. In particular, we would like to know how the fixed points of $\psi_{1}(\alpha,\sigma^2)$ behave for a given $\sigma^2$ and how the fixed points of $\psi_{2}(\alpha,\sigma^2;\delta)$ behave for a given value of $\alpha$ and $\delta$. The graphs of these functions are shown in Figure \ref{fig:psi1_psi2}. 
\begin{figure}[!htbp]
\centering
\subfloat{\includegraphics[width=.47\textwidth]{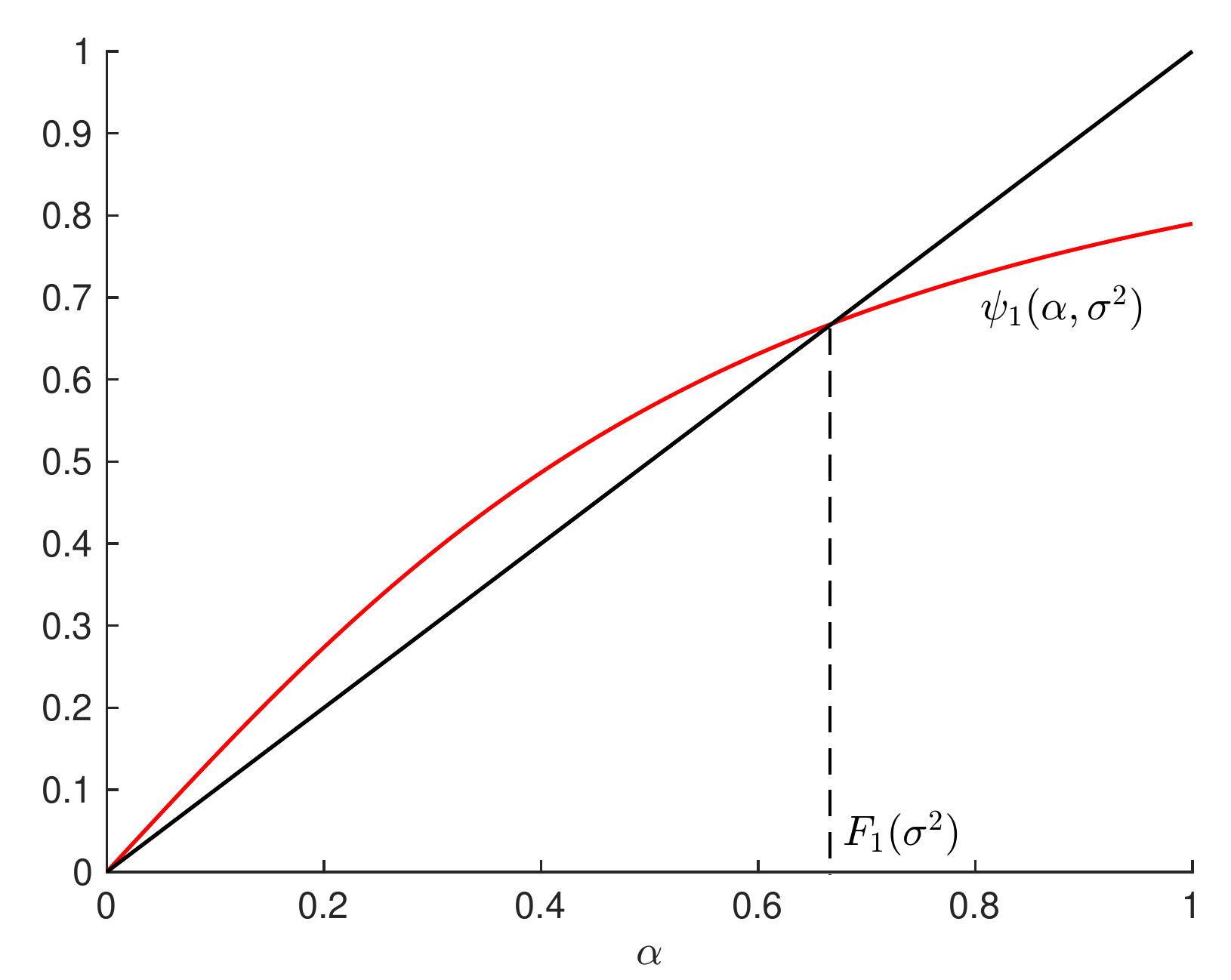}}
\subfloat{\includegraphics[width=.47\textwidth]{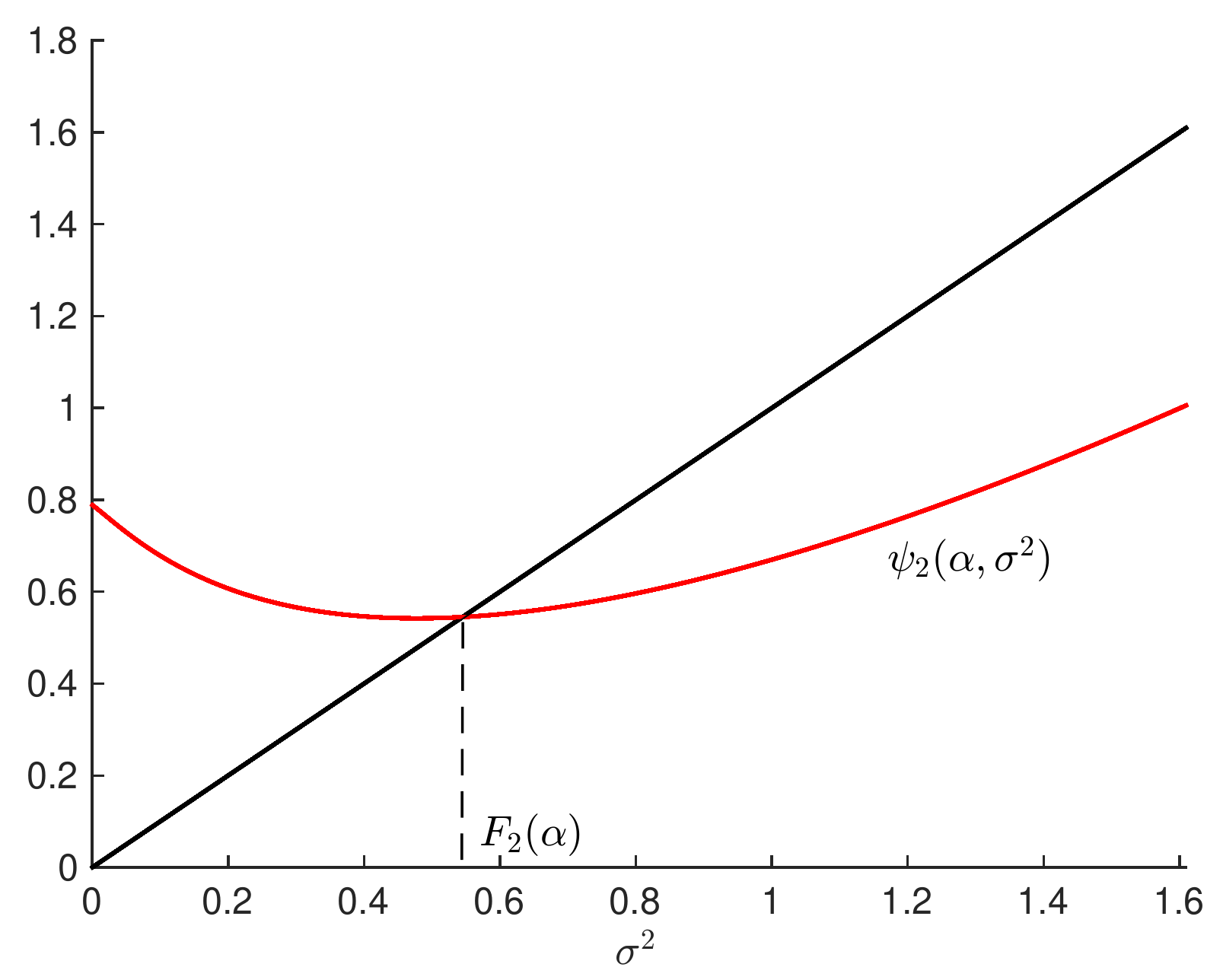}}
\caption{\textbf{Left:} plot of $\psi_1(\alpha,\sigma^2)$ against $\alpha$. $\sigma^2=0.3$. \textbf{Right:} plot of $\psi_2(\alpha,\sigma^2;\delta)$ against $\sigma^2$. $\alpha=0.3$ and $\delta=\deltaAMP$.}
\label{fig:psi1_psi2}
\end{figure}
We list some of the important properties of these two functions. We refer the reader to Section \ref{ssec:psi1_2prop} to see more accurate statement of these claims. 
\begin{enumerate}
\item $\psi_1\left(\alpha,\sigma^2\right)$ is a concave and strictly increasing function of $\alpha>0$, for any $\sigma^2>0$: This implies that $\psi_1\left(\alpha,\sigma^2\right)$ can have two fixed points: one at zero and one at $\alpha>0$. Also, as is clear from the figure, the second fixed point is the stable one. 

\item If $\delta> \deltaAMP$, then $\psi_2$ has always one stable fixed point. It may have one unstable fixed points (as a function of $\sigma^2$). See Fig.~\ref{Fig:psi2_small_delta} for an example of this situation.

\end{enumerate}

For the moment assume that the unstable fixed points do not affect the dynamics of $\ampa$. Let $F_1(\sigma^2)$ denote the non-zero fixed point of $\psi_1$ and $F_2(\sigma^2)$ the stable fixed point of $\psi_2$.\footnote{In the literature of dynamical systems, these functions are sometimes called \textit{nullclines}. Nullclines are useful for qualitatively analyzing local dynamical behavior of two-dimensional maps (which is the case for the SE in this paper).} We will prove in Lemma \ref{Lem:2} that $F_1(\sigma^2)$  is a decreasing function and hence $F_1^{-1}(\alpha)$ is well-defined on $0<\alpha\le 1$. Moreover, we will show that by choosing $F_1^{-1}(0)=\frac{\pi^2}{16}$, $F_1^{-1}(\alpha)$ is continuous on [0,1]. $F_1^{-1}(\alpha)$ and $F_2(\alpha; \delta)$ are shown in Fig. \ref{Fig:nullclines}. Note that the places these curves intersect correspond to the fixed points of \eqref{eq:FPequation}. Depending on the value of $\delta$ the two curves show the following different behaviors:
\begin{enumerate}
\item When $\delta>\deltaAMP$, the dashed curve (see Fig. \ref{Fig:nullclines}) is entirely below the solid curve except at $(\alpha,\sigma^2)=(1,0)$. $\deltaAMP$ is the critical value of $\delta$ at which $F_2(0;\delta) = F_1^{-1}(0)$. Formally, we will prove the following lemma:

\begin{lemma}\label{Lem:F1_F2_complex}
If $\delta \ge\deltaAMP=\frac{64}{\pi^2}-4$, then $F_1^{-1}(\alpha)>F_2(\alpha;\delta)$ holds for any $\alpha\in(0,1)$.
\end{lemma}

You may find the proof of this lemma in Section \ref{proof:lemmadominationF_1}. Intuitively speaking, in this case we expect the state evolution to converge to the fixed point $(\alpha,\sigma^2)=(1,0)$, meaning that AMP.A achieves exact recovery. 

\item When $2<\delta<\deltaAMP$, the two curves intersect at multiple locations, but $F_2(\alpha) <F_1^{-1}(\alpha)$ for the values of $\alpha$ that are close to one. This implies that AMP.A can still exactly recover $\bm{x}_*$ if the initialization is close enough to $\bm{x}_*$. However, this does not happen with spectral initialization. We will discuss this case in Theorem \ref{Lem:fixed_point} and we do not pursue it further here.

\end{enumerate}

\begin{figure}[!htbp]
\centering
\includegraphics[width=.55\textwidth]{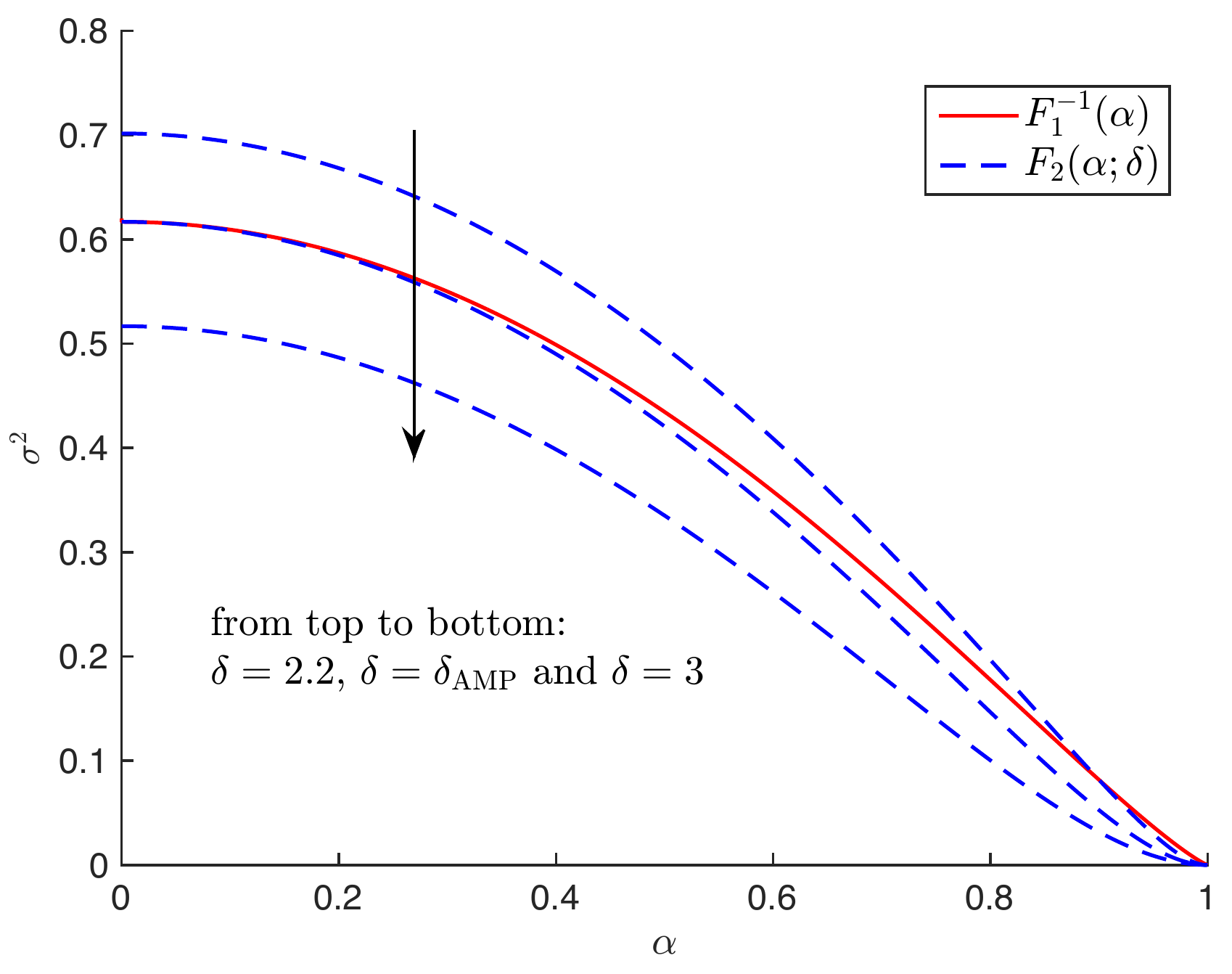}
\caption{Plots of $F_1^{-1}(\alpha)$ and $F_2(\alpha)$ for different values of $\delta$. When $\delta=\deltaAMP$, $F_1^{-1}(\alpha)$ and $F_2(\alpha;\delta)$ intersect at $\alpha=0$.}
\label{Fig:nullclines}
\end{figure}

So far, we have studied the solutions of \eqref{eq:FPequation}. But the ultimate goal of analysis of $\ampa$ is the analysis of \eqref{eq:sedynamics}. In particular, it is important to show that the estimates $(\alpha_t, \sigma_t^2)$ converge to $(1,0)$ and do not oscillate. Unfortunately, the dynamics of $(\alpha_t, \sigma_t^2)$ do not monotonically move toward the fixed point $(1,0)$, which makes the analysis of SE complicated. 

Suppose that $\delta>\deltaAMP$. We first show that $(\alpha_t,\sigma^2_t)$ lies within a bounded region if the initialization falls into that region.
\begin{lemma}\label{Lem:region_bound}
Suppose that $\alpha_0>0$ and $\sigma^2_0\le1$. If $\delta>\deltaAMP=\frac{64}{\pi^2}-4$, then the sequences $\{\alpha_t\}_{t\ge1}$ and $\{\sigma^2_t\}_{t\ge1}$ generated by \eqref{Eqn:SE_complex} satisfy the following:
\[
0\le\alpha_t\le1\quad \text{and}\quad 0\le\sigma^2_t\le\sigma^2_{\max},\quad \forall t\ge1,
\]
where $\sigma^2_{\max}\Mydef\max\left\{ 1,\frac{4}{\delta}\right\}$.
\end{lemma}
\begin{proof}
As discussed in Lemma~\ref{Lem:psi_phase_invariance_0}, the assumption $\alpha_0>0$ implies that $\alpha_t>0$, $\forall t\ge1$. Further, from the property that $0<\psi_1(\alpha,\sigma^2)<1$ for $\alpha>0$ and $\sigma^2>0$ (see Lemma~\ref{lem:psi1} (ii)), we readily have $0\le\alpha_t\le1$. Similarly, Lemma~\ref{lem:psi2} (iii) shows that if $\delta>\deltaAMP$, $ \alpha\in[0,1]$ and $\sigma^2\in[0,\sigma^2_{\max}]$, then $0\le\psi_2(\alpha,\sigma^2;\delta)\le\sigma^2_{\max}$. By our assumption, we have $\sigma^2_0\le1\le\sigma^2_{\max}$, and using induction we prove $0\le\sigma^2_t\le\sigma^2_{\max}$.
\end{proof}

From the above lemma, we see that to understand the dynamics of the SE, we only focus on the region $\mathcal{R}\Mydef\left\{(\alpha,\sigma^2)\big|0<\alpha\le1,0< \sigma^2 \le \sigma^2_{\max}\right\}$. Since the dynamic of $\ampa$ is complicated, we divide this region into smaller regions. See Figure \ref{fig:regions_I&III} for an illustration. 

\begin{figure}[!htbp]
\centering
\includegraphics[width=.57\textwidth]{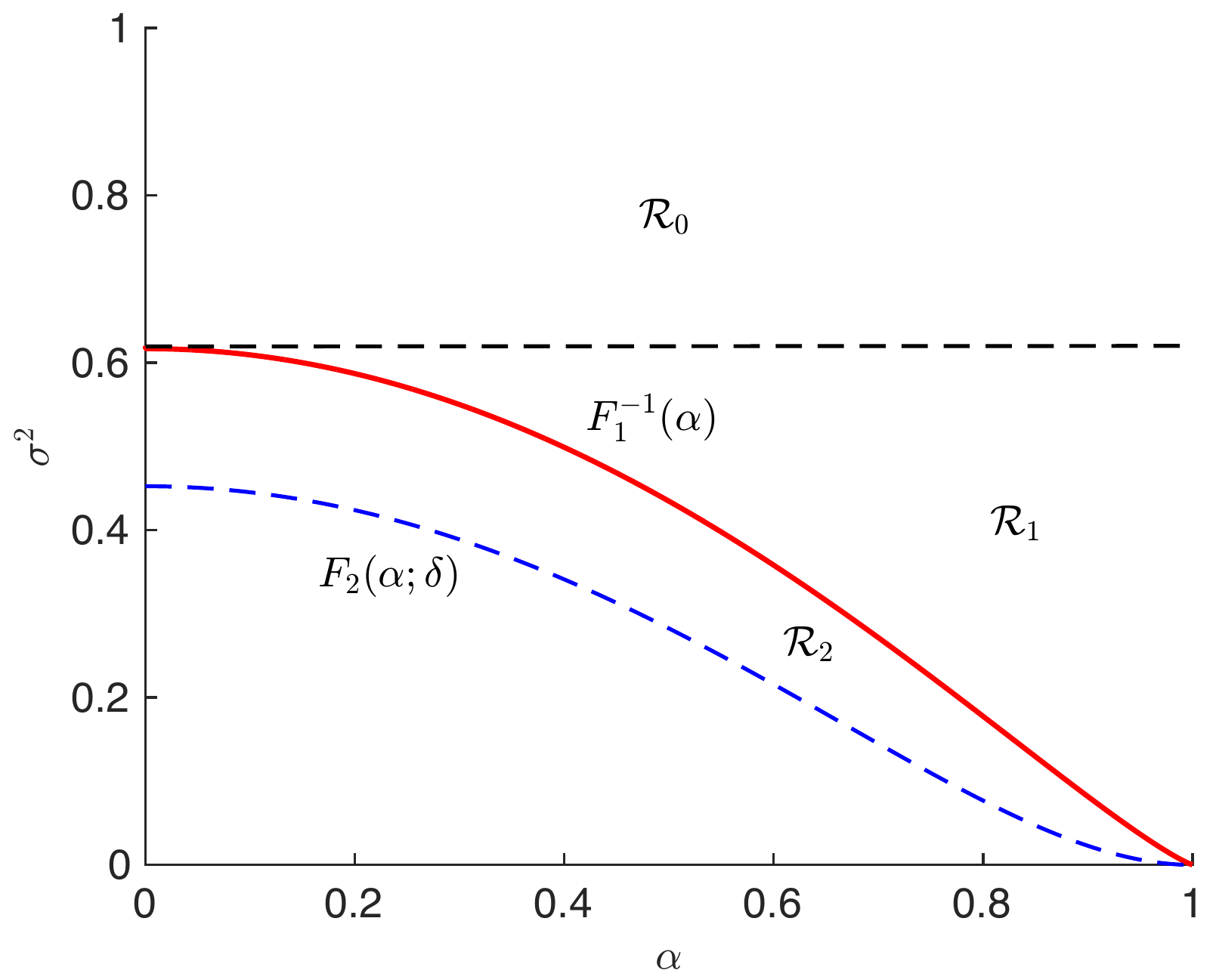}
\caption{Illustration of the three regions in Definition~\ref{Def:regions}. Note that $\mathcal{R}_2$ also includes the region below $F_2(\alpha;\delta)$.}
\label{fig:regions_I&III}
\end{figure}

 \begin{definition}\label{Def:regions}
We divide $\mathcal{R}\Mydef\left\{(\alpha,\sigma^2)\big|0<\alpha\le1,0< \sigma^2 \le \sigma^2_{\max}\right\}$ into the following three sub-regions:
\BE
\begin{split}
 \mathcal{R}_0&\Mydef\left\{(\alpha,\sigma^2)\big|0<\alpha\le1,\frac{\pi^2}{16}< \sigma^2 \le\sigma^2_{\max}\right\}, \\
\mathcal{R}_1&\Mydef\left\{(\alpha,\sigma^2)\big|0<\alpha\le1,F_1^{-1}(\alpha)< \sigma^2 \le\frac{\pi^2}{16}\right\},\\
\mathcal{R}_2 &\Mydef\left\{(\alpha,\sigma^2)\big|0<\alpha\le1,0\le\sigma^2 \le F_1^{-1}(\alpha)\right\}.
\end{split}
\EE
\end{definition}

Our next lemma shows that if $(\alpha_t, \sigma_t^2)$ is in $\mathcal{R}_1$ or $\mathcal{R}_2$ for $t\ge1$, then  $(\alpha_t, \sigma_t^2)$ converges to $(1,0)$. The following lemma demonstrates this claim.

\begin{lemma}\label{Lem:regionI_III}
Suppose that $\delta>\deltaAMP$. If $(\alpha_{t_0},\sigma^2_{t_0})$ is in $\mathcal{R}_1\cup\mathcal{R}_2$ at time $t_0$ (where $t_0\ge1$), and $\{\alpha_t\}_{t\ge t_0}$ and $\{\sigma^2_t\}_{t\ge t_0}$ are obtained via the SE in \eqref{Eqn:SE_complex}, then
\begin{enumerate}
\item[(i)] $(\alpha_t,\sigma^2_t)$ remains in $\mathcal{R}_1\cup\mathcal{R}_2$ for all $t>t_0$;
\item[(ii)] $(\alpha_t,\sigma^2_t)$ converges:
\[
\lim_{t\to\infty} \alpha_t = 1\quad{and}\quad\lim_{t\to\infty} \sigma^2_t = 0.
\]
\end{enumerate}
\end{lemma}

This claim will be proved in Section \ref{Sec:proof_lem_regionI_III}. Notice that the condition $t_0\ge1$ is important for part (i) to hold: if $(\alpha_0,\sigma^2_0)$ is close to the origin (and thus in $\mathcal{R}_2$), then $(\alpha_1,\sigma^2_1)$ can move to $\mathcal{R}_0$. However, this cannot happen when $t\ge1$. In the proof given in Section \ref{Sec:proof_lem_regionI_III}, we showed that for any $(\alpha_0,\sigma^2_0)\in\mathcal{R}$ the possible locations of $(\alpha_1,\sigma^2_1)$ are bounded from below by a curve, and once $(\alpha,\sigma^2)$ is above this curve and also in region $\mathcal{R}_1$ or $\mathcal{R}_2$, then we will prove that it cannot go to $\mathcal{R}_0$. Finally, we will prove the following Lemma that completes the proof. 

\begin{lemma}\label{Lem:regionII_IV}
Suppose that $\delta>\deltaAMP$. Let $\{\alpha_t\}_{t\ge1}$ and $\{\sigma^2_t\}_{t\ge1}$ be the  sequences generated according to \eqref{Eqn:SE_complex} from any $(\alpha_0,\sigma^2_0)\in\mathcal{R}_0$. Then, there exists a finite number $T\ge1$ such that $(\alpha_T,\sigma_T^2)\in\mathcal{R}_1\cup\mathcal{R}_2$.
\end{lemma}
The proof of this result is in Section \ref{ssec:proofLemma8}. Combining the above two lemmas, it is straightforward to see that $(\alpha_t, \sigma_t^2) \rightarrow (1,0)$, and hence the proof is complete. 

Below we present the missing details. 

\subsubsection{Properties of $\psi_1$ and $\psi_2$}\label{ssec:psi1_2prop}

In this section we derive all the main properties of $\psi_1$ and $\psi_2$ that are used throughout the paper. 

\begin{lemma}\label{lem:psi1} 
$\psi_1\left(\alpha,\sigma^2\right)$ has the following properties (for $\alpha\ge0$): 
\begin{enumerate}
\item[(i)] $\psi_1\left(\alpha,\sigma^2\right)$ is a concave and strictly increasing function of $\alpha>0$, for any given $\sigma^2>0$.
\item[(ii)] $0<\psi_1(\alpha,\sigma^2)\le1$, for $\alpha>0$ and $\sigma^2>0$.
\item[(iii)] If $0<\sigma^2 < \pi^2/16$, then there are two nonnegative solutions to $\alpha=\psi_1(\alpha,\sigma^2)$: $\alpha=0$ and $\alpha=F_1(\sigma^2)>0$. Further, $F_1(\sigma^2)$ is strongly globally attracting, meaning that
\BS\label{Eqn:psi1_contracting}
\BE
\alpha<\psi_1(\alpha,\sigma^2)<F_1(\sigma^2),\quad\alpha\in(0,F_1(\sigma^2)),
\EE
and
\BE
F_1(\sigma^2) < \psi_1(\alpha,\sigma^2) <  \alpha,\quad \alpha \in( F_1(\sigma^2),\infty).
\EE
\ES
On the other hand, if $\sigma^2 \ge \pi^2/16$ then $\alpha =0$ is the unique nonnegative fixed point and it is strongly globally attracting.
\end{enumerate}
\end{lemma}
\begin{proof}

\noindent \textit{Part (i):} From \eqref{Eqn:map_expression_complex}, it is easy to verify that $\psi_1(\alpha,\sigma^2)$ is an increasing function of $\alpha>0$. We now prove its concavity. To this end, we calculate its first and second partial derivatives:
\BS\label{Eqn:psi1_dif}
\begin{align}
\frac{\partial\psi_1(\alpha,\sigma^2) }{\partial \alpha}&=\int_0^{\frac{\pi}{2}} \frac{\sin^2\theta\cdot\sigma^2}{(\alpha^2\sin^2\theta+\sigma^2)^{\frac{3}{2}}}\mr{d}\theta, \label{Eqn:psi1_dif_a}\\
\frac{\partial_1^2 \psi_1(\alpha,\sigma^2)}{\partial \alpha^2}&= \int_0^{\frac{\pi}{2}}\frac{-3 \sin^4\theta\cdot\sigma^2\alpha}{(\alpha^2\sin^2\theta+\sigma^2 )^{\frac{5}{2}}}\mr{d}\theta<0,\quad\forall \alpha>0,\sigma^2>0.
\end{align}
\ES
Hence, $\psi_2(\alpha,\sigma^2)$ is a concave function of $\alpha$ for $\alpha>0$. \\

\noindent \textit{Part (ii):}  Positivity of $\psi_1$ is obvious. Also, note that 
\[
\psi_1(\alpha, \sigma^2) = \int_{0}^{\pi/2} \frac{\sin^2 \theta}{( \sin^{2} (\theta) + \frac{\sigma^2}{\alpha^2})^\frac{1}{2}} d\theta \leq \int_{0}^{\pi/2} \sin\theta d \theta = 1. 
 \]

\noindent \textit{Proof of (iii)}: The claim is a consequence of the concavity of $\psi_1$ (with respect to $\alpha$) and the following condition:
\[
\left. \frac{\partial \psi_1(\alpha,\sigma^2)}{\partial \alpha} \right|_{\alpha=0}=1\Longleftrightarrow \sigma^2=\frac{\pi^2}{16}.
\]

The detailed proof is as follows. First, it is straightforward to verify that $\alpha=0$ is always a solution to $\alpha=\psi_1(\alpha,\sigma^2)$. 
Define
\BE
\Psi_1(\alpha,\sigma^2)\Mydef  \psi_1(\alpha,\sigma^2) -\alpha.\nonumber
\EE
Since $\Psi_1(\alpha,\sigma^2)$ is a concave function of $\alpha$ (as $\psi_1(\alpha,\sigma^2)$ is concave), $\frac{\partial \Psi_1(\alpha,\sigma^2)}{\partial\alpha}$ is decreasing. Let's first consider $\sigma^2> \pi^2/16$. In this case we know that

\BE\label{Eqn:psi1_local_stability}
\frac{ \partial\Psi_1(\alpha,\sigma^2)}{\partial\alpha} \leq \frac{ \partial\Psi_1(\alpha,\sigma^2)}{\partial\alpha}\Big|_{\alpha=0} =\frac{ \partial\psi_1(\alpha,\sigma^2)}{\partial\alpha}\Big|_{\alpha=0}-1=\frac{\pi}{4\sigma}-1<0,
\EE
where the second equality can be calculated from \eqref{Eqn:psi1_dif_a}. Since $\Psi_1 (\alpha, \sigma^2)$ is a decreasing function of $\alpha$ and is equal to zero at zero, and it does not have any other solution. Now, consider case $\sigma^2< \pi^2/16$. It is straightforward to confirm that 
\[
\frac{ \partial\Psi_1(\alpha,\sigma^2)}{\partial\alpha}\Big|_{\alpha=0} =\frac{ \partial\psi_1(\alpha,\sigma^2)}{\partial\alpha}\Big|_{\alpha=0}-1=\frac{\pi}{4\sigma}-1 >0.
\]
Furthermore,  from \eqref{Eqn:psi1_dif_a}  we have $\frac{\partial \psi_1(\alpha,\sigma^2)}{\partial\alpha}\Big|_{\alpha\to\infty} =0$, and so
\BE
\frac{\partial \Psi_1(\alpha,\sigma^2)}{\partial\alpha}\Big|_{\alpha\to\infty}\to-1.\nonumber
\EE
Hence, $\Psi_1(\alpha, \sigma^2)=0$ has exactly one more solution for $\alpha>0$. Note that since from part (ii) $\psi_1(\alpha, \sigma^2) <1$, the solution of $\alpha = \psi_1(\alpha, \sigma^2)$ also satisfies $\alpha \leq 1$.  

Finally, the strong global attractiveness follows from the fact that $\psi_1$ is a strictly increasing function of $\alpha$.

\end{proof}

\begin{lemma}\label{lem:psi2}
$\psi_2\left(\alpha,\sigma^2;\delta\right)$ has the following properties: 
\begin{enumerate}

\item[(i)] If $\delta<2$, then $\sigma^2=0$ is a locally unstable fixed point to $\sigma^2=\psi_2\left(\alpha,\sigma^2;\delta\right)$, meaning that
\[
\frac{\partial \psi_2(\alpha,\sigma^2;\delta)}{\partial\sigma^2}\Big|_{\alpha=1,\sigma^2=0}>1.
\]
\item [(ii)] For any $\delta>2$, $\sigma^2=\psi_2\left(\alpha,\sigma^2;\delta\right)$ has a unique fixed point in $\sigma^2\in[0,1]$ for any $\alpha\in[0,1]$. Further, the fixed point is (weakly) globally attracting in $\sigma^2\in[0,1]$:
\BS\label{Eqn:psi2_contracting}
\BE
\sigma^2<\psi_2(\alpha,\sigma^2;\delta) ,\quad \sigma^2\in(0,F_2(\alpha)),
\EE
and
\BE
\psi_2(\alpha,\sigma^2) < \sigma^2,\quad \sigma^2\in( F_2(\alpha),1).
\EE
\ES
 \item[(iii)] 
 If $\delta\ge \deltaAMP$, then for any $\alpha\in[0,1]$, we have
 \BE
0\le\psi_2(\alpha,\sigma^2;\delta)\le \sigma^2_{\max},\quad  \sigma^2\in[0,\sigma^2_{\max}],\nonumber
 \EE
where $\sigma^2_{\max}\Mydef \max\{1,4/\delta\}$.
\item[(iv)] If $\delta\ge \deltaAMP$, then for any $\alpha\in[0,1]$, $F_2(\alpha)$ is the unique (weakly) globally attracting fixed point of $\sigma^2=\psi_2(\alpha,\sigma^2;\delta)$ in $\sigma^2\in[0,\sigma^2_{\max}]$. Namely,
\BS\label{Eqn:psi2_contracting2}
\BE
\sigma^2<\psi_2(\alpha,\sigma^2;\delta) ,\quad \sigma^2\in(0,F_2(\alpha)),
\EE
and
\BE
\psi_2(\alpha,\sigma^2) < \sigma^2,\quad \sigma^2 \in(F_2(\alpha), \sigma^2_{\max}).
\EE
\ES
\item[(v)] For any $\delta>0$, $\psi_2(\alpha,\sigma^2;\delta)$ is an increasing function of $\sigma^2>0$ if
\BE\label{Eqn:alpha_ast_def}
\alpha>\alpha_{\ast}\Mydef\frac{1}{2\sqrt{1+s_{\ast}^2 }} E\left( \frac{1}{1 +s_{\ast}^2} \right)\approx0.53,
\EE
where $s_{\ast}^2$ is the unique solution to
\BE
2E\left(\frac{1}{1+s_{\ast}^2}\right)=K\left(\frac{1}{1+s_{\ast}^2}\right).\nonumber
\EE
Here, $K(\cdot)$ and $E(\cdot)$ denote the complete elliptic integrals introduced in \eqref{Eqn:def_elliptic}. Further, when $\alpha>\alpha_{\ast}$ and $\delta>\deltaAMP$, then $F_2(\sigma^2)$ is strongly globally attracting in $[0,\sigma^2_{\max}]$. Specifically,
\[
\sigma^2<\psi_2(\alpha,\sigma^2;\delta) <F_2(\alpha), \quad\sigma^2\in(0,F_2(\alpha)),
\]
and
\[
F_2(\alpha)<\psi_2(\alpha,\sigma^2) < \sigma^2,\quad \sigma^2 \in(F_2(\alpha), \sigma^2_{\max}).
\]
\end{enumerate}
\end{lemma}

\begin{proof}
First note that the partial derivative of $\psi_2$ w.r.t. $\sigma^2$ is given by
\BE\label{Eqn:psi2_dif_sigma2}
\frac{\partial \psi_2(\alpha,\sigma^2;\delta)}{\partial\sigma^2} = 
\frac{4}{\delta}\left( 1-\frac{1}{2}\int_0^{\frac{\pi}{2}} \frac{\sigma^2}{(\alpha^2\sin^2\theta+\sigma^2)^{\frac{3}{2}}} \mr{d}\theta\right).
\EE


\textit{Part (i):} Before we proceed, we first comment on the discontinuity of  the partial derivative $\frac{\partial \psi_2(\alpha,\sigma^2;\delta)}{\partial\sigma^2}$ at $\sigma^2=0$. Note that the formula in \eqref{Eqn:psi2_dif_sigma2} was derived for non-zero values of $\sigma^2$. Naively, one may plug in $\sigma^2=0$ in the equation and assume that $\frac{\partial \psi_2(\alpha,\sigma^2;\delta)}{\partial\sigma^2}\Big|_{\alpha=1,\sigma^2=0}=\frac{4}{\delta}$. This is not the case since the integral $\int_{0}^{\pi/2} \frac{d\theta}{\sin \theta}$ is divergent. It turns out that the derivative $\frac{\partial \psi_2(\alpha,\sigma^2;\delta)}{\partial\sigma^2}$ is a continuous function of $\sigma^2$. The technical details can be found in Appendix \ref{sec:proofscontinuity}. 

Since $\frac{\partial \psi_2(\alpha,\sigma^2;\delta)}{\partial\sigma^2}$ is continuous at $\sigma^2=0$, we have 
\[
\frac{\partial \psi_2(\alpha,\sigma^2;\delta)}{\partial\sigma^2}\Big|_{\alpha=1,\sigma^2=0}= \lim_{\sigma^2 \rightarrow 0} \frac{\partial \psi_2(1,\sigma^2;\delta)}{\partial\sigma^2}.
\]
Note that if we set $m= 1/\sigma^2$, then from \eqref{Eqn:integral_identities} we have
\begin{eqnarray*}
\frac{\partial \psi_2(1,\sigma^2;\delta)}{\partial\sigma^2} = 
\frac{4}{\delta}\left( 1-\frac{1}{2}\int_0^{\frac{\pi}{2}} \frac{\sigma^2}{(\sin^2\theta+\sigma^2)^{\frac{3}{2}}} \mr{d}\theta\right) = \frac{4}{\delta} \left(1 - \frac{1}{2} \sqrt{\frac{m}{1+m}} E \left( {\frac{m}{m+1}}\right) \right).
\end{eqnarray*}
It is then straightforward to use Lemma \ref{Lem:elliptic} to prove that
\[
\lim_{m \rightarrow \infty}\frac{4}{\delta} \left(1 - \frac{1}{2} \sqrt{\frac{m}{1+m}} E \left( {\frac{m}{m+1}}\right) \right) = \frac{2}{\delta}. 
\]
Hence, $\frac{\partial \psi_2(\alpha,\sigma^2;\delta)}{\partial\sigma^2}\Big|_{\alpha=1,\sigma^2=0}>1$ for $\delta<2$. \\

\noindent \textit{Part (ii):} We first prove that the following equation has at least one solution for any $\alpha\in[0,1]$ and $\delta>2$:
\[
\sigma^2=\psi_2(\alpha,\sigma^2;\delta),\quad\sigma^2\in[0,1].
\]
It is straightforward to verify that 
\BE\label{eq:psi2zero}
\psi_2(\alpha,\sigma^2;\delta)|_{\sigma^2=0}=\frac{4}{\delta}(1-\alpha)^2\ge0.
\EE
We next prove our claim by proving the following:
\BE\label{Eqn:psi2_sigma2_1}
\psi_2(\alpha,\sigma^2;\delta)|_{\sigma^2=1}<1,\quad\forall \alpha\in[0,1]\text{ and }\delta>2.
\EE
From \eqref{Eqn:map_expression_complex_b}, we have
\BE\label{Eqn:lem_psi2_2}
\psi_2(\alpha,\sigma^2;\delta)|_{\sigma^2=1}<1\Longleftrightarrow\underbrace{\int_0^{\frac{\pi}{2}} \frac{2\alpha^2\sin^2\theta+1}{(\alpha^2\sin^2\theta+1)^{\frac{1}{2}}}\mr{d}\theta -\alpha^2}_{g(\alpha^2)}>2-\frac{\delta}{4}.
\EE
We next show that $g(\alpha^2)$ in \eqref{Eqn:lem_psi2_2} is a concave function of $\alpha^2$, and hence the minimum can only happen at either $\alpha=0$ or $\alpha=1$. The first two derivatives w.r.t. $\alpha^2$ are given by:
\[
\frac{\mr{d}g(\alpha^2)}{\mr{d}\alpha^2}
=\int_0^{\frac{\pi}{2}} \frac{\sin^2\theta\left(\alpha^2\sin^2\theta+\frac{3}{2}\right)}{(\alpha^2\sin^2\theta+1)^{\frac{3}{2}}}\mr{d}\theta -1,
\]
and
\[
\frac{\mr{d}^2g(\alpha^2)}{\mr{d}(\alpha^2)^2}=-\int_0^{\frac{\pi}{2}}\frac{\sin^4\theta\left(\frac{1}{2}\alpha^2\sin^2\theta+\frac{5}{4}\right)}{(\alpha^2\sin^2\theta+1)^{\frac{5}{2}}}\mr{d}\theta<0.
\]
The concavity of $g(\alpha^2)$ implies that its minimum happens at either $\alpha=0$ or $\alpha=1$. Hence, to prove \eqref{Eqn:lem_psi2_2}, it suffices to prove that
\[
g(0)=\frac{\pi}{2}>2-\frac{\delta}{4}\quad\text{and}\quad  g(1)\approx 1.509 >2-\frac{\delta}{4},
\]
which holds for $\delta>2$. Hence, \eqref{Eqn:lem_psi2_2} holds. By combining  \eqref{eq:psi2zero} and \eqref{Eqn:psi2_sigma2_1} we conclude that $\psi_2(\alpha, \sigma^2;\delta)$ has at least one fixed point between $\sigma^2= 0$ and $\sigma^2 =1$. The next step is to prove the uniqueness of this fixed point. For the rest of the proof, we discuss two cases separately: a) $\delta>4$ and b) $2<\delta\le 4$.

\begin{itemize}
\item [(a)] $\delta>4$. 
Define
\BE\label{Eqn:Psi2_def}
\Psi_2(\alpha,\sigma^2;\delta)\Mydef \psi_2(\alpha,\sigma^2;\delta)-\sigma^2.
\EE
From \eqref{Eqn:psi2_dif_sigma2}, if $\delta>4$, then $\frac{\partial \psi_2(\alpha,\sigma^2;\delta)}{\partial\sigma^2}<1$, $\forall\sigma^2>0$. This means that $\Psi_2(\alpha,\sigma^2;\delta)$ defined in \eqref{Eqn:Psi2_def} is monotonically decreasing in $\sigma^2>0$. Hence, the solution to $\Psi_2(\alpha,\sigma^2;\delta)=0$ is unique. Furthermore, the following property is a direct consequence of the monotonicity of $\Psi_2(\alpha,\sigma^2;\delta)$:
\BS\label{Eqn:psi2_attacting_2}
\BE
\Psi_2(\alpha,\sigma^2;\delta) < 0,\quad\forall 0<\sigma^2<F_2(\alpha),
\EE
and
\BE
\Psi_2(\alpha,\sigma^2;\delta) > 0 > \sigma^2,\quad\forall F_2(\alpha)<\sigma^2 < 1,
\EE
\ES
where $F_2(\alpha)$ denotes the solution to $\Psi_2(\alpha,\sigma^2;\delta)=0$.

\item [(b)] $2<\delta\le 4$. In this case, we will prove that there exists a threshold on $\sigma^2$, denoted as $\sigma^2_{\star}(\alpha;\delta)$ below, such that the following hold:
\BE\label{Eqn:psi2_dif_sigma2_cases}
\frac{\partial \psi_2(\alpha,\sigma^2;\delta)}{\partial\sigma^2}<1,\quad \forall \sigma^2<\sigma_{\star}^2(\alpha;\delta)\quad\text{and}\quad \frac{\partial \psi_2(\alpha,\sigma^2;\delta)}{\partial\sigma^2}>1,\quad \forall  \sigma^2\in(\sigma_{\star}^2(\alpha;\delta),\infty).
\EE
This means that $\Psi_2(\alpha,\sigma^2;\delta)=\psi_2(\alpha,\sigma^2;\delta)-\sigma^2$ is strictly decreasing on $\sigma^2\in(0,\sigma^2_{\star}(\alpha;\delta))$ and increasing on $\sigma^2\in(\sigma^2_{\star}(\alpha;\delta),\infty)$. Note that since we have proved that $\Psi_2(\alpha,\sigma^2;\delta) =0$ has at least one solution, we conclude that there exist exactly two solutions to $\Psi_2(\alpha,\sigma^2;\delta)=0$, one in $(0,\sigma^2_{\star}(\alpha;\delta))$ and the second in $(\sigma^2_{\star}(\alpha;\delta),\infty)$, if  $\Psi_2(\alpha,\sigma^2;\delta)|_{\sigma^2=\sigma^2_{\star}(\alpha;\delta)}<0$. This is the case since $\Psi_2(\alpha,\sigma^2;\delta)|_{\sigma^2=1}<0$ (see \eqref{Eqn:psi2_sigma2_1}), and that $\Psi_2(\alpha,\sigma^2;\delta)|_{\sigma^2=1}<\Psi_2(\alpha,\sigma^2;\delta)|_{\sigma^2=\sigma^2_{\star}(\alpha;\delta)}$ (since the latter is the global minimum of $\Psi_2(\alpha,\sigma^2;\delta)$ in $\sigma^2\in(0,\infty)$). 

Also, it is easy to prove \eqref{Eqn:psi2_attacting_2}. In fact, the following holds:
\[
\Psi_2(\alpha,\sigma^2;\delta) < 0,\quad\forall 0<\sigma^2<F_2(\alpha),
\]
and
\[
\Psi_2(\alpha,\sigma^2;\delta) > 0 > \sigma^2,\quad\forall F_2(\alpha)<\sigma^2 < \hat{F}_2(\alpha;\delta),
\]
where $\hat{F}_2(\alpha;\delta)>1$ denotes the larger solution to $\Psi_2(\alpha,\sigma^2;\delta)=0$. See Fig.~\ref{Fig:psi2_small_delta} for an illustration.

\begin{figure}[!htbp]
\centering
\includegraphics[width=.57\textwidth]{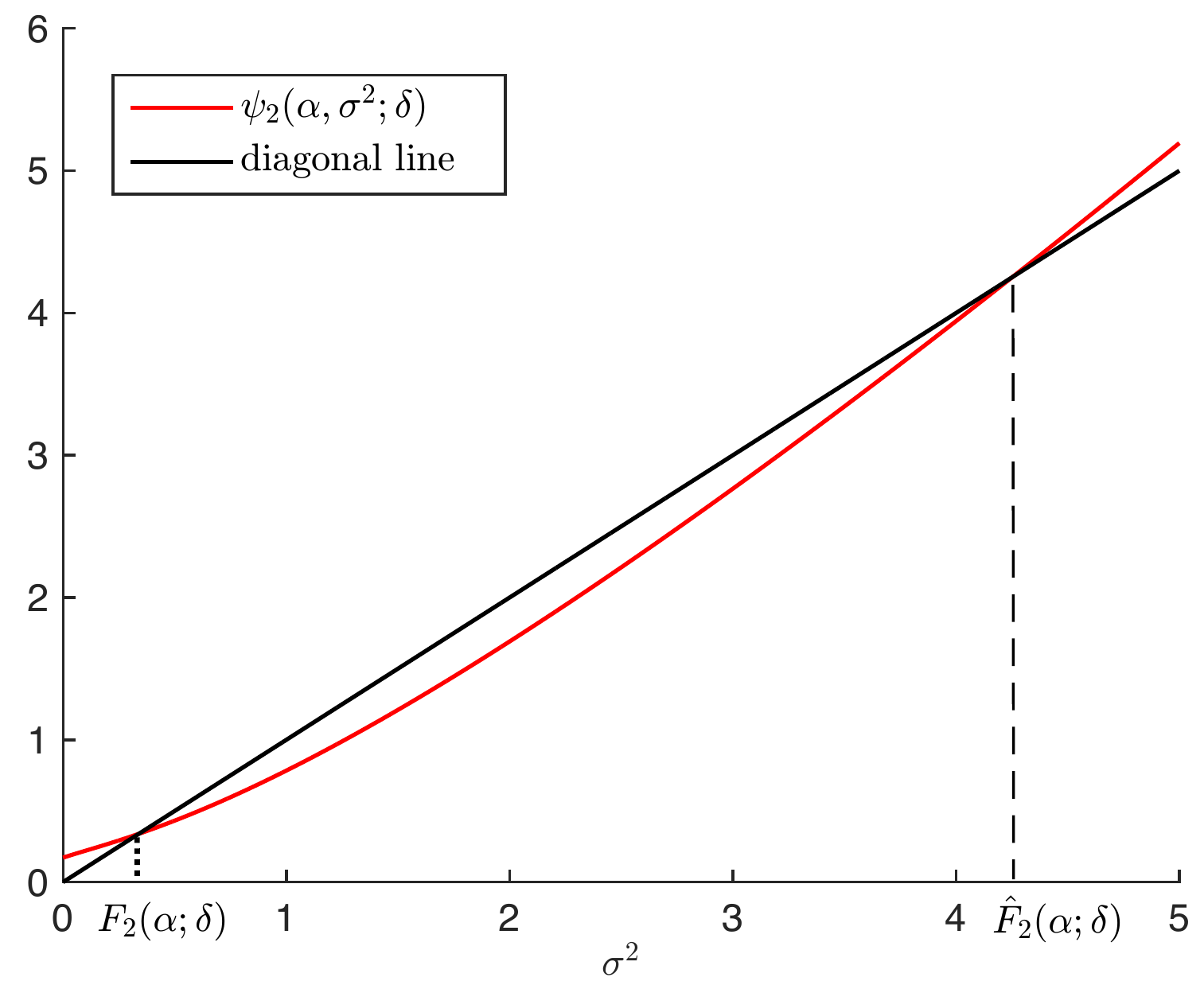}
\caption{Plot of $\psi_2(\alpha,\sigma^2;\delta)$ for $\alpha=0.7$ and $\delta=2.1$.}
\label{Fig:psi2_small_delta}
\end{figure}
From the above discussions, it remains to prove \eqref{Eqn:psi2_dif_sigma2_cases}. To this end, it is more convenient to express \eqref{Eqn:psi2_dif_sigma2} using elliptic integrals discussed in Section \ref{ssec:ellipticintegrals}:
\BS\label{Eqn:psi2_dif_sigma2_fs}
\begin{align}
\frac{\partial \psi_2(\alpha,\sigma^2;\delta)}{\partial\sigma^2} &= 
\frac{4}{\delta}\left( 1-\frac{1}{2}\int_0^{\frac{\pi}{2}} \frac{\sigma^2}{(\alpha^2\sin^2\theta+\sigma^2)^{\frac{3}{2}}} \mr{d}\theta\right)\\
&=\frac{4}{\delta\alpha}\Bigg(\alpha-\underbrace{ \frac{1}{2\sqrt{1+s^2}} E\left( \frac{1}{1 + s^2} \right)}_{f(s)} \Bigg), \label{Eqn:f_def}
\end{align}
\ES
where we introduced a new variable $s\Mydef\frac{\sigma}{\alpha}$ and the last step is derived using the identities in Lemma \ref{Lem:Aux_2}. Based on \eqref{Eqn:psi2_dif_sigma2_fs} we can now rewrite \eqref{Eqn:psi2_dif_sigma2_cases} as
\BE\label{Eqn:psi2_dif_cases_fs}
f(s)>\alpha\left(1-\frac{\delta}{4}\right),\quad\forall s<\frac{\sigma_{\star}(\alpha;\delta)}{\alpha}\quad\text{and}\quad f(s)<\alpha\left(1-\frac{\delta}{4}\right),\quad\forall s\in\left(\frac{\sigma_{\star}(\alpha;\delta)}{\alpha},\infty\right).
\EE
To prove this, we first show that there exists $s^*$ such that $f(s)$ is strictly increasing on $(0,s_{\ast})$ and decreasing on $(s_{\ast},\infty)$, namely,
\BS\label{Eqn:lem_psi2_fs_cases}
\BE
f'(s)>0,\  \text{for }s<s_{\ast},\quad\text{and}\quad f'(s)<0,\  \text{for }s>s_{\ast}.
\EE
$s_{\ast}$ is in fact the unique solution to the following equation:
\BE\label{Eqn:lem_psi2_1}
2E\left(\frac{1}{1+s_{\ast}^2}\right)=K\left(\frac{1}{1+s_{\ast}^2}\right).
\EE
\ES
This can be seen from $f'(s)$ derived below:
\BE
\begin{split}
f'(s)&=\frac{\mr{d}}{\mr{d}s}\frac{1}{2\sqrt{1+s^2}} E\left( \frac{1}{1 + s^2} \right)\\
&=\frac{s}{2(1+s^2)^{\frac{3}{2}}}\left[K\left(\frac{1}{1+s^2}\right)-2E\left(\frac{1}{1+s^2}\right)\right].
\end{split}\nonumber
\EE
Further noting that $E(\cdot)$ is strictly decreasing in $(0,1)$ while $K(\cdot)$ is increasing, we proved \eqref{Eqn:lem_psi2_fs_cases}. 
\begin{figure}[!htbp]
\centering
\includegraphics[width=.55\textwidth]{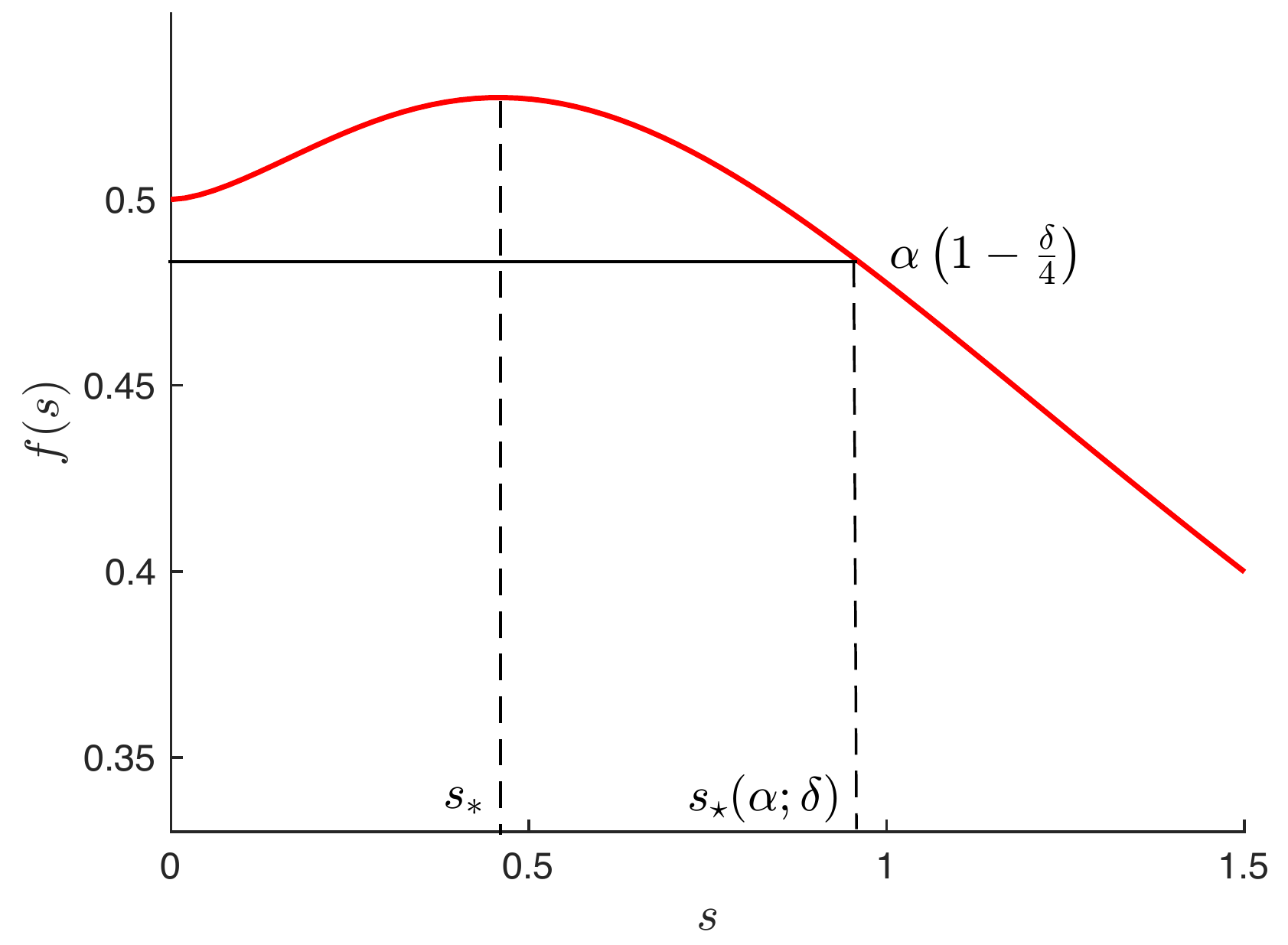}
\caption{Illustration of $f(s)$. }
\label{Fig:fs}
\end{figure}

Based on the above discussions, we can finally turn to the proof of \eqref{Eqn:psi2_dif_cases_fs}. From \eqref{Eqn:f_def}, it is straightforward to verify that $f(0)=\frac{1}{2}$. Therefore, when $\delta>2$, we have
\[
\alpha\left(1-\frac{\delta}{4}\right) \le 1-\frac{\delta}{4}<\frac{1}{2}=f(0),\quad \forall \delta>2\text{ and }0\le\alpha\le1.
\]
Hence, the following equation admits a unique solution (denoted as $s_{\star}(\alpha;\delta)$ below):
\[
f(s)=\alpha\left(1-\frac{\delta}{4}\right),\quad \forall \delta>2\text{ and }0\le\alpha\le1.
\]
See Fig.~\ref{Fig:fs} for an illustration.
Also, from our above discussions on the monotonicity of $f(s)$ it is straightforward to show that
\[
f(s)>\alpha\left(1-\frac{\delta}{4}\right),\quad \forall s<s_{\star}(\alpha;\delta)\quad\text{and}\quad f(s)<\alpha\left(1-\frac{\delta}{4}\right),\forall s\in\left(s_{\star}(\alpha;\delta),\infty\right),
\]
which proves \eqref{Eqn:psi2_dif_cases_fs} by setting $\sigma_{\star}(\alpha;\delta)\Mydef \alpha\cdot s_{\star}(\alpha;\delta)$. This proves \eqref{Eqn:psi2_dif_sigma2_cases}, which completes the proof.
\end{itemize}

\textit{Part (iii):}
We will prove a stronger result: $\psi_2\le 4/\delta$. From \eqref{Eqn:map_expression_complex_b}, $\psi_{2}(\alpha,\sigma^2;\delta) \le 4/\delta$ is equivalent to 
\[
\alpha^2+\sigma^2- \int_0^{\frac{\pi}{2}} \frac{2\alpha^2\sin^2\theta+\sigma^2}{(\alpha^2\sin^2\theta+\sigma^2)^{\frac{1}{2}}}\mr{d}\theta\le0,
\]
which can be further reformulated as
\BE\label{Eqn:Lemma1_1}
\alpha^2 \le \int_0^{\frac{\pi}{2}} \frac{2\alpha^2\sin^2\theta}{(\alpha^2\sin^2\theta+\sigma^2)^{\frac{1}{2}}}\mr{d}\theta+\sigma^2\left( \int_0^{\frac{\pi}{2}} \frac{1}{(\alpha^2\sin^2\theta+\sigma^2)^{\frac{1}{2}}}\mr{d}\theta - 1   \right).
\EE
For $0\le\alpha\le1$ and $\sigma^2\le\sigma^2_{\max}$ we have
\BE\label{Eqn:Lemma1_2}
\begin{split}
\int_0^{\frac{\pi}{2}} \frac{1}{(\alpha^2\sin^2\theta+\sigma^2)^{\frac{1}{2}}}\mr{d}\theta &\ge \int_0^{\frac{\pi}{2}} \frac{1}{\left(\sin^2\theta+\sigma_{\max}^2\right)^{\frac{1}{2}}}\mr{d}\theta,\\
&\overset{(a)}{=} \int_0^{\frac{\pi}{2}} \frac{1}{\left(\sin^2\theta+\frac{4}{\deltaAMP}\right)^{\frac{1}{2}}}\mr{d}\theta\\
&\approx1.09>1,
\end{split}
\EE
where step (a) from $\sigma^2_{\max}=\max\left\{1,4/\delta\right\}\ge\max\left\{1,4/\deltaAMP\right\}=4/\deltaAMP\approx1.6$.
Due to \eqref{Eqn:Lemma1_2}, to prove \eqref{Eqn:Lemma1_1}, it suffices to prove
\[
\alpha^2 \le \int_0^{\frac{\pi}{2}} \frac{2\alpha^2\sin^2\theta}{(\alpha^2\sin^2\theta+\sigma^2)^{\frac{1}{2}}}\mr{d}\theta,
\]
or
\[
1\le \int_0^{\frac{\pi}{2}} \frac{2\sin^2\theta}{(\alpha^2\sin^2\theta+\sigma^2)^{\frac{1}{2}}}\mr{d}\theta,
\]
which, similar to \eqref{Eqn:Lemma1_2}, can be proved by
\[
\int_0^{\frac{\pi}{2}} \frac{2\sin^2\theta}{(\alpha^2\sin^2\theta+\sigma^2)^{\frac{1}{2}}}\mr{d}\theta \ge \int_0^{\frac{\pi}{2}} \frac{2\sin^2\theta}{\left(\sin^2\theta+\frac{4}{\deltaAMP}\right)^{\frac{1}{2}}}\mr{d}\theta\approx1.02>1.
\]

\textit{Part (iv):} 
We bound the partial derivative of $\psi_{2}(\alpha,\sigma^2;\delta)$ for $\sigma^2\in[0,\sigma^2_{\max}]$ as:
\BE\label{Eqn:psi2_dif}
\begin{split}
\frac{ \psi_2(\alpha,\sigma^2;\delta)}{\partial\sigma^2} &= 
\frac{4}{\delta}\left( 1-\frac{1}{2}\int_0^{\frac{\pi}{2}} \frac{\sigma^2}{(\alpha^2\sin^2\theta+\sigma^2)^{\frac{3}{2}}} \mr{d}\theta\right)\\
& \overset{(a)}{\le} \frac{4}{\delta}\left( 1-\frac{1}{2}\int_0^{\frac{\pi}{2}} \frac{\sigma^2}{(\theta^2+\sigma^2)^{\frac{3}{2}}} \mr{d}\theta\right)\\
& \overset{(b)}{=} \frac{4}{\delta}\left( 1-\frac{1}{2}\int_0^{\frac{\pi}{2\sigma}} \frac{1}{(\tilde{\theta}^2+1)^{\frac{3}{2}}} \mr{d}\tilde{\theta}\right)\\
&\overset{(c)}{\le} \frac{4}{\deltaAMP}\left( 1-\frac{1}{2}\int_0^{\frac{\pi}{2\sqrt{\frac{4}{\deltaAMP}}}} \frac{1}{(\tilde{\theta}^2+1)^{\frac{3}{2}}} \mr{d}\tilde{\theta}\right)\\
&\approx0.98 <1,
\end{split}
\EE
where step $(a)$ follows from the constraint $0\le \alpha\le 1$ and the inequality $\sin \theta\le \theta$; $(b)$ is due to the variable change $\tilde{\theta}=\theta/\sigma$; $(c)$ is a consequence of the constraint $\sigma^2\le \sigma^2_{\max}=\max\{1,4/\delta\}\le \max\{1,4/\deltaAMP\}=4/\deltaAMP$. As a result of \eqref{Eqn:psi2_dif}, $\Psi_2(\alpha,\sigma^2;\delta)=\psi_2(\alpha,\sigma^2;\delta)-\sigma^2$ is decreasing in $\sigma^2\in[0,\sigma^2_{\max}]$.
It is easy to verify that $\psi_2(0,\alpha;\delta)\ge 0$ for $\alpha\in [0,1]$. Further, Lemma~\ref{lem:psi2} (iii) implies that 
\[
\psi_2(\sigma_{\max}^2,\alpha;\delta) - \sigma_{\max}^2\le 0.
\]
Hence, there exists a unique solution (which we denote as $F_2(\alpha)$) to the following equation:
\[
\psi_2(\sigma,\alpha;\delta) = \sigma^2,\quad0\le \sigma^2\le \sigma_{\max}^2.
\]
Finally, the property in \eqref{Eqn:psi2_contracting2} is a direct consequence of the fact that $\Psi_2(\alpha,\sigma^2;\delta) =\psi_2(\alpha,\sigma^2;\delta)-\sigma^2$ is a decreasing function of $\sigma^2\le \sigma^2_{\max}$.\\

\textit{Part (v):} In \eqref{Eqn:psi2_dif_sigma2_fs}, we have derived the following:
\BE
\frac{  \psi_2(\alpha,\sigma^2;\delta)}{\partial\sigma^2}=\frac{4}{\delta\alpha}\left(\alpha-f(s)\right),\nonumber
\EE
where $s\Mydef\frac{\sigma}{\alpha}$.
From \eqref{Eqn:f_def}, we see that $\psi_2(\alpha,\sigma^2;\delta)$ is an increasing function of $\sigma^2$ if the following holds:
\BE
\alpha> f(s).\nonumber
\EE
Further, \eqref{Eqn:lem_psi2_fs_cases} implies that the maximum of $f(s)$ happens at $s_*$, i.e.,
\BE\label{Eqn:monotonicity_condi2}
\max_{s>0}\ f(s)=\frac{1}{2\sqrt{1+s_{\ast}^2 }} E\left( \frac{1}{1 +s_{\ast}^2} \right)\Mydef\alpha_{\ast},
\EE
where $s_{\ast}^2$ is the unique solution to
\BE\label{Eqn:lem_psi2_1}
2E\left(\frac{1}{1+s_{\ast}^2}\right)=K\left(\frac{1}{1+s_{\ast}^2}\right).
\EE
Clearly, $\alpha>\alpha_{\ast}$ immediately implies $\alpha >f(s)$, which further guarantees that $\psi_2(\alpha,\sigma^2;\delta)$ is monotonically increasing on $\sigma^2>0$. Finally, the strong global attractiveness of $F_2(\alpha)$ is a direct consequence of part (iv) of this lemma together with the monotonicity of $\psi_2$.
\end{proof}


\subsubsection{Properties of $F_1$ and $F_2$}
In this section we derive the main properties of the functions $F_1$ and $F_2$ introduced in Section \ref{Sec:roadmap_complex}. These properties play major roles in the results of the paper. 

\begin{lemma}\label{Lem:2}
The following hold for $F_1(\sigma^2)$ and $F_2(\alpha;\delta)$ (for $\delta>2$):
\begin{enumerate}
\item[(i)] $F_1(0)=1$ and $\lim_{\sigma^2\rightarrow \frac{\pi^2}{16}^{-}}F_1(\sigma^2)=0$. Further, by choosing $F_1(\frac{\pi^2}{16})=0$, we have $F_{1}(\sigma^2)$ is continuous on $\left[0, \frac{\pi^2}{16}\right]$ and strictly decreasing in $\left(0,\frac{\pi^2}{16}\right)$;
\item[(ii)] $F_2$ is a continuous function of $\alpha \in [0,1]$ and $\delta\in (2, \infty)$. $F_2(1;\delta)=0$, and $F_2(0;\delta)=\left(\frac{-\pi+\sqrt{\pi^2+4(\delta-4)}}{\delta-4}\right)^2$ for $\delta\neq4$ and $F_2(0;4)=4/\pi^2$.
\end{enumerate}
\end{lemma}
\begin{proof}

\textit{Part (i):} We first verify $F_1(0)=1$ and $\lim_{\sigma^2\rightarrow \frac{\pi^2}{16}^{-}}F_1(\sigma^2)=0$. First, $F_1(0)=1$ can be seen from the following facts: (a) $\psi_1(\alpha,0)=1$ for $\alpha>0$, see \eqref{Eqn:map_expression_complex_a}; and (b) By definition, $F_1(0)$ is the non-zero solution to $\alpha=\psi_1(\alpha,0)$. Then, by Lemma~\ref{lem:psi1} (iii) and continunity of $\psi_1$, we know $F_1$ is continuous on $[0,\frac{\pi^2}{16})$, and further $\lim_{\sigma^2\rightarrow \frac{\pi^2}{16}^{-}}F_1(\sigma^2)=0$ since $\sigma^2=\frac{\pi^2}{16}$ corresponds to a case where the non-negative solution to $\psi_1(\alpha,\sigma^2)=\alpha$ decreases to zero. Next, we prove the monotonicity of $F_1$. Note that
\[
F_1(\sigma^2)=\psi_1(F_1(\sigma^2),\sigma^2),
\]
Differentiation w.r.t. $\sigma^2$ yields
\[
F_1'(\sigma^2)=\partial_2 \psi_1(F_1(\sigma^2),\sigma^2) +\partial _1\psi_1(F_1(\sigma^2),\sigma^2)\cdot F_1'(\sigma^2),
\]
where $\partial_2 \psi_1(F_1(\sigma^2),\sigma^2)\Mydef\frac{\partial \psi_1(\alpha,\sigma^2)}{\partial\sigma^2}\Big|_{\alpha=F_1(\sigma^2)}$ and $\partial_1 \psi_1(F_1(\sigma^2),\sigma^2)\Mydef\frac{\partial \psi_1(\alpha,\sigma^2)}{\partial\alpha}\Big|_{\alpha=F_1(\sigma^2)}$. Hence,
\BE\label{Eqn:F1_prime}
\left[1-\partial_1 \psi_1(F_1(\sigma^2),\sigma^2)\right]\cdot F_1'(\sigma^2)=\partial_2 \psi_1(F_1(\sigma^2),\sigma^2).
\EE
We have proved in \eqref{Eqn:psi1_local_stability} that $\frac{\partial \psi_1(\alpha,\sigma^2)}{\partial\alpha}\Big|_{\alpha=0}<1$ when $\sigma^2<\frac{\pi^2}{16}$. Together with the concavity of $\psi_1$ w.r.t. $\alpha$ (cf. Lemma~\ref{lem:psi1} (i)), we have
\BE\label{Eqn:F1_prime_den}
\frac{\partial \psi_1(\alpha,\sigma^2)}{\partial\alpha}\Big|_{\alpha=F_1(\sigma^2)}<\frac{\partial \psi_1(\alpha,\sigma^2)}{\partial\alpha}\Big|_{\alpha=0}<1,\quad\forall\sigma^2<\frac{\pi^2}{16}.
\EE
Further, from \eqref{Eqn:map_expression_complex_a}, it is straightforward to see that $\psi_1$ is a strictly decreasing function of $\sigma^2$, and thus 
\BE\label{Eqn:F1_prime_num}
\partial_2 \psi_1(F_1(\sigma^2),\sigma^2)=\frac{\partial \psi_1(\alpha,\sigma^2)}{\partial\alpha}\Big|_{\alpha=F_1(\sigma^2)}<0.
\EE
Substituting \eqref{Eqn:F1_prime_den} and \eqref{Eqn:F1_prime_num} into \eqref{Eqn:F1_prime}, we obtain
\[
F_1'(\sigma^2)<0,\quad \forall \sigma^2<\frac{\pi^2}{16}.
\]


\textit{Proof of (ii):} By Lemma~\ref{lem:psi2} (ii) and continuity of $\psi_2$, it is straightforward to check that $F_2$ is continuous. Moreover, we have proved that $\sigma^2=F_2(\alpha;\delta)$ is the unique solution to the following equation (for $\delta>2$):
\BE\label{Eqn:psi2_fixed}
\sigma^2 = \frac{4}{\delta}\left( \alpha^2+\sigma^2+1-\int_0^{\frac{\pi}{2}} \frac{2\alpha^2\sin^2\theta+\sigma^2}{(\alpha^2\sin^2\theta+\sigma^2)^{\frac{1}{2}}}\mr{d}\theta \right),\quad \sigma^2\in[0,1].
\EE
When $\alpha=0$, \eqref{Eqn:psi2_fixed} reduces
\[
\sigma^2 = \frac{4}{\delta}\left(\sigma^2+1-\frac{\pi}{2}\sigma \right),\quad \sigma^2\in[0,1],
\]
which has two possible solutions (for $\delta\neq4$):
\[
\sigma_1 = \frac{-\pi+\sqrt{\pi^2+4(\delta-4)}}{\delta-4}\quad\text{and}\quad\sigma_2 = \frac{-\pi-\sqrt{\pi^2+4(\delta-4)}}{\delta-4}.
\]
(For the special case $\delta=4$, $\sigma_1=2/\pi$.)
However, $\sigma_2$ is invalid due to our constraint $0<\sigma^2<1$. This can be seen as follows. First, $\sigma_2<0$ for $\delta>4$ and hence invalid. When $2<\delta<4$, we have
\[
\sigma_2=\frac{\pi+\sqrt{\pi^2-4(4-\delta)}}{4-\delta}>\frac{\pi}{4-\delta}>1.
\]
Hence, $F_2(0;\delta)=\sigma_1$. When $\alpha=1$, \eqref{Eqn:psi2_fixed} becomes:
\[
\sigma^2 = \frac{4}{\delta}\left( 2+\sigma^2-\int_0^{\frac{\pi}{2}} \frac{2\sin^2\theta+\sigma^2}{(\sin^2\theta+\sigma^2)^{\frac{1}{2}}}\mr{d}\theta \right),\quad \sigma^2\in[0,1].
\]
It is straightforward to verify that $\sigma^2=0$ is a solution. Also, from Lemma~\ref{lem:psi2} (ii), $\sigma^2=0$ is a also the unique solution. Hence, $F_2(1;\delta)=0$.

\end{proof}

\subsubsection{Proof of Lemma \ref{Lem:F1_F2_complex}}\label{proof:lemmadominationF_1}

In Lemma~\ref{lem:psi2}, we have proved that $F_2(\alpha;\delta)$ is the unique globally attracting fixed point of $\psi_2$ in $\sigma^2\in[0,1]$ (for $\delta>2$), and from \eqref{Eqn:psi2_contracting} we have
\BE\label{Eqn:Lem_F1F2_1}
\sigma^2>F_2(\alpha;\delta)\Longleftrightarrow \psi_2(\alpha,\sigma^2;\delta)<\sigma^2,\quad\sigma^2\in[0,1].
\EE
Here, our objective is to prove that $F_1^{-1}(\alpha)<F_2(\alpha;\delta)$ holds for any $\alpha\in(0,1)$ when $\delta\ge\deltaAMP$. From \eqref{Eqn:Lem_F1F2_1} and noting that $F_1^{-1}(\alpha)\le\pi^2/16<1$ (from Lemma~\ref{Lem:2}), our problem can be reformulated as proving the following inequality (for $\delta>\deltaAMP$):
\BE\label{Eqn:Lem3_proof0}
\psi_2(\alpha,F_1^{-1}(\alpha);\delta)<F_1^{-1}(\alpha),\quad \forall\alpha\in(0,1).
\EE
Since $\psi_2(\alpha,F_1^{-1}(\alpha);\delta)$ is a strictly decreasing function of $\delta$ (see \eqref{Eqn:map_expression_complex_b}), it suffices to prove that \eqref{Eqn:Lem3_proof0} holds for $\delta=\deltaAMP$:
\BE\label{Eqn:Lem3_proof1}
\psi_2(\alpha,F_1^{-1}(\alpha);\deltaAMP)<F_1^{-1}(\alpha),\quad \forall\alpha\in(0,1).
\EE

We now make some variable changes for \eqref{Eqn:Lem3_proof1}. From \eqref{Eqn:map_expression_complex_a}, $\psi_1$ in  can be rewritten as the following for $\alpha>0$:
\[
\psi_1(\alpha,\sigma^2)=\int_0^{\frac{\pi}{2}}\frac{\sin^2\theta}{\left(\sin^2\theta +\frac{\sigma^2}{\alpha^2} \right)^{\frac{1}{2}}}\mr{d}\theta.
\]
By definition, $F_1(\sigma^2)$ is the solution to $\alpha=\psi_1(\alpha,\sigma^2)$, and hence the following holds:
\[
\alpha=\int_0^{\frac{\pi}{2}}\frac{\sin^2\theta}{\left(\sin^2\theta +\frac{F_1^{-1}(\alpha)}{\alpha^2} \right)^{\frac{1}{2}}}\mr{d}\theta.
\]
At this point, it is more convenient to make the following variable change:
\BE\label{Eqn:Lem_F1F2_s_def}
s\Mydef \frac{\sqrt{F_1^{-1}(\alpha)}}{\alpha},
\EE
from which we get
\BE \label{Eqn:Lem_F1F2_phi}
\alpha = \phi_1(s)\Mydef\int_0^{\frac{\pi}{2}}\frac{\sin^2\theta}{\left(\sin^2\theta +s^2 \right)^{\frac{1}{2}}}\mr{d}\theta.
\EE
Notice that $\phi_1:\mathbb{R}_{+}\mapsto[0,1]$ is a monotonically decreasing function, and it defines a one-to-one map between $\alpha$ and $s$. From the above definitions, we have
\BE\label{Eqn:Lem3_proof2}
F_1^{-1}(\alpha) = s^2 \alpha^2=s^2\phi_1^2(s),
\EE
where the first equality is from \eqref{Eqn:Lem_F1F2_s_def} and the second step from \eqref{Eqn:Lem_F1F2_phi}. Using the relationship in \eqref{Eqn:Lem3_proof2}, we can reformulate the inequality in \eqref{Eqn:Lem3_proof1} into the following equivalent form:
\BE\label{Eqn:Lem3_proof0_b}
\psi_2(\phi_1(s),s^2\phi_1^2(s);\deltaAMP)<s^2\phi_1^2(s),\quad\forall s>0.
\EE

Substituting \eqref{Eqn:Lem_F1F2_phi} and \eqref{Eqn:map_expression_complex_b} into \eqref{Eqn:Lem3_proof0_b} and after some manipulations, we can finally write our objective as:
\BE\label{Eqn:Lem3_proof3}
     \int_0^{\frac{\pi}{2}} \frac{\sin^2 \theta}{(\sin^2 \theta +
     s^2)^{\frac{1}{2}}} d \theta \cdot \int_0^{\frac{\pi}{2}} \frac{(1 -
     \gamma^{} s^2) \sin^2 \theta + s^2}{(\sin^2 \theta + s^2)^{\frac{1}{2}}}
     d \theta >  1, \quad \forall s > 0.
\EE
where we defined 
\BE \label{Eqn:gamma_def}
\gamma \Mydef1-\frac{\deltaAMP}{4}=2-\frac{16}{\pi^2}.
\EE
In the next two subsections, we prove \eqref{Eqn:Lem3_proof3} for $s^2 > 0.07$ and $s^2 \le 0.07$ using different techniques.
\begin{itemize}

\item[(i)] Case I: 
We make another variable change:
\[
  t \Mydef  \frac{1}{s^2} .
\]
Using the variable $t$, we can rewrite \eqref{Eqn:Lem3_proof3} into the following:
\BS
\BE \label{Eqn:Lem3_proof4}
  G (t) \Mydef \frac{g_1 (t)}{g_2 (t)} - \frac{1}{g^2_2 (t)} \ge 
\gamma, \quad \forall t \in [0, 14.3).
\EE
where $\gamma$ is defined in \eqref{Eqn:gamma_def}, and
\begin{align}
  g_1 (t) & \Mydef  \int_0^{\frac{\pi}{2}} (1 + t \sin^2 \theta)^{\frac{1}{2}} d
  \theta,\\
  g_2 (t) & \Mydef  \int_0^{\frac{\pi}{2}} \frac{\sin^2 \theta}{(1 + t
  \sin^2 \theta)^{\frac{1}{2}}} d \theta .
\end{align}
\ES
Notice that if we could prove \eqref{Eqn:Lem3_proof4} for $t<14.3$, we would have proved \eqref{Eqn:Lem3_proof3} for $s^2>0.07$, since $14.3>1/0.07\approx14.28$.
For the ease of later discussions, we define
\[
\begin{split}
  g_3 (t) & \Mydef  \int_0^{\frac{\pi}{2}} \frac{\sin^4 \theta}{(1 + t \sin^2
  \theta)^{\frac{3}{2}}} d \theta,\\
  g_4 (t) & \Mydef  \int_0^{\frac{\pi}{2}} \frac{\sin^6 \theta}{(1 + t \sin^2
  \theta)^{\frac{5}{2}}} d \theta .
\end{split}
\]
The following identities related to $\{g_1(t),g_2(t),g_3(t),g_4(t)\}$ will be used in our proof:
\BE\label{Eqn:Lem3_proof5}
\begin{split}
  g_1' (t) & = \frac{1}{2} g_2 (t),\\
  g_2' (t) & =  - \frac{1}{2} g_3 (t),\\
  g'_3 (t) & =  - \frac{3}{2} g_4 (t) .
\end{split}
\EE

We now prove \eqref{Eqn:Lem3_proof4}. First, it is straightforward to verify that equality holds for \eqref{Eqn:Lem3_proof4} at $t=0$, i.e.,
\BE
G (0) =\gamma.
\EE
Hence, to prove that $G(t)\ge\gamma$ for $t\in[0,14.3)$, it is sufficient to prove that $G (t)$ is an increasing function of $t$ on 
$t\in[0,14.3)$. To this end, we calculate the derivative of $G(t)$:
\[
\begin{split}
  G' (t) & =  \frac{g_1' (t) g_2 (t) - g_1 (t) g_2' (t)}{g_2^2 (t)} - \left(
  \frac{- 2 g'_2 (t)}{g_2^3 (t)} \right)\\
  & \overset{(a)}{=}  \frac{\frac{1}{2} g_2^2 (t) + \frac{1}{2} g_1 (t) g_3 (t)}{g_2^2 (t)}
  - \frac{g_3 (t)}{g_2^3 (t)}\\
  & =  1 + \frac{1}{2} \frac{g_1 (t) g_3 (t)}{g_2^2 (t)} - \frac{g_3
  (t)}{g_2^3 (t)}\\
  & =  \frac{1}{2} \frac{g_3 (t)}{g_2^3 (t)} \bigg( \underbrace{\frac{g_2^3
  (t)}{g_3 (t)}}_{G_1 (t)} + \underbrace{g_1 (t) g_2 (t)}_{G_2 (t)} - 2
  \bigg),
\end{split}
\]
where step (a) follows from the identities listed in \eqref{Eqn:Lem3_proof5}. Since $g_3 (t) > 0$, we have
\[
  G' (t) > 0  \Longleftrightarrow G_1 (t) + G_2 (t) - 2 > 0.
\]
It remains to prove that $G_1 (t) + G_2 (t) - 2 > 0$ for $t <
14.3$. Our numerical results suggest that $G_1 (t) + G_2 (t)$ is a
monotonically decreasing function for $t > 0$, and $G_1 (t) + G_2 (t)\to 2$ as $t
\rightarrow \infty .$ However, directly proving the monotonicity of $G_1 (t) +
G_2 (t)$ seems to be quite complicated. We use a different strategy here. We will prove that (at the end of this section)
\begin{itemize}
  \item $G_1 (t)$ is monotonically increasing;
  
  \item $G_2 (t)$ is monotonically decreasing.
\end{itemize}
As a consequence, the following hold true
for any $c_2 > c_1 > 0$:
\[
 G_1 (t) + G_2 (t) - 2 \ge G_1 (c_1) + G_2 (c_2) - 2, \quad \forall t \in
  [c_1, c_2] .
\]
Hence, if we verify that $G_1 (c_1) + G_2 (c_2) - 2 > 0$, we will be
proving the following:
\[
G_1 (t) + G_2 (t) - 2 > 0, \quad \forall t \in [c_1, c_2] .
\]
To this end, we verify that $G_1 (c_1) + G_2 (c_2) - 2 > 0$ hold for a sequence of $c_1$
and $c_2$: $[c_1,c_2]=[0,0.49]$, $[c_1,c_2]=[0.49,1.08]$, $[c_1,c_2]=[1.08,1.78]$, $[c_1,c_2]=[1.78,2.56]$, $[c_1,c_2]=[2.56,3.47]$, $[c_1,c_2]=[3.47,4.47]$, $[c_1,c_2]=[4.47,5.56]$, $[c_1,c_2]=[5.56,6.77]$, $[c_1,c_2]=[6.67,8.08]$, $[c_1,c_2]=[8.08,9.5]$, $[c_1,c_2]=[9.5,11]$, $[c_1,c_2]=[11,12.6]$, $[c_1,c_2]=[12.6,14.3]$.
Combining all the above results proves
\[
   G_1 (t) + G_2 (t) - 2 > 0, \quad \forall t \in [0, 14.3] .
\]

From the above discussions, it only remains to prove the monotonicity of $G_1(t)$ and $G_2(t)$.
Consider $G_1 (t)$ first:
\BE\label{Eqn:Lem3_proof6}
\begin{split}
  G' _1(t) & = \left(\frac{g_2^3
  (t)}{g_3 (t)}\right)'\\
  &=\frac{3 g_2^2 (t) g_2' (t) g_3 (t) - g_2^3 (t) g_3' (t)}{g_3^2
  (t)}\\
  & =  \frac{- \frac{3}{2} g_2^2 (t) g_3^2 (t) + \frac{3}{2} g_2^3 (t) g_4
  (t)}{g_3^2 (t)}\\
  & =  - \frac{3}{2} g_2^2 (t) + \frac{3}{2} \frac{g_2^3 (t) g_4 (t)}{g_3^2
  (t)}\\
  & =  \frac{3}{2} \frac{g_2^2 (t)}{g_3^2 (t)} \cdot [- g_3^2 (t) + g_2 (t)
  g_4 (t)] .
\end{split}
\EE
Applying the Cauchy-Schwarz inequality yields:
\BE\label{Eqn:Lem3_proof7}
\begin{split}
  g_2 (t) g_4 (t) & =  \int_0^{\frac{\pi}{2}} \frac{\sin^2 \theta}{(1 +
  t \sin^2 \theta)^{\frac{1}{2}}} d \theta \cdot \int_0^{\frac{\pi}{2}} \frac{\sin^6
  \theta}{(1 + t \sin^2 \theta)^{\frac{5}{2}}} d \theta\\
  & \ge  \left( \int_0^{\frac{\pi}{2}} \frac{\sin^4 \theta}{(1 + t
  \sin^2 \theta)^{\frac{3}{2}}} d \theta \right)^2\\
  & =  g_3^2 (t) .
\end{split}
\EE
Combining \eqref{Eqn:Lem3_proof6} and \eqref{Eqn:Lem3_proof7}, we proved that $G'_1 (t) \ge 0$, and therefore $G_1 (t)$ is monotonically increasing. For $G_2 (t)$, we have
\[
\begin{split}
  G_2' (t) & = g_1' (t) g_2 (t) + g_1 (t) g_2' (t)\\
  & =  \frac{1}{2} g_2^2 (t) + g_1 (t) \left( - \frac{1}{2} g_3 (t)
  \right)\\
  & =  \frac{1}{2} [g_2^2 (t) - g_1 (t) g_3 (t)] .
\end{split}
\]
Again, using Cauchy-Schwarz we have
\[
\begin{split}
  g_1 (t) g_3 (t) & =  \int_0^{\frac{\pi}{2}} (1 + t \sin^2 \theta)^{\frac{1}{2}} d
  \theta \cdot \int_0^{\frac{\pi}{2}} \frac{\sin^4 \theta}{(1 + t \sin^2
  \theta)^{\frac{3}{2}}} d \theta\\
  & \ge \left( \int_0^{\frac{\pi}{2}} \frac{\sin^2 \theta}{(1 + t
  \sin^2 \theta)^{\frac{1}{2}}} d \theta \right)^2\\
  & = g_2^2 (t) .
\end{split}
\]
Combining the previous two equations leads to $G'_2 (t) \ge 0$, which completes our proof.
\item[(ii)] Case II:
We next prove \eqref{Eqn:Lem3_proof3} for $s^2 \le 0.07$, which is based on a different strategy. Some manipulations
of the RHS of \eqref{Eqn:Lem3_proof3} yields:
\BS\label{Eqn:Lem3_proof7.1}
\BE
     \int_0^{\frac{\pi}{2}} \frac{\sin^2 \theta}{(\sin^2 \theta +
     s^2)^{\frac{1}{2}}} d \theta \cdot \int_0^{\frac{\pi}{2}} \frac{\left(1 -
   \gamma s^2\right) \sin^2 \theta + s^2}{(\sin^2 \theta + s^2)^{\frac{1}{2}}}
     d \theta = \frac{E (x) T (x)}{x} - \frac{\gamma (1 - x) T^2 (x)}{x^2}, 
     \EE
where $E(\cdot)$, $K(\cdot)$ and $T(\cdot)$ are elliptic integrals defined in \eqref{Eqn:def_elliptic}, $\gamma$ is a constant defined in \eqref{Eqn:gamma_def}, and $x$ is a new variable:
\BE\label{Eqn:Lem3_proof7.2}
 x \Mydef \frac{1}{1 + s^2} . 
 \EE
\ES
From our reformulation in \eqref{Eqn:Lem3_proof7.1}, the inequality in \eqref{Eqn:Lem3_proof3} for $s^2<0.07$ becomes
\BE\label{Eqn:Lem3_proof7.3}
\frac{E (x) T (x)}{x} -\gamma\frac{ (1 - x) T^2 (x)}{x^2} > 1, \quad x \in
   [0.93, 1).
\EE
Note that $0.93<1/(1+0.07)$ and thus proving the above inequality for $x\in[0.93, 1)$ is sufficient to prove the original inequality for $s^2\le0.07$ (note that $x\Mydef1/(1+s^2)$, see \eqref{Eqn:Lem3_proof7.2}).

With some further calculations, \eqref{Eqn:Lem3_proof7.3} can be reformulated as
\BE\label{Eqn:Lem3_proof8} 
\frac{x}{T^2 (x)} \frac{E (x) T (x) - x}{(1 - x)} >\gamma, \quad x \in
   [0.93, 1) . 
\EE   
The following inequality is due to \cite[Eqn.~(1)]{Anderson1985}
\[
T (x) < x < 1, \quad \forall x \in (0, 1) . 
\]
Hence,
\[
 \frac{x}{T^2 (x)} \frac{E (x) T (x) - x}{(1 - x)} > \frac{E (x) T (x) -
   x}{1 - x}, \quad \forall x \in (0, 1), 
\]
and to prove \eqref{Eqn:Lem3_proof8} it suffices to prove the following
\BE\label{Eqn:Lem3_proof9} 
 \frac{E (x) T (x) - x}{1 - x} > \gamma, \quad \forall x \in [0.93, 1) . 
 \EE
To this end, we will prove that the LHS of \eqref{Eqn:Lem3_proof9} is a strictly increasing
function of $x \in [0.93, 1)$. If this is true, we would have
\[
 \frac{E (x) T (x) - x}{1 - x} > \frac{E (x) T (x) - x}{1 - x} |_{x = 0.93}
    \approx 0.385 > \gamma=2-\frac{16}{\pi^2} \approx 0.3789, \quad \forall x \in
   [0.93, 1) . 
\]

We next prove the monotonicity of $\frac{E (x) T (x) - x}{1 - x}$. From the identities in Lemma~\ref{Lem:elliptic}, we derive the following
\[
  {}[E (x) T (x) - x]'  =  \frac{E^2 (x) - 2 (1 - x) E (x) K (x) + (1 - x)
  K^2 (x)}{2 x} - 1.
\]
Hence, to prove that $\frac{E (x) T (x) - x}{1 - x}$ is monotonically
increasing, it is sufficient to prove the following inequality:
\BE \label{Eqn:Lem3_proof10}
 \left( \frac{E^2 (x) - 2 (1 - x) E (x) K (x) + (1 - x) K^2 (x)}{2 x} - 1
   \right) (1 - x) - [E (x) T (x) - x] (- 1) > 0. 
 \EE 
Now, substituting $T (x) = E (x) - (1 - x) K (x)$ into \eqref{Eqn:Lem3_proof10} and after some
manipulations, we finally reformulate the inequality to be proved into the
following form:
\[
  T (x)^{ 2}  >  2 x - x E^2 (x).
\]

It can be verified that
equality holds at $x = 1$. We next prove that $T (x)^2 + x E (x)^2 - 2 x$ is
monotonically decreasing on $[0.93, 1)$. We differentiate once more:
\[
  (T (x)^2 + x E (x)^2 - 2 x)'  =  2 E (x)^2 - (1 - x) K (x)^2 - 2.
\]
Our problem boils down to proving $2 E (x)^2 - (1 - x) K (x)^2 - 2 < 0$ for $x
\in [0.93, 1)$. We can verify that $2 E (x)^2 - (1 - x) K (x)^2 - 2 = 0$ holds
at $x = 1$. We finish by showing that $2 E (x)^2 - (1 - x) K (x)^2 - 2$ is
monotonically increasing in $x \in [0.93, 1)$. To this end, we differentiate
again:
\BE\label{Eqn:Lem3_proof11}
\begin{split}
[2 E (x)^2 - (1 - x) K (x)^2 - 2]' & = \frac{K (x)^2 - 3 E (x) K (x) + 2
  E (x)^2}{x}\\
   &= \frac{\left[ K (x) - \frac{3}{2} E (x) \right]^2 - \frac{1}{2} E
  (x)^2}{x}.
\end{split}
\EE
We note that $K (x) - \left( \frac{3}{2} + \frac{1}{\sqrt{2}} \right) E (x)$
is a monotonically increasing function in (0,1) since $K (x)$ is monotonically
increasing and $E (x)$ is monotonically decreasing. We verify that \ $K (x) -
\left( \frac{3}{2} + \frac{1}{\sqrt{2}} \right) E (x)>0$ when $x \ge0.93$.
Hence,
\[
  K (x) - \left( \frac{3}{2} + \frac{1}{\sqrt{2}} \right) E (x)  >  0, \quad
  \forall x \in [0.93, 1),
\]
and therefore
\BE\label{Eqn:Lem3_proof12}
  \left( K (x) - \frac{3}{2} E (x) \right)^2  >  \frac{1}{2} E (x)^2, \quad
  \forall x \in [0.93, 1) .
\EE
Substituting \eqref{Eqn:Lem3_proof12} into \eqref{Eqn:Lem3_proof11}, we prove that $[2 E (x)^2 - (1 - x) K (x)^2 - 2]' > 0$
for $x \in [0.93, 1)$, which completes the proof.

\end{itemize}

\subsubsection{Proof of Lemma~\ref{Lem:regionI_III}} \label{Sec:proof_lem_regionI_III}

First, we introduce a function that will be crucial for our proof. 

\begin{definition}\label{def:Lalpha}
Define
\BE\label{Eqn:left_bound}
L(\alpha;\delta)\Mydef \frac{4}{\delta}\left( 1 - \frac{\phi_2^2(\phi_1^{-1}(\alpha)) }{4\left[1+(\phi_1^{-1}(\alpha))^2\right]}\right),\quad\alpha\in(0,1),
\EE
where $\phi_1:\mathbb{R}_{+}\mapsto[0,1]$ and $\phi_2:\mathbb{R}_{+}\mapsto\mathbb{R}_{+}$ below:
\BS
\begin{align}
\phi_1(s)&\Mydef\int_0^{\frac{\pi}{2}}\frac{\sin^2\theta}{\left(\sin^2\theta +s^2 \right)^{\frac{1}{2}}}\mr{d}\theta, \label{Eqn:phi}\\
\phi_2(s)&\Mydef \int_0^{\frac{\pi}{2}} \frac{ 2\sin^2\theta +s^2 }{ \left( \sin^2\theta +s^2 \right)^{\frac{1}{2}} }\mr{d}\theta,\label{Eqn:phi2}
\end{align}
\ES
where $\phi_1^{-1}$ is the inverse functions of  $\phi_1$. The existence of $\phi_1^{-1}$ follows from its monotonicity, which can be seen from its definition.  
\end{definition}

In the following, we list some preliminary properties of $L(\alpha;\delta)$. The main proof for Lemma~\ref{Lem:regionI_III} comes afterwards.

\begin{itemize}

\item \textbf{Preliminaries:}

The following lemma helps us clarify the importance of $L$ in the analysis of the dynamics of SE: 

\begin{lemma}\label{Lem:left_bound}
For any $\alpha>0$, $\sigma^2>0$ and $\delta>0$, the following holds:
\BE\label{Eqn:left_bound_1}
 L\left[ \psi_1(\alpha,\sigma^2);\delta\right]\le\psi_2(\alpha,\sigma^2;\delta),
\EE
where $\psi_1$ and $\psi_2$ are the SE maps defined in \eqref{Eqn:map_expression_complex}, and $L(\alpha;\delta)$ is defined in \eqref{Eqn:left_bound}.
\end{lemma}
\begin{proof}
Define $\mathcal{X}\Mydef\{(\alpha,\sigma^2)|\alpha>0,\sigma^2>0\}$. Let $\mathcal{Y}$ be the image of $\mathcal{X}$ under the SE map in \eqref{Eqn:map_expression_complex}. We will prove that the following holds for an arbitrary $C\in[0,1]$:
\BE\label{Eqn:left_bound_2}
\begin{split}
L\left( C;\delta\right)=\min_{(\hat{\alpha},\hat{\sigma}^2)\in\mathcal{X}}&\  \psi_2(\hat{\alpha},\hat{\sigma}^2;\delta),
\end{split}
\EE
where $(\hat{\alpha},\hat{\sigma}^2)$ satisfies the constraint
\[
\psi_1(\hat{\alpha},\hat{\sigma}^2) =C.
\]
If \eqref{Eqn:left_bound_2} holds, we would have proved \eqref{Eqn:left_bound_1}. To see this, consider arbitrary $(\alpha,\sigma^2)$ such that $\psi_1(\alpha,\sigma^2) = C$. Then, we have
\[
\begin{split}
L\left[ \psi_1(\alpha,\sigma^2);\delta\right] \overset{(a)}{=} \min_{(\hat{\alpha},\hat{\sigma}^2)}\ \psi_2(\hat{\alpha},\hat{\sigma}^2;\delta) \overset{(b)}{\le} \psi_2({\alpha},{\sigma}^2;\delta),
\end{split}
\]
where step (a) follows from \eqref{Eqn:left_bound_2} and $\psi_1(\alpha,\sigma^2) = C$, and step (b) holds since the choice $\hat{\alpha}=\alpha$ and $\hat{\sigma}^2=\sigma^2$ is feasible for the constraint $\psi_1(\hat{\alpha},\hat{\sigma}^2) = \psi_1({\alpha},{\sigma}^2)$. This is precisely \eqref{Eqn:left_bound_1}.

We now prove \eqref{Eqn:left_bound_2}. From \eqref{Eqn:map_expression_complex_a}  
we have
\[
\psi_1(\alpha, \sigma^2) = \int_0^{\pi/2} \frac{\alpha \sin^2 \theta}{ (\alpha^2 \sin^2 \theta + \sigma^2)^{1/2}}d\theta. 
\]
Furthermore, from the definition of $\phi_1$ in \eqref{Eqn:phi} we have
\BE\label{Eqn:alpha_new_s}
\psi_1(\hat{\alpha},\hat{\sigma}^2) =\phi_1\left(\frac{\hat{\sigma}}{\hat{\alpha}}\right)=C \Longrightarrow s \Mydef\frac{\hat{\sigma}}{\hat{\alpha}}= \phi_1^{-1}(C).
\EE
Similarly, from \eqref{Eqn:map_expression_complex_b}, i.e. the definition of $\psi_2$, and the definition of $\phi_2$ in \eqref{Eqn:phi2}, we can express $\psi_2(\hat{\alpha},\hat{\sigma}^2;\delta)$ as
\[
\begin{split}
\psi_2(\hat{\alpha},\hat{\sigma}^2;\delta) &= \frac{4}{\delta}\left[ \hat{\alpha}^2+\hat{\sigma}^2 + 1 -  \hat{\alpha}\cdot \phi_2\left(\frac{\hat{\sigma}}{\hat{\alpha}}\right)\right]\\
&=\frac{4}{\delta}\left[ (1+s^2)\hat{\alpha}^2 + 1 -  \hat{\alpha} \cdot \phi_2(s)\right].
\end{split}
\]
From \eqref{Eqn:alpha_new_s}, we see that fixing $\psi_1(\hat{\alpha},\hat{\sigma}^2) =C$ is equivalent to fixing $s=\phi_1^{-1}(C)$. Further, for a fixed $s$, $\psi_2(\hat{\alpha},\hat{\sigma}^2)$ is a quadratic function of $\hat{\alpha}$, and the minimum happens at
\[
\hat{\alpha}_{\min} = \frac{\phi_2(s)}{2(1+s^2)}=\frac{\phi_2(\phi_1^{-1}(C))}{2\left[1+\left(\phi_1^{-1}(C)\right)^2\right]},
\]
and $\psi_2(\hat{\alpha}_{\min},\hat{\sigma}^2;\delta)$ is
\[
\psi_2(\hat{\alpha}_{\min},\hat{\sigma}^2;\delta)=\frac{4}{\delta}\left( 1 - \frac{\phi_2^2(s) }{4(1+s^2)}\right)=\frac{4}{\delta}\left( 1 - \frac{\phi_2^2(\phi_1^{-1}(C)) }{4\left(1+\left[\phi_1^{-1}(C)\right]^2\right)}\right)=L\left(C;\delta\right),
 \]
where the last step is from the definition of $L$ is \eqref{Eqn:left_bound}. This completes the proof.
\end{proof}

To understand the implication of this lemma, let us consider the $t^{\rm th}$ iteration of the SE:
\[
\begin{split}
\alpha_{t+1} &= \psi_{1}(\alpha_t,\sigma_t^2),\\
\sigma_{t+1}^2 &= \psi_{2}(\alpha_t,\sigma_t^2;\delta),
\end{split}
\]
Note that according to Lemma \ref{Lem:left_bound}, no matter where $(\alpha_t, \sigma_t^2)$ is, $(\alpha_{t+1}, \sigma_{t+1}^2)$ will fall above the $\sigma^2 = L(\alpha; \delta)$ curve. This function is a key component in the dynamics of $\ampa$. Before we proceed further we discuss two main properties of the function $L(\alpha; \delta)$.

\begin{lemma}\label{Lem:monotonicity_L}
$L(\alpha;\delta)$ is a strictly decreasing function of $\alpha\in(0,1)$.
\end{lemma}
\begin{proof}

Recall from \eqref{Eqn:left_bound} that $L(\alpha;\delta)$ is defined as
\BE
\begin{split}
L(\alpha;\delta)&\Mydef \frac{4}{\delta}\left( 1 - \frac{\phi_2^2(\phi_1^{-1}(\alpha)) }{4(1+(\phi_1^{-1}(\alpha))^2)}\right)\\
&= \frac{4}{\delta}\left( 1 - I_2[\phi_1^{-1}(\alpha)]\right),
\end{split}\nonumber
\EE
where $I_2:\mathbb{R}_{+}\mapsto\mathbb{R}_{+}$ is defined as
\BE\label{Eqn:I2_def}
I_2(s) \Mydef \frac{\phi^2_2(s)}{4(1+s^2)}.
\EE
From \eqref{Eqn:phi}, it is easy to see that $\phi_1(s)$ is a decreasing function. Hence, to prove that $L(\alpha;\delta)$ is a decreasing function of $\alpha$, it suffices to prove that $I_2(s)$ is strictly decreasing. 

Substituting \eqref{Eqn:phi2} into \eqref{Eqn:I2_def} yields:
\BE
\begin{split}
I_2(s) &= \frac{\phi^2_2(s)}{4(1+s^2)}\\
&=\frac{1}{4(1+s^2)}\left(\int_0^{\frac{\pi}{2}} \frac{ 2\sin^2\theta +s^2}{ \left( \sin^2\theta +s^2 \right)^{\frac{1}{2}} }\right)^2 \\
&\overset{(a)}{=}\frac{1}{4}\left[2E\left(\frac{1}{1+s^2}\right)-\frac{s^2}{1+s^2}K\left(\frac{1}{1+s^2}\right)\right]^2\\
&=\frac{1}{4}\left[ 2E(x)-(1-x)K(x) \right]^2,
\end{split}\nonumber
\EE
where step (a) is obtained through similar calculations as those in \eqref{Eqn:integral_identities}, and in the last step we defined $x=\frac{1}{1+s^2}$. Hence, to prove that $I_2(s)$ is a decreasing function of $s$, it suffices to prove that $[2E(x)-(1-x)K(x)]^2$ is an increasing function of $x$.
Further, $2E(x)-(1-x)K(x)=T(x)+E(x)>0$ (form the definition of $T(x)$ in \eqref{Eqn:def_elliptic}), our problem reduces to proving that $2E(x)-(1-x)K(x)$ is increasing. To this end, differentiation yields
\BE
\left[2E(x)-(1-x)K(x)\right]'\overset{(a)}{=}\frac{E(x)-(1-x)K(x)}{2x}\overset{(b)}{=}\frac{1}{2}T(x)\overset{(c)}{>}0,\nonumber
\EE
where (a) is from the differentiation identities in Lemma~\ref{Lem:elliptic}, (b) is from \eqref{Eqn:def_elliptic}, and $T(x)>0$ follows from Lemma~\ref{Lem:elliptic} (ii) together with the fact that $T(0)=0$.
\end{proof}

The next lemma compares the function $L(\alpha;\delta)$ with $F_1^{-1}(\alpha)$.

\begin{lemma}\label{Lem:F1_L}
If $\delta>\deltaAMP$, then 
\BE
F_1^{-1}(\alpha)>L(\alpha;\delta),\quad\forall\alpha\in(0,1).\nonumber
\EE
\end{lemma}
\begin{proof}
We prove by contradiction. Suppose that $L(\hat{\alpha};\delta)\ge F_1^{-1}(\hat{\alpha})$ at some $\hat{\alpha}\in(0,1)$. If this is the case, then there exists a $\hat{\sigma}^2$ such that
\BE\label{Eqn:Lem_F1L_1}
F_1^{-1}(\hat{\alpha}) \le \hat{\sigma}^2\leq L(\hat{\alpha};\delta).
\EE
Since $F_1$ is a decreasing function (see Lemma~\ref{Lem:2}), the first inequality implies that $\hat{\alpha}\textcolor{red}{\ge} F_1(\hat{\sigma}^2)$. Then, based on the global attractiveness property in Lemma~\ref{lem:psi1} (iii), we have
\BE\label{Eqn:Lem_F1_L_1}
\psi_1(\hat{\alpha},\hat{\sigma}^2)\le \hat{\alpha}.
\EE
Further, Lemma~\ref{Lem:F1_F2_complex} shows that $F_1^{-1}(\hat{\alpha})>F_2(\hat{\alpha};\delta)$ for $\delta>\deltaAMP$, and using \eqref{Eqn:Lem_F1L_1} we also have $\hat{\sigma}^2\ge F_1^{-1}(\hat{\alpha})> F_2(\hat{\alpha};\delta)$. Also, from \eqref{Eqn:Lem_F1L_1}, 
\[
\hat{\sigma}^2\le L(\hat{\alpha};\delta)\overset{(a)}{<} L(0;\delta)=\frac{4}{\delta}\left(1-\frac{\pi^2}{16}\right)<\frac{4}{\delta}\le\sigma^2_{\max},
\]
where (a) is due to the monotonicity of $L(\alpha;\delta)$ (see Lemma~\ref{Lem:monotonicity_L}). From the above discussions, $F_2(\hat{\alpha};\delta)<\hat{\sigma}^2<\sigma^2_{\max}$. We then have (for $\delta>\deltaAMP$):
\BE\label{Eqn:Lem_F1_L_2}
 \psi_2(\hat{\alpha},\hat{\sigma}^2;\delta)\overset{(a)}{<}\hat{\sigma}^2\overset{(b)}{\le}L(\hat{\alpha};\delta)\overset{(c)}{\le}L\left[ \psi_1(\hat{\alpha},\hat{\sigma}^2);\delta\right],
\EE
where step (a) follows from the global attractiveness property in Lemma~\ref{lem:psi2} (iv), step (b) is due to the hypothesis in \eqref{Eqn:Lem_F1L_1}, step (c) is from \eqref{Eqn:Lem_F1_L_1} together with the monotonicity of $L(\alpha;\delta)$ (see Lemma~\ref{Lem:monotonicity_L}). Note that \eqref{Eqn:Lem_F1_L_2} shows that $ \psi_2(\hat{\alpha},\hat{\sigma}^2;\delta)<L\left[ \psi_1(\hat{\alpha},\hat{\sigma}^2);\delta\right]$, which contradicts Lemma \ref{Lem:left_bound}, where we proved that $ \psi_2(\alpha,\sigma^2;\delta)\ge L\left[ \psi_1(\alpha,\sigma^2);\delta\right]$ for any $\alpha>0$, $\sigma^2>0$ and $\delta>0$. Hence, we must have that $L(\alpha;\delta)<F_1^{-1}(\alpha)$ for any $\alpha\in(0,1)$. 
\end{proof}

\begin{lemma}\label{Lem:L_bound}
The following holds for any $\alpha\in(0,1)$ and $\delta>0$, 
\BE\label{Eqn:lem_L_bound0}
L(\alpha;\delta)>\frac{4}{\delta}\left ( 1 - \frac{\pi^2}{16 }- \frac{1}{2}\alpha^2\right ),
\EE
where $L(\alpha,\delta)$ is defined in \eqref{Eqn:left_bound}.
\end{lemma}
\begin{proof}
From \eqref{Eqn:left_bound}, proving \eqref{Eqn:lem_L_bound0} is equivalent to proving:
\BE\label{Eqn:lem_L_bound1}
1 - \frac{\phi_2^2(\phi_1^{-1}(\alpha)) }{4\left[1+(\phi_1^{-1}(\alpha))^2\right]} > 1 - \frac{\pi^2}{16 }- \frac{1}{2}\alpha^2,\quad\forall \alpha\in(0,1),
\EE
where $\phi_1:[0,\infty)\mapsto[0,1]$ and $\phi_2:[0,\infty)\mapsto[0,\infty)$ are defined as (see \eqref{Eqn:phi} and \eqref{Eqn:phi2}):
\BS \label{Eqn:lem_L_bound_1}
\begin{align}
\phi_1(s)&=\int_0^{\frac{\pi}{2}}\frac{\sin^2\theta}{\left(\sin^2\theta +s^2 \right)^{\frac{1}{2}}}\mr{d}\theta, \label{Eqn:lem_L_bound_1a}\\
\phi_2(s)&= \int_0^{\frac{\pi}{2}} \frac{ 2\sin^2\theta +s^2 }{ \left( \sin^2\theta +s^2 \right)^{\frac{1}{2}} }\mr{d}\theta.  \label{Eqn:lem_L_bound_1b}
\end{align}
\ES
We make a variable change:
\[
\alpha = \phi_1(s).
\]
Simple calculations show that \eqref{Eqn:lem_L_bound1} can be reformulated as the following
\BE\label{Eqn:lem_L_bound_2}
\frac{1}{1+s^2}\phi_2^2(s) < \frac{\pi^2}{4} + 2  \phi_1^2(s),\quad s\in(0,\infty).
\EE
Let us further define
\BE \label{Eqn:lem_L_bound_3}
\phi_3(s)\equiv \int_0^{\frac{\pi}{2}} (\sin^2 \theta + s^2)^{\frac{1}{2}} \mr{d} \theta.
\EE
From \eqref{Eqn:lem_L_bound_1} and \eqref{Eqn:lem_L_bound_3}, we have
\[
\phi_2(s)=\phi_1(s) + \phi_3(s),
\]
and \eqref{Eqn:lem_L_bound_2} can be reformulated as
\BE\label{Eqn:lem_L_bound_4}
\left[\phi_1(s) + \phi_3(s)\right]^2 - (1 + s^2) \left[ \frac{\pi^2}{4} + 2\phi_1^2(s)\right]<0.
\EE
To this end, we can write the LHS of \eqref{Eqn:lem_L_bound_4} into a quadratic form of $\phi_1(s)$:
\[
\begin{split}
&\left[\phi_1(s) + \phi_3(s)\right]^2- (1 + s^2) \left[ \frac{\pi^2}{4} + 2\phi_1^2(s)\right] \\
& = \phi_1^2(s) + \phi_3^2(s) + 2\phi_1(s) \phi_3(s)-(1 + s^2) \left[ \frac{\pi^2}{4} + 2\phi_1^2(s)\right]\\
&=-(1+2s^2)\phi_1^2(s) + 2\phi_1(s) \phi_3(s)-\frac{\pi^2}{4}(1+s^2) + \phi_3^2(s).
\end{split}
\]
Hence, to prove that this quadratic form is negative everywhere, it suffices to prove that the discriminant is negative, i.e.,
\[
4\phi_3^2(s) + 4  (1+2s^2)\left[-\frac{\pi^2}{4}(1+s^2) + \phi_3^2(s)\right]<0,
\]
or
\[
  \phi_3^2(s) < \frac{\pi^2}{8} (1+2 s^2 ).
\]
Finally, by Cauchy-Schwarz we have
\[
\begin{split}
  \phi_3^2(s)  & =\left[ \int_0^{\frac{\pi}{2}} (\sin^2 \theta + s^2)^{\frac{1}{2}} \mr{d} \theta\right]^2\\
  &\le \int _0^{\frac{\pi}{2}}1 \mr{d} \theta \cdot \int_0^{\frac{\pi}{2}} \left(\sqrt{\sin^2 \theta + s^2}\right)^2  \mr{d}
  \theta\\
&=\frac{\pi}{2}\left(\frac{\pi}{4}+\frac{\pi}{2}s^2\right)=\frac{\pi^2}{8} (1+2 s^2 ) ,
\end{split}
\]
which completes our proof.
\end{proof}

\begin{lemma}\label{Lem:monotonicity_Psi2_G1}
For any $\alpha\in[0,1]$, $\psi_2(\alpha,\sigma^2;\deltaAMP)$ is an increasing function of $\sigma^2$ on $\sigma^2\in[L(\alpha;\deltaAMP),\infty)$, where the function $L(\alpha;\delta)$ is defined in \eqref{def:Lalpha}. 
\end{lemma}
\begin{proof}
From Lemma~\ref{lem:psi2} (v), the case $\alpha>\alpha_{\ast}\approx0.53$ is trivial since then $\psi_2(\sigma^2,\alpha;\deltaAMP)$ is strictly increasing in $\sigma^2\in\mathbb{R}_{+}$. In the rest of this proof, we assume that $\alpha<\alpha_{\ast}$. We have derived in \eqref{Eqn:psi2_dif_sigma2} that

\BE\label{Eqn:lem_mono_Psi2_G1_1}
\frac{\partial \psi_2(\alpha,\sigma^2;\delta)}{\partial\sigma^2}>0 \Longleftrightarrow \alpha >\frac{1}{2\sqrt{1+s^2}} E\left( \frac{1}{1 + s^2} \right)=f(s),
\EE
where
\[
s\Mydef \frac{\sigma}{\alpha}.
\]
Hence, the result of Lemma \ref{Lem:monotonicity_Psi2_G1} can be reformulated as proving the following:   
\BE
 \alpha >f(s),\quad \forall s \ge \frac{\sqrt{ L(\alpha;\deltaAMP)}}{\alpha}, \ \ \ \ \ \ \ \ \alpha\in[0,\alpha_{\ast}).\nonumber
\EE
We proceed in three steps:
\begin{enumerate}
\item[(i)] In Lemma \ref{Lem:L_bound}, we proved that the following holds for any $\alpha\in[0,1]$:
\BE\label{Eqn:L_hat}
{L}(\alpha;\deltaAMP)\ge\hat{L}(\alpha,\deltaAMP)\Mydef\frac{4}{\deltaAMP}\left ( 1 - \frac{\pi^2}{16 }- \frac{1}{2}\alpha^2\right ).
\EE
For convenience, define
\BE\label{Eqn:s_hat_def}
\hat{s}(\alpha) \Mydef \frac{\sqrt{\hat{L}(\alpha;\deltaAMP)}}{\alpha}.
\EE
\item[(ii)] We prove that $f(s)$ is monotonically decreasing on $s\in\left [\hat{s}(\alpha),\infty\right)$ for $\alpha<\alpha_{\ast}$.
\item[(iii)] We prove that the following holds for $\alpha<\alpha_{\ast}$:
\BE
 \alpha >f(\hat{s}(\alpha)).\nonumber
\EE
\end{enumerate}
Clearly, \eqref{Eqn:lem_mono_Psi2_G1_1} follows from the above claims. Here, we introduce the function $\hat{L}$ since $\hat{L}$ has a simple closed-form formula and is easier to manipulate than $L(\alpha)$. We next prove step (ii). From \eqref{Eqn:lem_psi2_fs_cases}, it suffices to prove that
\[
\hat{s}(\alpha) > s_{\ast},\quad \forall\alpha<\alpha_{\ast},
\]
where $s_{\ast}$ and $\alpha_{\ast}$ are defined in \eqref{Eqn:lem_psi2_1} and \eqref{Eqn:monotonicity_condi2} respectively. To this end, we note that the following holds for $\alpha < \alpha_{\ast}$:
\[
\begin{split}
\hat{s}(\alpha) =  \frac{\sqrt{\hat{L}(\alpha;\deltaAMP)}}{\alpha} >\frac{\sqrt{\hat{L}(\alpha_{\ast};\deltaAMP)}}{\alpha_{\ast}}\approx 1.18,
\end{split}
\]
where the inequality follows from the fact that $\hat{L}$ in \eqref{Eqn:L_hat} is strictly decreasing in $\alpha$, and the last step is calculated from \eqref{Eqn:L_hat} and $\alpha_{\ast}\approx0.527$ . Finally, numerical evaluation of \eqref{Eqn:lem_psi2_1} shows that $s_{\ast}\approx0.458$.  Hence, $\hat{s}(\alpha) > s_{\ast}$, which completes the proof.

We next prove step (iii). First, simple manipulations yields
\BE\label{Eqn:lem_mono_Psi2_G1_4}
\begin{split}
\hat{s}^2(\alpha)\overset{(a)}{=}\frac{\hat{L}(\alpha)}{\alpha^2}\overset{(b)}{=}\frac{4}{\deltaAMP}\left [ \left(1 - \frac{\pi^2}{16 }\right)\cdot \frac{1}{\alpha^2}- \frac{1}{2}\right ],
\end{split}
\EE
where (a) is from the definition of $\hat{s}(\alpha)$ in \eqref{Eqn:s_hat_def} and (b) is due to \eqref{Eqn:L_hat}. 
Using \eqref{Eqn:lem_mono_Psi2_G1_4}, we further obtain
\BE\label{Eqn:L_hat2}
\alpha=\sqrt{ \frac{16-\pi^2}{ 4\deltaAMP\hat{s}^2(\alpha) + 8 } }.
\EE
Now, from \eqref{Eqn:L_hat2} and \eqref{Eqn:f_def}, we have
\BE \label{Eqn:lem_mono_Psi2_G1_5}
\begin{split}
\alpha - f(\hat{s}(\alpha))>0 &\Longleftrightarrow \sqrt{ \frac{16-\pi^2}{ 4\deltaAMP\hat{s}^2(\alpha) + 8 } } - \frac{1}{2\sqrt{1+\hat{s}^2(\alpha)}} E\left( \frac{1}{1 + \hat{s}^2(\alpha)} \right) >0.
\end{split}
\EE

We prove \eqref{Eqn:lem_mono_Psi2_G1_5} by showing that the following stronger result holds:
\BE\label{Eqn:lem_mono_Psi2_G1_6}
\sqrt{ \frac{16-\pi^2}{ 4\deltaAMP t^2 + 8 } } - \frac{1}{2\sqrt{1+t^2 }} E\left( \frac{1}{1 + t^2} \right) >0,\quad \forall t\in\mathbb{R}_{+}.
\EE
For convenience, we make a variable change:
\[
x\Mydef \frac{1}{1+t^2}.
\]
With some straightforward calculations, we can rewrite \eqref{Eqn:lem_mono_Psi2_G1_6} as
\[
\begin{split}
E(x) & < \sqrt{ \frac{16-\pi^2}{ \deltaAMP(1-x) + 2 x } }
\end{split}
\]
The following upper bound on $E(x)$ is due to \cite[Eqn.~(1.2)]{Wang2013}:
\[
E(x)<\frac{\pi}{2}\sqrt{1-\frac{x}{2}},\quad \forall x\in(0,1].
\]
Hence, it is sufficient to prove that
\[
\frac{\pi}{2}\sqrt{1-\frac{x}{2}}< \sqrt{ \frac{16-\pi^2}{ \deltaAMP(1-x) + 2 x } },
\]
which can be reformulated as
\[
\left(1-\frac{x}{2}\right)\left(\deltaAMP - (\deltaAMP-2)x\right)<\frac{4}{\pi^2}(16-\pi^2)=\deltaAMP
\]
where the second equality follows from the definition $\deltaAMP=\frac{64}{\pi^2}-4$. The above inequality holds since $0<1-\frac{x}{2}<1$ and $0<\deltaAMP - (\deltaAMP-2)x<\deltaAMP$. This completes the proof.
\end{proof}

\vspace{5pt}

\begin{lemma}\label{Lem:monotonicity_delta}
For any $\alpha\in[0,1]$, $\psi_2\left(\alpha,L(\alpha;\delta);\delta\right)$ is a strictly decreasing function of $\delta>0$, where $L(\alpha;\delta)$ is defined in \eqref{Eqn:left_bound}.
\end{lemma}
\begin{proof}
From the definition of $L(\alpha;\delta)$ in \eqref{Eqn:left_bound}, we can write
\[
\psi_2\left(\alpha,L(\alpha;\delta);\delta\right) = \psi_2\left(\alpha,\frac{1}{\delta}\bar{\sigma}^2;\delta\right),
\]
where (note that $\bar{\sigma}$ is not the conjugate of $\sigma$)
\[
 \bar{\sigma}^2 \Mydef  4\left( 1 - \frac{\phi_2^2(\phi_1^{-1}(\alpha)) }{4\left[1+(\phi_1^{-1}(\alpha))^2\right]}\right).
\]
A key observation here is that $ \bar{\sigma}^2$ does not depend on $\delta$. Clearly, Lemma~\ref{Lem:monotonicity_delta} is implied by the following stronger result:
\[
\frac{\partial \psi_2\left(\alpha,\frac{1}{\delta}\bar{\sigma}^2;\delta\right)}{\partial \delta} < 0,\quad\forall \bar{\sigma}^2>0,\alpha>0,\delta>0,
\]
which we will prove in the sequel. For convenience, we define
\BE\label{Eqn:lem_mono_delta_1}
\bar{s}\Mydef\frac{\bar{\sigma}}{\alpha},\ \gamma\Mydef\frac{1}{\delta}\text{ and } s=\sqrt{\gamma}\bar{s}.
\EE
Using these new variables, we have
\BE
\begin{split}
\psi_2\left(\alpha,\frac{1}{\delta}\bar{\sigma}^2;\delta\right)&=\psi_2\left(\alpha,\gamma\bar{\sigma}^2;\gamma^{-1}\right)\\
&= 4\gamma\left((1+\gamma\bar{s}^2)\alpha^2+1- \alpha\int_0^{\frac{\pi}{2}} \frac{ 2\sin^2\theta +\gamma\bar{s}^2 }{ \left( \sin^2\theta + \gamma\bar{s}^2\right)^{\frac{1}{2}} }\mr{d}\theta\right),
\end{split}\nonumber
\EE
where the last equality is from the definition of $\psi_2$ in \eqref{Eqn:map_expression_complex_b}.
It remains to prove that $\psi_2\left(\alpha,\gamma\bar{\sigma}^2;\gamma^{-1}\right)$ is an increasing function of $\gamma$. The partial derivative of $\psi_2(\alpha,\sigma^2;\delta)$ w.r.t. $\gamma$ is given by
\BE \label{Eqn:psi2_dif_gamma}
\begin{split}
\frac{\partial \psi_2\left(\alpha,\gamma\bar{\sigma}^2;\gamma^{-1}\right)}{\partial \gamma} &= 4(1+2\gamma\bar{s}^2)\alpha^2-4\alpha\left(\int_0^{\frac{\pi}{2}} \frac{2\sin^2\theta+\gamma\bar{s}^2}{(\sin^2\theta+\gamma\bar{s}^2)^{\frac{1}{2}}}\mr{d}\theta+\frac{1}{2}\int_0^{\frac{\pi}{2}} \frac{\gamma^2\bar{s}^4}{(\sin^2\theta+\gamma\bar{s}^2)^{\frac{3}{2}}}\mr{d}\theta \right)+4\\
&\overset{(a)}{=}(1+2{s}^2)\alpha^2-4\alpha\left(\int_0^{\frac{\pi}{2}}\frac{2\sin^2\theta+s^2}{\left(\sin^2\theta+s^2\right)^{\frac{1}{2}}}\mr{d}\theta+\frac{1}{2}\int_0^{\frac{\pi}{2}} \frac{s^4}{(\sin^2\theta+s^2)^{\frac{3}{2}}}\mr{d}\theta\right)+4 \\
&\overset{(b)}{=}4(1+2s^2)\alpha^2-4\alpha\left(\frac{(5s^2+4)E\left(\frac{1}{1+s^2}\right)-2s^2K\left(\frac{1}{1+s^2}\right)}{2\sqrt{1+s^2}}\right)+4 ,
\end{split}
\EE
where in step (a) we used the relationship $s^2=\gamma\bar{s}^2$ (see \eqref{Eqn:lem_mono_delta_1}), and step (b) is from the identities in \eqref{Eqn:integral_identities}. From \eqref{Eqn:psi2_dif_gamma}, we see that $\frac{\partial \psi_2\left(\alpha,\gamma\bar{\sigma}^2;\gamma^{-1}\right)}{\partial \gamma} $ is a quadratic function of $\alpha$. Therefore, to prove $\frac{\partial \psi_2\left(\alpha,\gamma\bar{\sigma}^2;\gamma^{-1}\right)}{\partial \gamma} >0$, it suffices to show that the discriminant is negative:
\BE\label{Eqn:lem_mono_delta_3}
\left(\frac{(5s^2+4)E\left(\frac{1}{1+s^2}\right)-2s^2K\left(\frac{1}{1+s^2}\right)}{2\sqrt{1+s^2}}\right)^2-4(1+2s^2)<0.
\EE
Further, to prove \eqref{Eqn:lem_mono_delta_3}, it is sufficient to prove that the following two inequalities hold:
\BS
\BE\label{Eqn:Aux_1_ineq1}
(5s^2+4)E\left(\frac{1}{1+s^2}\right)-2s^2K\left(\frac{1}{1+s^2}\right)>0,
\EE
and
\BE\label{Eqn:Aux_1_ineq2}
(5s^2+4)E\left(\frac{1}{1+s^2}\right)-2s^2K\left(\frac{1}{1+s^2}\right)<4\sqrt{1+s^2}\sqrt{1+2s^2}.
\EE
\ES
We first prove \eqref{Eqn:Aux_1_ineq1}. It is sufficient to prove the following
\BE\label{Eqn:Aux_1_ineq2.1}
(4s^2+4)E\left(\frac{1}{1+s^2}\right)-2s^2K\left(\frac{1}{1+s^2}\right)>0.
\EE
Applying a variable change $x=\frac{1}{1+s^2}$, we can rewrite \eqref{Eqn:Aux_1_ineq2.1} as
\[
\frac{4E(x)-2(1-x)K(x)}{x}>0.
\]
The above inequality holds since
\[
4E(x)-2(1-x)K(x)>2E(x)-2(1-x)K(x)=2T(x)>0,
\]
where the last equality is from the definition of $T(x)$ in \eqref{Eqn:def_elliptic}.\par

We next prove \eqref{Eqn:Aux_1_ineq2}. Again, applying the variable change $x=\frac{1}{1+s^2}$ and after some straightforward manipulations, we can rewrite \eqref{Eqn:Aux_1_ineq2} as
\[
h(x)/x<0,\quad x\in(0,1),
\]
where
\[
h(x)\Mydef (5-x)E(x)-2(1-x)K(x)-4\sqrt{2-x}<0.
\]
Hence, we only need to prove $h(x)<0$ for $0<x<1$. First, we note that $\lim_{x\to1^-}h(x)=0$, from the fact that $E(1)=1$ and $\lim_{x\to1^-}(1-x)K(x)=0$ (see Lemma~\ref{Lem:elliptic} (i)). We finish the proof by showing that $h(x)$ is strictly increasing in $x\in(0,1)$. Using the identities in \eqref{Eqn:elliptic_dif}, we can obtain
\BE
h'(x)=\frac{3}{2}\frac{(1-x)(E(x)-K(x))}{x}+\frac{2}{\sqrt{2-x}}.\nonumber
\EE
To prove $h'(x)>0$, it is equivalent to prove
\BE\label{Eqn:Aux_1_ineq3}
\begin{split}
\frac{4x}{3(1-x)\sqrt{2-x}}&>K(x)-E(x)\\
&=\int_0^{\frac{\pi}{2}}\frac{1}{(1-x\sin^2\theta)^{\frac{1}{2}}}\mr{d}\theta-\int_0^{\frac{\pi}{2}}(1-x\sin^2\theta)^{\frac{1}{2}}\mr{d}\theta\\
&=\int_0^{\frac{\pi}{2}}\frac{x\sin^2\theta}{(1-x\sin^2\theta)^{\frac{1}{2}}}\mr{d}\theta.
\end{split}
\EE
Noting $0<x<1$, we can get the following
\[
\begin{split}
\int_0^{\frac{\pi}{2}}\frac{x\sin^2\theta}{(1-x\sin^2\theta)^{\frac{1}{2}}}\mr{d}\theta<\int_0^{\frac{\pi}{2}}\frac{x\sin^2\theta}{1-x\sin^2\theta}\mr{d}\theta=\frac{\pi}{2}\left(\frac{1}{\sqrt{1-x}}-1\right).
\end{split}
\]
Hence, to prove \eqref{Eqn:Aux_1_ineq3}, it suffices to prove
\[
\frac{4x}{3(1-x)\sqrt{2-x}}>\frac{\pi}{2}\left(\frac{1}{\sqrt{1-x}}-1\right),
\]
which can be reformulated as
\[
\frac{8}{3\pi}\frac{1}{\sqrt{2-x}}>\frac{\sqrt{1-x}}{1+\sqrt{1-x}}.
\]
The inequality holds since
\BS
\[
\frac{8}{3\pi}\frac{1}{\sqrt{2-x}}>\frac{8}{3\pi}\frac{1}{\sqrt{2}}>\frac{1}{2},\quad \forall x\in(0,1),
\]
and
\[
\frac{\sqrt{1-x}}{1+\sqrt{1-x}}<\frac{1}{2},\quad \forall x\in(0,1).
\]
\ES
\end{proof}

\begin{figure}[!htbp]
\centering
\includegraphics[width=.55\textwidth]{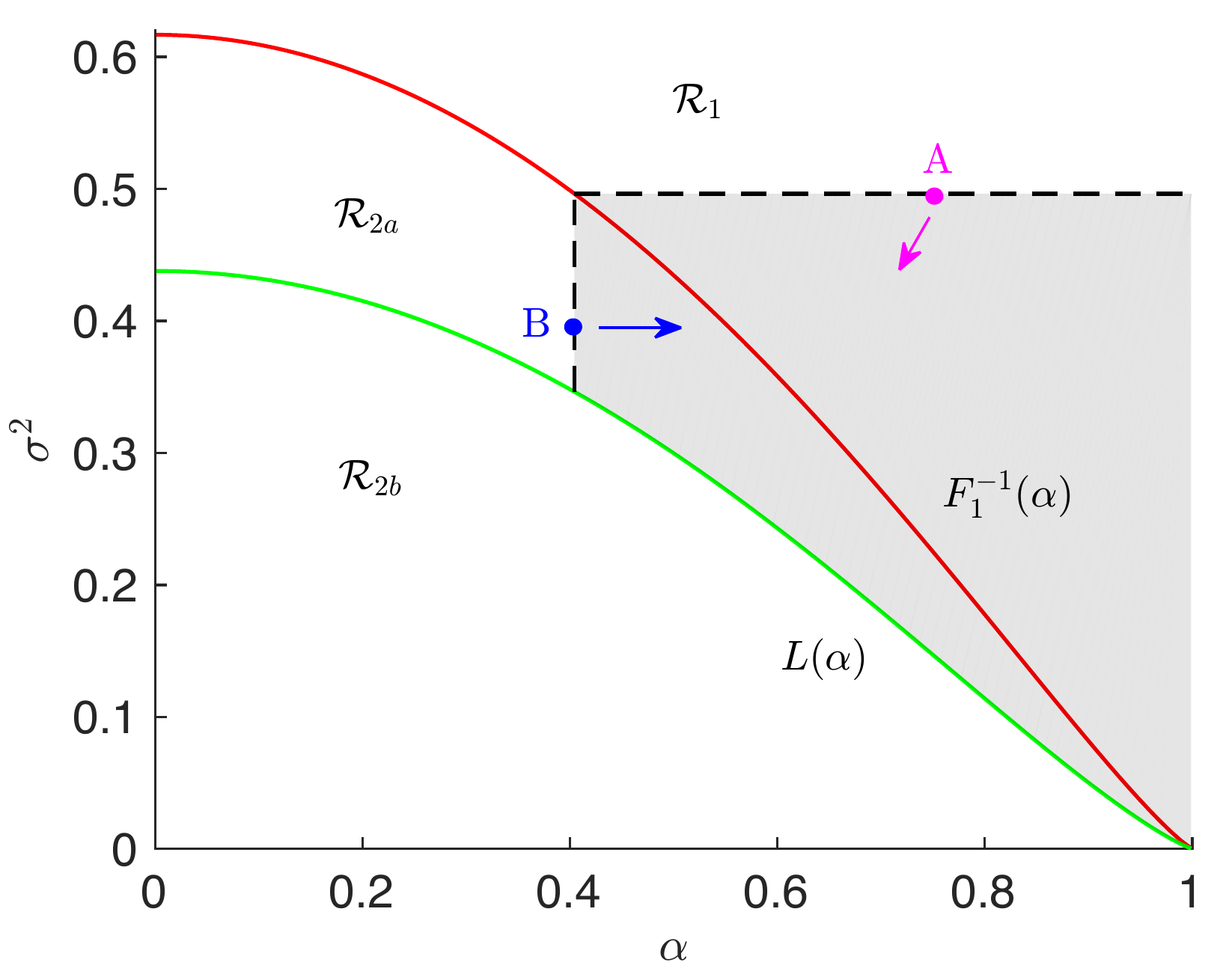}
\caption{Illustration of the convergence behavior. $\mathcal{R}_1$ and $\mathcal{R}_2$ are defined in Definition \ref{Def:regions}. For both point A and point B, $B_1(\alpha,\sigma^2)$ and $B_2(\alpha,\sigma^2)$ are given by the two dashed lines. After one iteration, $\mathcal{R}_{2b}$ will not be achievable and we can focus on $\mathcal{R}_{2a}$.}
\label{fig:regions_I&III_bounds}
\end{figure}

\item \textbf{Main proof}

We now return to the main proof for Lemma~\ref{Lem:regionI_III}. Notice that by Lemma \ref{Lem:left_bound}, $(\alpha_{t_0},\sigma^2_{t_0})$ cannot fall below the curve $L(\alpha;\delta)$ for $t_0\ge1$. Hence, for $\mathcal{R}_2$, we can focus on the region above $L(\alpha;\delta)$ (including $L(\alpha;\delta)$), which we denote as $\mathcal{R}_{2a}$. See Fig.~\ref{fig:regions_I&III_bounds} for illustration.

We will first prove that if $(\alpha,\sigma^2)\in\mathcal{R}_1\cup \mathcal{R}_{2a}$, then the next iterates $\psi_1(\alpha,\sigma^2)$ and $\psi_2(\alpha,\sigma^2)$ satisfy the following:
\BS\label{Eqn:lem_regionI_III_1}
\BE\label{Eqn:lem_regionI_III_1a}
\psi_1(\alpha,\sigma^2) \ge B_1(\alpha,\sigma^2),
\EE
and
\BE\label{Eqn:lem_regionI_III_1b}
\psi_2(\alpha,\sigma^2) \le B_2(\alpha,\sigma^2),
\EE
\ES
where $B_1(\alpha,\sigma^2) $ and $B_2(\alpha,\sigma^2) $ are defined as
\BE\label{Eqn:bounds_def}
\begin{split}
B_1(\alpha,\sigma^2)& \Mydef\min\left\{ \alpha, F_1(\sigma^2) \right\},\\
B_2(\alpha,\sigma^2) & \Mydef\max\left\{ \sigma^2, F_1^{-1}(\alpha) \right\}.
\end{split}
\EE

Note that when $(\alpha,\sigma^2)$ is on $F_1^{-1}$ (i.e., $\sigma^2=F_1^{-1}(\alpha)$), equalities in \eqref{Eqn:lem_regionI_III_1a} and \eqref{Eqn:lem_regionI_III_1b} can be achieved. Further, this is the only case when either of the equality is achieved. Also, it is easy to see that if $(\alpha,\sigma^2)$ is on $F_1^{-1}$, then $(\psi_1(\alpha,\sigma^2),\psi_2(\alpha,\sigma^2))$ cannot be on $F_1^{-1}$.

Since $F_1^{-1}$ separates $\mathcal{R}_1$ and $\mathcal{R}_{2a}$,  \eqref{Eqn:bounds_def} can also be written as
\BE\label{Eqn:bounds_def2}
\big[ B_1(\alpha,\sigma^2), B_2(\alpha,\sigma^2)\big]=
\begin{cases}
 [ F_1(\sigma^2) , \sigma^2 ] & \text{if $(\alpha,\sigma^2)\in\mathcal{R}_1$},\\
  [ \alpha , F^{-1}_1(\alpha) ]  & \text{if $(\alpha,\sigma^2)\in\mathcal{R}_{2a}$}.
  \end{cases}
\EE
As a concrete example, consider the situation shown in Fig.~\ref{fig:regions_I&III_bounds}. In this case, for both point A and point B, $B_1(\alpha,\sigma^2)$ and $B_2(\alpha,\sigma^2)$ are given by the two dashed lines. This directly follows from \eqref{Eqn:bounds_def2} by noting that point A is in region $\mathcal{R}_1$ and point B is in region $\mathcal{R}_{2a}$. Let $\mathcal{R}_{2a}\backslash F_1^{-1}(\alpha)$ be a shorhand for $\{(\alpha,\sigma^2)|(\alpha,\sigma^2)\in\mathcal{R}_{2a},\alpha\neq F_1(\sigma^2)\}$. To prove the strict inequality in \eqref{Eqn:lem_regionI_III_1}, we deal with $(\alpha,\sigma^2)\in\mathcal{R}_1$ and $(\alpha,\sigma^2)\in\mathcal{R}_{2a}\backslash F_1^{-1}(\alpha)$ separately.
\begin{enumerate}
\item  Assume that $(\alpha,\sigma^2)\in\mathcal{R}_1$. Using \eqref{Eqn:bounds_def2}, the inequality in \eqref{Eqn:lem_regionI_III_1} can be rewritten as
\BE\label{Eqn:lem_regionI_III_7}
\psi_1(\alpha,\sigma^2) > F_1(\sigma^2)\quad\text{and}\quad \psi_2(\alpha,\sigma^2) < \sigma^2.
\EE
Since $(\alpha,\sigma^2)\in\mathcal{R}_1$, we have $\sigma^2 > F_1^{-1}(\alpha)$. Then, applying \eqref{Eqn:psi1_contracting} proves $\psi_1(\alpha,\sigma^2) > F_1(\sigma^2)$. Further, using Lemma~\ref{Lem:F1_F2_complex}, we have $\sigma^2 > F_1^{-1}(\alpha)>F_2(\alpha)$. Also, Lemma~\ref{Lem:region_bound} guarantees that $\sigma^2<\sigma^2_{\max}$. Hence, $F_1^{-1}(\alpha)<\sigma^2 < \sigma^2_{\max}$ and applying Lemma~\ref{lem:psi2} (iv) yields $\psi_2(\alpha,\sigma^2) < \sigma^2$.
 
\item We now consider the case where $(\alpha,\sigma^2)\in\mathcal{R}_{2a}\backslash F_1^{-1}(\alpha)$. Similar to \eqref{Eqn:lem_regionI_III_7}, we need to prove
\BE\label{Eqn:lem_regionI_III_7b}
\psi_1(\alpha,\sigma^2) > \alpha\quad\text{and}\quad \psi_2(\alpha,\sigma^2) < F_1^{-1}(\alpha).
\EE
The inequality $\psi_1(\alpha,\sigma^2) > \alpha$ can be proved by the global attractiveness in Lemma~\ref{lem:psi1} (iii) and the fact that $\sigma^2 < F_1^{-1}(\alpha)$ when $(\alpha,\sigma^2)\in\mathcal{R}_{2a}\backslash F_1^{-1}(\alpha)$. The proof for $\psi_2(\alpha,\sigma^2) < F_1^{-1}(\alpha)$ is considerably more complicated and is detailed in Lemma~\ref{Lem:RegionIII_bound} below.  

\begin{lemma}\label{Lem:RegionIII_bound}
For any $(\alpha,\sigma^2)\in\mathcal{R}_{2a}$ (see Definition~\ref{Def:regions}) and $\delta\ge\deltaAMP$, the following holds:
\BE\label{Eqn:lem_regionIII_bound_1}
\psi_2(\alpha,\sigma^2;\delta)<F_1^{-1}(\alpha),
\EE
where $\psi_2$ is the SE map in \eqref{Eqn:map_expression_complex_b} and $F_1^{-1}$ is the inverse of $F_1$ defined in Lemma~\ref{lem:psi1}.
\end{lemma}

\begin{proof}

The following holds when $(\alpha,\sigma^2)\in\mathcal{R}_{2a}$:
\BE
\psi_2(\alpha,\sigma^2;\delta)\le \max_{ \hat{\sigma}^2\in\mathcal{D}_{\alpha}}\psi_2({\alpha},\hat{\sigma}^2;\delta),\nonumber
\EE
where
\BE\label{Eqn:D_def}
\mathcal{D}_{\alpha} \Mydef \left\{ \hat{\sigma}^2 \big| L({\alpha;\delta}) \le \sigma^2 \le F_1^{-1}({\alpha})\right\}.
\EE
Hence, to prove \eqref{Eqn:lem_regionIII_bound_1}, it suffices to prove that the following holds for any $\delta\ge\deltaAMP$ and $\alpha\in[0,1]$:
\BE \label{Eqn:lem_regionIII_bound_3}
\max_{ \hat{\sigma}^2\in\mathcal{D}_{\alpha}}\psi_2({\alpha},\hat{\sigma}^2;\delta) < F_1^{-1}(\alpha).
\EE 
We next prove \eqref{Eqn:lem_regionIII_bound_3}. We consider the three different cases:
\begin{enumerate}
\item[(i)] $\alpha \in[ \alpha_{\ast},1]$ and all $\delta\in[\deltaAMP,\infty)$, where $\alpha_{\ast}$ is defined in \eqref{Eqn:alpha_ast_def}.
\item[(ii)] $\alpha\in[0, \alpha_{\ast})$ and $\delta\in[\deltaAMP,17]$.
\item[(iii)] $\alpha \in[0, \alpha_{\ast})$ and $\delta\in(17,\infty)$.
\end{enumerate}

\textit{Case (i):} Lemma~\ref{lem:psi2} (v) shows that $\psi_2$ is an increasing function of $\sigma^2$ in $\mathbb{R}_{+}$. Hence, by noting \eqref{Eqn:D_def}, we have
\BE
\max_{ \hat{\sigma}^2\in\mathcal{D}_{\alpha}}\psi_2({\alpha},\hat{\sigma}^2;\delta) = \psi_2({\alpha},F_1^{-1}(\alpha);\delta).\nonumber
\EE 
Therefore, proving \eqref{Eqn:lem_regionIII_bound_8} reduces to proving
\BE\label{Eqn:lem_regionIII_bound_3.1}
\psi_2({\alpha},F_1^{-1}(\alpha);\delta)\le F_1^{-1}(\alpha).
\EE
Finally, \eqref{Eqn:lem_regionIII_bound_3.1} follows from the global attractiveness property in Lemma~\ref{lem:psi2} (iv) and the inequality $F_1^{-1}(\alpha)>F_2(\alpha;\delta)$ in Lemma~\ref{Lem:F1_F2_complex}.

\textit{Case (ii):} We will prove that the following holds for $\alpha\in[0, \alpha_{\ast})$ and $\delta\in[\deltaAMP,17]$ (at the end of this proof)
\BE\label{Eqn:lem_regionIII_bound_6}
\max_{ \hat{\sigma}^2\in\mathcal{D}_{\alpha}}\psi_2({\alpha},\sigma^2;\delta)=\max\left\{\psi_2({\alpha},L(\alpha;\delta);\delta), \ \psi_2({\alpha},F_1^{-1}(\alpha);\delta) \right\}.
\EE
Namely, the maximum of $\psi_2$ over $\sigma^2$ is achieved at either $\sigma^2=L(\alpha;\delta)$ or $\sigma^2=F_1^{-1}(\alpha)$.
Hence, we only need to prove that the following holds for any $\alpha\in[0,\alpha_{\ast})$ and $\delta\ge\deltaAMP$:
\BE\label{Eqn:lem_regionIII_bound_7}
\max\left\{\psi_2({\alpha},L(\alpha;\delta);\delta), \ \psi_2({\alpha},F_1^{-1}(\alpha);\delta) \right\}\le F_1^{-1}(\alpha).
\EE
In the sequel, we first use \eqref{Eqn:lem_regionIII_bound_6} to prove \eqref{Eqn:lem_regionIII_bound_3}, and the proof for \eqref{Eqn:lem_regionIII_bound_6} will come at the end of this proof. 

Firstly, it is easy to see that $\psi_2(\alpha,F_1^{-1}(\alpha);\delta)$ is a decreasing function of $\delta$, since $\psi_2(\alpha,\sigma^2;\delta)$ is a decreasing function of $\delta$ and $F_1^{-1}(\alpha)$ does not depend on $\delta$. Further, Lemma~\ref{Lem:monotonicity_delta} shows that $\psi_2({\alpha},L(\alpha;\delta);\delta)$ is also a decreasing function of $\delta$. (Notice that unlike $F_1^{-1}(\alpha)$, $L(\alpha;\delta)$ depends on $\delta$, and thus Lemma~\ref{Lem:monotonicity_delta} is nontrivial.) Hence, to prove \eqref{Eqn:lem_regionIII_bound_7} for $\delta\ge\deltaAMP$, it suffices to prove \eqref{Eqn:lem_regionIII_bound_7} for $\delta=\deltaAMP$, namely,
\BE\label{Eqn:lem_regionIII_bound_8}
\max\left\{\psi_2({\alpha},L(\alpha;\delta);\deltaAMP), \ \psi_2({\alpha},F_1^{-1}(\alpha);\deltaAMP) \right\}\le F_1^{-1}(\alpha).
\EE
When $\delta=\deltaAMP$, we prove in Lemma~\ref{Lem:monotonicity_Psi2_G1} that $\psi_2$ is an increasing function of $\sigma^2$ in $\sigma^2\in[L(\alpha;\deltaAMP),\infty)$. (Such  monotonicity generally does not hold if $\delta$ is too large.) Further, Lemma~\ref{Lem:F1_L} shows that $F_1^{-1}(\alpha)>L(\alpha;\deltaAMP)$. Hence,
\[
\psi_2({\alpha},L(\alpha;\delta);\deltaAMP) \le \psi_2({\alpha},F_1^{-1}(\alpha);\deltaAMP),
\]
and thus proving \eqref{Eqn:lem_regionIII_bound_8} reduces to proving
\BE
\psi_2({\alpha},F_1^{-1}(\alpha);\deltaAMP)\le F_1^{-1}(\alpha),\nonumber
\EE
which follows from the same argument as that for \eqref{Eqn:lem_regionIII_bound_3.1}.

\textit{Case (iii):} Lemma~\ref{lem:psi2} (iii) shows that $\psi_2(\alpha;\sigma^2;\delta) \le \frac{4}{\delta}$ for any $\sigma^2\in[0,\sigma^2_{\max}]$.
It is easy to see that $\mathcal{D}_{\alpha}\subset [0,\sigma^2_{\max}]$, and thus
\BE \label{Eqn:lem_regionIII_bound_4}
\max_{ {\sigma}^2\in\mathcal{D}_{\alpha}}\psi_2({\alpha},{\sigma}^2;\delta)\le\frac{4}{\delta}  \le \frac{4}{17}\approx 0.235.
\EE
Further, Lemma~\ref{Lem:2} shows that $F_1^{-1}:[0,1]\mapsto[0,\pi^2/16]$ is monotonically decreasing. Hence,
\BE\label{Eqn:lem_regionIII_bound_5}
F_1^{-1}(\alpha) > F_1^{-1}(\alpha_{\ast})\approx 0.415,
\EE
where the numerical constant is calculated from the closed form formula $F_1^{-1}(\alpha)=\alpha^2\cdot\left[\phi_1^{-1}(\alpha)\right]^2$ (see \eqref{Eqn:Lem3_proof2}) and $\alpha_{\ast}\approx0.5274$ (from \eqref{Eqn:alpha_ast_def}). Comparing \eqref{Eqn:lem_regionIII_bound_4} and \eqref{Eqn:lem_regionIII_bound_5} shows that \eqref{Eqn:lem_regionIII_bound_3} holds in this case.

It only remains to prove \eqref{Eqn:lem_regionIII_bound_6}. We have shown in \eqref{Eqn:psi2_dif_sigma2_fs} that 
\BE\label{Eqn:lem_regionIII_bound_10}
\frac{\partial \psi_2(\alpha,\sigma^2;\delta)}{\partial\sigma^2}=\frac{4}{\delta\alpha}\Bigg(\alpha-\underbrace{ \frac{1}{2\sqrt{1+s^2}} E\left( \frac{1}{1 + s^2} \right)}_{f(s)} \Bigg),
\EE
where $s\Mydef\sigma/\alpha$. Further, we have proved in \eqref{Eqn:lem_psi2_fs_cases} that $f(s)$ is strictly increasing on $[0,s_{\ast})$ and strictly decreasing on $(s_{\ast},\infty)$, where $s_{\ast}$ is defined in \eqref{Eqn:lem_psi2_1}. Hence, when $f(0)=0.5<\alpha<f(s_{\ast})=\alpha_{\ast}$, there exist two solutions to 
\[
\alpha = f(s),
\]
denoted as $s_1(\alpha)$ and $s_2(\alpha)$, respectively. Also, from \eqref{Eqn:lem_regionIII_bound_10} and noting the definition $s=\sigma/\alpha$, we have
\[
\begin{split}
\frac{\partial \psi_2(\alpha,\sigma^2;\delta)}{\partial\sigma^2} &> 0 \Longleftrightarrow \sigma^2\in\left[0,\sigma^2_1(\alpha)\right)\cup\left(\sigma^2_2(\alpha),\infty\right),\\
\frac{\partial \psi_2(\alpha,\sigma^2;\delta)}{\partial\sigma^2} &\le 0 \Longleftrightarrow \sigma^2\in\left[\sigma^2_1(\alpha),\sigma^2_2(\alpha)\right],
\end{split}
\]
where $\sigma^2_1(\alpha)\Mydef \alpha^2 s^2_1(\alpha)$ and $\sigma^2_2(\alpha)\Mydef \alpha^2 s^2_2(\alpha)$.
Hence, for fixed $\alpha$ where $\alpha\in(f(0),f(s_{\ast}))$, $\sigma^2_1(\alpha)$ is a local maximum of $\psi_2$ and $\sigma^2_2(\alpha)$ is a local minimum. Clearly, if
\BE\label{Eqn:lem_regionIII_bound_11}
L(\alpha;\delta) \ge \sigma^2_1(\alpha),
\EE
then the maximum of $\psi_2$ over $\sigma^2\in[L(\alpha;\delta),F_1^{-1}(\alpha)]$ can only happen at either $L(\alpha;\delta)$ or $F_1^{-1}(\alpha)$, which will prove \eqref{Eqn:lem_regionIII_bound_6}. Further, for the degenerate case $\alpha\in(0,f(0))$, $\psi_2$ only has a local minimum, and it is easy to see that \eqref{Eqn:lem_regionIII_bound_6} also holds. Thus, we only need to prove that \eqref{Eqn:lem_regionIII_bound_11} holds when $\delta<17$. This can be proved as follows:
\BE\label{Eqn:lem_regionIII_bound_12}
\begin{split}
\sigma^2_1(\alpha)& \overset{(a)}{\le} s_{\ast}^2\cdot \alpha^2 \overset{(b)}{\le} s_{\ast}^2\cdot \alpha_{\ast}^2,
\end{split}
\EE
where (a) is from the fact that $s_1(\alpha)\le s_{\ast}$ and (b) is from our assumption $\alpha\le\alpha_{\ast}$. On the other hand, since $L(\alpha)$ is a decreasing function of $\alpha$ (see Lemma~\ref{Lem:monotonicity_L}), and thus for $\alpha\le\alpha_{\ast}$ we have
\BE\label{Eqn:lem_regionIII_bound_13}
\begin{split}
L(\alpha;\delta) & \ge L(\alpha_{\ast};\delta) \\
&=\frac{4}{\delta}\left( 1 - \frac{\phi_2^2(\phi_1^{-1}(\alpha_{\ast})) }{4\left[1+(\phi_1^{-1}(\alpha_{\ast}))^2\right]}\right),
\end{split}
\EE
where the last step is from Definition~\ref{Eqn:left_bound}. Based on \eqref{Eqn:lem_regionIII_bound_12} and \eqref{Eqn:lem_regionIII_bound_13}, we see that $L(\alpha;\delta)>\sigma^2_1(\alpha)$ for $\alpha\le\alpha_{\ast}$ if
\[
\delta \le \frac{4}{s_{\ast}^2\cdot\alpha_{\ast}^2}\left( 1 - \frac{\phi_2^2(\phi_1^{-1}(\alpha_{\ast})) }{4\left[1+(\phi_1^{-1}(\alpha_{\ast}))^2\right]}\right)\approx 17.04,
\]
where the numerical constant is calculated based on the definition of $\alpha_{\ast}$ in \eqref{Eqn:monotonicity_condi2}, the definition of $s_{\ast}$ in \eqref{Eqn:lem_psi2_1}, and that of $\phi_1$ and $\phi_2$ in Definition~\ref{Eqn:left_bound}. Hence, the condition $\delta<17$ is enough for our purpose. This concludes our proof.
\end{proof}

\end{enumerate}

Now we turn our attention to the proof of part (i) of Lemma \ref{Lem:regionI_III}.
Suppose that $(\alpha,\sigma^2)\in\mathcal{R}_1\cup\mathcal{R}_{2a}$. Then, using \eqref{Eqn:lem_regionI_III_1} and based on the fact that $F_1(\alpha)$ is a strictly decreasing function, we know that $(\psi_1(\alpha,\sigma^2),\psi_2(\alpha,\sigma^2))\in\mathcal{R}_1\cup\mathcal{R}_2$. (See Definition~\ref{Def:regions}.) Further, Lemma~\ref{Lem:regionII_IV} shows that $(\psi_1(\alpha,\sigma^2),\psi_2(\alpha,\sigma^2))\notin\mathcal{R}_{2b}$. Hence, $(\psi_1(\alpha,\sigma^2),\psi_2(\alpha,\sigma^2))\in\mathcal{R}_1\cup\mathcal{R}_{2a}$. Applying this argument recursively shows that if $(\alpha_{t_0},\sigma_{t_0}^2)\in\mathcal{R}_1\cup\mathcal{R}_{2a}$, then $(\alpha_t,\sigma_t^2)\in\mathcal{R}_1\cup\mathcal{R}_{2a}$ for all $t> t_0$. An illustration of the situation is shown in Fig.~\ref{fig:regions_I&III_bounds}.\\

Now we can discuss the proof of part (ii) of Lemma \ref{Lem:regionI_III}. To proceed, we introduce two auxiliary sequences $\{\tilde{\alpha}_{t+1}\}_{t\ge t_0}$ and $\{\tilde{\sigma}^2_{t+1}\}_{t\ge t_0}$, defined as:
\BE\label{Eqn:lem_regionI_III_1.1}
\tilde{\alpha}_{t+1} = B_1(\alpha_t,\sigma^2_t)\quad\text{and} \quad \tilde{\sigma}^2_{t+1} = B_2(\alpha_t,\sigma^2_t),
\EE
where $B_1$ and $B_2$ are defined in \eqref{Eqn:bounds_def}. Note that the definitions of $B_1(\alpha,\sigma^2)$ and $B_2(\alpha,\sigma^2)$ require $(\alpha,\sigma^2)\in\mathcal{R}_1\cup\mathcal{R}_{2a}$, and such requirement is satisfied here due to part (i) of this lemma. Noting the SE update $\alpha_{t+1}=\psi_1(\alpha_t,\sigma^2_t)$ and $\sigma^2_{t+1}=\psi_2(\alpha_t,\sigma^2_t)$, and recall the inequalities in \eqref{Eqn:lem_regionI_III_1}, we obtain the following:
\BE\label{Eqn:lem_regionI_III_2}
\alpha_{t+1} \ge \tilde{\alpha}_{t+1}  \quad\text{and}\quad \sigma^2_{t+1} \le \tilde{\sigma}^2_{t+1},\quad\forall t\ge t_0.
\EE 
Namely, $\{\tilde{\alpha}_{t+1}\}_{t\ge t_0}$ and $\{\tilde{\sigma}^2_{t+1}\}_{t\ge t_0}$ are ``worse'' than $\{{\alpha}_{t+1}\}_{t\ge t_0}$ and $\{{\sigma}^2_{t+1}\}_{t\ge t_0}$, respectively, at each iteration. We next prove that 
\BE\label{Eqn:lem_regionI_III_a1}
\lim_{t\to\infty} \tilde{\alpha}_{t+1} =1\quad\text{and}\quad\lim_{t\to\infty} \tilde{\sigma}^2_{t+1} =0,
\EE
which together with \eqref{Eqn:lem_regionI_III_2}, and the fact that $\alpha_{t+1}\le1$ and $\sigma_{t+1}>0$ (since $(\alpha_t,\sigma^2_t)\in\mathcal{R}_{2a}$), leads to the results we want to prove:
\[
\lim_{t\to\infty} \alpha_{t+1} =1\quad\text{and}\quad\lim_{t\to\infty} \sigma^2_{t+1} =0.
\]

It remains to prove \eqref{Eqn:lem_regionI_III_a1}. First, notice that $\tilde{\alpha}_{t+1}\le1$ and $\tilde{\sigma}^2_{t+1}\ge0$ ($\forall t\ge t_0$), from the definition in \eqref{Eqn:bounds_def}. We then show that the sequence $\{\tilde{\alpha}_{t+1}\}_{t\ge t_0}$ is monotonically non-decreasing and $\{\tilde{\sigma}^2_{t+1}\}_{t\ge t_0}$ is monotonically non-increasing, namely,
\BE\label{Eqn:lem_regionI_III_2.2}
 \tilde{\alpha}_{t+2}\ge \tilde{\alpha}_{t+1} \quad\text{and}\quad \tilde{\sigma}^2_{t+2}\le \tilde{\sigma}^2_{t+1}, \quad\forall t\ge t_0,
\EE
and equalities of \eqref{Eqn:lem_regionI_III_2.2} hold only when the equalities in \eqref{Eqn:lem_regionI_III_1} hold. Then we can finish the proof by the fact that $\tilde{\alpha}$ and $\tilde{\sigma}^2$ will improve strictly in at most two consecutive iterations and the ratios $\frac{\tilde{\alpha}_{t+2}}{\tilde{\alpha}_t}, \frac{\tilde{\sigma}^2_{t+2}}{\tilde{\sigma}^2_t}$ are continuous functions of $(\alpha_t,\sigma_t^2)$ on $[\tilde{\alpha}_{t_0},1]\times [0,\sigma_{\max}^2]$. (This is essentially due to the fact that equalities in \eqref{Eqn:lem_regionI_III_1} can be achieved when $\sigma^2=F_1^{-1}(\alpha)$, but this cannot happen in two consecutive iterations. See the discussions below \eqref{Eqn:bounds_def}.) 

To prove \eqref{Eqn:lem_regionI_III_2.2}, we only need to prove the following (based on the definition in \eqref{Eqn:lem_regionI_III_1.1})
\BE
B_1\left[ \psi_1,\psi_2 \right] \ge B_1(\alpha,\sigma^2)\quad\text{and}\quad B_2\left[ \psi_1,\psi_2 \right] \le B_2(\alpha,\sigma^2),\quad\forall(\alpha,\sigma^2)\in\mathcal{R}_1\cup\mathcal{R}_{2a},\nonumber
\EE
where $\psi_1$ and $\psi_2$ are shorthands for $\psi_1(\alpha,\sigma^2)$ and $\psi_2(\alpha,\sigma^2;\delta)$. From \eqref{Eqn:bounds_def}, the above inequalities are equivalent to
\BE\label{Eqn:lem_regionI_III_3}
\min\left\{ \psi_1,F_1(\psi_2) \right\}\ge B_1(\alpha,\sigma^2),
\EE
and
\BE\label{Eqn:lem_regionI_III_4}
\max\left\{ \psi_2,F_1^{-1}(\psi_1) \right\}\le B_2(\alpha,\sigma^2).
\EE
Note that \eqref{Eqn:lem_regionI_III_1} already proves the following
\[
\psi_1 \ge B_1(\alpha,\sigma^2)\quad\text{and}\quad\psi_2 \le B_2(\alpha,\sigma^2).
\]
Hence, to prove \eqref{Eqn:lem_regionI_III_3} and \eqref{Eqn:lem_regionI_III_4}, we only need to prove
\[
F_1(\psi_2) \ge B_1(\alpha,\sigma^2)\quad\text{and}\quad F_1^{-1}(\psi_1)\le  B_2(\alpha,\sigma^2).
\] 
To prove $F_1(\psi_2) \ge B_1(\alpha,\sigma^2)$, we note that
\[
\begin{split}
\psi_2 &\overset{(a)}{\le} B_2(\alpha,\sigma^2)\\
&\overset{(b)}{=} \max\left\{ \sigma^2,F_1^{-1}(\alpha) \right\}\\
&\overset{(c)}{=} F_1^{-1}\left(\min\left\{F_1( \sigma^2),\alpha \right\}\right)\\
&\overset{(d)}{=} F_1^{-1}\left(B_1(\alpha,\sigma^2)\right),
\end{split}
\]
where (a) is from \eqref{Eqn:lem_regionI_III_1b}, (b) is from \eqref{Eqn:bounds_def}, and (c) is due to the fact that $F_1^{-1}$ is strictly decreasing, and (d) from \eqref{Eqn:lem_regionI_III_1}. Hence, since $F_1$ is strictly decreasing, we have 
\[
F_1(\psi_2)\ge F_1\left[F_1^{-1}\left(B_1(\alpha,\sigma^2)\right)\right]=B_1(\alpha,\sigma^2).
\]
Further, it is straightforward to see that if both inequalities are strict in \eqref{Eqn:lem_regionI_III_1} then
 \[
 \min\left\{ \psi_1,F_1(\psi_2) \right\} > B_1(\alpha,\sigma^2).
 \]
This shows that equalities of \eqref{Eqn:lem_regionI_III_2.2} hold only when the equalities in \eqref{Eqn:lem_regionI_III_1} hold.

The proof for $F_1^{-1}(\psi_1)\le  B_2(\alpha,\sigma^2)$ is similar and omitted. 
\end{itemize}


\subsubsection{Proof of Lemma~\ref{Lem:regionII_IV}} \label{ssec:proofLemma8} 

Suppose that $(\alpha,\sigma^2)\in \mathcal{R}_0$. From Definition ~\ref{Def:regions}, we have
\BE\label{Eqn:lem_regionII_1}
\frac{\pi^2}{16}<\sigma^2\le\sigma_{\max}^2.
\EE
Further, $F_1^{-1}$ is monotonically decreasing and hence (for $\delta>\deltaAMP$)
\BE\label{Eqn:lem_regionII_2}
\frac{\pi^2}{16}=F_1^{-1}(0)>F_1^{-1}(\alpha) \ge F_2({\alpha};\delta),
\EE
where the last inequality is due to Lemma~\ref{Lem:F1_F2_complex}. Combining \eqref{Eqn:lem_regionII_1} and \eqref{Eqn:lem_regionII_2} yields
\BE\label{Eqn:lem_regionII_3}
F_2({\alpha};\delta) < \sigma^2 \le\sigma^2_{\max}.
\EE
By the global attractiveness property in Lemma~\ref{lem:psi2} (iv), \eqref{Eqn:lem_regionII_3} implies
\BE
\psi_2(\alpha;\sigma^2;\delta)<\sigma^2.\nonumber
\EE

From the above analysis, we see that as long as $\frac{\pi^2}{16}<\sigma^2_t\le \sigma^2_{\max}$ (and also $0<\alpha_t<1$), $\sigma^2_{t+1}$ will be strictly smaller than $\sigma^2_t$:
\[
\sigma^2_{t+1}= \psi_2(\alpha_t;\sigma_t^2;\delta)<\sigma^2_t.
\]
Hence, there exists a finite number $T\ge1$ such that 
\[
\sigma^2_{T-1}>\frac{\pi^2}{16}\quad\text{and}\quad\sigma^2_{T}\le\frac{\pi^2}{16}.
\]
Otherwise, $\sigma_t^2$ will converge to a $\bar{\sigma}^2$ in $\mathcal{R}_0$. This implies that $\bar{\sigma}^2$ is a fixed point of $\psi_2$ for certain value of $0 < \alpha \leq 1$. However, we know from part (i) of Lemma \ref{Lem:2} and Lemma \ref{Lem:F1_F2_complex} that this cannot happen. 

Based on a similar argument, we also have $\psi_1(\alpha;\sigma^2)<\alpha$ and so $\alpha_{t+1}<\alpha_t$ for $t\le T-1$. Further, we can show that $\alpha_t>0$ (i.e., $\alpha_t\neq0$) for all $0\le t\le T$. First, $\alpha_0>0$ follows from our assumption. Further, from  \eqref{Eqn:map_expression_complex_a} we see that $\alpha_{t+1}>0$ if $\alpha_{t}>0$. Then, using a simple induction argument we prove that $\alpha_t>0$ for all $0\le t\le T$. Putting things together, we showed that there exists a finite number $T\ge1$ such that
\[
0<\alpha_T\le1 \quad\text{and}\quad\sigma^2_{T}\le\frac{\pi^2}{16}.
\]
(Recall that we have proved in Lemma~\ref{Lem:region_bound} that $\alpha_T\le1$.)
From Definition ~\ref{Def:regions}, $(\alpha_T,\sigma^2_T)\in\mathcal{R}_1\cup\mathcal{R}_2$.


\subsection{Proof of Theorem~\ref{Lem:fixed_point}} \label{Sec:proof_global_complex}

We consider the two different cases separately: (1) $\delta>\deltaGlobal$ and (2) $\delta<\deltaGlobal$.

\subsubsection{Case $\delta>\deltaGlobal$}

In this section, we will prove that when $\delta>\deltaGlobal$ the state evolution converges to the fixed point $(\alpha,\sigma^2)=(1,0)$ if initialized close enough to the fixed point. We first prove the following lemma, which shows that $F_1^{-1}$ is larger than $F_2(\alpha;\delta)$ for $\alpha$ close to one.

\begin{lemma}\label{Lem:F1_F2_complex_local}
Suppose that $\delta>\deltaGlobal=2$. Then, there exists an $\epsilon>0$ such that the following holds:
\BE\label{Eqn:F1_F2_complex_local}
F_1^{-1}(\alpha)>F_2(\alpha;\delta),\quad\forall \alpha\in(1-\epsilon,1).
\EE
\end{lemma}
\begin{proof}
In Lemma \ref{Lem:F1_F2_complex}, we proved that $F_1^{-1}(\alpha)>F_2(\alpha;\delta)$ holds for all $\alpha\in(0,1)$ when $\delta>\deltaAMP\approx2.5$. Here, we will prove that $F_1^{-1}(\alpha)>F_2(\alpha;\delta)$ holds for $\alpha$ close to 1 when $\delta>\deltaGlobal=2$. Similar to the manipulations given in Section \ref{proof:lemmadominationF_1}, the inequality \eqref{Eqn:F1_F2_complex_local} can be re-parameterized into the following: 
\BE\label{Eqn:F1_F2_complex_local2}
     \int_0^{\frac{\pi}{2}} \frac{\sin^2 \theta}{(\sin^2 \theta +
     s^2)^{\frac{1}{2}}} d \theta \cdot \int_0^{\frac{\pi}{2}} \frac{(1 -
     \gamma^{} s^2) \sin^2 \theta + s^2}{(\sin^2 \theta + s^2)^{\frac{1}{2}}}
     d \theta >  1, \quad \forall s \in (0,\xi),
\EE
where $\gamma\Mydef 1-\delta/4$ and
$\xi=\phi_1^{-1}(\epsilon)$ (see \eqref{Eqn:Lem_F1F2_phi} for the definition of $\phi_1$). Again, it is more convenient to express \eqref{Eqn:F1_F2_complex_local2} using elliptic integrals (cf. \eqref{Eqn:Lem3_proof7.3})
\BE\label{Eqn:F1_F2_complex_local3}
\frac{E(x)T(x)}{x}-\frac{\gamma (1-x)T^2(x)}{x^2} > 1,\quad \forall x\in\left( \frac{1}{1+\xi},1 \right),
\EE
where we made a variable change $x\Mydef 1/(1+s^2)$. To this end, we can verify that
\[
\lim_{x\to1}\frac{E(x)T(x)}{x}-\frac{\gamma (1-x)T^2(x)}{x^2}=1.
\]
To complete the proof, we only need to show that the derivative of the LHS of \eqref{Eqn:F1_F2_complex_local3} in a small neighborhood of $x=1$ is strictly negative when $\delta>\deltaGlobal=2$. Using the formulas listed in Section \ref{ssec:ellipticintegrals}, we can derive the following:
\[
\begin{split}
&\frac{\mr{d}}{\mr{d}x}\left(\frac{E(x)T(x)}{x}-\frac{\gamma (1-x)T^2(x)}{x^2}\right)\Big|_{x\to1}\\
&=\frac{2\gamma (x-4)E(x)\cdot(1-x)K(x)+[4\gamma(1-x)+x]\cdot (1-x)K^2(x)+[2\gamma(2-x)-x]E^2(x)}{2x^3}\Big|_{x\to1}\\
&=\gamma-\frac{1}{2},
\end{split}
\]
where the last step is due to the facts that $E(x)=1$ and $\lim_{x\to1}(1-x)K(x)=0$. See Section \ref{ssec:ellipticintegrals} for more details. Hence, the above derivative is negative if $\gamma<\frac{1}{2}$ or $\delta>2$ by noting the definition $\gamma=1-\delta/4$. 
\end{proof}

\begin{figure}[!htbp]
\centering
\includegraphics[width=.50\textwidth]{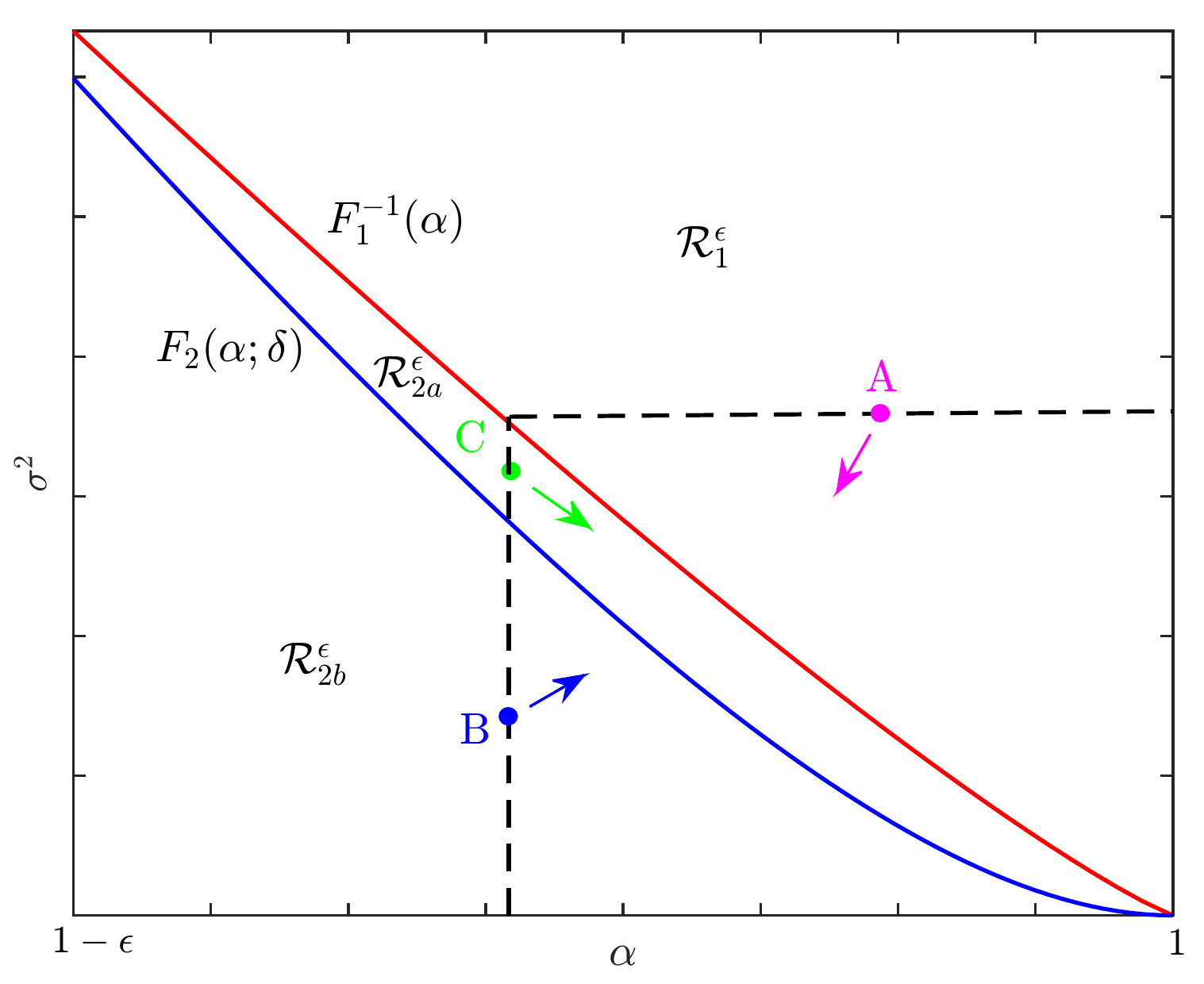}
\caption{Illustration of the local convergence behavior when $\delta>\deltaGlobal$. For all the three points shown in the figure, $B_1$ and $B_2$ are given by the dashed lines. }
\label{Fig:local_convergence}
\end{figure}

We now turn to the proof of Lemma~\ref{Lem:fixed_point}. The idea of the proof is similar to that of Theorem \ref{Theo:PhaseTransition_complex}. There are some differences though, since now $\delta$ can be smaller than $\deltaAMP$ and some results in the proof of Theorem \ref{Theo:PhaseTransition_complex} do not hold for the case considered here. On the other hand, as we focus on the range $\alpha\in(1-\epsilon,1)>\alpha_{\ast}$, and under this condition we know that $F_2(\sigma^2;\delta)$ is strongly globally attracting (see Lemma \ref{lem:psi2}-(v)), which means that $\psi_2(\alpha,\sigma^2)$ moves towards the fixed point $F_2(\alpha;\delta)$, but cannot move to the other side of $F_2(\alpha;\delta)$.

We continue to prove the local convergence of the state evolution. We divide the region $\mathcal{R}^{\epsilon}\Mydef \{(\alpha,\sigma^2)|1-\epsilon\le\alpha\le1,0\le\sigma^2\le F_1^{-1}(1-\epsilon)\}$ into the following sub-regions:
\BE\label{Eqn:regions_local}
\begin{split}
\mathcal{R}^{\epsilon}_1&\Mydef\left\{(\alpha,\sigma^2)\big|1-\epsilon\le\alpha\le1,F_1^{-1}(\alpha)< \sigma^2 \le F_1^{-1}(1-\epsilon)\right\},\\
\mathcal{R}^{\epsilon}_{2a} &\Mydef\left\{(\alpha,\sigma^2)\big|1-\epsilon\le\alpha\le1,F_2(\alpha;\delta)<\sigma^2 \le F_1^{-1}(\alpha)\right\}\\
 \mathcal{R}^{\epsilon}_{2b} &\Mydef\left\{(\alpha,\sigma^2)\big|1-\epsilon\le\alpha\le1,0 \le \sigma^2 \le F_2(\alpha;\delta)\right\}.
\end{split}
\EE
Similar to the proof of Lemma~\ref{Lem:regionI_III} discussed in Section \ref{Sec:proof_lem_regionI_III}, we will show that if $(\alpha,\sigma^2)\in\mathcal{R}^{\epsilon}$ then the new states $(\psi_1,\psi_2)$ can be bounded as follows:
\BE\label{Eqn:psi_bounds_local}
\psi_1(\alpha,\sigma^2)\ge B_1(\alpha,\sigma^2)\quad\text{and}\quad \psi_2(\alpha,\sigma^2)\le B_2(\alpha,\sigma^2),\quad\forall(\alpha,\sigma^2)\in\mathcal{R}^{\epsilon},
\EE
where 
\[ 
B_1(\alpha,\sigma^2)=\min\left\{\alpha,F_1(\sigma^2) \right\}\quad\text{and}\quad B_2(\alpha,\sigma^2)=\max\left\{\sigma^2,F_1^{-1}(\alpha) \right\}.
\]
Based on the strong global attractiveness of $\psi_1$ (Lemma \ref{lem:psi1}-iii) and $\psi_2$ (Lemma \ref{lem:psi2}-v) and the additional result \eqref{Eqn:psi2_contracting}, it is straightforward to show the following:
\[
\begin{split}
\psi_1(\alpha,\sigma^2)\ge F_1(\sigma^2)&\quad\text{and}\quad \psi_2(\alpha,\sigma^2)\le \sigma^2,\quad\forall(\alpha,\sigma^2)\in\mathcal{R}^{\epsilon}_1,\\
\psi_1(\alpha,\sigma^2)\ge \alpha&\quad\text{and}\quad \psi_2(\alpha,\sigma^2)\le \sigma^2,\quad\forall(\alpha,\sigma^2)\in\mathcal{R}^{\epsilon}_{2a},\\
\psi_1(\alpha,\sigma^2)\ge \alpha&\quad\text{and}\quad \psi_2(\alpha,\sigma^2)\le F_2(\alpha;\delta),\quad\forall(\alpha,\sigma^2)\in\mathcal{R}^{\epsilon}_{2b},
\end{split}
\]
which, together with the definitions given in \eqref{Eqn:regions_local} and the fact that $F_2(\alpha;\delta)<F_1^{-1}(\alpha)$ (cf. Lemma \ref{Lem:F1_F2_complex_local}), proves \eqref{Eqn:psi_bounds_local}. The rest of the proof follows that in Section \ref{Sec:proof_lem_regionI_III}. Namely, we construct two auxiliary sequences $\{\tilde{\alpha}_{t+1}\}$ and $\{\tilde{\sigma}^2_{t+1}\}$ where
\[
\tilde{\alpha}_{t+1}=B_1(\alpha_t,\sigma^2_t)\quad\text{and}\quad \tilde{\sigma}^2_{t+1}=B_2(\alpha_t,\sigma^2_t),
\]
and show that $\{\tilde{\alpha}_{t+1}\}$ and $\{\tilde{\sigma}^2_{t+1}\}$ monotonically converge to $1$ and $0$ respectively. The detailed arguments can be found in Section \ref{Sec:proof_lem_regionI_III} and will not be repeated here.

\subsubsection{Case $\delta<\deltaGlobal$}\label{Sec:complex_local_converse}
We proved in \eqref{Eqn:psi2_dif_sigma2_fs} that 
\[
\begin{split}
\frac{\partial \psi_2(\alpha,\sigma^2;\delta)}{\partial\sigma^2} &= 
\frac{4}{\delta\alpha}\Bigg(\alpha-\underbrace{ \frac{1}{2\sqrt{1+s^2}} E\left( \frac{1}{1 + s^2} \right)}_{f(s)} \Bigg),
\end{split}
\]
where $s = \frac{\sigma}{\alpha}$. Hence, we have (note that $E(1)=1$)
\BE\label{Eqn:psi2_dif_local}
\partial_2 \psi_2(\alpha,0)\Mydef \frac{\partial\psi_2(\alpha,\sigma^2) }{\partial \sigma^2}\Big|_{\sigma^2=0}=\frac{4}{\delta}\left(1-\frac{1}{2\alpha}\right),\quad \forall \alpha>0.
\EE
Therefore,
\[
\partial_2 \psi_2(\alpha,0)>1,\quad\forall \alpha>\frac{2}{4-\delta}.
\]
When $\delta<\deltaGlobal=2$, we have $\frac{2}{4-\delta}<1$ and therefore there exists a constant $\alpha^{\ast}$ that satisfies the following:
\[
\frac{2}{4-\delta}<\alpha^{\ast}<1,
\]
which together with \eqref{Eqn:psi2_dif_local} yields
\[
\partial_2 \psi_2(\alpha^{\ast},0) >1.
\]
Further, as discussed in the proof of Lemma \ref{lem:psi2}-(i), $\partial_2 \psi_2(\alpha^{\ast},\sigma^2)$ is a continuous function of $\sigma^2$. Hence, there exists $\xi^{\ast}>0$ such that
\BE\label{Eqn:psi2_dif_local2}
\partial_2 \psi_2(\alpha^{\ast},\sigma^2) >1,\quad \forall \sigma^2\in[0,\xi^{\ast}].
\EE
Further, we have shown in \eqref{Eqn:psi2_dif_sigma2} that
\[
\frac{\partial \psi_2(\alpha,\sigma^2;\delta)}{\partial\sigma^2}=
\frac{4}{\delta}\left( 1-\frac{1}{2}\int_0^{\frac{\pi}{2}} \frac{\sigma^2}{(\alpha^2\sin^2\theta+\sigma^2)^{\frac{3}{2}}} \mr{d}\theta\right),
\]
and it is easy to see that $\partial_2 \psi_2(\alpha,\sigma^2;\delta)$ is an increasing function of $\alpha\in(0,\infty)$. Hence, together with \eqref{Eqn:psi2_dif_local2} we get the following
\[
\partial_2 \psi_2(\alpha,\sigma^2;\delta)>1,\quad \forall (\alpha,\sigma^2)\in[\alpha^{\ast},1]\times[0,\xi^{\ast}],
\]
which means that $\psi_2(\alpha,\sigma^2)-\sigma^2$ is a strictly increasing function of $\sigma^2$ for $(\alpha,\sigma^2)\in[\alpha^{\ast},1]\times[0,\xi^{\ast}]$. Hence,
\[
\psi_2(\alpha,\sigma^2)-\sigma^2>\psi_2(\alpha,0)=\frac{4}{\delta}(1-\alpha)^2\ge0,\quad \forall (\alpha,\sigma^2)\in[\alpha^{\ast},1]\times[0,\xi^{\ast}].
\]
This implies that $\sigma^2$ moves away from $0$ in a neighborhood of the fixed point $(1,0)$.



\bibliographystyle{unsrt}	
\bibliography{Phase_retrieval}		

\appendix

\section{Derivations of AMP.A}\label{Sec:AMP_derivations}
For the convenience of the readers (especially those who are not familiar with AMP), we provide a sketch of the derivations of the AMP.A algorithm in this appendix. Our derivations follow the approach proposed in \cite{Rangan11}. However, there are some differences specially in the last steps of our derivation.

For simplicity, we focus on the real-valued case. Consider the following optimization problem:
\BE\label{Eqn:penalize}
\underset{\bm{x}}{\min}\quad  \sum_{a=1}^m \left(y_a-|(\bm{Ax})_a|\right)^2+\frac{\mu}{2} \|\bm{x}\|_2^2,
\EE
where $\mu$ is a penalization parameter.  We now sketch the derivations of the AMP.A algorithm intended for solving \eqref{Eqn:penalize}. First, we construct the following joint pdf for \eqref{Eqn:penalize}:
\BE\label{Eqn:factor_penalization}
\begin{split}
\ell(\bm{x})&=\frac{1}{Z}\prod_{a=1}^m \exp\left[ -\beta\left(y_a - | (\bm{Ax})_a |\right)^2 \right]\cdot \prod_{i=1}^n\exp\left(-\beta\cdot \frac{\mu}{2}x_i^2\right),
\end{split}
\EE
where $Z$ is a normalizing constant, $(\bm{Ax})_a$ and $y_a$ denote the $a$-th entries of $\bm{Ax}$ and $\bm{y}$, and $\beta>0$ is parameter (the inverse temperature). Define
\BE\label{Eqn:factors_def}
\begin{split} 
f(y,z) &= \exp\left( -\beta\left(y-|z|\right)^2 \right).
\end{split}
\EE  
Following \cite[Chapter 5]{ArianThesis}, we proceed in three steps:
\begin{itemize}
\item Derive the sum-product belief propagation (BP) algorithm for \eqref{Eqn:factor_penalization}.
\item Approximate the BP update rules.
\item Find the message update rules in the limit of $\beta\to\infty$.
\end{itemize}
The above procedure is slightly different from the original derivations in \cite{Rangan11} (which is derived directly from the max-sum belief propagation algorithm) but equivalent. The sum-product BP algorithm reads
\BS\label{Eqn:BP}
\begin{align}
\hat{m}_{a\to i}^t(x_i)&\simeq\int_{\bm{x}\backslash i} f(y_a,(\bm{Ax})_a) \prod_{j\neq i} \mr{d}m_{j\to a}^t(x_j), \label{Eqn:BP_a}\\
m_{i \to a}^{t + 1}\left( {{x_i}} \right) &\simeq \prod\limits_{b \ne a}^{} {\hat m_{b \to i}^t\left( {{x_i}} \right)}\cdot \exp \left( -\beta\cdot \frac{\mu}{2}x_i^2 \right).\label{Eqn:BP_b}
\end{align}
\ES

We next simplify the above BP update rules.
\subsection{Messages from factor nodes to variable nodes}
Let $x_{j\to a}^t$ and $v_{j\to a}^t/\beta$ be the mean and variance of the incoming message $m_{j\to a}^t$ (here $v_{j\to a}^t$ is $O(1)$ and the variance of $m_{j\to a}^t$ is $O(1/\beta)$ as $\beta\to\infty$ \cite{ArianThesis}). Note that the calculation of the message $\hat{m}^t_{a \rightarrow i} (x_i)$ in \eqref{Eqn:BP_a} can be interpreted as the expectation of $ f(y_a,(\bm{Ax})_a)$ with respect to random vector $\bm{x}\backslash i$ that has product measure $\prod_{j\neq i} \mr{d}m_{j\to a}^t(x_j)$. Since in this interpretation the elements of $\bm{x}\backslash i$ are independent, based on a heuristic central limit theorem argument, we assume that $Z_a\Mydef (\bm{Ax})_a$ is Gaussian distributed, with mean and variance respectively given by \cite[Chapter 5.2]{ArianThesis}
\BE\label{Eqn:s_a_def}
\begin{split}
s_a^t&\Mydef\sum_{j\neq i}A_{aj}x_{j\to a}^t+A_{ai}x_i,\\
&=A_{ai}(x_i-x_{i\to a}^t)+\underbrace{ \sum\limits_{j = 1}^n {{A_{aj}}x_{j \to a}^t}}_{p_a^t},\\
\frac{\tau_a^t}{\beta}&\Mydef \frac{1}{\beta}\sum_{j\neq i}A_{aj}^2 v_{j\to a}^t\approx \frac{1}{\beta}\sum_{j=1}A_{aj}^2 v_{j\to a}^t.
\end{split}
\EE
Based on this approximation, the message $\hat{m}_{a\to i}^t(x_i)$ in \eqref{Eqn:BP_a} can be expressed as follows
\BE\label{Eqn:a_to_i}
\begin{split}
\hat{m}_{a\to i}^t(x_i)
&= \mathbb{E}\left\{ \exp\left[-\beta(y_a-|Z_a|)^2\right] \right\} \\
&=\int \exp\left[-\beta(y_a-|z|)^2\right] \cdot \mathcal{N} \left(z; A_{ai}(x_i-x_{i\to a}^t)+p_a^t , 
\tau_a^t/\beta\right)\mr{d}z,
\end{split}
\EE
where the expectation in step (a) is over $Z_a=(\bm{Ax})_a$ (with respect to the product distribution $\prod_{j\neq i} \mr{d}m_{j\to a}^t(x_j)$). Following \cite{Rangan11}, we define
\BE\label{Eqn:def_score}
H\left( {p,y,v /\beta } \right) \buildrel \Delta \over = \log \left[ {\int \exp\big( -\beta(y-|z|)^2 \big)  \cdot \mathcal{N}\left( {z;p,v /\beta } \right){\rm{d}}z} \right].
\EE
Using this definition, we can write $\log \left[ {\hat m_{a \to i}^t\left( {{x_i}} \right)} \right]$ in \eqref{Eqn:a_to_i} as
\[
\log \left[ {\hat m_{a \to i}^t\left( {{x_i}} \right)} \right] = H\left( {{A_{ai}}\left( {{x_i} - x_{i \to a}^t} \right) + p_a^t,{\tau}_a^t/\beta} \right).
\]
Noting $A_{ai}=O_p\left(\frac{1}{\sqrt{n}}\right)$, following \cite{Rangan11} we apply a second order Taylor expansion to $\log\left[\hat{m}_{a\to i}^t(x_i)\right]$ (amounts to a Gaussian approximation of ${\hat m_{a \to i}^t(x_i)} $) : 
\BS\label{Eqn:Taylor}
\begin{align}
H\left( {{A_{ai}}\left( {{x_i} - x_{i \to a}^t} \right) + p_a^t,{y_a},{\tau}_a^t/\beta} \right)& \approx H_a(t) + {A_{ai}}\left( {{x_i} - x_{i \to a}^t} \right)H'_a(t) + \frac{1}{2}A_{ai}^2{\left( {{x_i} - x_{i \to a}^t} \right)^2}H''_a(t)\\
&= \frac{1}{2}A_{ai}^2H''_a(t) x_i^2+ \left[ {{A_{ai}}H'_a(t) - A_{ai}^2x_{i \to a}^tH''_a(t)} \right] x_i+ {\rm{const}},
\end{align}
\ES
where we have omitted constant terms (relative to $x_i$), and $H_a(t)$, $H'_a(t)$ and $H''_a(t)$ are short-hands for
\[
\begin{split}
H_a(t)&= H(p_a^t,y_a,\tau_a^t/\beta),\\
H'_a(t) &=\frac{\partial H(p,y,\tau/\beta)}{\partial p}\big|_{p=p_a^t,y=y_a,\tau=\tau_a^t}\\
H''_a(t) &=\frac{\partial ^2H(p,y,\tau/\beta)}{\partial p^2}\big|_{p=p_a^t,y=y_a,\tau=\tau_a^t}.
\end{split}
\]
\subsection{Messages from variable nodes to factor nodes}
The message from $x_i$ to $F_a$ is 
\BE\label{Eqn:V2F}
m_{i \to a}^{t + 1}\left( {{x_i}} \right) \simeq \prod\limits_{b \ne a}^{} {\hat m_{b \to i}^t\left( {{x_i}} \right)} \cdot \exp \left( -\beta\cdot \frac{\mu}{2}x_i^2 \right).
\EE
From the Gaussian approximation in \eqref{Eqn:Taylor}, $m_{i \to a}^{t + 1}\left( {{x_i}} \right)$ is also Gaussian. Consider the following term:
\BE \label{Eqn:V2F_2}
\begin{split}
\log\left[m_{i \to a}^{t + 1}\right] &  \simeq \sum_{b\neq a} \log \left[ \hat{m}_{b\to i}^t(x_i)\right]-\frac{\mu\beta}{2} x_i^2 \\
&\approx \frac{1}{2}\left( \sum\limits_{b \ne a} {A_{bi}^2  H''_b(t) } -\beta\mu \right)x_i^2 +   \left( {\sum\limits_{b \ne a} {{A_{bi}} H'_b(t) }  - \sum\limits_{b \ne a} {A_{bi}^2} H''_b(t)  x_{i \to b}^t} \right){x_i},
\end{split}
\EE
where the second approximation comes from \eqref{Eqn:Taylor}. Comparing \eqref{Eqn:V2F_2} with the exponent of a Gaussian pdf, we find that its variance (which we denote by $v_{i\to a}^{t+1}/\beta$) and mean are respectively given by
\BE\label{Eqn:belief_var}
\frac{{v_{i \to a}^{t + 1}}}{\beta }  = \frac{1}{ - \sum\limits_{b \ne a} {A_{bi}^2 H''_b(t)}  +\beta\mu}
= \underbrace{\frac{1}{{ - \sum\limits_{b = 1}^m {A_{bi}^2 \cdot H''_b(t)} } +\beta\mu}}_{v_i^{t+1}/\beta}+O_p\left( \frac{1}{n} \right),
\EE
and
\BE \label{Eqn:V2F_3}
\begin{split}
x_{i \to a}^{t + 1}& = \frac{{v_i^{t + 1}}}{\beta } \cdot \left( {\sum\limits_{b \ne a} {{A_{bi}} \cdot H'_b(t)}  - \sum\limits_{b \ne a} {A_{bi}^2 \cdot H''_b(t) \cdot x_{i \to b}^t} } \right).
\end{split}
\EE
The approximation in \eqref{Eqn:belief_var} is due to our assumption $\mathbb{E}[A_{bi}^2]=1/m$. In \eqref{Eqn:V2F_3}, we have approximated $v_{i\to a}^{t+1}$ by $v_i^{t+1}$ and omit the $O_p(1/n)$ error term.
\subsection{From BP to AMP}
We assume that the message $x_{i \to a}^{t + 1}$ has the following structure \cite[Chapter 5.2.4]{ArianThesis}:
\[
x_{i \to a}^{t + 1} = x_i^{t + 1} + \delta x_{i \to a}^{t + 1} + O_p\left( \frac{1}{n}\right),
\]
where $ x_i^{t + 1}=O_p(1)$ and $\delta x_{i \to a}^{t + 1} \sim O_p\left( {1/\sqrt n } \right)$. From \eqref{Eqn:V2F_3}, we can identify $x_i^{t + 1}$ and $\delta x_{i \to a}^{t + 1}$ (which is the term that depends on the index $a$) to be the following
\BS\label{Eqn:x_i2a}
\begin{align}
x_{i \to a}^{t + 1} &= \underbrace {\frac{{v_i^{t + 1}}}{\beta } \cdot \left( {\sum\limits_{b = 1}^m {{A_{bi}}  \cdot H'_b(t)}  - \sum\limits_{b = 1}^m {A_{bi}^2 \cdot H''_b(t) \cdot x_{i \to b}^t} } \right)}_{x_i^{t + 1}}\\
&\underbrace { - \frac{{v_i^{t + 1}}}{\beta } \cdot {A_{ai}} \cdot H'_a(t)}_{\delta x_{i \to a}^{t + 1}} + \underbrace {\frac{{v_i^{t + 1}}}{\beta } \cdot A_{ai}^2 \cdot H''_a(t) \cdot x_{i \to a}^t}_{O_p\left( {1/n} \right)}.
\end{align}
\ES
We further simplify $x_i^{t+1}$ (i.e., the first term in the above equation) as follows
\BS \label{Eqn:belief_mean}
\begin{align}
x_i^{t + 1} &= \frac{{v_i^{t + 1}}}{\beta } \cdot \left[ {\sum\limits_{b = 1}^m {{A_{bi}} \cdot H'_b(t)}  - \sum\limits_{b = 1}^m {A_{bi}^2 \cdot H''_b(t) \cdot x_{i \to b}^t} } \right]\\
&= \frac{{v_i^{t + 1}}}{\beta } \cdot \left[ {\sum\limits_{b = 1}^m {{A_{bi}} \cdot H'_b(t)}  - \sum\limits_{b = 1}^m {A_{bi}^2 \cdot H''_b(t)\cdot x_i^t} } \right] + O_p\left( {\frac{1}{{n }}} \right).
\end{align}
\ES
The approximation error in the above is $O_p(1/n)$ since
\[
\begin{split}
\sum_{b=1}^m A_{bi}^2\cdot H''_b(t)\cdot \delta x_{i\to b}^t&=-\frac{v_i^{t}}{\beta}\sum_{b=1}^m A_{bi}^3\cdot H''_b(t)\cdot H'_b(t)=O_p\left(\frac{1}{n}\right),
\end{split}
\]
where we used $\delta x_{i\to b}^t=-v_i^{t+1}/\beta\cdot A_{bi} H'_b(t)$ in the previous equation. 
Ignoring the $O_p(1/n)$ term, the update in \eqref{Eqn:belief_mean} becomes
\BE\label{Eqn:x_i_t+1}
\begin{split}
x_i^{t + 1} &= \frac{{v_i^{t + 1}}}{\beta } \cdot \sum\limits_{b = 1}^m {{A_{bi}}  H'_b(t)}  + \frac{v_i^{t+1}}{\beta}  \left(- \sum_{b=1}^m A_{bi}^2 H_b''(t) \right) \cdot x_i^t\\
&= \frac{v_i^{t+1} }{ \beta}\cdot\left( -\sum_{b=1}^m A_{bi}^2H_b''(t) \right)\cdot \left( \frac{\sum_{b=1}^m A_{bi}H_b'(t)}{ - \sum_{b=1}^m A_{bi}^2 H_b''(t)  }    +  x_i^t \right).
\end{split}
\EE

We now return to the update of $p_a^{t+1}$ defined in \eqref{Eqn:s_a_def}:
\BE\label{Eqn:p_update}
\begin{split}
p_a^{t+1} &\Mydef \sum_{j=1}^n A_{aj} x_{j\to a}^{t+1} \\
&\overset{(a)}{\approx}  \sum_{j=1}^n A_{aj}\left( x_j^{t+1} -\frac{v_j^{t+1}}{\beta}\cdot A_{aj}\cdot H_a'(t) \right)\\
&=\Big(\sum_{j=1}^n A_{aj}x_j^{t+1}\Big) - \frac{\left(\sum_{j=1}A_{aj}^2 v_j^{t+1}\right)}{\beta}\cdot H_a'(t)\\
&\overset{(b)}{=}\Big(\sum\limits_{j = 1}^n {{A_{aj}}x_j^{t + 1}}\Big)  - \frac{{\tau _a^{t + 1}}}{\beta } \cdot H'_a(t)
\end{split}
\EE
where step (a) is due to \eqref{Eqn:x_i2a} and step (b) is from the definition in \eqref{Eqn:s_a_def}.
\subsection{Large $\beta$ Limit}\label{Sec:zero_T}
Putting \eqref{Eqn:s_a_def}, \eqref{Eqn:p_update}, \eqref{Eqn:belief_var}, \eqref{Eqn:x_i_t+1}, we obtain the following simplified BP update rules ($\forall a=1,\ldots,m$ and $\forall i=1,\ldots,n$):
\BS\label{Eqn:AMP_summary}
\begin{align}
\tau_a^t & =  \sum\limits_{j = 1}^n {A_{aj}^2v_j^{t}}, \label{Eqn:AMP_summary_a}\\
p_a^{t } &= \sum\limits_{j = 1}^n {{A_{aj}}x_j^{t }}  -\frac{\tau_a^t}{\beta}\cdot H' _a(t-1), \\
v_i^{t+1}  &=\frac{\beta}{{ - \sum\limits_{b = 1}^m {A_{bi}^2 \cdot H''_b(t)} } +\beta\mu},  \label{Eqn:AMP_summary_c}\\
x_i^{t + 1} &= \frac{v_i^{t+1}}{\beta}\cdot\left(-\sum_{b=1}^m A_{bi}^2 H_b''(t)\right)\cdot \left( x_i^t + \frac{\sum_{b=1}^m A_{bi} H_b'(t) }{-\sum_{b=1}^m A_{bi}^2 H_b''(t) }  \right),
\end{align}
\ES
where $H'_b(t)$ and $H''_b(t)$ are shorthands for $H'(p_b^t,y_b,\tau_b^t/\beta)$ and $H''(p_b^t,y_b,\tau_b^t/\beta)$ respectively. The algorithm summarized above is a special form the generalized AMP (GAMP) algorithm derived in \cite{Rangan11} (see Algorithm 1).

We further approximate the variance updates in \eqref{Eqn:AMP_summary_a} and \eqref{Eqn:AMP_summary_c} by averaging over $\bm{A}$ (based on some heuristic concentration arguments). After this approximation, $\tau_a^t$ becomes invariant to the index $a$ (denoted as $\tau^t$ below). We can then write \eqref{Eqn:AMP_summary} into the following vector form:
\BE\label{Eqn:AMP_summary_2}
\begin{split}
\tau^{t} &=\frac{1}{\delta}\frac{1}{ - \mr{div}_p (\hat{g}_{t-1}) } \cdot \lambda_{t-1},\\
\bm{p}^t &=\bm{A}\bm{x}^t  -\frac{1}{\delta} \frac{ \hat{g}\left( \bm{p}^{t-1},\bm{y},\tau^{t-1}/\beta \right) }{- \mr{div}_p (\hat{g}_{t-1}) }\cdot \lambda_{t-1} ,\\
\bm{x}^{t+1} &= \lambda_t\cdot\left(\bm{x}^t + \frac{\bm{A}^\UT  \hat{g}\left( \bm{p}^{t},\bm{y},\tau^{t}/\beta \right) }{- \mr{div}_p (\hat{g}_{t})}\right),
\end{split}
\EE
where we defined:
\[
\begin{split}
\hat{g}(p,y,\tau/\beta) &\Mydef \frac{H'(p,y,\tau/\beta)}{\beta},\\
\mr{div}_p(\hat{g}_t) & \Mydef \frac{1}{m}\sum_{a=1}^m \partial_p \hat{g}(p_a^t,y_a,\tau^t/\beta),\\
\lambda_t &\Mydef \frac{ -\mr{div}_p (\hat{g}_{t}) }{  -\mr{div}_p (\hat{g}_{t})+ \mu}.
\end{split}
\]

We next consider the zero-temperature limit, i.e., $\beta\to\infty$. From the definition of $H$ in \eqref{Eqn:def_score}, it can be verified that \cite{Rangan11}:
\[
\hat{g}(p,y,\tau/\beta) = \frac{\mathbb{E}[z,p,y,\tau/\beta] - p}{\tau},
\]
where $\mathbb{E}[z,p,y,\tau/\beta]$ denotes the posterior mean estimator of $z$ w.r.t. the distribution $p(z|p,y,\tau/\beta)\propto \exp\left[ -\beta (y-|z|)^2  -\beta \frac{1}{2\tau}(z-p)^2  \right]$. As $\beta\to\infty$, the posterior mean concentrates around the minimum of the posterior probability, i.e., $\mathbb{E}[z,p,y,\tau/\beta]\to \mr{prox}(p,y,\tau)$ where
\BE\label{Eqn:prox_lambda_tau}
 \mr{prox}(p,y,\tau)\Mydef \underset{z}{\mr{argmin}}\ (y-|z|)^2 +\frac{ (z - p)^2}{2\tau},
\EE
which has the following closed-form expression (for $\tau>0$):
\[
 \mr{prox}(p,y,\tau)=\frac{2\tau y+ |p|}{1+2\tau}\cdot\mr{sign}(p).
\]
Here, $\mr{sign}(0)$ can be arbitrarily defined to be $+1$ or $-1$.
The function $\hat{g}$ becomes:
\BE\label{Eqn:g_ghat}
\hat{g}(p,y,\tau)=\frac{ \mr{prox}(p,y,\tau)- p}{\tau}=\frac{2}{1+2\tau}\cdot \underbrace{\left( y\cdot\mr{sign}(p)-p \right)}_{g(p,y)}.
\EE
\subsection{Summary of AMP.A}
After some algebra, we can finally express \eqref{Eqn:AMP_summary_2} using $g$ (instead of $\hat{g}$, see \eqref{Eqn:g_ghat}) as the following:
\BE\label{Eqn:AMP_summary_3}
\begin{split}
\tau^{t} &=\frac{1}{\delta}\frac{\tau^{t-1} + \frac{1}{2}}{ -\mr{div}_p(g_{t-1}) }\cdot \lambda_{t-1},\\
\bm{p}^t &=\bm{A}\bm{x}^t  -\frac{1}{\delta} \frac{ {g}\left( \bm{p}^{t-1},\bm{y} \right) }{- \mr{div}_p(g_{t-1})} \cdot \lambda_{t-1},\\
\bm{x}^{t+1} &= \lambda_t\cdot\left(\bm{x}^t + \frac{\bm{A}^\UT  {g}\left( \bm{p}^{t},\bm{y} \right) }{- \mr{div}_p(g_{t})}\right),
\end{split}
\EE
where
\[
\lambda_t=\frac{-\mr{div}_p(g_{t-1})}{-\mr{div}_p(g_{t-1})+\mu \left(\tau_t +\frac{1}{2}\right)}
\]
There are a couple of points we want to emphasize:
\begin{itemize}
\item
When $\mu=0$, the update of $\bm{p}^t$ and $\bm{x}^{t+1}$ are independent of the parameter $\tau$. This is why we prefer to use $g(p,y)$ instead of $\hat{g}(p,y,\tau)$, see \eqref{Eqn:g_ghat}.
\item Calculating the divergence term $\mr{div}_p(g)$ is tricky due to the discontinuity of $g(p,y)$ at $p=0$. Unlike the complex-valued case, a simple empirical average does not work well. We postpone our discussions on this issue to a forthcoming paper \cite{Plan}.
\end{itemize}

\subsection{Heuristic derivations of the state evolution}\label{sec:HeuristicDerSE}

According to \eqref{Eqn:AMP_complex}, the complex-valued version of AMP.A proceeds as follows
\BS\label{Eqn:SE_heuristic_derive_0}
\BE\label{Eqn:SE_heuristic_derive_2}
\begin{split}
x_i^{t+1} &= -2\mr{div}_p(g_t)\cdot{x}^t_i + 2\underbrace{\sum_{a=1}^m \bar{A}_{ai} g(p_a^t,y_a)}_{T},
\end{split}
\EE
where 
\BE \label{Eqn:SE_heuristic_derive_0b}
\mr{div}_p(g_t) \Mydef \frac{1}{m}\sum_{a=1}^m  \frac{1}{2}\left( \frac{\partial g(p_a^t,y_a)}{\partial p_a^R} -\mr{i} \frac{\partial g(p_a^t,y_a)}{\partial p_a^I}\right).
\EE
\ES
Suppose that at each iteration the elements of $\bm{x}^t$ are distributed as
\BE\label{Eqn:SE_heuristic_derive_1}
x_i^t \overset{d}{=}\alpha_t x_{*,i} + \sigma_t h_i,\quad\forall i=1,\ldots,n,
\EE
where $x_{*,i}$ represents the $i$th entry of the true signal vector $\bm{x}_{\ast}$ and $h_i\sim\mathcal{CN}(0,1)$ is independent of $x_i^t$. Rigorous proof of the state evolution framework is based on the conditioning technique developed in \cite{Bayati&Montanari11,Rangan11,Javanmard2013}. Here, our goal is show the reader how to heuristically derive the state evolution (SE) recursion, namely, \textit{given $\alpha_t$ and $\sigma_t$, how to derive $\alpha_{t+1}$ and $\sigma_{t+1}$.} Following \cite{DoMaMo09,Bayati&Montanari11}, we make the following heuristic assumptions to derive the SE:
\begin{itemize}
\item[(i)] We ignore the Onsager correction term, i.e., we assume that $\bm{p}^t$ is generated as (cf. \eqref{Eqn:AMP_complex}):
\[
p_a^t = \sum_{j}^n A_{aj} x_j^t,\quad \forall a=1,\ldots,m.
\]
\item[(ii)] We assume that $\bm{x}^t$ is independent of $\bm{A}$.
\end{itemize}
We derive $\alpha_{t+1}$ and $\sigma_{t+1}$ separately in the following two subsections.

\subsubsection{Derivations of $\alpha_{t+1}$} \label{Sec:heuristic_alpha}
To derive $\alpha_{t+1}$, we will calculate the expectation of the term $T$ in \eqref{Eqn:SE_heuristic_derive_2} by treating $\bm{x}_{*}$ and $\bm{x}^t$ as constants. In other words, the expectations in this section are conditioned on $\bm{x}_{\ast}$ and $\bm{x}^t$. We now consider the expectation of a single entry in $T$:
\BE\label{Eqn:SE_heuristic_derive_2b}
\begin{split}
\mathbb{E}\bigg[ \bar{A}_{ai} g\left(p_a^t,y_a\right)  \bigg] &=\mathbb{E}\left[ \bar{A}_{ai} \cdot g\bigg(\sum_{j=1}^n A_{aj}x_j^t,\ \Big|\sum_{j=1}^n A_{aj}x_{\ast,j}\Big| +w_a\bigg)  \right] \\
&=\mathbb{E} \left[ \bar{A}_{ai}\sum_{j=1}^n A_{aj}x_j^t \right]  \cdot \mathbb{E}\left[\partial_p g(p_a^t,y_a)\right] + \mathbb{E} \left[ \bar{A}_{ai}\sum_{j=1}^n A_{aj}x_{*,j} \right]  \cdot \mathbb{E}\left[\partial_z g(p_a^t, y_a)\right]\\
&=\frac{1}{m} x_i^t  \cdot \mathbb{E}\left[\partial_p g(p_a^t, y_a)\right] + \frac{1}{m} x_{*,i}  \cdot  \mathbb{E}\left[\partial_z g(p_a^t, y_a)\right],
\end{split}
\EE
where the last step is from Stein's lemma (for complex Gaussian random variables) \cite[Lemma 2.3]{Campese2015}, and $\partial_p g(p_a^t, y_a)$ and $\partial_z g(p_a^t, |z_a|+w_a)$ are defined as
\[
\begin{split}
\partial_p g(p, y) &\Mydef  \frac{1}{2}\left( \frac{\partial }{\partial p_R}g(p,y) -\mr{i}    \frac{\partial }{\partial p_I}g(p, y)\right),\\
\partial_z g(p, |z|+w) &\Mydef  \frac{1}{2}\left( \frac{\partial }{\partial z_R}g(p, |z|+w) -\mr{i}  \frac{\partial }{\partial z_I}g(p, |z|+w)\right),
\end{split}
\]
where $p_R$ and $p_I$ are the real and imaginary parts of $p$ (i.e., $p=p_R+\mr{i} p_I$) and $z_R$ and $z_I$ are the real and imaginary parts of $z$. Similar expressions also appeared in the complex AMP algorithm (CAMP) developed for solving the LASSO problem \cite{Maleki2013}. In AMP.A, $g(p,y)=y\cdot p/|p|-p$ and based on the above definitions we can derive that
\[
\begin{split}
\partial_p g(p,y) &= \frac{y}{2|p|}-1,\\
\partial_z g(p,|z|+w) &= \frac{\bar{z}p}{2|z|\, |p|}=\frac{1}{2}e^{\mr{i}(\theta_p-\theta_z)},
\end{split}
\]
where $\theta_p$ and $\theta_z$ are the phases of $p$ and $z$ respectively. Note that in rigorous calculations we should be careful about the discontinuity of $g$. In this heuristic calculations we have ignored this issue. We will discuss this issue in our forthcoming paper \cite{Plan}. Substituting \eqref{Eqn:SE_heuristic_derive_2b} into \eqref{Eqn:SE_heuristic_derive_2} yields
\BE\label{Eqn:SE_heuristic_derive_4}
\begin{split}
\mathbb{E}[T]&= \frac{1}{m}\sum_{a=1}^m \mathbb{E}\left[\partial_p g(p_a^t, y_a)\right]  \cdot x_i^t+ \frac{1}{m}   \sum_{a=1}^m  \mathbb{E}\left[\partial_z g(p_a^t, y_a|)\right]\cdot x_{*,i}\\
&\approx \mr{div}_p(g_t)\cdot x_i^t +\mr{div}_z(g_t)\cdot  x_{*,i},
\end{split}
\EE
where in the last step we assumed that the empirical averages of the partial derivatives $\mr{div}_p(g_t)=\frac{1}{m}\sum_{a=1}^m \partial_p g(p_a^t, y_a)$ and $\mr{div}_z(g_t)= \frac{1}{m}\sum_{a=1}^m \partial_z g(p_a^t, |z_a|+w_a)$ converge to their expectations. Substituting \eqref{Eqn:SE_heuristic_derive_4} into 
\eqref{Eqn:SE_heuristic_derive_2} yields
\[
\begin{split}
\mathbb{E}[x_i^{t+1}] &=  -2\mr{div}_p(g_t)\cdot {x}^t_i +2 \mathbb{E}[T]\\
&=2\mr{div}_z(g_t)\cdot  x_{*,i}.
\end{split}
\]
From our assumption in \eqref{Eqn:SE_heuristic_derive_1}, we have $\mathbb{E}[x_i^{t+1}]=\alpha_{t+1}\cdot x_{*,i}$. This result combined with \eqref{Eqn:SE_heuristic_derive_4} leads to 
\BE\label{Eqn:SE_heuristic_derive_4b}
\alpha_{t+1}=2\mr{div}_z(g_t).
\EE
Finally, when $\bm{x}$ and $\bm{x}^t$ are independent of $\bm{A}$, and by central limit theorem we can assume that both $p_a^t=\sum_{i=1}^n A_{ai}x_i^t$ and $z_a=\sum_{i=1}^n A_{ai}x_{*,i}$ are Gaussian, and their joint distribution is specified by the relationship $p_a^t \overset{d}{=}\alpha_t z_a + \sigma_t b_a$ where $z_a\sim\mathcal{CN}(0,1/\delta)$ and $b_i\sim\mathcal{CN}(0,1/\delta)$ are independent. 

\subsubsection{Derivations of $\sigma^2_{t+1}$} \label{Sec:heuristic_alpha}

From \eqref{Eqn:SE_heuristic_derive_1}, $\sigma^2_{t+1}$ can be derived as
\BE\label{Eqn:SE_heuristic_derive_5}
\begin{split}
\sigma^2_{t+1} &=\mr{var}[x_i^{t+1}]=\mr{var}[-2\mr{div}_p(g_t)\cdot{x}^t_i +2T ]=4\cdot\mr{var}[T].
\end{split}
\EE
Further,
\BE\label{Eqn:SE_heuristic_derive_6}
\begin{split}
\mathbb{E}[|T|^2] &=\mathbb{E}\Bigg[ \bigg| \sum_{a=1}^m A_{ai}g(p_a^t,|z_a|)\bigg|^2 \Bigg] \\
 & =\sum_{a=1}^m \mathbb{E}\Big[|A_{ai}|^2\cdot  |g_a |^2 \Big] +  \sum_a\sum_{b\neq a} \mathbb{E}\left[ \bar{A}_{ia}\bar{g}_a A_{ib}g_b  \right]\\
&\overset{(a)}{\approx} \frac{1}{m} \sum_{a=1}^m\mathbb{E}\Big[  |g_a |^2 \Big] +  \sum_a\sum_{b\neq a} \mathbb{E}\left[ \bar{A}_{ia}\bar{g}_a  \right]\cdot \mathbb{E}[A_{ib}g_b] \\
&\approx \frac{1}{m} \sum_{a=1}^m\mathbb{E}\Big[  |g_a |^2 \Big] +  \frac{m(m-1)}{m^2}\cdot\left| \mathbb{E}[T]\right|^2\\
&\approx  \frac{1}{m} \sum_{a=1}^m\mathbb{E}\Big[  |g_a |^2 \Big] +\left| \mathbb{E}[T]\right|^2,
\end{split}
\EE
where $g_a$ and $g_b$ are shorthands for $g(p_a^t,y_a)$ and $g(p_b^t,y_b)$ respectively, and step (a) follows from the heuristic assumption that the correlation between $|A_{ai}|^2$ and $|g_a|^2$, and the correlation between $A_{ia}g_a$ and $A_{ib}g_b$ can be ignored. Hence, combining \eqref{Eqn:SE_heuristic_derive_5} and \eqref{Eqn:SE_heuristic_derive_6} we obtain
\[
\sigma^2_{t+1}=4\left(\mathbb{E}\left[|T|^2\right]-\left|\mathbb{E}\left[T\right]\right|^2\right)\approx \frac{4}{m} \sum_{a=1}^m\mathbb{E}\Big[  |g_a(p_a^t,y_a )^2 \Big],
\]
where as argued below \eqref{Eqn:SE_heuristic_derive_4b} the joint distribution of $p_a^t$ and $z_a$ are specified by $p_a^t \overset{d}{=}\alpha_t z_a + \sigma_t b_a$ where $z_a\sim\mathcal{CN}(0,1/\delta)$ and $b_a\sim\mathcal{CN}(0,1/\delta)$ are independent.

\section{Simplifications of SE maps}\label{Sec:SE_derivations}

\subsection{Auxiliary Results}\label{Sec:auxiliary}

Here we collect some auxiliary results that will be used in the simplification of the state evolution equation.  
\begin{lemma}\label{Lem:Aux_1}
The following identities hold for any $a\in\mathbb{R}$ and $b\in\mathbb{R}_+$:
\BS
\begin{align}
&\int_0^{2\pi}\int_0^{\infty} r\cos\theta\exp\left( -\frac{r^2-2ar\cos\theta}{b} \right)\mr{d}r\mr{d}\theta = 2a\sqrt{b}\sqrt{\pi}\int_0^{\frac{\pi}{2}}\cos^2\theta\exp\left( \frac{a^2\cos^2\theta}{b} \right)\mr{d}\theta,\label{Eqn:Aux_1_1} \\
&\int_0^{2\pi}\int_0^{\infty} r\sin\theta\exp\left( -\frac{r^2-2ar\cos\theta}{b} \right)\mr{d}r\mr{d}\theta = 0.\label{Eqn:Aux_1_2}
\end{align}
\ES
\end{lemma}
\begin{proof}
We first consider \eqref{Eqn:Aux_1_1}:
\BE \label{Eqn:Aux_1_3}
\begin{split}
&\int_0^{2\pi}\int_0^{\infty}r\cos\theta\exp\left(-\frac{r^2-2a\cdot r\cos\theta}{b}\right)\mr{d}\theta\mr{d}r \\
&=\int_0^{2\pi}\cos\theta\exp\left(\frac{a^2\cos^2\theta}{b}\right)\mr{d}\theta\int_0^{\infty}r\exp\left(-\frac{(r-a\cos\theta)^2}{b}\right)\mr{d}r\\
&\overset{(a)}{=}\int_0^{2\pi}\cos\theta\exp\left(\frac{a^2\cos^2\theta}{b}\right)\left[\frac{1}{2}b\exp\left(\frac{-a^2\cos^2\theta}{b}\right)+a\cos\theta\sqrt{b\pi}\Phi\left(\frac{\sqrt{2}a\cos\theta}{\sqrt{b}}\right)\right]\mr{d}\theta\\
&=\int_0^{2\pi}\frac{1}{2}b\cos\theta\mr{d}\theta + \int_0^{2\pi}a\cos^2\theta\sqrt{b\pi}\exp\left(\frac{a^2\cos^2\theta}{b}\right)\Phi\left(\frac{\sqrt{2}a\cos\theta}{\sqrt{b}}\right)\mr{d}\theta\\
&\overset{(b)}{=}\int_0^{\pi}a\cos^2\theta\sqrt{b\pi}\exp\left(\frac{a^2\cos^2\theta}{b}\right)\Phi\left(\frac{\sqrt{2}a\cos\theta}{\sqrt{b}}\right)\mr{d}\theta+\int_{0}^{\pi}a\cos^2\hat{\theta}\sqrt{b\pi}\exp\left(\frac{a^2\cos^2\hat{\theta}}{b}\right)\Phi\left(-\frac{\sqrt{2}a\cos\hat{\theta}}{\sqrt{b}}\right)\mr{d}\hat{\theta}\\
&=\int_0^{\pi}a\cos^2\theta\sqrt{b\pi}\exp\left(\frac{a^2\cos^2\theta}{b}\right)\left[\Phi\left(\frac{\sqrt{2}a\cos\theta}{\sqrt{b}}\right)+\Phi\left(-\frac{\sqrt{2}a\cos\theta}{\sqrt{b}}\right)\right]\mr{d}\theta\\
&\overset{(c)}{=}a\sqrt{b\pi}\int_0^{\pi}\cos^2\theta\exp\left(\frac{a^2\cos^2\theta}{b}\right)\mr{d}\theta\\
&\overset{(d)}{=}2a\sqrt{b\pi}\int_0^{\frac{\pi}{2}}\cos^2\theta\exp\left(\frac{a^2\cos^2\theta}{b}\right)\mr{d}\theta,
\end{split}
\EE
where step (a) is from the integral ($\Phi(x)$ denotes the CDF of the standard Gaussian distribution):
\[
\int_{0}^{\infty}r\exp\left(-\frac{(r-m)^2}{v}\right)\mr{d}r =\frac{1}{2}b\exp\left(\frac{-m^2}{v}\right)+m\sqrt{v\pi}\Phi\left(\frac{\sqrt{2}m}{\sqrt{v}}\right),\quad \forall m\in\mathbb{R}, v\in\mathbb{R}_+,
\]
step (b) is from the variable change $\hat{\theta}=\theta-\pi$, step (c) is from the fact that $\Phi(x)+\Phi(-x)=1$, and step (d) is from
\[
\begin{split}
&\int_0^{\pi}\cos^2\theta\exp\left(\frac{a^2\cos^2\theta}{b}\right)\mr{d}\theta\\
&=\int_0^{\frac{\pi}{2}}\cos^2\theta\exp\left(\frac{a^2\cos^2\theta}{b}\right)\mr{d}\theta+\int_{\frac{\pi}{2}}^{\pi}\cos^2\theta\exp\left(\frac{a^2\cos^2\theta}{b}\right)\mr{d}\theta\\
&=\int_0^{\frac{\pi}{2}}\cos^2\theta\exp\left(\frac{a^2\cos^2\theta}{b}\right)\mr{d}\theta+\int_{\frac{\pi}{2}}^{0}\cos^2\hat{\theta}\exp\left(\frac{a^2\hat{\cos}^2\theta}{b}\right)(-\mr{d}\hat{\theta})\quad(\hat{\theta}=\pi-\theta)\\
&=2\int_0^{\frac{\pi}{2}}\cos^2\theta\exp\left(\frac{a^2\cos^2\theta}{b}\right)\mr{d}\theta.
\end{split}
\]

The identity in \eqref{Eqn:Aux_1_2} can be derived based on similar calculations:
\[
\begin{split}
\int_0^{2\pi}\int_0^{\infty}r\sin\theta\exp\left(-\frac{r^2-2b\cdot r\cos\theta}{b}\right)\mr{d}\theta\mr{d}r &=a\sqrt{b\pi}\int_0^{\pi}\frac{1}{2}\sin2\theta\exp\left(\frac{a^2\cos^2\theta}{b}\right)\mr{d}\theta\\
&=0.
\end{split}
\]
\end{proof}


\vspace{3pt}

\begin{lemma}
Let $\tilde{Z}\sim\mathcal{N}(0,1)$ be a standard Gaussian random variable. Then, for any $x\in\mathbb{R}$, the following identities hold:
\BE\label{Eqn:Aux_f}
\begin{split}
\mathbb{E}\left[|\tilde Z|\cdot \phi\left( x|\tilde Z| \right) \right] &=\frac{1}{\pi}\frac{1}{1+x^2},\\
\mathbb{E}\left[ \Phi\left( x|\tilde Z| \right) \right] &=\frac{1}{\pi}\mr{arctan}(x)+\frac{1}{2},\\
\mathbb{E}\left[\tilde{Z}^2\cdot \Phi\left( x|\tilde Z| \right) \right] &=\frac{1}{\pi}\mr{arctan}(x)+\frac{1}{2} + \frac{1}{\pi}\frac{x}{1+x^2},
\end{split}
\EE
where $\phi(\cdot)$ and $\Phi(\cdot)$ are, respectively, PDF and CDF functions of the standard Gaussian distribution.
\end{lemma}

\begin{proof}
Consider the first identity:
\BE\label{Eqn:Aux_f1}
\begin{split}
\mathbb{E}\left[|\tilde Z|\cdot \phi\left( x|\tilde Z| \right) \right]  &=\int_{-\infty}^{\infty} |z|\phi(x\,|z|)\phi(z)\mr{d}z\\
& \overset{(a)}{=} 2\int_{0}^{\infty}z\phi(x\,z)\phi(z)\mr{d}z\\
& \overset{(b)}{=}\frac{1}{\pi}\int_{0}^{\infty}z \exp\left[ -(1+x^2)\frac{z^2}{2} \right] \mr{d}z\\
&=\frac{1}{\pi}\frac{1}{1+x^2},
\end{split}
\EE
where (a) is from the symmetry of $\phi$ and (b) from the definition $\phi(x)=1/\sqrt{2\pi}e^{-x^2/2}$. Further, 
\BE
\begin{split}
\frac{\mr{d}}{\mr{d}x}\mathbb{E}\left[\Phi\left( x|\tilde Z| \right) \right] &=\frac{\mr{d}}{\mr{d}x}\int_{-\infty}^{\infty} \Phi(x\,|z|)\phi(z)\mr{d}z=\frac{\mr{d}}{\mr{d}x}\int_{0}^{\infty} 2\Phi(x\,z)\phi(z)\mr{d}z \\
&=\int_{0}^{\infty} 2 \frac{\mr{d}}{\mr{d}x}\Phi(x\,z)\phi(z)\mr{d}z=\int_{0}^{\infty} 2 z\phi(x\,z)\phi(z)\mr{d}z\\
&=\frac{1}{\pi}\frac{1}{1+x^2},
\end{split}
\EE
where the last equality is from \eqref{Eqn:Aux_f1}. Hence,
\BE\label{Eqn:Aux_f2}
\mathbb{E}\left[\Phi\left( x|\tilde Z| \right) \right]   = \int_{-\infty}^x \frac{1}{\pi}\frac{1}{1+t^2}\mr{d}t\\
=\frac{1}{\pi}\mr{arctan}(x)+\frac{1}{2}.
\EE

Finally, the third identity in \eqref{Eqn:Aux_f} can be derived as follows:
\BE
\begin{split}
\mathbb{E}\left[\tilde{Z}^2\cdot \Phi\left( x|\tilde Z| \right) \right] &= \int_{-\infty}^{\infty}z^2\Phi(x\,|z|)\phi(z)\mr{d}z\\
&= \int_{0}^{\infty}z^2\Phi(xz)\phi(z)\mr{d}z\\
&\overset{(a)}{=}-2\int_0^{\infty}z \Phi(xz)\mr{d}\phi(z)\\
&=-2\left\{ z\Phi(xz)\phi(z)\big|_{0}^{\infty} - \int_0^{\infty}\phi(z)\left[ \Phi(xz)+xz\phi(xz)\mr{d}z \right] \right\}\\
&=2 \int_0^{\infty}\phi(z)\Phi(xz)\mr{d}z+x\cdot2\int_0^{\infty}z\phi(xz)\phi(z)\mr{d}z\\
&\overset{(b)}{=}\frac{1}{\pi}\mr{arctan}(x)+\frac{1}{2} + \frac{1}{\pi}\frac{x}{1+x^2},
\end{split}
\EE
where (a) is from the identity $\phi'(z)=z\phi(z)$ and (b) from our previously derived identities in \eqref{Eqn:Aux_f1} and \eqref{Eqn:Aux_f2}.
\end{proof}

\subsection{Complex-valued $\ampa$}\label{Sec:SE_derivations_complex}

From Definition \ref{Def:SE_map_complex}, the SE equations are given by
\BE\label{Eqn:psi2_complex_repeat}
\begin{split}
\psi_1(\alpha,\sigma^2) &=2\cdot \mathbb{E}\left[ \partial_z g(p,Y) \right]\\
&=\mathbb{E}\left[ \frac{\bar{Z} P}{|Z|\, |P|} \right],\\
\psi_2(\alpha,\sigma^2;\delta,\sigma^2_w) &=4\cdot\mathbb{E}\left[ \left|g(P,Y)\right|^2 \right]\\
&=4\cdot\mathbb{E}\left[ \left( |Z| -|P|+W \right)^2\right]\\
&=\underbrace{4\cdot\mathbb{E}\left[ \left( |Z| -|P| \right)^2\right] }_{\psi_2(\alpha,\sigma^2;\delta)}+4\sigma^2_w.
\end{split}
\EE
In the above, $Z\sim\mathcal{CN}(0,1/\delta)$, $P=\alpha Z + \sigma B$ where $B\sim\mathcal{CN}(0,1/\delta)$ is independent of $Z$, and $Y=|Z|+W$ where $W\sim\mathcal{CN}(0,\sigma^2_w)$ independent of both $Z$ and $B$. We first consider a special case $\sigma^2=0$ ($\alpha\neq0$). When $\sigma=0$, we have $P=\alpha Z+\sigma B=\alpha Z$, and therefore
\[
\begin{split}
\psi_1(\alpha,0) &=\mathbb{E}\left[ \frac{\alpha \bar{Z}Z}{\alpha |Z|\, |Z|} \right]=1,\\
\psi_2(\alpha,0;\delta,\sigma^2_w) &=4\cdot\mathbb{E}\left[ \left( |Z| -|\alpha Z| \right)^2\right]+4\sigma^2_w =\frac{4}{\delta}\left(1-|\alpha| \right)^2+4\sigma^2_w.
\end{split}
\]
We next turn to the general case where $\sigma^2\neq 0$. Later, we will see that our formulas derived for positive $\sigma^2$ covers the special case $\sigma^2=0$ as well. Lemma~\ref{Lem:psi_phase_invariance} can simplify our derivations. 
\begin{lemma}\label{Lem:psi_phase_invariance}
$\psi_1$ and $\psi_2$ in \eqref{Eqn:psi2_complex_repeat} have the following properties (for any $\alpha\in\mathbb{C}\backslash0$ and $\sigma^2\ge0$):
\begin{enumerate}
\item [(i)] $\psi_1(\alpha,\sigma^2)=\psi_1(|\alpha|,\sigma^2)\cdot e^{\mr{i}\theta_{\alpha}}$, with $e^{\mr{i}\theta_{\alpha}}$ being the phase of $\alpha$;
\item[(ii)] $\psi_2(\alpha,\sigma^2;\delta)=\psi_2(|\alpha|,\sigma^2;\delta)$.
\end{enumerate}
\end{lemma}
\begin{proof}
Note that for $\psi_1$ and $\psi_2$ defined in \eqref{Eqn:psi2_complex_repeat}, we have $P|Z\sim\mathcal{CN}(\alpha Z,\sigma^2/\delta)$. Consider the random variable $\tilde{P}\Mydef P\cdot e^{-\mr{i}\theta_{\alpha}}$. Based on the rotational invariance of circularly-symmetric Gaussian, we have $\tilde{P}|Z\sim\mathcal{CN}(|\alpha|Z;\sigma^2/\delta)$.
Hence,
\[
\begin{split}
\psi_1(\alpha,\sigma^2)&=\mathbb{E}\left[ \frac{\bar{Z} P}{|Z|\, |P|} \right] =e^{\mr{i}\theta_{\alpha}}\cdot \mathbb{E}\left[ \frac{\bar{Z} \tilde{P}}{|Z|\, |\tilde{P}|} \right]=e^{\mr{i}\theta_{\alpha}}\cdot\psi_1(|\alpha|,\sigma^2).
\end{split} 
\]
The proof of $\psi_2(\alpha,\sigma^2;\delta)=\psi_2(|\alpha|,\sigma^2;\delta)$ follows from a similar argument: the joint distribution of $|Z|$ and $|P|$ does not depend on $\theta_{\alpha}$, and thus $\psi_2(\alpha,\sigma^2)=4\mathbb{E}\left[ \left( |Z| -|P| \right)^2\right]$ does not depend on $\theta_{\alpha}$.
\end{proof}

Note that Lemma~\ref{Lem:psi_phase_invariance} also holds for $\alpha=0$ if we define $\angle 0=0$. 
\begin{remark}
In the following, we will derive $\psi_1$ and $\psi_2$ for the case where \textit{$\alpha$ is real and nonnegative}. The results for complex-valued $\alpha$ can be easily derived from those for nonnegative $\alpha$, based on Lemma~\ref{Lem:psi_phase_invariance}.
\end{remark}

We can also write $\psi_1$ as
\[
\psi_1(\alpha,\sigma^2)=\mathbb{E}\left[ \frac{\bar{Z} P}{|Z|\, |P|} \right]=\mathbb{E}[e^{\mr{i}(\theta_p-\theta_z)}].
\]
Note that $\theta_p-\theta_z$ is the phase of an auxiliary variable $\hat{P}\Mydef e^{-\mr{i}\theta_z}P = \alpha |Z| + \sigma e^{-\mr{i}\theta_z}B$. Further, from the rotational invariance, conditioned on $|Z|$, $\hat{P}$ is distributed as $\hat{P}\sim\mathcal{CN}(\alpha |Z|,\sigma^2/\delta)$. Hence, the expectation of its phase can be calculated as 
\BE\label{Eqn:E_theta_PZ}
\begin{split}
\mathbb{E}\left[e^{\mr{i}(\theta_p-\theta_z)}\big|\, |Z|\right]=&\int_0^{2\pi}\int_0^{\infty}e^{\mr{i}\theta}\cdot\frac{1}{\pi\sigma^2/\delta}\exp\left(-\frac{\left| re^{\mr{i}\theta}-\alpha |Z| \right|^2}{\sigma^2/\delta}\right)\cdot r\mr{d}r\mr{d}\theta\\
=& \frac{1}{\pi\sigma^2/\delta}\exp\left(-\frac{\alpha^2|Z|^2}{\sigma^2/\delta}\right) \cdot\int_{0}^{2\pi}\int_0^{\infty}re^{\mr{i}{\theta}}\cdot\frac{1}{\pi\sigma^2/\delta}\exp\Bigg(-\frac{r^2-2\alpha |Z|\cos\theta r}{\sigma^2/\delta}\Bigg)\mr{d}r\mr{d}{\theta}\\ 
=& \frac{1}{\pi\sigma^2/\delta}\exp\left(-\frac{\alpha^2|Z|^2}{\sigma^2/\delta}\right) \cdot\int_{0}^{2\pi}\int_0^{\infty}r\cos\theta\cdot\frac{1}{\pi\sigma^2/\delta}\exp\Bigg(-\frac{r^2-2\alpha |Z|\cos\theta r}{\sigma^2/\delta}\Bigg)\mr{d}r\mr{d}{\theta}\\ 
&+\mr{i}\frac{1}{\pi\sigma^2/\delta}\exp\left(-\frac{\alpha^2|Z|^2}{\sigma^2/\delta}\right) \cdot\int_{0}^{2\pi}\int_0^{\infty}r\sin\theta\cdot\frac{1}{\pi\sigma^2/\delta}\exp\Bigg(-\frac{r^2-2\alpha |Z|\cos\theta r}{\sigma^2/\delta}\Bigg)\mr{d}r\mr{d}{\theta}\\
=&2\int_0^{\frac{\pi}{2}} \frac{\alpha |Z|}{\sqrt{\pi}\sqrt{\sigma^2/\delta}}\cos^2\theta\exp\left( -\frac{\alpha^2|Z|^2\sin^2\theta}{\sigma^2/\delta} \right)\mr{d}\theta,
\end{split}
\EE
where the last step follow the following two identities together with some straightforward manipulations:
\BS\label{Eqn:identities_E_theta}
\begin{align}
&\int_0^{2\pi}\int_0^{\infty}r\cos\theta\exp\left( -\frac{r^2-2\alpha|Z|\cos\theta r}{\sigma^2/\delta} \right)\mr{d}r\mr{d}\theta= \frac{2\alpha\sigma\sqrt{\pi}}{\sqrt{\delta}}\int_0^{\frac{\pi}{2}}\cos^2\theta\exp\left(\frac{\alpha^2|Z|^2\cos^2\theta}{\sigma^2/\delta}\right)\mr{d}\theta,\\
&\int_0^{2\pi}\int_0^{\infty}r\sin\theta\exp\left( -\frac{r^2-2\alpha|Z|\cos\theta r}{\sigma^2/\delta} \right)\mr{d}r\mr{d}\theta=0.
\end{align}
\ES
The above identities are proved in Lemma~\ref{Lem:Aux_1} in Appendix~\ref{Sec:auxiliary}.
Using \eqref{Eqn:E_theta_PZ} and noting that $Z\sim\mathcal{CN}(0,1/\delta)$, we further average our result over $|Z|$:
\BE\label{Eqn:E_Z2_theta}
\begin{split}
\mathbb{E}\left[e^{\mr{i}(\theta_p-\theta_z)} \right] &=\mathbb{E}\left\{ 2\int_0^{\frac{\pi}{2}} \frac{\alpha |Z|}{\sqrt{\pi}\sqrt{\sigma^2/\delta}}\cos^2\theta\exp\left( -\frac{\alpha^2|Z|^2\sin^2\theta}{\sigma^2/\delta} \right)\mr{d}\theta\right\}\\
&\overset{(a)}{=} \int_0^{\infty}2\delta r\exp\left(-\delta r^2\right)\cdot \left(2\int_0^{\frac{\pi}{2}} \frac{\alpha r}{\sqrt{\pi}\sqrt{\sigma^2/\delta}}\cos^2\theta\exp\left( -\frac{\alpha^2r^2\sin^2\theta}{\sigma^2/\delta} \right)\mr{d}\theta\right)\mr{d}r\\
&=\frac{4\alpha\delta^{3/2}}{\sqrt{\pi}\sigma}\int_0^{\frac{\pi}{2}}\cos^2\theta\mr{d}\theta\int_0^{\infty}r^2\exp\left(-\delta\left(1+\frac{\alpha^2\sin^2\theta}{\sigma^2}\right)r^2\right)\mr{d}r\\
&\overset{(b)}{=}\frac{\alpha}{\sigma}\int_0^{\frac{\pi}{2}}\cos^2\theta\left(1+\frac{\alpha^2\sin^2\theta}{\sigma^2}\right)^{-\frac{3}{2}}\mr{d}\theta\\
&\overset{(c)}{=}\frac{\alpha}{\sigma}\int_0^{\frac{\pi}{2}}\frac{\sin^2\theta}{\left( 1+\frac{\alpha^2}{\sigma^2}\sin^2\theta \right)^{\frac{1}{2}}}\mr{d}\theta\\
&=\int_0^{\frac{\pi}{2}} \frac{ \alpha \,\sin^2\theta }{(\alpha^2\sin^2\theta +\sigma^2)^{\frac{1}{2}}} \mr{d}\theta,
\end{split}
\EE
where step (a) follows since the density of $|Z|$ is $f_{|Z|}(r)=\int_0^{2\pi}\delta/\pi\exp(-\delta r^2)r\mr{d}\theta=2\delta r\exp(-\delta r^2)$, and step (b) follows from the identity $\int_0^{\infty} r^2\exp(-ar^2)\mr{d}r=\sqrt{\pi}/4\cdot a^{-3/2}$, and (c) is derived in \eqref{Lem:Aux_2_c}. \\

We next derive $\psi_2(\alpha,\sigma^2;\delta)$. From \eqref{Eqn:psi2_complex_repeat}, we have
\[
\begin{split}
\psi_2(\alpha,\sigma^2;\delta)&=4\mathbb{E}\left[\left(|Z|-|P|\right)^2\right]\\
&=4\left(\frac{1+\alpha^2+\sigma^2}{\delta}-2\cdot\mathbb{E}\left\{|ZP|\right\}\right),
\end{split}
\]
where the last step is from $Z\sim\mathcal{CN}(0,1/\delta)$ and $P\sim\mathcal{CN}(0,(\alpha^2+\sigma^2)/\delta)$. We next calculate $\mathbb{E}[|ZP|]$. Again, conditioned on $|Z|$, $P$ is distributed as ${P}\sim\mathcal{CN}(\alpha|Z|,\sigma^2/\delta)$. We first calculate $\mathbb{E}[|P|\big|\, |Z|]$:
\BE\label{Eqn:E_P_abs}
\begin{split}
\mathbb{E}\left[|P|\big|\, |Z|\right]&=\int_{\mathbb{C}}|P| \frac{1}{\pi \sigma^2/\delta}\exp\left( -\frac{\big|P - \alpha |Z|\big|^2}{\sigma^2/\delta} \right)\mr{d}P\\
&= \int_0^{2\pi}\int_0^{\infty} r\frac{1}{\pi\sigma^2/\delta}\exp\left(-\frac{|re^{\mr{i}\theta}-\alpha|Z||^2}{\sigma^2/\delta}\right)\cdot r\mr{d}r\mr{d}\theta\\
&=\frac{1}{\pi\sigma^2/\delta}\int_0^{2\pi}\exp\left( -\frac{\alpha^2|Z|^2\sin^2\theta}{\sigma^2/\delta} \right)\mr{d}\theta\int_0^{\infty}r^2\exp\left(-\frac{(r-\alpha|Z|\cos\theta)^2}{\sigma^2/\delta}\right)\mr{d}r\\
&= \frac{2}{\sqrt{\pi\sigma^2/\delta}}\int_0^{\frac{\pi}{2}}\left(\alpha^2|Z|^2\cos^2\theta+\frac{\sigma^2}{2\delta}\right)\exp\left( -\frac{\alpha^2|Z|^2\sin^2\theta}{\sigma^2/\delta} \right)\mr{d}\theta  ,
\end{split}
\EE
where in the last step we used the following indentity
\[
\int_{0}^{\infty}r^2\exp\left(-\frac{(r-m)^2}{v}\right)\mr{d}r =\frac{mv}{2}\exp\left(-\frac{m^2}{v}\right)+\sqrt{v\pi}\left(m^2+\frac{v}{2}\right)\Phi\left(\sqrt{\frac{2}{v}}\cdot m\right),\quad \forall m\in\mathbb{R}, v\in\mathbb{R}_+
\]
and some manipulations similar to those in \eqref{Eqn:Aux_1_3}. Following the same procedure as that in \eqref{Eqn:E_Z2_theta}, we further calculate $\mathbb{E}[|ZP|]$ as:
\BE\label{Eqn:E_ZP}
\begin{split}
\mathbb{E}[|ZP|]&=\int_0^{\infty}r\cdot 2r \delta \exp\left(-\delta r^2\right)\cdot \left( \frac{2}{\sqrt{\pi\sigma^2/\delta}}\int_0^{\frac{\pi}{2}}\left(\alpha^2r^2\cos^2\theta+\frac{\sigma^2}{2\delta}\right)\exp\left( -\frac{\alpha^2r^2\sin^2\theta}{\sigma^2/\delta} \right)\mr{d}\theta\right)\mr{d}r\\
&=\int_0^{\frac{\pi}{2}}\int_0^{\infty}\frac{4\delta^{3/2}}{\sqrt{\pi}\sigma}\left(\alpha^2\cos^2\theta\cdot r^4 +\frac{\sigma^2}{2\delta}\cdot r^2\right)\exp\left(- \delta\left(1+\frac{\alpha^2\sin^2\theta}{\sigma^2}\right)r^2\right)\mr{d}r\mr{d}\theta\\
&=\frac{3\alpha^2}{2\sigma\delta}\int_0^{\frac{\pi}{2}}\cos^2\theta\left(1+\frac{\alpha^2}{\sigma^2}\sin^2\theta\right)^{-\frac{5}{2}}\mr{d}\theta+\frac{\sigma}{2\delta}\int_0^{\frac{\pi}{2}}\left(1+\frac{\alpha^2}{\sigma^2}\sin^2\theta\right)^{-\frac{3}{2}}\mr{d}\theta,
\end{split}
\EE
where in the last step we used the following identities: $\int_0^{\infty} r^4\exp(-ar^2)\mr{d}r=3\sqrt{\pi}/8\cdot a^{-5/2}$ and $\int_0^{\infty} r^2\exp(-ar^2)\mr{d}r=\sqrt{\pi}/4\cdot a^{-3/2}$. Finally, using \eqref{Eqn:E_ZP} we have
\[
\begin{split}
\psi_2(\alpha,\sigma^2;\delta)&=4\left(\frac{1+\alpha^2+\sigma^2}{\delta}-2\cdot\mathbb{E}\left\{|Z||P|\right\}\right)\\
&\overset{(a)}{=}4\left\{\frac{1+\alpha^2+\sigma^2}{\delta}-2\left[\frac{3\alpha^2}{2\sigma\delta}\int_0^{\frac{\pi}{2}}\cos^2\theta\left(1+\frac{\alpha^2}{\sigma^2}\sin^2\theta\right)^{-\frac{5}{2}}\mr{d}\theta+\frac{\sigma}{2\delta}\int_0^{\frac{\pi}{2}}\left(1+\frac{\alpha^2}{\sigma^2}\sin^2\theta\right)^{-\frac{3}{2}}\mr{d}\theta\right]\right\}\\
&=\frac{4}{\delta}\left\{1+\alpha^2+\sigma^2-\frac{\sigma}{2}\left[\frac{3\alpha^2}{\sigma^2}\int_0^{\frac{\pi}{2}}\cos^2\theta\left(1+\frac{\alpha^2}{\sigma^2}\sin^2\theta\right)^{-\frac{5}{2}}\mr{d}\theta+\int_0^{\frac{\pi}{2}}\left(1+\frac{\alpha^2}{\sigma^2}\sin^2\theta\right)^{-\frac{3}{2}}\mr{d}\theta\right]\right\}\\
&\overset{(b)}{=}\frac{4}{\delta}\Bigg(1+\alpha^2+\sigma^2-\sigma\int_0^{\frac{\pi}{2}}\frac{1+2\frac{\alpha^2}{\sigma^2}\sin^2\theta}{\left(1+\frac{\alpha^2}{\sigma^2}\sin^2\theta\right)^{\frac{1}{2}}}\mr{d}\theta\Bigg)\\
&=\frac{4}{\delta}\Bigg(1+\alpha^2+\sigma^2-\int_0^{\frac{\pi}{2}}\frac{2\alpha^2\sin^2\theta+\sigma^2}{\left(\alpha^2\sin^2\theta+\sigma^2\right)^{\frac{1}{2}}}\mr{d}\theta\Bigg),
\end{split}
\]
where (a) is from \eqref{Eqn:E_ZP}, and the derivations of step (b) is more involved and are given in Lemma~\ref{Lem:Aux_2}.

\subsection{Real-valued $\ampa$}\label{Sec:SE_derivations_real}
For the real-valued case, the SE maps are given by
\BE\label{Eqn:SE_maps_real_original}
\begin{split}
\psi_1(\alpha,\sigma^2) &=  \mathbb{E}[\partial_z  g(P,Y)],\\
\psi_2(\alpha,\sigma^2;\delta,\sigma^2_w) &=\mathbb{E}\left[ g^2(P,Y) \right],
\end{split}
\EE
where$Z\sim\mathcal{N}(0,1/\delta)$, $P=\alpha Z + \sigma B$ where $B\sim\mathcal{N}(0,1/\delta)$ is independent of $Z$, and $Y=|Z|+W$ where $W\sim\mathcal{N}(0,\sigma^2_w)$ independent of both $Z$ and $B$. Substituting $g(p,y)=y\cdot \mr{sign}(p)-p$ into \eqref{Eqn:SE_maps_real_original} yields
\BE\label{Eqn:SE_maps_real_original2}
\begin{split}
\psi_1(\alpha,\sigma^2) &= \mathbb{E}\left[\partial_z |Z|\cdot \mr{sign}(P)\right]\\
&=\mathbb{E}\left[\mr{sign}(ZP)\right],\\
\psi_2(\alpha,\sigma^2) &=\mathbb{E}\left[ \left( |Z| - |P| +W\right)^2 \right]\\
&= \underbrace{\mathbb{E}\left[ \left( |Z| - |P| \right)^2 \right] }_{\psi_2(\alpha,\sigma^2)}+\sigma^2_w.
\end{split}
\EE
Further $\mathbb{E}\left[ \left( |Z| - |P| \right)^2 \right] =\frac{1}{\delta}(\alpha^2+\sigma^2+1)-2\mathbb{E}[|ZP|]$. It remains to derive the following terms: $\mathbb{E}\left[\mr{sign}(ZP) \right]$ and $\mathbb{E}[|ZP|]$. We first consider $\mathbb{E}\left[\mr{sign}(ZP) \right]$. Similar to the derivations in Section \eqref{Sec:SE_derivations_complex}, we will first calculate the expectation conditioned on $Z$. Note that conditioned on $Z$, we have $P|Z\sim\mathcal{N}(\alpha Z,\sigma^2/\delta)$ and 
\BE\label{Eqn:S}
S\Mydef ZP|Z\sim\mathcal{N}(\alpha Z^2,\sigma^2 Z^2/\delta).
\EE
We then have
\[
\begin{split}
\mathbb{E}\left[\mr{sign}(ZP)\big| Z\right] &=2\mr{Pr}(S>0) -1\\
 &=2\Phi\left( \frac{\alpha }{\sigma}|Z|\sqrt{\delta} \right) -1\\
&=2\Phi\left( \frac{\alpha }{\sigma}|\tilde{Z}| \right) -1,
\end{split}
\]
where $\Phi(\cdot)$ denotes the CDF function of a standard Gaussian random variable and $\tilde{Z}\Mydef Z\cdot\sqrt{\delta}\sim\mathcal{N}(0,1)$. We further average $\mathbb{E}_{\cdot |Z}[\mr{sign}(ZP)]$ over $Z$:
\BE\label{eq:expectedZP}
\begin{split}
\mathbb{E}\left[ \mr{sign}(ZP) \right]&=\mathbb{E}\left[ \mathbb{E}_{\cdot |Z}[\mr{sign}(ZP)] \right]\\
&=\mathbb{E}\left[2\Phi\left( \frac{\alpha }{\sigma}|\tilde{Z}| \right) \right]-1\\
&=\frac{2}{\pi}\mr{arctan}\left( \frac{\alpha}{\sigma} \right),
\end{split}
\EE
where the last step is due to the identity derived in \eqref{Eqn:Aux_f}. We next derive $\mathbb{E}[|ZP|]$. Conditioned on $Z$, $|ZP|$ is the magnitude of a Gaussian random variable (see \eqref{Eqn:S}), and its mean is given by \cite[(3)]{Leone1961}
\[
\begin{split}
\mathbb{E}\left[|ZP|\big| Z\right] &= 2\frac{\sigma|Z|}{\sqrt{\delta}}\cdot\phi\left( \frac{\alpha Z^2}{|Z|\sigma /\sqrt{\delta}} \right) +\alpha Z^2\left( 1-2\Phi\left( -\frac{\alpha Z^2}{|Z|\sigma /\sqrt{\delta}} \right) \right)\\
&= 2\frac{\sigma|Z|}{\sqrt{\delta}}\cdot\phi\left( \frac{\alpha Z^2}{|Z|\sigma /\sqrt{\delta}} \right) +\alpha Z^2\left( 2\Phi\left( \frac{\alpha Z^2}{|Z|\sigma /\sqrt{\delta}} \right) -1\right)\\
&=\frac{1}{\delta}\cdot\left[2\sigma\cdot|\tilde{Z}|\phi\left( \frac{\alpha |\tilde{Z}|}{\sigma} \right) +\alpha \cdot\tilde{Z}^2\left( 2\Phi\left( \frac{\alpha |\tilde{Z}|}{\sigma } \right) -1\right)\right].
\end{split}
\]
Again, in the last step we defined $\tilde{Z}\Mydef\sqrt{\delta}Z$. Averaging the above equality over $|\tilde{Z}|$ yields
\BE\label{Eqn:ZP_real}
\begin{split}
\mathbb{E}[|ZP|]&=\frac{1}{\delta}\cdot\mathbb{E}\left[2\sigma\cdot|\tilde{Z}|\phi\left( \frac{\alpha |\tilde{Z}|}{\sigma} \right) +\alpha \cdot\tilde{Z}^2\left( 2\Phi\left( \frac{\alpha |\tilde{Z}|}{\sigma } \right) -1\right)\right]\\
&\overset{(a)}{=} \frac{1}{\delta}\left\{ \frac{1}{\pi}\frac{2\sigma}{1+\frac{\alpha^2}{\sigma^2}}+2\alpha\cdot\left[ \frac{1}{\pi}\mr{arctan}\left(\frac{\alpha}{\sigma}\right)+\frac{1}{2} + \frac{1}{\pi}\frac{\alpha/\sigma}{1+\frac{\alpha^2}{\sigma^2}} \right] -\alpha \right\}\\
&=\frac{1}{\delta}\left\{ \frac{1}{\pi}\left(\frac{2\sigma^3+2\alpha^2\sigma}{\alpha^2+\sigma^2}\right)+\frac{2\alpha}{\pi}\mr{arctan}\left(\frac{\alpha}{\sigma}\right) \right\}\\
&=\frac{1}{\delta}\left\{ \frac{2\sigma}{\pi}+\frac{2\alpha}{\pi}\mr{arctan}\left(\frac{\alpha}{\sigma}\right) \right\},
\end{split}
\EE
where (a) is derived using the identities in \eqref{Eqn:Aux_f}.
Finally, combining \eqref{Eqn:SE_maps_real_original2}, \eqref{eq:expectedZP} and \eqref{Eqn:ZP_real}, and after some calculations, we finally obtain the following
\[
\begin{split}
\psi_1(\alpha,\sigma^2) &=\frac{2}{\pi}\mr{arctan}\left(\frac{\alpha}{\sigma}\right),\\
\psi_2(\alpha,\sigma^2;\delta,\sigma^2_w) &=\underbrace{ \frac{1}{\delta}\left[ \alpha^2+\sigma^2 +1-\frac{4\sigma}{\pi}-\frac{4\alpha}{\pi}\mr{arctan}\left(\frac{\alpha}{\sigma}\right)  \right] }_{\psi_2(\alpha,\sigma^2;\delta)} +\sigma^2_w
\end{split}
\]

\section{Continuity of the partial derivative $\frac{\partial \psi_2(\alpha, \sigma^2)}{\partial \sigma^2}$ at $(\alpha, \sigma^2) = (1,0)$} \label{sec:proofscontinuity}
Note that in the proof of Lemma \ref{lem:psi2}-(i) we showed that the $\lim_{(\alpha, \sigma^2) \rightarrow (1,0)} \frac{\partial \psi_2(\alpha, \sigma^2)}{\partial \sigma^2} = \frac{2}{\delta}$. Our goal here is to show that the derivative exists at $(\alpha, \sigma^2) = (1,0)$ and it is equal to $\frac{2}{\delta}$.

\subsection{Proof of the main claim}
Our goal in this section is to show that $\left. \frac{\partial \psi_2 (\alpha, \sigma^2)}{\partial \sigma^2} \right|_{(1, 0)} = \frac{2}{\delta}$. From the definition of the partial derivative, we have
\begin{eqnarray}
\left. \frac{\partial \psi_2 (\alpha, \sigma^2)}{\partial \sigma^2} \right|_{(1, 0)} &=& \lim_{\sigma^2 \rightarrow 0} \frac{1}{\sigma^2} (\psi_2(1, \sigma^2) - \psi_2(1,0)) \nonumber \\
&=&\lim_{\sigma^2 \rightarrow 0} \frac{4}{\delta \sigma^2} (1+\sigma^2 +1 - \int_{0}^{\pi/2} \frac{2 \sin^2 \theta + \sigma^2}{(\sin^2 \theta + \sigma^2)^{\frac{1}{2}}}d \theta -2 + \int_{0}^{\pi/2} 2 \sin \theta d\theta) \nonumber \\
&=&\lim_{\sigma^2 \rightarrow 0} \frac{4}{\delta \sigma^2} (\sigma^2 - \int_{0}^{\pi/2} \frac{2 \sin^2 \theta + \sigma^2}{(\sin^2 \theta + \sigma^2)^{\frac{1}{2}}}  d\theta+2)
\end{eqnarray}
Define $m \triangleq 1/\sigma^2$. Then,
\begin{eqnarray}
\left. \frac{\partial \psi_2 (\alpha, \sigma^2)}{\partial \sigma^2} \right|_{(1, 0)} &=& \lim_{m \rightarrow \infty} \frac{4m }{\delta} (\frac{1}{m} - \int_{0}^{\pi/2} \frac{2\sqrt{m} \sin^2 \theta + 1/\sqrt{m}}{(m \sin^2 \theta + 1)^{\frac{1}{2}}} d\theta +2) \nonumber \\
&\overset{(a)}{=}&  \lim_{m \rightarrow \infty} \frac{4m }{\delta} \left(\frac{1}{m} -2 \frac{(m+1) E(\frac{m}{m+1}) - K(\frac{m}{m+1} )}{\sqrt{m(m+1)}} - \frac{1}{\sqrt{m(m+1)}} K \left(\frac{m}{m+1}\right)+2 \right) \nonumber\\
&=&  \lim_{m \rightarrow \infty} \frac{4m }{\delta} \left(\frac{1}{m} -2 \frac{(m+1) E(\frac{m}{m+1}) }{\sqrt{m(m+1)}} + \frac{1}{\sqrt{m(m+1)}} K \left(\frac{m}{m+1}\right)+2 \right).
\end{eqnarray}
To obtain Equality (a) we have used \eqref{Eqn:integral_identities}. By employing Lemma \ref{Lem:elliptic} (i) we have\begin{eqnarray}
\lefteqn{ \lim_{m \rightarrow \infty} \frac{4m }{\delta} \left(\frac{1}{m} -2 \frac{(m+1) E(\frac{m}{m+1}) }{\sqrt{m(m+1)}} + \frac{1}{\sqrt{m(m+1)}} K \left(\frac{m}{m+1}\right)+2 \right)} \nonumber \\
&=& \lim_{m \rightarrow \infty} \frac{4m }{\delta} \left(\frac{1}{m} -2 \frac{(m+1)}{\sqrt{m(m+1)}} \left(1 + \frac{1}{2} \frac{\log 4 \sqrt{m+1}}{m+1} - \frac{1}{4 (m+1)}\right) + \frac{1}{\sqrt{m(m+1)}} \log 4 \sqrt{m+1}+2 \right) \nonumber \\
&=& \lim_{m \rightarrow \infty} \frac{4m }{\delta} \left(\frac{1}{m} -2 \frac{(m+1)}{\sqrt{m(m+1)}} \left(1 - \frac{1}{4 (m+1)}\right) +2 \right) \nonumber \\
&=& \lim_{m \rightarrow \infty} \frac{4m }{\delta} \left(\frac{1}{m} -2 \frac{(m+1)}{\sqrt{m(m+1)}}+2  + \frac{1}{2 \sqrt{m(m+1)}}\right)  \nonumber \\
 &=&\lim_{m \rightarrow \infty} \frac{4m }{\delta} \left(\frac{1}{m} -2 \frac{(m+1)}{\sqrt{m(m+1)}}+2  \right)   +  \lim_{m \rightarrow \infty} \frac{4m }{\delta} \left(\frac{1}{2 \sqrt{m(m+1)}}\right)   = 0 + \frac{2}{\delta}. 
\end{eqnarray}
Again we emphasize that we have also shown in the proof of Lemma \ref{lem:psi2} that $\lim_{(\alpha, \sigma^2) \rightarrow (1, 0)} \frac{\partial \psi_2(\alpha, \sigma^2)}{\partial \sigma^2} = \frac{2}{\delta}$. Hence, $\frac{\partial \psi_2(\alpha, \sigma^2)}{\partial \sigma^2}$ is continuous at $(\alpha, \sigma^2)= (1,0)$. 

%


\section{Asymptotic analysis of real-valued $\ampa$}

\subsection{Proof of Theorem \ref{The:PT_real}}\label{Sec:SE_conver_proof}
The proof of Theorem \ref{The:PT_real} is in parallel to that for Theorem~\ref{Theo:PhaseTransition_complex}. For this reason, we will only report the discrepancies. For intuition and more discussions, please refer to Section~\ref{ssec:proofTheo:PhaseTransition_complex}.

\subsubsection{Roadmap of the proof}
Again, we define $F_1(\sigma^2)$ to be the non-negative fixed point of $\psi_1$ and $F_2(\alpha,\delta)$ to be the fixed point of $\psi_2$, where $\psi_1$ and $\psi_2$ are now defined in \eqref{Eqn:map_expression_real}. Different from the complex-valued case, $\psi_2$ now has a unique fixed point. Properties of $\psi_1$ and $\psi_2$ are detailed in Section \ref{ssec:psi12_real}. Similar to complex-valued case, $F_1^{-1}(\alpha)$ and $F_2(\alpha;\delta)$ satisfy the following property:
 \begin{lemma}\label{Lem:F1_F2_real}
 If $\delta>\deltaAMP=\frac{\pi^2}{4}-1$, then $F_1^{-1}(\alpha)>F_2(\alpha;\delta)$ for $\alpha\in(0,1)$.
 \end{lemma}
 This lemma is proved in Section \ref{ssec:proofF1F2relation}. We will later use this lemma to show that when $\delta>\deltaAMP=\frac{\pi^2}{4}-1$ the state evolution converges to the desired fixed point $(\alpha,\sigma^2)=(1,0)$ for all initialization as long as $\alpha_0\neq0$. This means that AMP.A recovers the signal perfectly as long as the initial estimate is not orthogonal to the true signal.

Our next step is to analyze the dynamics of $\ampa$ for $\delta> \deltaAMP$. The following lemma implies that we only need to focus on the region where $\alpha\in[0,1]$.
\begin{lemma}
Let $\{\alpha_t\}_{t\ge1}$ and $\{\sigma^2_t\}_{t\ge1}$ be two sequences generated according to \eqref{Eqn:SE_real}. Then for any $\alpha_0\ge0$ and $\sigma^2_0\in\mathbb{R}_+$, we have $\alpha_t\in[0,1]$ for any $t\ge1$.
\end{lemma}
This lemma is a direct consequence of Lemma \ref{lem:psi1_real}-ii proved in Section \ref{ssec:psi12_real}. Hence, we skip its proof. Similar to \eqref{Eqn:left_bound}, the following function characterizes the lower boundary of the region that $(\alpha_t,\sigma^2_t)$ ($\forall t\ge1$) can fall into. 
\begin{definition}\label{def:L_realvaleued}
For any $\delta>0$ and $\alpha\in[0,1]$, define
\BE\label{Eqn:L_real}
\begin{split}
L(\alpha;\delta) \Mydef \frac{1}{\delta}\left\{1-\left[\frac{2}{\pi}\cos\left(\frac{\pi}{2}\alpha\right) +\alpha\sin\left(\frac{\pi}{2}\alpha\right)\right]^2\right\}.
\end{split}
 \EE
 \end{definition}
 
For the intuition about $L$ the reader may refer to Section \ref{ssec:proofTheo:PhaseTransition_complex}. As in the complex-valued signals case, the following  properties of this function play critical roles in the dynamics of the SE: 

\vspace{3pt}
\begin{lemma}\label{Lem:monotonicity_L_real}
$L(\alpha;\delta)$ defined in \eqref{Eqn:L_real} is a strictly decreasing function of $\alpha\in(0,1)$.
\end{lemma}
This is straightforward to see and hence the proof is skipped.

 \begin{lemma}\label{Lem:F1_F2_real2}
 If $\delta>\deltaAMP=\frac{\pi^2}{4}-1$, then $F_1^{-1}(\alpha)> L(\alpha;\delta)$ for any $\alpha\in(0,1)$
  \end{lemma}
\vspace{3pt}

We skip the proofs of this Lemma. The arguments are similar to Lemma \ref{Lem:F1_F2_complex} and the calculations are straightforward too. Similar to Definition \ref{Def:region_real}, we divide $\left\{(\alpha,\sigma^2):\alpha\in(0,1], \sigma^2\ge0\right\}$ into four subregions.
 
 \begin{definition}\label{Def:region_real}
We divide $\left\{(\alpha,\sigma^2):\alpha\in(0,1], \sigma^2\ge0\right\}$ into the following four sub-regions:
 \BE\label{Eqn:regions_real}
\begin{split}
\mathcal{R}_0&\Mydef\left\{(\alpha,\sigma^2)\big|0<\alpha\le1,\frac{4}{\pi^2}< \sigma^2 <\infty\right\},\\
\mathcal{R}_1&\Mydef\left\{(\alpha,\sigma^2)\big|0<\alpha\le1,F_1^{-1}(\alpha)< \sigma^2 \le\frac{4}{\pi^2}\right\},\\
\mathcal{R}_{2a} &\Mydef\left\{(\alpha,\sigma^2)\big|0<\alpha\le1,L(\alpha)\le\sigma^2 \le F_1^{-1}(\alpha)\right\},\\
 \mathcal{R}_{2b} &\Mydef\left\{(\alpha,\sigma^2)\big|0<\alpha\le1,0 \le \sigma^2 <L(\alpha;\delta)\right\}.
\end{split}
\EE
\end{definition}
Note that there are two differences between Definition \ref{Def:region_real} and Definition \ref{Def:regions}. First, the upper limit of $\sigma^2$ for $\mathcal{R}_1$ is changed from $\frac{\pi^2}{16}$ to $\frac{4}{\pi^2}$. Second, in Definition \ref{Def:regions}, $\sigma^2<\sigma^2_{\max}=\max\{1,\delta/4\}$ for $\mathcal{R}_0$, but in Definition \ref{Def:region_real}, the value of $\sigma^2$ for $\mathcal{R}_2$ is not upper bounded. Our next lemma shows that for any $(\alpha_0, \sigma_0^2)\in\mathcal{R}$, the states of the dynamical system \eqref{Eqn:SE_real} will eventually move to $\mathcal{R}_1$ or $\mathcal{R}_{2a}$. 

\begin{lemma}\label{Lem:regionII_IV_real}
Suppose that $\delta>\deltaAMP$. Let $\{\alpha_t\}_{t\ge1}$ and $\{\sigma^2_t\}_{t\ge1}$ be the  sequences generated according to \eqref{Eqn:SE_real} from any $\alpha_0>0$ and $\sigma_0\in\mathbb{R_+}$. 
\begin{enumerate}
\item[(i)] Starting from $t\ge1$, $(\alpha_t,\sigma^2_t)$ cannot be in $\mathcal{R}_{2b}$ for any $\alpha_0\neq0$ and $\sigma_0^2\ge0$. 
\item[(ii)] Let $(\alpha_0,\sigma_0^2)$ be an arbitrary point in $\mathcal{R}_0$. Then, there exists a finite number $T\ge1$ such that $(\alpha_T,\sigma_T^2)\in\mathcal{R}_1\cup\mathcal{R}_{2a}$.
\end{enumerate}
\end{lemma}

The proof of Lemma \ref{Lem:regionII_IV_real} is very similar to that of Lemma \ref{Lem:regionII_IV} and therefore skipped here. Finally, we complete the proof by proving the following lemma.

\begin{lemma}\label{Lem:regionI_III_real}
Suppose that $\delta>\deltaAMP$. If $(\alpha_{t_0},\sigma^2_{t_0})$ is in $\mathcal{R}_1\cup\mathcal{R}_{2a}$ at time $t_0$ (where $t_0\ge0$), and $\{\alpha_t\}_{t\ge t_0}$ and $\{\sigma^2_t\}_{t\ge t_0}$ are obtained via the SE in \eqref{Eqn:map_expression_real}, then
\begin{enumerate}
\item[(i)] $(\alpha_t,\sigma^2_t)$ remains in $\mathcal{R}_1\cup\mathcal{R}_{2a}$ for all $t>t_0$;
\item[(ii)] $(\alpha_t,\sigma^2_t)$ converges:
\[
\lim_{t\to\infty} \alpha_t = 1\quad{and}\quad\lim_{t\to\infty} \sigma^2_t = 0.
\]
\end{enumerate}
\end{lemma}

The proof of this lemma is presented in Section \ref{ssec:prooflemmaconvergenceR1R3_real}. 
\subsubsection{Properties of $\psi_1$ and $\psi_2$}\label{ssec:psi12_real}
In this section, we discuss several properties of $\psi_1$ and $\psi_2$.

\begin{lemma}\label{lem:psi1_real} 
$\psi_1\left(\alpha,\sigma^2\right)$ in \eqref{Eqn:map_expression_real_a} has the following properties (for $\alpha\ge0$): 
\begin{enumerate}
\item[(i)] $\psi_1\left(\alpha,\sigma^2\right)$ is a concave and strictly increasing function of $\alpha>0$, for any given $\sigma^2>0$.
\item[(ii)] $0<\psi_1(\alpha,\sigma^2)<1$, for $\alpha>0$ and $\sigma^2>0$.
\item[(iii)] If $\sigma^2 < 4/\pi^2$, then there are two nonnegative solutions to $\alpha=\psi_1(\alpha,\sigma^2)$: $\alpha=0$ and $\alpha=F_1(\sigma^2)>0$. Further, $F_1(\sigma^2)$ is strongly globally attracting.
On the other hand, if $\sigma^2 \ge 4/\pi^2$ then $\alpha =0$ is the unique nonnegative fixed point and it is strongly globally attracting.
 
\end{enumerate}
\end{lemma}
\begin{proof}
The proof strategy is similar to the one given in Section \ref{ssec:psi1_2prop}. Also, the calculations are straightforward. Hence, to save some space we skip the proof of this lemma. 
\end{proof}

\begin{lemma}\label{lem:psi2_real}
$\psi_2\left(\alpha,\sigma^2;\delta\right)$ has the following properties: 
\begin{enumerate}
\item[(i)] If $\delta<1$, then $\sigma^2=0$ is a locally unstable fixed point to $\sigma^2=\psi_2\left(\alpha,\sigma^2;\delta\right)$ for any $\alpha>0$, meaning that
\[
\frac{\partial \psi_2(\alpha,\sigma^2;\delta)}{\partial\sigma^2}\Big|_{\sigma^2=0}>1.
\]
\item [(ii)] For any $\delta>1$, $\sigma^2=\psi_2\left(\alpha,\sigma^2;\delta\right)$ has a unique fixed point, denoted as $F_2(\alpha;\delta)$, in $\sigma^2\in[0,\infty)$ for any $\alpha\in[0,1]$. Further, the fixed point is weakly globally attracting in $\sigma^2\in[0,\infty)$.

\item[(iii)] For any $\delta\ge0$, $\psi_2(\alpha,\sigma^2;\delta)$ is an increasing function of $\sigma^2\ge0$ if
\BE\label{Eqn:alpha_ast_def_real}
\alpha>\alpha_{\ast}=\frac{1}{\pi}.
\EE
Further, in this case $F_2(\alpha;\delta)$ is strongly globally attracting in $\sigma^2\in[0,\infty)$.

\end{enumerate}
\end{lemma}

\begin{proof} Recall from \eqref{Eqn:map_expression_real_b} that $\psi_2(\alpha,\sigma^2;\delta)$ is defined as
\[
\psi_2(\alpha,\sigma^2;\delta)=
\frac{1}{\delta}\left[\alpha^2+\sigma^2+1-\frac{4\sigma}{\pi} -\frac{4\alpha}{\pi}\mr{arctan}\left( \frac{\alpha}{\sigma} \right)\right]. 
\]

\textit{Proof of (i):} The partial derivative of $\psi_2$ w.r.t. $\sigma^2$ is 
\BE\label{Eqn:psi2_dif_real}
\begin{split}
\frac{\partial \psi_2(\alpha,\sigma^2;\delta)}{\partial \sigma^2}=\frac{1}{\delta}\left( 1-\frac{2}{\pi}\frac{\sigma}{\alpha^2+\sigma^2} \right).
\end{split}
\EE
The claims follows from the following fact:
\[
\begin{split}
\frac{\partial \psi_2(\alpha,\sigma^2;\delta)}{\partial \sigma^2}\Big|_{\sigma^2=0}=\frac{1}{\delta},\quad\forall \alpha>0.
\end{split}
\]

\textit{Proof of (ii):}
From \eqref{Eqn:psi2_dif_real}, we see that the following holds for any $\alpha\ge0$ and $\delta>0$:
\[
\frac{\partial \psi_2(\alpha,\sigma^2;\delta)}{\partial \sigma^2}<1,\quad\forall\sigma^2>0.
\]
Hence, the function $\Psi_2(\alpha,\sigma^2;\delta)=\psi_2(\alpha,\sigma^2;\delta)-\sigma^2$ is strictly decreasing on $\sigma^2\in\mathbb{R}_{+}$. Since $\Psi_2(\alpha,0;\delta)=\frac{1}{\delta}(\alpha-1)^2\ge0$ and $\Psi_2(\alpha,\infty;\delta)=-\infty$ for $\delta>1$ (which is easy to show from the definition of $\psi_2$), it follows that there exists a unique fixed point, denoted as $F_2(\alpha;\delta)$, to the following equation:
\[
\psi_2(\alpha,\sigma^2;\delta)-\sigma^2=0.
\]
Further, using similar arguments as those in the proof of Lemma \ref{lem:psi2}, we can prove that $F_2(\alpha;\delta)$ is globally attracting in $\sigma^2\in[0,\infty)$.

\textit{Proof of (iii):} When $\psi_2(\alpha,\sigma^2;\delta)$ is an increasing function of $\sigma^2$ in $[0,\infty)$, we have
\[
\frac{\partial \psi_2(\alpha,\sigma^2;\delta)}{\partial \sigma^2}=1-\frac{2}{\pi}\frac{\sigma}{\alpha^2+\sigma^2}>0,\quad\forall \sigma^2\ge0.
\]
or
\[
\alpha^2 >\frac{2}{\pi}\sigma-\sigma^2,\quad\sigma^2\ge0.
\]
It is easy to show that the maximum of the RHS over $\sigma^2\ge0$ is $\frac{1}{\pi^2}$. Hence, $\psi_2(\alpha,\sigma^2)$ is a strictly increasing function of $\sigma^2$ in $[0,\infty)$ if $\alpha>\frac{1}{\pi}$.
\end{proof}

\subsubsection{Properties of $F_1$ and $F_2$}

In this section we derive the main properties of the functions $F_1$ and $F_2$.  

\begin{lemma}\label{Lem:2_real}
The following hold for $F_1(\sigma^2)$ and $F_2(\alpha;\delta)$ (for $\delta>1$):
\begin{enumerate}
\item[(i)] $F_1(0)=1$ and $\lim_{\sigma^2\rightarrow \frac{4}{\pi^2}^{-}}F_1(\sigma^2)=0$. Further, by defining $F_1(\frac{4}{\pi^2})=0$, we have $F_{1}(\sigma^2)$ is continuous on $\left[0, \frac{4}{\pi^2}\right]$ and strictly decreasing in $\left(0,\frac{4}{\pi^2}\right)$;
\item[(ii)] $F_2(0;\delta)=\left(\frac{-\frac{2}{\pi} +\sqrt{ \frac{4}{\pi^2} +\delta - 1 }}{\delta - 1}\right)^2$ and $F_2(1;\delta)=0$.
\end{enumerate}
\end{lemma}
\begin{proof}
The proof is similar to the proof of Lemma \ref{Lem:2}.
\end{proof}

\subsubsection{Proof of Lemma \ref{Lem:F1_F2_real}}\label{ssec:proofF1F2relation}

It is straightforward to show that $F_2(\alpha;\delta)$ is a decreasing function of $\delta$ for any $\alpha\in[0,1]$. Hence, we only need to prove the lemma for the case where $\delta = \deltaAMP$. Based on the same arguments detailed in Section \ref{proof:lemmadominationF_1}, it suffices to prove the following inequality:
\BE\label{Eqn:F1_F2_real}
\psi_2\left(\alpha,F_1^{-1}(\alpha);\deltaAMP\right)<F_1^{-1}(\alpha),\quad \forall\alpha\in(0,1).
\EE
We make the following variable change:
\[
t=g^{-1}(\alpha),
\]
where $g:(0,\infty)\mapsto(0,1)$ is defined as (with some abuse of notations)
 \BE\label{Eqn:phi1_real}
 \begin{split}
g(t)&\Mydef \frac{2}{\pi}\mr{arctan}\left(t\right).
\end{split}
\EE 
Based on this re-parameterization, \eqref{Eqn:F1_F2_real} becomes
\BE\label{Eqn:F1_F2_real2}
 \psi_2\left( g(t), \frac{g^2(t)}{t^2};\deltaAMP \right) < \frac{g^2(t)}{t^2},\quad\forall t>0.
\EE
Substituting the definition of $\psi_2$ in \eqref{Eqn:map_expression_real_b} into \eqref{Eqn:F1_F2_real2} and after some straightforward calculations, it can be shown that \eqref{Eqn:F1_F2_real2} is implied by the following:
\BE\label{Eqn:G_real}
G(t)\Mydef t^2+\frac{4}{\pi}\frac{t}{g(t)}-\frac{t^2}{g^2(t)}>1-\deltaAMP,\quad\forall t>0.
\EE
From \eqref{Eqn:phi1_real} and \eqref{Eqn:G_real} and noting $\deltaAMP=\frac{\pi^2}{4}-1$, we can verify that $\lim_{t\to0_+}G(t)=1-\deltaAMP$. Consequently, it suffices to prove that $G(t)$ is strictly increasing on $(0,\infty)$. To this end, we calculate $G'(t)$:
\BE\label{Eqn:G_prime}
\begin{split}
G'(t)&=2t+\frac{4}{\pi}h'(t)-2h(t)\cdot h'(t) \\
&=2h(t) h'(t)\cdot\left(\frac{t}{h(t)h'(t)}+\frac{2}{\pi}\frac{1}{h(t)} -1 \right),
\end{split}
\EE
where for convenience we defined
\BE\label{Eqn:h_def}
h(t)\Mydef \frac{t}{g(t)}=\frac{\pi}{2}\frac{t}{\mr{arctan}(t)}.
\EE
We first note that $h'(t)>0$:
\BE\label{Eqn:h_prime}
\begin{split}
h'(t)=\frac{\pi}{2}\cdot\frac{\mr{arctan}(t) -\frac{t}{1+t^2}}{\mr{arctan}^2(t)}>0,\quad t>0,
\end{split}
\EE
where the inequality follows since $\mr{arctan}(t)-\frac{t}{1+t^2}$ is strictly increasing on $(0,\infty)$ and $[\mr{arctan}(t)-\frac{t}{1+t^2}]|_{t=0}=0$. Hence, to prove $G'(t)>0$, we only need to prove that (cf.~\eqref{Eqn:G_prime})
\BE\label{Eqn:F1_F2_ineq}
\underbrace{\frac{t}{h(t)h'(t)}}_{G_1(t)}+\underbrace{\frac{2}{\pi}\frac{1}{h(t)}}_{G_2(t)} -1>0,\quad\forall t>0.
\EE
Similar to the treatment in Section \ref{proof:lemmadominationF_1}, we consider two different cases: (1) $0<t\le0.75$ and (2) $t\ge0.75$.
\begin{itemize}
\item[(i)] Case I: $0<t\le0.75$. From \eqref{Eqn:h_prime}, $h(t)$ is a strictly increasing function of $t>0$, and thus $G_2(t)=\frac{2}{\pi h(t)}$ is strictly decreasing.

We next show that $G_1(t)$ is an increasing function of $t>0$. The derivative of $G_1(t)$ is given by:
\[
\begin{split}
G_1'(t) &\overset{(a)}{=} \left( \frac{4}{\pi^2}\frac{\mr{arctan}^3(t)}{\mr{arctan}(t)-\frac{t}{1+t^2}} \right)' \\
&=\frac{ 4 }{\pi^2}\cdot \frac{ \mr{arctan}^2(t) \left[(3+t^2)\mr{arctan}(t)-3t\right] }{\left[  t-(1+t^2)\mr{arctan}(t)\right]^2  }\\
&\overset{(b)}{>}0,
\end{split}
\]
where (a) is from \eqref{Eqn:F1_F2_ineq}, \eqref{Eqn:h_def} and \eqref{Eqn:h_prime}, and (b) is a consequence of the following facts: (i) $[(3+t^2)\mr{arctan}(t)-3t]_{t=0}=0$, (ii) $[(3+t^2)\mr{arctan}(t)-3t]'=2t\left(\mr{arctan}(t)-\frac{t}{1+t^2}\right)>0$ (similar to \eqref{Eqn:h_prime}).

The following proof is based on the idea introduced in Section \ref{proof:lemmadominationF_1}: since $G_1(t)$ is an increasing function and $G_2(t)$ is a decreasing function, the following holds for any $c_2>c_1>0$:
\[
G_1(c_1)+G_2(c_2)-1>0\Longrightarrow G_1(t)+G_2(t)-1>0,\quad \forall t\in[c_1,c_2].
\]
We verified that $G_1(c_1)+G_2(c_2)-1>0$ holds for a sequence of intervals: $[c_1,c_2]=[0,0.32]$, $[c_1,c_2]=[0.32,0.45]$, $[c_1,c_2]=[0.45,0.55]$, $[c_1,c_2]=[0.55,0.64]$, $[c_1,c_2]=[0.64,0.7]$, $[c_1,c_2]=[0.7,0.75]$. Altogether, we proved $G_1(t)+G_2(t)-1>0$ for $t\in(0,0.75]$.

\item[(ii)]
Case II: $t\ge0.75$. From the definitions in \eqref{Eqn:F1_F2_ineq}, \eqref{Eqn:h_def} and \eqref{Eqn:h_prime}, and based on some calculations not shown here, we write the LHS of \eqref{Eqn:F1_F2_ineq} as
\BE\label{Eqn:F1_F2_ineq2}
\begin{split}
G_1(t)+G_2(t) -1& =\frac{4}{\pi^2} \cdot \underbrace{\frac{(t^3 + t) \cdot \mr{arctan}^3 (t) + (t^2
  + 1)\mr{arctan}^2 (t) - \mr{arctan} (t) \cdot t}{(t^3 + t)\mr{arctan}
  (t) - t^2} }_{R(t)} - 1.
\end{split}
\EE
From $\lim_{t\to\infty}\mr{arctan}(t)=\pi/2$, it is easy to see that 
\[
\lim_{t\to\infty}G_1(t)+G_2(t) -1 = 0.
\]
Hence, to prove $G_1(t)+G_2(t) -1>0$ for $t\ge0.75$, it suffices to show that $R(t)$ in \eqref{Eqn:F1_F2_ineq2} is strictly decreasing on $[0.75,\infty)$. To this end, we calculate $R'(t)$ below:
\[
R'(t) =  \frac{(t^4 - 1) \mr{arctan}^3 (t) + t^3 + 3 (t^3 + t) \mr{arctan}^2 (t) - 3
   (t^4 + t^2) \mr{arctan} (t)}{ t^2(1+t^2)\left[ t-(1+t^2)\mr{arctan}(t) \right]^2 }\Mydef\frac{N(t)}{D(t)}.
\]
Since $D(t)>0$, we have
\[
R'(t)<0\Longleftrightarrow N(t)<0.
\]
To this end, it can be shown that
\[
N'(t) =4 t^2 \cdot \mr{arctan} (t) \cdot [ t \cdot
  \mr{arctan}^2 (t) + 3 \cdot \mr{arctan} (t) - 3 t ] .
\]
Hence, to prove $N'(t)<0$ for $t\ge0.75$, we only need to prove
\[
t \cdot \mr{arctan}^2 (t) + 3 \cdot \mr{arctan} (t) - 3 t<0,\quad \forall t\ge0.75,
\]
which is equivalent to proving
\[
\mr{arctan}(t)<\frac{-3+\sqrt{9+12t^2}}{2t},\quad \forall t\ge0.75.
\]
It is proved in \cite[Theorem 3]{Zhu2008} that
\[
\mr{arctan}(t)< \frac{8 t}{3 + \sqrt{25 + \frac{256}{\pi^2} t^2}},\quad \forall t>0.
\]
Hence, it suffices to prove
\[
\frac{8 t}{3 + \sqrt{25 + \frac{256}{\pi^2} t^2}}< \frac{-3+\sqrt{9+12t^2}}{2t}=\frac{6 t}{3 + \sqrt{9 + 12 t^2}},\quad \forall t\ge0.75.
\] 
which, after some straightforward manipulations, reduces to 
\[
  3 + 4 \sqrt{9 + 12 t^2}-3 \sqrt{25 +\frac{256}{\pi^2} t^2}<0,\quad \forall t>0.75.
\]
We can verify that the above inequality holds for $t=0.75$. We complete our proof by showing that the LHS of the above inequality is decreasing in $t\in[0.75,\infty)$:
\[
\begin{split}
\left(  3 + 4 \sqrt{9 + 12 t^2}-3 \sqrt{25 +\frac{256}{\pi^2} t^2}\right)' 
&=48t\cdot\left(\frac{1}{\sqrt{9+12t^2}} -\frac{1}{\sqrt{\frac{25\pi^4}{256}+\pi^2t^2}} \right)\\
&=48t\cdot\frac{ (\pi^2-12)t^2+\frac{25\pi^4}{256}-9 }{T_1^2T_2+T_1T_2^2}\\
&<0,\quad \forall t>0.75,
\end{split}
\]
where $T_1\Mydef \sqrt{9+12t^2}$ and $T_2\Mydef \sqrt{\frac{25\pi^4}{256}+\pi^2t^2}$, and the last inequality can be easily proved since $(\pi^2-12)t^2+\frac{25\pi^4}{256}-9<0 $ is a strictly decreasing function of $t$ and $ [(\pi^2-12)t^2+\frac{25\pi^4}{256}-9]_{t=0.75}<0 $.

\end{itemize}

\subsubsection{Proof of Lemma \ref{Lem:regionI_III_real}} \label{ssec:prooflemmaconvergenceR1R3_real}
\begin{itemize}

\item \textbf{Preliminaries}

\begin{lemma}\label{lem:LhatL_real}
For any $\alpha>0$ and $\delta>0$, $L(\alpha;\delta)$ satisfies 
 \BE\label{Eqn:Lhat_real}
L(\alpha, \delta) \geq \hat{L}(\alpha;\delta) \Mydef \frac{1}{\delta}\left( 1-\frac{4}{\pi^2}-\alpha^2 \right).
 \EE
 \end{lemma}
 \begin{proof}
 According to Definition \ref{def:L_realvaleued} we have 
 \BE
\begin{split}
L(\alpha;\delta) \Mydef \frac{1}{\delta}\left\{1-\left[\frac{2}{\pi}\cos\left(\frac{\pi}{2}\alpha\right) +\alpha\sin\left(\frac{\pi}{2}\alpha\right)\right]^2\right\}.
\end{split}
 \EE
Then, the inequality $ \hat{L}(\alpha;\delta) \leq L(\alpha;\delta) $ is equivalent to
 \[
 \left[\cos\left(\frac{\pi}{2}\right),\sin\left(\frac{\pi}{2}\right)\right]\left[\frac{2}{\pi},\alpha\right]^\UT\le \sqrt{ \frac{4}{\pi^2} +\alpha^2},
 \]
 which is clear from the Cauchy-Schwartz Inequality. 
 \end{proof} 

 \begin{lemma}\label{lem:psi2deltaAMP}
For any $\alpha\in[0,1]$, $\psi_2(\alpha,\sigma^2;\deltaAMP)$ in \eqref{Eqn:map_expression_real} is an increasing function of $\sigma^2$ in $\sigma^2\in[L(\alpha;\deltaAMP),\infty)$.
 \end{lemma}
 \begin{proof}
In Lemma \ref{lem:psi2_real}, we proved that $\psi_2$ is strictly increasing on $\sigma^2>0$ for $\alpha\ge1/\pi$. Hence, we only need to consider the case $\alpha<\alpha_{\ast}=\frac{1}{\pi}$. From the expression of $\psi_2$ in \eqref{Eqn:map_expression_real}, it is straightforward to see that $\psi_2$ is increasing on $\sigma^2\in[\sigma^2_2(\alpha),\infty)$ (for $\alpha<1/\pi$), where
\[
\sigma^2_2(\alpha)\Mydef\left(  \frac{1}{\pi}+\sqrt{ \frac{1}{\pi^2} -\alpha^2}\right)^2.
\]
Lemma \ref{lem:LhatL_real} shows that $\hat{L}(\alpha;\delta)$ is a lower bound of $L(\alpha;\delta)$ for any $\alpha\in(0,1)$. Hence, it suffices to prove that
 \BE\label{Eqn:psi2_L_mono_real_1}
 \hat{L}(\alpha;\deltaAMP)=\frac{1}{\deltaAMP}\left(1-\frac{4}{\pi^2} -\alpha^2\right) \ge \sigma^2_2(\alpha),\quad\alpha\in\left[0,\pi^{-1}\right].
 \EE
Noting $\deltaAMP=\frac{\pi^2}{4}-1$, it can be shown that to prove \eqref{Eqn:psi2_L_mono_real_1} it suffices to prove
\[
\begin{split}
\frac{1}{\pi} +\frac{\pi}{2}\frac{\pi^2-8}{\pi^2-4}\alpha^2 &\ge \sqrt{ \frac{1}{\pi^2} -\alpha^2},\quad\alpha\in[0,\pi^{-1}].
\end{split}
\]
which holds since the LHS is lower bounded by $1/\pi$ while the RHS is upper bounded by $1/\pi$.
 \end{proof}

  \begin{lemma}\label{Lem:mono_delta_real}
$\psi_2(\alpha,L(\alpha, \delta);\delta)$ is a decreasing function of $\delta>0$ for any $\alpha>0$. 
 \end{lemma}
 \begin{proof}
 Note that we can represent $L(\alpha, \delta)$ as $\frac{1}{\delta}\bar{\sigma}^2$, where $\bar{\sigma}^2$ is a number that does not depend on $\delta$. Hence, we will prove that  $\psi_2\left(\alpha,\frac{1}{\delta}\bar{\sigma}^2;\delta\right)$ is a decreasing function of $\delta$ for any fixed $\alpha>0$ and $\bar{\sigma}^2>0$. From the definition of $\psi_2$ in \eqref{Eqn:map_expression_real_b}, we have
 \[
 \begin{split}
\psi_2\left(\alpha,\frac{1}{\delta}\bar{\sigma}^2;\delta\right) &=\frac{1}{\delta}\left[ \alpha^2+\frac{1}{\delta}\bar{\sigma}^2 +1-\frac{4\bar{\sigma}}{\pi\sqrt{\delta}}-\frac{4\alpha}{\pi}\mr{arctan}\left(\frac{\alpha\sqrt{\delta}}{\bar{\sigma}}\right)  \right] \\
&\overset{(a)}{=}\frac{1}{\delta}\left[ (\alpha-1)^2+\frac{1}{\delta}\bar{\sigma}^2 -\frac{4\bar{\sigma}}{\pi\sqrt{\delta}}+\frac{4\alpha}{\pi}\mr{arctan}\left(\frac{\bar{\sigma}}{\alpha\sqrt{\delta}}\right)  \right]\\
&\overset{(b)}{=}(\alpha-1)^2\beta^2+ \alpha^2\bar{s}^2\beta^4 -\frac{4\bar{s} \alpha }{\pi}\beta^3+\frac{4\alpha}{\pi}\mr{arctan}\left(\beta \bar{s}\right)\beta^2,
 \end{split}
 \]
 where (a) follows from the identity $\mr{arctan}\left(\frac{1}{s}\right)=\frac{\pi}{2}-\mr{arctan}(s)$, and in (b) we introduced the following definitions:
 \[
 \beta\Mydef \frac{1}{\sqrt{\delta}}\quad\text{and}\quad\bar{s}\Mydef \frac{\bar{\sigma}}{\alpha}.
 \]
 We then calculate the derivative of $\psi_2\left(\alpha,\frac{1}{\delta}\bar{\sigma}^2;\delta\right)=\psi_2\left(\alpha,\beta^2\bar{\sigma}^2;\beta^{-2}\right)$ w.r.t. $\beta$:
 \[
 \begin{split}
 \frac{\partial \psi_2\left(\alpha,\beta^2\bar{\sigma}^2;\beta^{-2}\right)}{\partial\beta}
 &=\beta\left[  2(\alpha-1)^2 + 4\alpha^2\bar{s}^2\beta^2-\frac{12\bar{s}\alpha}{\pi}\beta+ \frac{8\alpha}{\pi}\mr{arctan}\left(\beta \bar{s}\right)+\frac{4\alpha\beta}{\pi}\frac{\bar{s}}{1+\beta^2\bar{s}^2}\right]\\
 &=2\beta\left[  (\alpha-1)^2 + 2\alpha^2{s}^2-\frac{6{s}\alpha}{\pi}+ \frac{4\alpha}{\pi}\mr{arctan}\left({s}\right)+\frac{2\alpha}{\pi}\frac{s}{1+s^2}\right],
 \end{split}
 \]
 where in the last step we defined $s\Mydef\beta\bar{s}$.
It suffices to prove that
 \[
 (\alpha-1)^2 + 2\alpha^2{s}^2-\frac{6{s}\alpha}{\pi}+ \frac{4\alpha}{\pi}\mr{arctan}\left({s}\right)+\frac{2\alpha}{\pi}\frac{s}{1+s^2}>0,
 \]
 or
 \[
(1+2s^2)\alpha^2+\left[ \frac{4}{\pi}\mr{arctan}(s)+\frac{2s}{\pi(1+s^2)}-\frac{6s}{\pi} -2\right]\alpha+1>0.
 \]
 We prove by showing that the discriminant of the above quadratic function (of $\alpha$) is negative:
 \[
 \left[ 2+\frac{6s}{\pi}-\frac{4}{\pi}\mr{arctan}(s)-\frac{2s}{\pi(1+s^2)} \right]^2 - 4 (1+2s^2)<0.
 \]
We next prove that the following two inequalities hold:
\BE\label{Eqn:lem_mono_delta_real1}
2+ \frac{6s}{\pi}-\frac{4}{\pi}\mr{arctan}(s)-\frac{2s}{\pi(1+s^2)}  > 0,
\EE
and
\BE\label{Eqn:lem_mono_delta_real2}
2+\frac{6s}{\pi}-\frac{4}{\pi}\mr{arctan}(s)-\frac{2s}{\pi(1+s^2)}-2\sqrt{1+2s^2}<0.
\EE
First, \eqref{Eqn:lem_mono_delta_real1} follows from the following facts: (i) $2>\frac{4}{\pi}\mr{arctan}(s)$ and (ii) $3>1/(1+s^2)$.  We rewrite \eqref{Eqn:lem_mono_delta_real2} as
\BE\label{Eqn:lem_mono_delta_real3}
\begin{split}
\frac{3s}{\pi}-\frac{2}{\pi}\mr{arctan}(s)-\frac{s}{\pi(1+s^2)}<\sqrt{1+2s^2}-1=\frac{2s^2}{1+\sqrt{1+2s^2}}.
\end{split}
\EE
Using $\mr{arctan}(s)>s/(1+s^2)$ (see \eqref{Eqn:h_prime}), we can upper bound the LHS by
\[
\begin{split}
\frac{3s}{\pi}-\frac{2}{\pi}\mr{arctan}(s)-\frac{s}{\pi(1+s^2)} <\frac{3}{\pi}\frac{s^3}{1+s^2}.
\end{split}
\]
Hence, to prove \eqref{Eqn:lem_mono_delta_real3}, it is sufficient to prove
\[
\frac{3}{\pi}\frac{s^3}{1+s^2} < \frac{2s^2}{1+\sqrt{1+2s^2}},
\]
or 
\[
\frac{3}{\pi}\frac{s}{1+s^2} < \frac{2}{1+\sqrt{1+2s^2}},
\]
which holds since (i) LHS is an increasing function of $s$ while the RHS is a decreasing function, and (ii) equality holds when $s\to\infty$.

 \end{proof}

 \begin{lemma}\label{Lem:RegionIII_bound_real}
 For any $(\alpha,\sigma^2)\in\mathcal{R}_{2a}$ and $\delta\ge\deltaAMP=\frac{\pi^2}{4}-1$, we have $\psi_2(\alpha,\sigma^2;\delta)<F_1^{-1}(\alpha)$, where $\mathcal{R}_{2a}$ is defined in \eqref{Eqn:regions_real}.
 \end{lemma}
 \begin{proof}
 The proof is similar to that of Lemma \ref{Lem:RegionIII_bound}. We consider three different cases:
 \begin{enumerate}
 \item [(i)] $\alpha\in[\pi^{-1},1]$ and $\delta\in[\deltaAMP,\infty]$.
 \item [(ii)] $\alpha\in[0,\pi^{-1})$ and $\delta\in[\deltaAMP,\delta_*]$.
 \item [(iii)] $\alpha\in[0,\pi^{-1}]$ and $\delta\in[\delta_*,\infty)$,
 \end{enumerate}
 where $\delta_* = \frac{  1-\left[\frac{2}{\pi}\cos(0.5)+\frac{1}{\pi}\sin(0.5)\right]^2 }{\frac{1}{\pi^2}}\approx 4.87$.
 
 \textit{Case (i):} In Lemma \ref{lem:psi2_real}, we proved that $\psi_2$ is strictly increasing on $\sigma^2>0$ for $\alpha\ge1/\pi$. Since in $\mathcal{R}_{2a}$ $\sigma^2 < F_1^{-1} (\alpha)$, the proof of 
 \[
 \psi_2(\alpha,\sigma^2;\delta)<F_1^{-1}(\alpha)
 \]
 on $\mathcal{R}_{2a}$ reduces to the proof of 
 \[
 \max_{\sigma^2 < F_1^{-1} (\alpha)} \psi_2(\alpha,\sigma^2;\delta) = \psi_2(\alpha, F_1^{-1} (\alpha); \delta) <F_1^{-1}(\alpha).
 \]
 The last equality is clear from the global attractiveness of $F_2(\alpha)$ in $\psi_2$ that is proved in Lemma \ref{lem:psi2_real}-ii and the fact that $F_2(\alpha) < F_1^{-1} (\alpha)$ that is proved in Lemma \ref{Lem:F1_F2_real}.

 \textit{Case (ii):} As shown in \eqref{Eqn:psi2_dif_real} we have
 \BE\label{Eqn:psi2_dif_real}
\begin{split}
\frac{\partial \psi_2(\alpha,\sigma^2;\delta)}{\partial \sigma^2}=\frac{1}{\delta}\left( 1-\frac{2}{\pi}\frac{\sigma}{\alpha^2+\sigma^2} \right).
\end{split}
\EE
 Hence, $\psi_2$ has two stationary points if $\alpha\in[0,\pi^{-1})$:
 \[
 \begin{split}
 \sigma^2_1(\alpha) & = \left(\frac{1}{\pi}-\sqrt{\frac{1}{\pi^2} -\alpha^2 }\right)^2,\\
  \sigma^2_2(\alpha) & = \left(\frac{1}{\pi}+\sqrt{\frac{1}{\pi^2} -\alpha^2 }\right)^2,
 \end{split}
\]
where $\sigma^2_1(\alpha)$ is a local maximum and $\sigma^2_2(\alpha)$ is a local minimum. Then, the maximum of $\psi_2$ over $\sigma^2\in[L(\alpha;\delta),F_1^{-1}(\alpha)]$ can only happen at either $L(\alpha;\delta)$ or $F_1^{-1}(\alpha)$ if the following holds:
\[
L(\alpha;\delta) \ge  \sigma^2_1(\alpha),\quad\forall \alpha\in[0,\pi^{-1}).
\]
Since $L(\alpha;\delta)$ is a decreasing function of $\alpha$ (which can be confirmed with a straightforward calculation of the derivative), then the following holds for $\alpha<\pi^{-1}$:
\[
L(\alpha;\delta) \ge L(\pi^{-1};\delta)=\frac{1}{\delta}\left\{ 1-\left[\frac{2}{\pi}\cos(0.5)+\frac{1}{\pi}\sin(0.5)\right]^2 \right\}\approx\frac{0.494}{\delta}.
\]
Further, $\sigma_1^2(\alpha)$ is an increasing function of $\alpha$ and is upper bounded by
\[
\sigma^2_1(\alpha)<\frac{1}{\pi^2},\quad\forall \alpha\in[0,\pi^{-1})
\]
Hence, $L(\alpha;\delta)\ge \sigma_1^2(\alpha)$ when
\[
\delta <\frac{  1-\left[\frac{2}{\pi}\cos(0.5)+\frac{1}{\pi}\sin(0.5)\right]^2 }{\frac{1}{\pi^2}}=\delta^*\approx 4.87.
\]
Now, suppose that $\delta<\delta^*$. Then, proving that $\psi_2(\alpha,\sigma^2;\delta)<F_1^{-1}(\alpha)$ is equivalent to proving:
\[
\max \{\psi_2(\alpha,L(\alpha;\delta);\delta), \psi_2(\alpha,F_1^{-1}(\alpha);\delta)\} < F_1^{-1}(\alpha).
\]

The rest of the argument is similar to the ones used in the proof of Lemma \ref{Lem:RegionIII_bound}. Since according to Lemma \ref{Lem:mono_delta_real} $\psi_2(\alpha,L(\alpha;\delta);\delta)$ is a decreasing function of $\delta$, and trivially $\psi_2(\alpha,F_1^{-1}(\alpha);\delta)\}$ is a decreasing function of $\delta$ we need to prove that 
\begin{equation}\label{eq:almostfinaleq}
\max \{\psi_2(\alpha,L(\alpha;\delta_{\rm AMP});\delta_{\rm AMP}), \psi_2(\alpha,F_1^{-1}(\alpha_{\rm AMP});\delta_{\rm AMP})\} \leq F_1^{-1}(\alpha).
\end{equation}
Also, since according to Lemma \ref{lem:psi2deltaAMP}, we have $\max \{\psi_2(\alpha,L(\alpha;\delta_{\rm AMP});\delta_{\rm AMP}), \psi_2(\alpha,F_1^{-1}(\alpha_{\rm AMP});\delta_{\rm AMP})\} = \psi_2(\alpha,F_1^{-1}(\alpha_{\rm AMP});\delta_{\rm AMP})$,  \eqref{eq:almostfinaleq} simplifies to:
\[
\psi_2(\alpha,F_1^{-1}(\alpha_{\rm AMP});\delta_{\rm AMP}) \leq F_1^{-1}(\alpha),
\]
which is a simple implication of the global attractiveness of $F_2(\alpha)$ in $\psi_2$ that is proved in Lemma \ref{lem:psi2_real}-ii.

\textit{Case (iii):} Since $F_1(\sigma^2)$ is the solution of $\alpha = \psi_1(\alpha, \sigma^2)=\frac{2}{\pi}\mr{arctan}(\alpha/\sigma)$, we can show that $F_1^{-1}(\alpha)=\alpha^2\cdot\cot^2\left(\frac{\pi}{2}\alpha\right)$. Since $F_1^{-1}(\alpha)$ is a decreasing function, we have
\BE\label{eq:lbF_1minus1realthirdcase}
F_1^{-1}(\alpha)>F_1^{-1}(\pi^{-1})\approx 0.339,\quad\alpha\in[0,\pi^{-1}).
\EE
Further, if the following holds for $\alpha\in[0,\pi^{-1})$ we would have proved that $\psi_2(\alpha,\sigma^2;\delta)<0.25$ when $\delta>4$:
\BE\label{eq:indeedfinaleq}
\psi_2(\alpha,\sigma^2;\delta) \le \frac{1}{\delta},\quad\forall (\alpha,\sigma^2)\in\mathcal{R}_{2a}.
\EE
Noting that $F_1^{-1}(\alpha)>0.339>0.25>1/\delta$ for $\alpha\in[0,\pi^{-1}),\delta>4$. Comparing this result with \eqref{eq:lbF_1minus1realthirdcase} proves that
\[
\psi_2(\alpha,\sigma^2;\delta)<F_1^{-1}(\alpha),\quad\forall \alpha\in[0,\pi^{-1}),\delta>4.
\]
Finally, we prove \eqref{eq:indeedfinaleq}. Since $\psi_2(\alpha,\sigma^2)=\frac{1}{\delta}(\alpha^2+\sigma^2+1-\frac{4\sigma}{\pi}-\frac{4\alpha}{\pi}\mr{atan}\left(\frac{\alpha}{\sigma})\right)$, we only need to prove
\[
\alpha^2+\sigma^2 -\frac{4\sigma}{\pi}-\frac{4\alpha}{\pi}\mr{atan}\left(\frac{\alpha}{\sigma}\right)\le0,\quad\forall\alpha\in[0,\pi^{-1}),(\alpha,\sigma^2)\in\mathcal{R}_{2a}
\]
which is equivalent to
\[
\alpha\cdot\frac{\alpha}{\sigma}+\sigma -\frac{4}{\pi}-\frac{4}{\pi}\frac{\alpha}{\sigma}\mr{atan}\left(\frac{\alpha}{\sigma}\right)\le0,\quad\forall\alpha\in[0,\pi^{-1}),(\alpha,\sigma^2)\in\mathcal{R}_{2a},
\]
Since $\alpha<1$, it suffices to prove
\[
\frac{\alpha}{\sigma}+\sigma -\frac{4}{\pi}-\frac{4}{\pi}\frac{\alpha}{\sigma}\mr{atan}\left(\frac{\alpha}{\sigma}\right)\le0,\quad\forall(\alpha,\sigma^2)\in\mathcal{R}_{2a}.
\]
Simple differentiation shows that the maximum of the function $f(x)=x-\frac{4}{\pi}x\cdot\mr{atan}(x)$ happens at $x_{\ast}$ where $\frac{4}{\pi}\cdot \mr{atan}(x_{\ast})=1-\frac{4}{\pi}\cdot\frac{x_{\ast}}{1+x_{\ast}^2}$ ($x_*\approx0.44$) and hence
\[
\begin{split}
x-\frac{4}{\pi}x\cdot\mr{atan}(x) &\le x_{\ast}-\frac{4}{\pi}x_{\ast}\cdot\mr{atan}(x_{\ast})=\frac{4}{\pi}\frac{x_{\ast}^2}{1+x_{\ast}^2}\approx\frac{4}{\pi}\cdot0.17<\frac{2}{\pi}.
\end{split}
\]
Using the above inequality, we obtain
\[
\begin{split}
\frac{\alpha}{\sigma}+\sigma -\frac{4}{\pi}-\frac{4}{\pi}\frac{\alpha}{\sigma}\mr{atan}\left(\frac{\alpha}{\sigma}\right) &<\sigma-\frac{2}{\pi}<0,\\
\end{split}
\]
where the last inequality is due to the fact that $(\alpha,\sigma^2)\in\mathcal{R}_{2a}$ and hence $\sigma^2\le F_1^{-1}(0)=\left(\frac{2}{\pi}\right)^2$.
 \end{proof}
 

\item \textbf{Main part}
The proof is similar to that of Lemma \ref{Lem:regionI_III}. The only noticeable difference is the proof for the following inequality (cf.~\eqref{Eqn:lem_regionI_III_7b})
\BE
\psi_2(\alpha;\sigma^2) < F_1^{-1}(\alpha),\quad\forall (\alpha,\sigma^2)\in\mathcal{R}_{2a},
\EE
where $\mathcal{R}_{2a}$ is now defined in Definition \ref{Def:region_real}. We have dedicated Lemma \ref{Lem:RegionIII_bound_real} to the proof of the above inequality, which is in parallel to Lemma \ref{Lem:RegionIII_bound} for the complex-valued case.

\end{itemize}

\subsection{Proof of Theorem \ref{lem:global_real}}\label{ssec:prooflemmaglobalminimum}
The proof is similar to that of Lemma \ref{Lem:fixed_point}. Hence, we only focus on the discrepancies. 

\subsubsection{$\delta>\deltaGlobal$}
\begin{lemma}\label{Lem:F1_F2_real_local}
Suppose that $\delta>\deltaGlobal=1+\frac{4}{\pi^2}$. Then, there exists an $\epsilon>0$ such that the following holds:
\BE\label{Eqn:F1_F2_real_local}
F_1^{-1}(\alpha)>F_2(\alpha;\delta),\quad\forall \alpha\in(1-\epsilon,1).
\EE
\end{lemma}
\begin{proof}
\eqref{Eqn:F1_F2_real_local} can be re-parameterized as
\[
 \psi_2\left( g\left(s^{-1}\right), s^2\cdot g^2\left(s^{-1}\right);\delta \right) < s^2\cdot g^2\left(s^{-1}\right),\quad\forall s\in(0,\xi),
\]
where $g(x)\Mydef \frac{2}{\pi}\mr{arctan}(x)$, $s= \cot(\frac{\pi}{2}\alpha)$ and $\xi = \tan (\frac{\pi}{2} \epsilon)$.
\[
s\cdot \mr{arctan}\left(\frac{1}{s}\right) > \underbrace{  \frac{-s^2 + s\sqrt{\frac{\pi^2}{4} +\left(1+\frac{\pi^2}{4}(\delta-1)\right)s^2 } }{ 1 + (\delta-1)s^2 } }_{R(s)} ,\quad s\in(0,\xi).
\]
Taylor expansions of the LHS and the RHS are respectively given by
\[
\begin{split}
s\cdot \mr{arctan}\left(\frac{1}{s}\right) &=\frac{\pi s}{2} - s^2 + \frac{s^4}{3}+O(s^6),\\
R(s) &= \frac{\pi s}{2} - s^2+\left(\frac{1}{\pi } -\frac{\pi}{4} (\delta - 1) \right)s^3 + (\delta-1)s^4+O(s^5)
\end{split}
\]
Then,
\[
\delta > 1+\frac{4}{\pi^2}\Longrightarrow \frac{1}{\pi } -\frac{\pi}{4} (\delta - 1)<0,
\]
and in this case there exists a constant $\xi>0$ such that
\[
s\cdot \mr{arctan}\left(\frac{1}{s}\right) > R(s),\quad \forall s\in(0,\xi).
\]
\end{proof}
Since the rest of the proof is exactly similar to the proof of Lemma \ref{Lem:fixed_point} for the sake of brevity we skip it here. 

\subsubsection{$\delta<\deltaGlobal$}
It is straightforward to use an argument similar to the one presented in Section \ref{Sec:complex_local_converse} and show that there exists a neighborhood of $(\alpha, \sigma^2)= (1,0)$ in which $\psi_2(\alpha, \sigma^2) - \sigma^2>0$. Hence, the state evolution moves away from $(0,1)$. 

\section{Proofs of Theorems \ref{thm:noisesens_comp} and \ref{thm:noisesens_real} }\label{Sec:proof_noise} 
In light of Lemma \ref{Lem:psi_phase_invariance_0}, we assume that $\alpha_0\ge0$ throughout this Appendix.
\subsection{Discussion}
The goal of this section is to prove Theorems \ref{thm:noisesens_comp} and \ref{thm:noisesens_real}. The strategy is similar to the proof of Theorem \ref{Theo:PhaseTransition_complex}. We first construct the functions $F_1^{-1}$ and $F_2$. Then, we show that these two functions will intersect at exactly one point when $\delta> \deltaAMP$. Finally, we discuss the dynamics of the state evolution and show that $(\alpha_t, \sigma_t^2)$ converge to the intersection of $F_1^{-1}$ and $F_2$. However, there are a few differences that make the proof of the noisy case more challenging:

\begin{enumerate}
\item Recall that in the noiseless case, the curve $F_1^{-1}$ is entirely above $F_2$ (except for the fixed point $(1,0)$) if $\delta>\deltaAMP$. See the plot in Fig.~\ref{fig:regions_I&III}. On the other hand, when there is some noise, the curve $F_2$ will move up a little bit (while $F_1^{-1}$ is unchanged) and will cross $F_1$ at a certain $\alpha_{\star}\in(0,1)$. As shown in Fig.~\ref{Fig:G}, $F_1^{-1}$ is above $F_2$ for $\alpha<\alpha_{\star}$ and is below $F_2$ when $\alpha>\alpha_{\star}$.

\item In the noisy setting the dynamic of SE becomes more challenging. In fact $(\alpha_t, \sigma_t^2)$ can move in any direction around the fixed point. That makes the proof of convergence of $(\alpha_t, \sigma_t^2)$ more complicated. 

\item In the noiseless setting the location of the fixed point of SE was $(\alpha, \sigma^2)= (1,0)$. This is not the case for the noisy settings where the location of the fixed point depends on the noise variance. 
\end{enumerate}

In the sections below we go over the entire proof, but will skip the parts that are similar to the proof of the noiseless setting which was discussed in Section \ref{ssec:proofTheo:PhaseTransition_complex}. 
\subsection{Complex-valued case}
\subsubsection{Preliminaries}

In the noisy setting, $\psi_1(\alpha;\sigma^2)$ remains unchanged, and $\psi_2(\alpha,\sigma^2;\delta)$ is replaced by $\psi_2(\alpha,\sigma^2;\delta,\sigma^2_w)$ below:
\BS\label{Eqn:SE_fixed_noisy_complex1}
\begin{align}
\psi_2(\alpha,\sigma^2;\delta,\sigma^2_w) &= \psi_2(\alpha,\sigma^2;\delta) +4\sigma^2_w\\
&=\frac{4}{\delta}\left\{\alpha^2+\sigma^2+1- \alpha \left[\phi_1\left(\frac{\sigma}{\alpha}\right)+\phi_3\left(\frac{\sigma}{\alpha}\right)\right]\right\}+4\sigma^2_w,
\label{Eqn:SE_fixed_noisy_complex1b}
\end{align}
\ES
where 
\BE\label{Eqn:phi1_phi3}
\begin{split}
\phi_1(s)& \Mydef \int_0^{\frac{\pi}{2}}\frac{\sin^2\theta}{\left(\sin^2\theta +s^2 \right)^{\frac{1}{2}}}\mr{d}\theta,\\
\phi_3(s)&\Mydef \int_0^{\frac{\pi}{2}}\left(\sin^2\theta +s^2 \right)^{\frac{1}{2}}\mr{d}\theta.
\end{split}
\EE
Before we proceed to the analysis of $\psi_1, \psi_2, F_1,$ and $F_2$, we list a few identities for $\phi_1$ and $\phi_3$ which will be used in our proofs later. 
\begin{lemma}
$\phi_1$ and $\phi_3$ satisfy the following properties:
\BE\label{Eqn:noisy_identities}
\begin{split}
\phi_1(s) &= \frac{(1+s^2)E\left(\frac{1}{1+s^2}\right)-s^2K\left(\frac{1}{1+s^2}\right)}{\sqrt{1+s^2}},\\
\phi_3(s) &= \sqrt{ 1 + s^2 }E\left(\frac{1}{1+s^2}\right),\\
\phi_1(0) &=1,\\
\left.\frac{\mr{d}\phi_1(s)}{\mr{d}s^2}  s^2\right|_{s=0}&=\left.\frac{s^2(E-K)}{2\sqrt{1+s^2}}\right|_{s=0}=0,\\
\left.\frac{\mr{d}\phi_1(s)\phi_3(s)}{\mr{d}s^2} \right|_{s=0}&= \left.\frac{1}{2}\left(  \frac{(1+s^2)E^2 -s^2K^2}{1+s^2} \right)^2\right|_{s=0}=\frac{1}{2},
\end{split}
\EE
where $E$ and $K$ are shorthands for $E\left(\frac{1}{1+s^2}\right)$ and $K\left(\frac{1}{1+s^2}\right)$ respectively in the last two identities.
\end{lemma}
The proof of this lemma is a simple application of the identities we derived in Section \ref{ssec:ellipticintegrals}, and is hence skipped. 

Our next lemma summarizes the main properties of $\psi_1, \psi_2, F_1$ and $F_2$ in the noisy phase retrieval problem.

\begin{lemma}\label{Lem:noisy}
Let $\tilde{\sigma}^2_{\max}\Mydef\sigma^2_{\max}+4\sigma^2_w$, where $\sigma^2_{\max}=\max\{1,4/\delta\}$. For any $\delta>\deltaAMP$, there exists $\epsilon>0$ such that when $0<\sigma^2_w<\epsilon$ the following statements hold simultaneously:
\begin{itemize}
\item [(a)] For $0\le\alpha\le1$, we have $\psi_2(\alpha,\sigma^2;\delta,\sigma^2_w)\le \tilde{\sigma}^2_{\max}$, $\forall\sigma^2\in[0,\tilde{\sigma}^2_{\max}]$.
\item [(b)] For $0\le\alpha\le1$, $\sigma^2=\psi_2(\alpha,\sigma^2;\delta)+4\sigma^2_w$ admits a unique globally attracting fixed point, denoted as $F_2(\alpha;\delta,\sigma^2_w)$, in $\sigma^2\in[0,\tilde{\sigma}^2_{\max}]$. Further, if $\alpha\ge\alpha_{\ast}$ (note that $\alpha_{\ast}\approx0.53$ is defined in \eqref{Eqn:alpha_ast_def}), then $F_2(\alpha;\delta,\sigma^2_w)$ is strongly globally attractive. Finally, $F_2(\alpha;\delta,\sigma^2_w)$ is a continuous function of $\sigma^2_w$. 
\item [(c)] The equation $F_1^{-1}(\alpha)=F_2(\alpha;\delta,\sigma^2_w)$ has a unique nonzero solution in $\alpha\in[0,1]$. Let $\alpha_{\star}(\delta,\sigma^2_w)$ be that unique solution. Then, $F_1^{-1}(\alpha)>F_2(\alpha;\delta,\sigma^2_w)$ for $0\le\alpha<\alpha_{\star}(\delta,\sigma^2_w)$ and $F_1^{-1}(\alpha)<F_2(\alpha;\delta,\sigma^2_w)$ for $\alpha_{\star}(\delta,\sigma^2_w)<\alpha\le1$.
\item [(d)] There exists $\hat{\alpha}(\delta,\sigma^2_w)$, such that $F_2(\alpha;\delta,\sigma^2_w)$ is strictly decreasing on $\alpha\in(0,\hat{\alpha}(\delta,\sigma^2_w))$ and strictly increasing on $(\hat{\alpha}(\delta,\sigma^2_w),1)$. Further, ${\alpha}_{\star}(\delta,\sigma^2_w)<\hat{\alpha}(\delta,\sigma^2_w)<1$.
\item [(e)] Define $L(\alpha;\delta,\sigma^2_w)\Mydef L(\alpha;\delta)+4\sigma^2_w$, where $L(\alpha;\delta)$ is defined in \eqref{Eqn:left_bound}. Then, $L(\alpha;\delta,\sigma^2_w)<F_1^{-1}(\alpha)$ for all $\alpha\in(0,\alpha_{\ast}]$, where $\alpha_*\approx0.53$ is defined in \eqref{Eqn:alpha_ast_def}. 
\item [(f)] For any $\alpha\in(0,\alpha_{\ast}]$ and $\sigma^2\in[L(\alpha;\delta,\sigma^2_w),F_1^{-1}(\alpha)]$, we have $\psi_2(\alpha,\sigma^2;\delta,\sigma^2_w)\Mydef\psi_2(\alpha,\sigma^2;\delta)+4\sigma^2_w < F_1^{-1}(\alpha)$.

\item[(g)] $F_2(1;\delta,\sigma^2_w) < F_1^{-1}(\alpha_{\ast})$.
\end{itemize}

\end{lemma}

\begin{proof} In the following, we will prove that each part of the lemma holds when $\sigma^2_w$ is smaller than a constant. Hence, the statements hold simultaneously when $\sigma^2_w$ is smaller than the minimum of those constants. 

\textit{Part (a):} In Lemma \ref{lem:psi2}-(iii) we proved that, for the noiseless setting, $\psi_2(\alpha;\sigma^2;\delta)\le \sigma^2_{\max}$ for $\sigma^2\in[0,\sigma^2_{\max}]$. If fact, it is easy to verify that our proof can be strengthened to $\psi_2(\alpha;\sigma^2;\delta)\le \sigma^2_{\max}$ for $\sigma^2\in[0,2]$, see \eqref{Eqn:Lemma1_2}. Note that $\sigma^2_{\max}=\max\{1,4/\delta\}\le 4/\deltaAMP\approx1.6$. Hence, $\psi_2(\alpha;\sigma^2;\delta)\le \sigma^2_{\max}$ for $\sigma^2\in[0,\tilde{\sigma}_{\max}^2]=\sigma^2_{\max}+4\sigma^2_w$ when $\sigma^2_w$ is small. Further, $\psi_2(\alpha;\sigma^2;\delta,\sigma^2_w)=\psi_2(\alpha;\sigma^2;\delta)+4\sigma^2_w$, and hence $\psi_2(\alpha;\sigma^2;\delta,\sigma^2_w)\le \tilde{\sigma}_{\max}^2$ for $\sigma^2\in[0,\tilde{\sigma}_{\max}^2]$.

\vspace{5pt}
\textit{Part (b):} The claim is a consequence of three facts: (i) $\psi_2(\alpha,\sigma^2;\delta,\sigma^2_w)\le \sigma^2$ at $\sigma^2=\tilde{\sigma}^2_{\max}$; (ii) $\frac{\partial\psi_2(\alpha,\sigma^2;\delta,\sigma^2_w)}{\partial \sigma^2}<1$ when $\sigma^2\in[0,\tilde{\sigma}^2_{\max}]$, and (iii) if $\alpha\ge\alpha_{\ast}$, then $\frac{\partial\psi_2(\alpha,\sigma^2;\delta,\sigma^2_w)}{\partial \sigma^2}>0$ for any $\sigma^2\ge0$. Fact (i) has been proved in part (a) of this lemma. For Fact (ii), recall that in \eqref{Eqn:psi2_dif} we have proved  $\frac{\partial\psi_2(\alpha,\sigma^2;\delta)}{\partial \sigma^2}<1$ when $\sigma^2\in[0,{\sigma}^2_{\max}]$. Again, similar to part (a) of this lemma, we can argue that the result actually holds for $\sigma^2\in[0,\tilde{\sigma}^2_{\max}]$. We prove Fact (ii) by further noting $\psi_2(\alpha,\sigma^2;\delta,\sigma^2_w)=\psi_2(\alpha,\sigma^2;\delta)+4\sigma^2_w$ and hence $\frac{\partial\psi_2(\alpha,\sigma^2;\delta,\sigma^2_w)}{\partial \sigma^2}=\frac{\partial\psi_2(\alpha,\sigma^2;\delta)}{\partial \sigma^2}$. Fact (iii) follows from Lemma \ref{lem:psi2}-(v) and the fact that $\frac{\partial\psi_2(\alpha,\sigma^2;\delta,\sigma^2_w)}{\partial \sigma^2}=\frac{\partial\psi_2(\alpha,\sigma^2;\delta)}{\partial \sigma^2}$.

We now show that $F_2(\alpha;\delta,\sigma^2_w)$ is a continuous function of $\sigma^2_w$. Let $x$ be an arbitrary constant in $(0,\epsilon)$. Suppose that $\lim_{\sigma^2_w\to x^-}F_2(\alpha;\delta,\sigma^2_w)=y_1$ and $\lim_{\sigma^2_w\to x^+}F_2(\alpha;\delta,\sigma^2_w)=y_2$, where $y_1,y_2\in[0,\tilde{\sigma}^2_{\max}]$ and $y_1\neq y_2$. Since $F_2$ is the fixed point of $\psi_2$, we then have $y_1=\psi_2(\alpha,y_1;\delta)+4x$ and $y_2=\psi_2(\alpha,y_2;\delta)+4x$, which leads to $y_1-\psi_2(\alpha,y_1;\delta)=y_2-\psi_2(\alpha,y_2;\delta)$. However, we have shown in Lemma \ref{lem:psi2} that $\tilde{\Psi}_2(\alpha,\sigma^2;\delta)\Mydef \sigma^2-\psi_2(\alpha,\sigma^2;\delta)-C$ is a strictly increasing function of $\sigma^2$ in $[0,\tilde{\sigma}^2_{\max}]$, and hence for any $C\in\mathbb{R}$ there cannot be two solutions to $\tilde{\Psi}_2(\alpha,\sigma^2;\delta)=0$.
This leads to contradiction.

\vspace{5pt}
\textit{Part (c):} 
It is more convenient to introduce a variable change:
\[
s\Mydef \phi_1^{-1}(\alpha)\quad\text{and}\quad s_{\star}(\delta,\sigma^2_w) = \phi_1^{-1}\left(\alpha_{\star}(\delta,\sigma^2_w)\right).
\]
As have been argued in Section \ref{proof:lemmadominationF_1}, $F_1^{-1}(\alpha)\le F_1^{-1}(0)=\pi^2/16<\tilde{\sigma}^2_{\max}$. Then, by the global attractiveness of $F_2(\alpha;\delta,\sigma^2_w)$ (part (b) of this lemma) and noting that $\phi_1:[0,\infty]\mapsto[0,1]$ is a decreasing function, our claim can be equivalently refomulated as
\BE\label{Eqn:noise_sensitivity_1}
\psi_2\left( \phi_1(s),s^2\phi_1^2(s);\delta \right ) +4\sigma^2_w> s^2\phi_1^2,\quad\forall s\in[0,s_{\star}(\delta,\sigma^2_w)),
\EE
and
\[
\psi_2\left( \phi_1(s),s^2\phi_1^2(s);\delta \right ) +4\sigma^2_w< s^2\phi_1^2,\quad\forall s>s_{\star}(\delta,\sigma^2_w).
\]
From the definition of $\psi_2$ in \eqref{Eqn:SE_fixed_noisy_complex1} and after straightforward manipulations, we can write \eqref{Eqn:noise_sensitivity_1} into 
\BE\label{Eqn:SE_fixed_noisy_complex3}
T(s^2, \delta,\sigma_w^2 )<0,\quad\forall s\in[0,s_{\star}(\delta,\sigma^2_w))\quad\text{and}\quad T(s^2, \delta,\sigma_w^2 )>0,\quad\forall s>s_{\star}(\delta,\sigma^2_w),
\EE
where
\BE\label{Eqn:SE_fixed_noisy_complex3aa}
T(s^2, \delta,\sigma_w^2 ) \triangleq \left(1-\frac{4}{\delta}\right)\phi^2_1(s)s^2 + \frac{4}{\delta}\phi_1(s)\phi_3(s)-\left( \frac{4}{\delta}+4\sigma^2_w \right).
\EE
From \eqref{Eqn:SE_fixed_noisy_complex3}, we have
\BE\label{Eqn:SE_fixed_noisy_complex3b}
\frac{\partial T(s^2, \sigma_w^2 )}{\partial s^2} =  \left(1-\frac{4}{\delta}\right)  \left(\phi_1^2(s) + 2\phi_1(s)\frac{\mr{d}\phi_1(s)}{\mr{d}s^2}s^2 \right) + \frac{4}{\delta} \frac{\mr{d}\phi_1(s)\phi_3(s)}{\mr{d}s^2}.
\EE
Applying the identities listed in \eqref{Eqn:noisy_identities}, we obtain
\[
\left.\frac{\partial T(s^2, \sigma_w^2 )}{\partial s^2}\right|_{s=0}=1-\frac{2}{\delta}>0.
\]
Further, $\frac{\partial T(s^2, \sigma_w^2 )}{\partial s^2}$ is a continuous function at $s^2=0$, and thus there exists $\epsilon>0$ such that 
\[
\frac{\partial T(s^2, \sigma_w^2 )}{\partial s^2}>0,\quad\forall s^2\in[0,\epsilon].
\]
The above result shows that $T(s^2, \sigma_w^2 )$ is monotonically increasing in $s^2\in[0,\epsilon]$. Further, from \eqref{Eqn:SE_fixed_noisy_complex3aa} we have
\[
T(s^2, \delta,\sigma_w^2 )  = T(s^2, \delta,0) - 4 \sigma^2_w. 
\]
It is straightforward to show that $T(0,\delta,\sigma^2_w)=-\sigma^2_w<0$. Hence, $T(s^2,\delta,\sigma^2_w)=0$ has a unique solution if the following holds:
\[
\inf_{s^2\ge\epsilon}\ T(s^2, \delta,\sigma^2_w)>0,
\]
or equivalently
\BE\label{Eqn:noise_sensitivity_2}
4\sigma^2_w<\inf_{s^2\ge\epsilon}\ T(s^2, \delta,0).
\EE
Lemma \ref{Lem:F1_F2_complex} proves that $F_1^{-1}(\alpha)>F_2(\alpha;\delta)$ for $\alpha\in(0,1)$ for any $\delta>\deltaAMP$, which, after re-parameterization implies that $T(s^2,\delta,0)>0$ for $s>0$ if $\delta>\deltaAMP$. Hence, $\inf_{s^2\ge\epsilon}\ T(s^2, \delta,0)$ is strictly positive, and there exists sufficiently small $\sigma^2_w$ such that \eqref{Eqn:noise_sensitivity_2} holds.

\vspace{5pt}
\textit{Part (d):} 
From the fixed point equation $F_2=\psi_2(\alpha,F_2;\delta,\sigma^2_w)$ where ($F_2$ denotes $F_2(\alpha;\delta,\sigma^2_w)$), we can derive the following (cf.~\eqref{Eqn:F1_prime}) 
\[
\left(1- \partial_2 \psi_2(\alpha,F_2;\delta,\sigma^2_w)\right) \cdot \frac{\mr{d}F_2(\alpha;\delta,\sigma^2_w)}{\mr{d}\alpha} =  \partial_1 \psi_2(\alpha,F_2;\delta,\sigma^2_w) .
\]
Similar to the proof of part (b), $1-\partial_2 \psi_2(\alpha,F_2;\delta,\sigma^2_w)>0$ when $\sigma^2_w$ is sufficiently small.  Hence, proving $ \partial_1 \psi_2(\alpha,F_2;\delta,\sigma^2_w)<0$ is simplified to proving that there exists $\hat{\alpha}(\delta,\sigma^2_w)$ such that
\BS\label{Eqn:noise_sensitivity_3}
\BE
\partial_1 \psi_2(\alpha,F_2;\delta,\sigma^2_w) < 0,\quad\forall \alpha\in\left(0,\hat{\alpha}(\delta,\sigma^2_w)\right),
\EE
and
\BE
\partial_1 \psi_2(\alpha,F_2;\delta,\sigma^2_w) > 0,\quad\forall \alpha\in\left(\hat{\alpha}(\delta,\sigma^2_w),\,1\right).
\EE
\ES
From \eqref{Eqn:map_expression_complex} and after some calculations, we obtain the following
\BE\label{Eqn:noise_sensitivity_h}
\begin{split}
 \frac{\partial\psi_2(\alpha,\sigma^2;\delta,\sigma^2_w)}{\partial \alpha} &=\frac{4}{\delta}\left( 2\alpha-\int_0^{\frac{\pi}{2}}\frac{2\alpha^3\sin^4\theta +3\alpha\sigma^2\sin^2\theta }{  ( \alpha^2\sin^2\theta +\sigma^2 )^{\frac{3}{2}}  } \mr{d}\theta\right)\\
 &=\frac{4}{\delta}\bigg(2\alpha- 2\underbrace{\int_0^{\frac{\pi}{2}}\frac{\sin^4\theta +\frac{3}{2}s^2\sin^2\theta }{  ( \sin^2\theta +s^2 )^{\frac{3}{2}}  } \mr{d}\theta}_{h(s)} \bigg),
 \end{split}
\EE
where $s\Mydef \sigma/\alpha$. Then, we can reformulate \eqref{Eqn:noise_sensitivity_3} as
\[
\alpha<h\left( \frac{ \sqrt{F_2(\alpha;\delta,\sigma^2_w)} }{\alpha} \right),\quad\forall \alpha\in\left(0,\hat{\alpha}(\delta,\sigma^2_w)\right),
\]
and
\[
\alpha>h\left( \frac{ \sqrt{F_2(\alpha;\delta,\sigma^2_w)} }{\alpha} \right),\quad\forall \alpha\in\left(\hat{\alpha}(\delta,\sigma^2_w),\,1\right).
\]

\begin{figure}[!htbp]
\centering
\includegraphics[width=.55\textwidth]{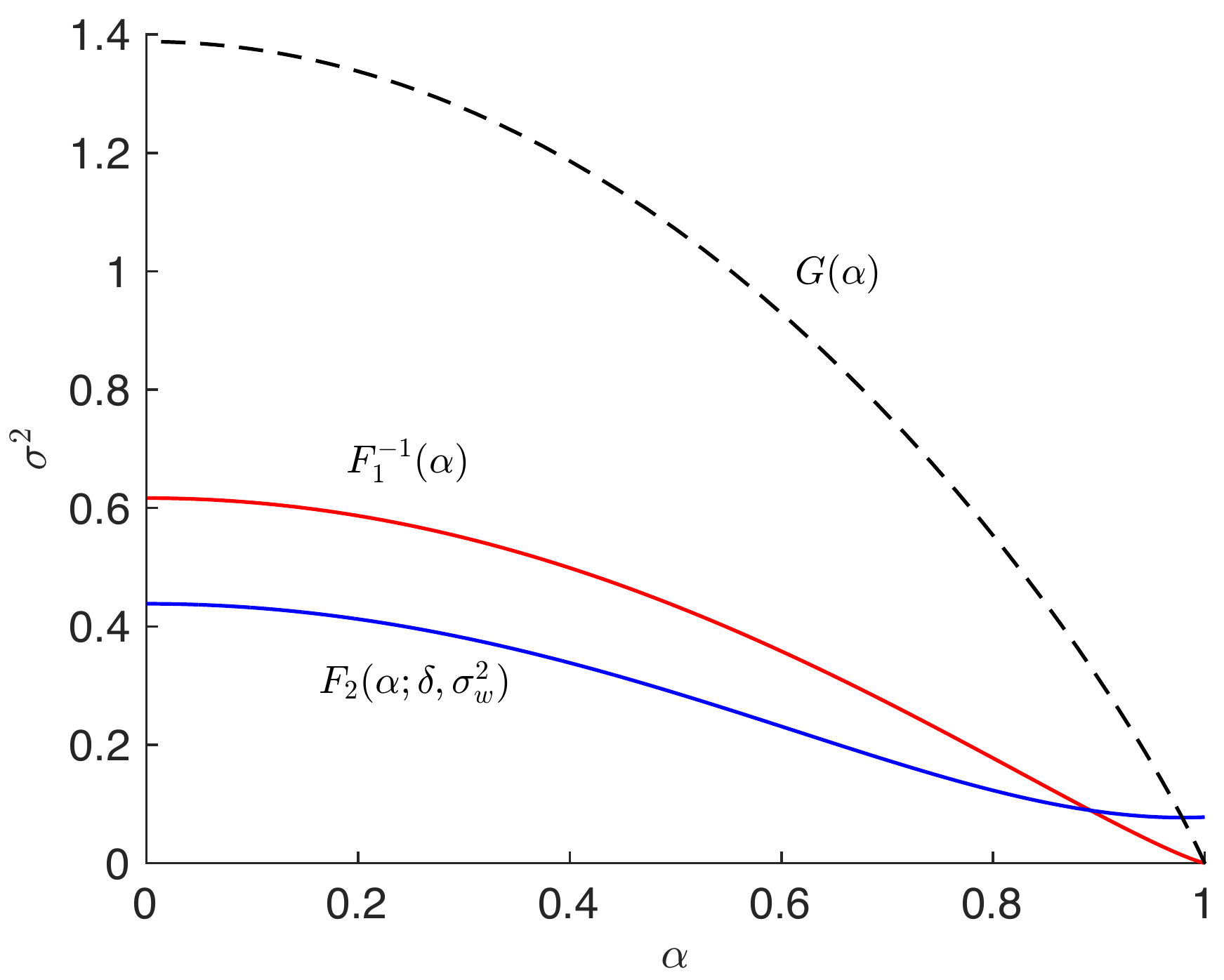}
\caption{Depiction of $F_1^{-1}(\alpha)$, $F_2(\alpha;\delta,\sigma^2_w)$ and $G(\alpha)$. $\alpha_{\star}(\delta,\sigma^2_w)$: solution to $F_1^{-1}(\alpha)=F_2(\alpha;\delta,\sigma^2_w)$. $\hat{\alpha}(\delta,\sigma^2_w)$: solution to $G^{-1}(\alpha)=F_2(\alpha;\delta,\sigma^2_w)$. }
\label{Fig:G}
\end{figure}

From the definition given in \eqref{Eqn:noise_sensitivity_h}, it is easy to show that $h:\mathbb{R}_+\mapsto[0,1]$ is a decreasing function. Then, the above inequality can be further simplified to
\BS\label{Eqn:noise_sensitivity_4}
\BE
F_2(\alpha;\delta,\sigma^2_w) < \left[\alpha\cdot h^{-1}(\alpha)\right]^2\Mydef G(\alpha),\quad\forall  \alpha\in\left(0,\hat{\alpha}(\delta,\sigma^2_w)\right),
\EE
and
\BE
F_2(\alpha;\delta,\sigma^2_w) > \left[\alpha\cdot h^{-1}(\alpha)\right]^2= G(\alpha),\quad\forall  \alpha\in\left(\hat{\alpha}(\delta,\sigma^2_w),\,1\right).
\EE
\ES
Similar to \eqref{Eqn:noise_sensitivity_1} and \eqref{Eqn:SE_fixed_noisy_complex3}, \eqref{Eqn:noise_sensitivity_4} can be re-parameterized as
\BE\label{Eqn:noise_sensitivity_5}
\psi_2\left( h(s),s^2\phi_1^2(s);\delta \right ) +4\sigma^2_w> s^2h^2,\quad\forall s<\hat{s}(\delta,\sigma^2_w),
\EE
and
\BE
\psi_2\left( \phi_1(s),s^2\phi_1^2(s);\delta \right ) +4\sigma^2_w< s^2h^2,\quad\forall s>\hat{s}(\delta,\sigma^2_w),
\EE
where $\hat{s}(\delta,\sigma^2_w)\Mydef h^{-1}\left( \hat{\alpha}(\delta,\sigma^2_w) \right)$. We skip the proof for \eqref{Eqn:noise_sensitivity_5} since it is very similar to the proof of part (c) of this lemma. (Note that to apply the above re-parameterization (which is based on the global attractiveness of $F_2$, i.e., part (b) of this lemma), we need to ensure $G(\alpha)<\tilde{\sigma}^2_{\max}$. This can be seen from the fact that $G(\alpha)\le G(0)=(3\pi/8)^2\approx 1.38$ while $ \tilde{\sigma}^2_{\max}+4\sigma_w^2$ and $ \sigma^2_{\max}=\max\{1,4/\delta\}>\max\{1,4/\deltaAMP\}\approx1.6$.)

Finally, to show $\hat{\alpha}(\delta,\sigma^2_w)>\alpha_{\star}(\delta,\sigma^2_w)$, we will prove that $G(\alpha)>F_1^{-1}(\alpha)$ for $\alpha\in[0,1)$. See the plot in Fig.~\ref{Fig:G}. Since $G(\alpha)=[\alpha\cdot h^{-1}(\alpha)]^2$ and $F_1^{-1}(\alpha)=[\alpha\cdot \phi_1^{-1}(\alpha)]^2$, we only need to prove $h^{-1}(\alpha)>\phi_1^{-1}(\alpha)$. Noting that both $\phi_1$ and $h$ are monotonically decreasing functions, it suffices to prove $h(s)>\phi_1(s)$ for $s>0$, which directly follows from their definitions (cf. \eqref{Eqn:noise_sensitivity_h} and \eqref{Eqn:phi}):
\[
\begin{split}
h(s) -\phi_1(s)&=\int_0^{\frac{\pi}{2}}\frac{\sin^4\theta +\frac{3}{2}s^2\sin^2\theta }{  ( \sin^2\theta +s^2 )^{\frac{3}{2}}  } \mr{d}\theta - \int_0^{\frac{\pi}{2}}\frac{\sin^2\theta}{\left(\sin^2\theta +s^2 \right)^{\frac{1}{2}}}\mr{d}\theta\\
&=\int_0^{\frac{\pi}{2}}\frac{\frac{1}{2}s^2\sin^2\theta }{  ( \sin^2\theta +s^2 )^{\frac{3}{2}}  }>0,\quad \forall s>0.
\end{split}
\]

\vspace{5pt}
\textit{Part (e):} First note that $L(\alpha;\delta,\sigma^2_w)=L(\alpha;\delta)+4\sigma^2_w$. Hence, the proof for the claim is straightforward if the inequality $L(\alpha;\delta)<F_1^{-1}(\alpha)$ is strict for $\alpha\le\alpha_{\ast}$. This is the case since Lemma \ref{Lem:F1_L} shows that $L(\alpha;\delta)\le F_1^{-1}(\alpha)$ for $\alpha\le1$, but equality only happends at $\alpha=1$. 

\vspace{5pt}
\textit{Part (f):} In Lemma \ref{Lem:RegionIII_bound}, we have proved the following result in the case of $\sigma^2_w=0$:
\[
\psi_2(\alpha,\sigma^2;\delta) < F_1^{-1}(\alpha),\quad \forall 0\le \alpha\le\alpha_{\ast},\ L(\alpha;\delta)<\sigma^2<F_1^{-1}(\alpha).
\]
(In fact, the above inequality holds for $\alpha$ up to one.) In the noisy case, $\psi_2$ increases a little bit: $\psi_2(\alpha,\sigma^2;\delta,\sigma^2_w)=\psi_2(\alpha,\sigma^2;\delta)+4\sigma^2_w$. Hence, when $\sigma^2_w$ is sufficiently small, we still have
\BE\label{Eqn:noisy_lemma_g1}
\psi_2(\alpha,\sigma^2;\delta,\sigma^2_w) < F_1^{-1}(\alpha),\quad \forall 0 \le \alpha\le\alpha_{\ast},\ L(\alpha;\delta)<\sigma^2<F_1^{-1}(\alpha).
\EE
Clearly, the inequality in \eqref{Eqn:noisy_lemma_g1} also holds for $L(\alpha;\delta,\sigma^2_w)<\sigma^2<F_1^{-1}(\alpha)$, since $L(\alpha;\delta,\sigma^2_w)=L(\alpha;\delta)+4\sigma^2_w>L(\alpha;\delta)$. 

\vspace{5pt}
\textit{Part (g):} Note that $F_1^{-1}(\alpha_{\ast})\approx F_1^{-1}(0.53) >0$ does not depend on $\sigma^2_w$. Further, $F_2(1;\delta,0)=0$ and $F_2(1;\delta,\sigma^2_w)$ is a continuous function of $\sigma^2_w$. Hence, $F_2(1;\delta,\sigma^2_w)<F_1^{-1}(\alpha_{\ast})$ for small enough $\sigma^2_w$.

\end{proof}

\subsubsection{Convergence of the SE}

Our next lemma proves that the state evolution still converges to the desired fixed point for $0<\alpha_0\le1$ and $\sigma^2_0\le1$ if $\delta>\deltaAMP$.


\begin{figure}[!htbp]
\centering
\subfloat{\includegraphics[width=.5\textwidth]{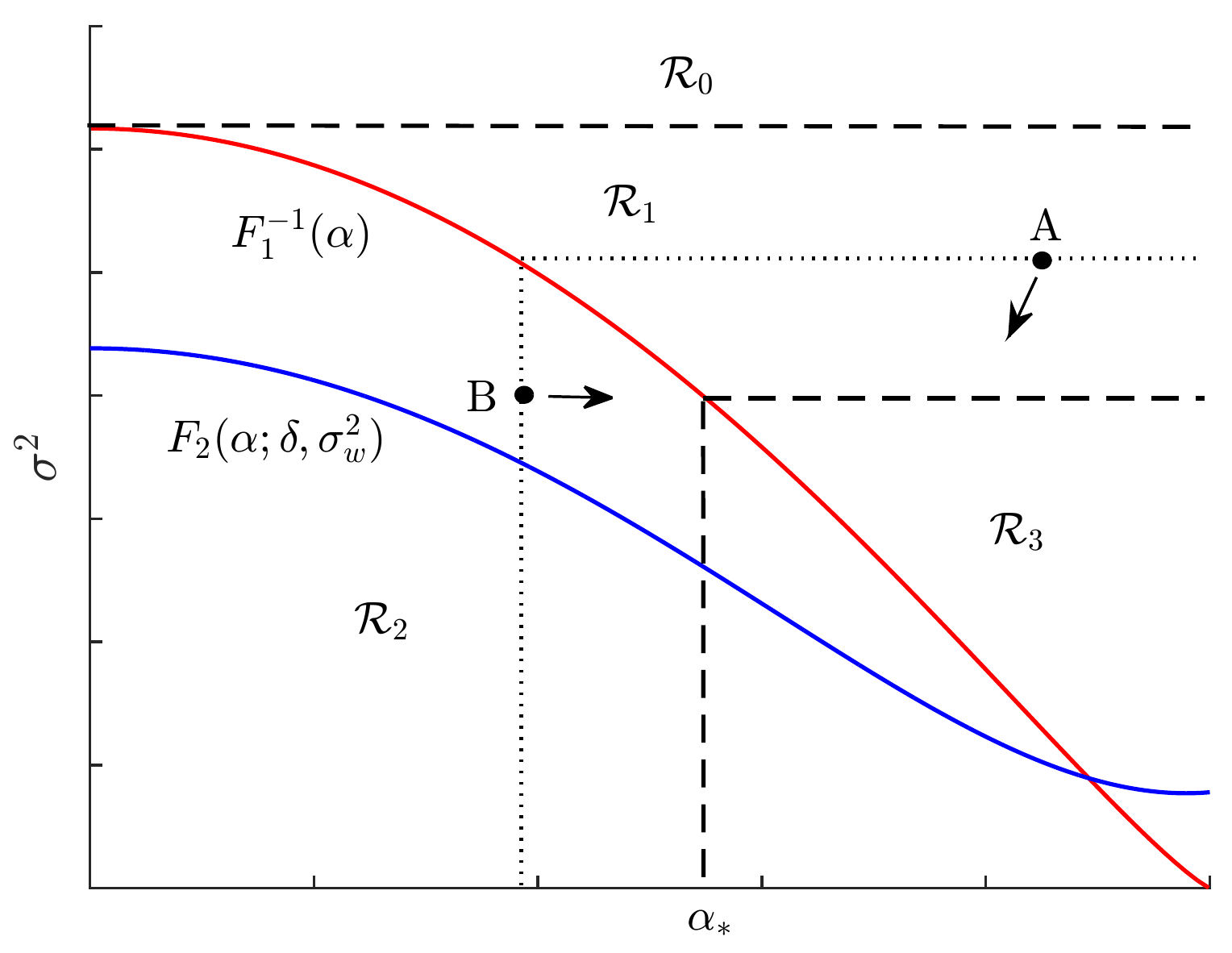}}
\subfloat{\includegraphics[width=.4\textwidth]{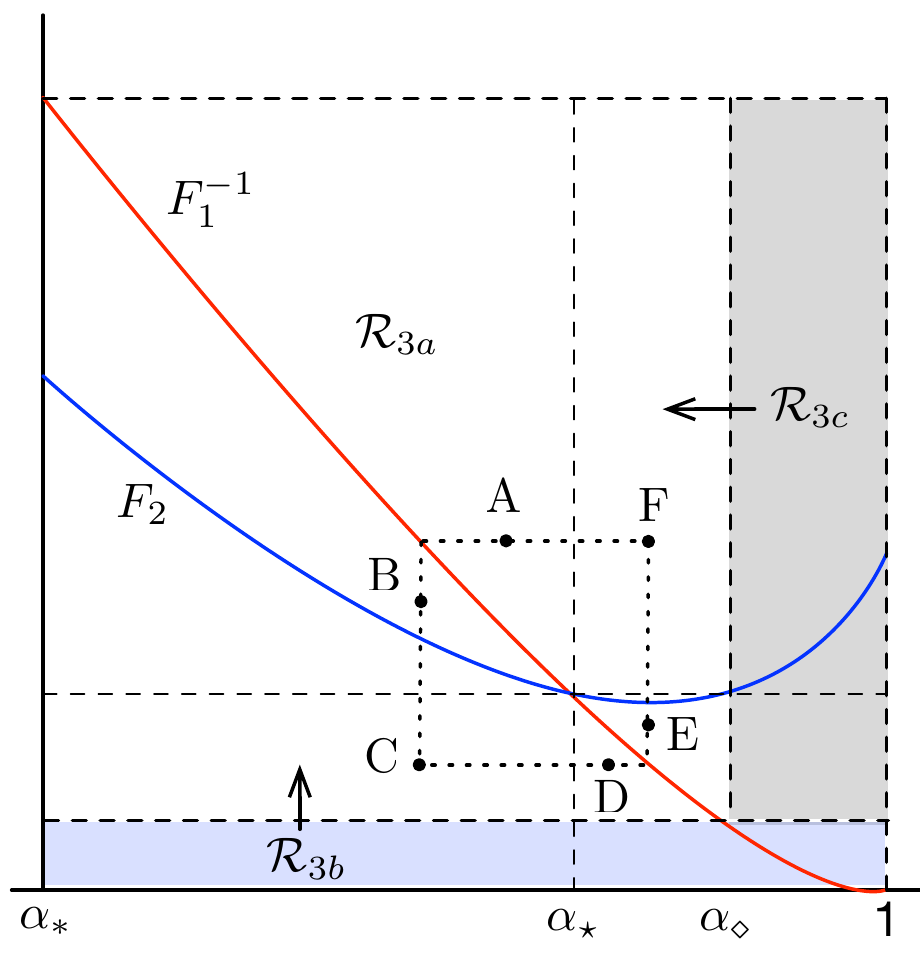}}
\caption{Dynamical behavior the state evolution in the low noise regime. \textbf{left:} points in $\mathcal{R}_1$ and $\mathcal{R}_2$ will eventually move to $\mathcal{R}_3$. Here, $\alpha_{\ast}\approx0.53$. \textbf{Right:} Illustration of $\mathcal{R}_3$. Points in $\mathcal{R}_{3b}$ and $\mathcal{R}_{3c}$ will eventually move to $\mathcal{R}_{3a}$. For points in $\mathcal{R}_{3a}$ (marked A, B, C, D, E, F), we can form a small rectangular region that bounds the remaining trajectory. Note that the lower and right bounds for A and B (and also the upper and left bounds for D and E) are given by $\sigma^2_{\star}$ and $\alpha_{\star}$ respectively.}
\label{Fig:regions_noisy}
\end{figure}

\begin{lemma}\label{Lem:SE_converge_noisy}
Let $\{\alpha_t\}_{t\ge1}$ and $\{\sigma^2_t\}_{t\ge1}$ be two state sequences generated according to \eqref{Eqn:SE_complex} from $\alpha_0$ and $\sigma^2_0$. Let $\epsilon$ be the constant required in Lemma \ref{Lem:noisy}. Then, the following holds for any $\delta>\deltaAMP$, $0<\sigma^2_w<\epsilon$, and $0<\alpha_0\le1$ and $\sigma^2_0\le1$:
\[
\lim_{t\to\infty} \alpha_t=\alpha_{\star}(\delta,\sigma^2_w)\quad\text{and}\quad \lim_{t\to\infty} \sigma^2_t=\sigma^2_{\star}(\delta,\sigma^2_w),
\]
where $\alpha_{\star}(\delta,\sigma^2_w)$ is the unique positive solution to $F_1^{-1}(\alpha)=F_2(\alpha;\delta,\sigma^2_w)$ and $\sigma^2_{\star}(\delta,\sigma^2_w)=F_1^{-1}(\alpha_{\star}(\delta,\sigma^2_w))$.
\end{lemma}

\begin{proof}
From Lemma \ref{Lem:noisy}-(a), when $\sigma^2_w$ is small enough, $(\alpha_t,\sigma^2_t)\in\mathcal{R}$ for all $t\ge1$, where $\mathcal{R}\Mydef\{(\alpha,\sigma^2)|0<\alpha\le1,\,0\le\sigma^2\le\tilde{\sigma}^2_{\max}\}$, where $\tilde{\sigma}^2_{\max}=\max\{1,4/\delta\}+4\sigma^2_w$. We divide $\mathcal{R}$ into several regions and discuss the dynamical behaviors of the state evolution for different regions separately. Specifically, we define
\BE\label{Eqn:regions_noisy}
\begin{split}
\mathcal{R}_0 &\Mydef \left\{ (\alpha,\sigma^2)|0<\alpha\le1,\  \pi^2/16<\sigma^2 \le\tilde{\sigma}^2_{\max} \right\},\\
\mathcal{R}_1 &\Mydef \left\{ (\alpha,\sigma^2)| F_1^{-1}(\alpha_{\ast})\le \sigma^2\le \pi^2/16, \ F_1(\sigma^2) \le\alpha \le 1 \right\},\\
\mathcal{R}_2 &\Mydef \left\{ (\alpha,\sigma^2)|0<\alpha\le\alpha_{*}, \ 0\le\sigma^2 < F_1^{-1}(\alpha)\right\},\\
\mathcal{R}_3 &\Mydef \left\{ (\alpha,\sigma^2)|\alpha_{*}\le\alpha\le1, \ 0\le \sigma^2 < F_1^{-1}(\alpha_{*})\right\},
\end{split}
\EE
where $\alpha_{\ast}\approx0.53$ was defined in \eqref{Eqn:alpha_ast_def}. Notice that $\alpha_{\star}(\delta,0)=1$, and therefore it is guaranteed that $\alpha_{\star}(\delta,\sigma^2_w)>\alpha_{\ast}$ for small enough $\sigma^2_w$. See Fig.~\ref{Fig:regions_noisy} for illustration. To prove the lemma, we will prove the following arguments: 
 \begin{itemize}
 
\item[(i)] If $(\alpha_{t_0},\sigma^2_{t_0})\in\mathcal{R}_0$, then there exists a finite $T_1\ge1$ such that $(\alpha_{t_0+T_1},\sigma^2_{t_0+T_1})\in\mathcal{R}\backslash\mathcal{R}_0$.

 \item[(ii)] If $(\alpha_{t_0},\sigma^2_{t_0})\in\mathcal{R}_1\cup\mathcal{R}_2$ for $t_0\ge1$ (i.e., after one iteration), then there exists a finite $T_2\ge1$ such that $(\alpha_{t_0+T_2},\sigma^2_{t_0+T_2})\in\mathcal{R}_3$. 

 \item[(iii)] We show that if $(\alpha_{t_0},\sigma^2_{t_0})\in\mathcal{R}_3$ for $t_0\ge0$, then $(\alpha_{t},\sigma^2_{t})\in\mathcal{R}_3$ for all $t>t_0$, and $(\alpha_{t},\sigma^2_{t})$ converges to $(\alpha_{\star},\sigma^2_{\star})$. 
 
 \end{itemize}
 
 The proof of (i) is similar to that of Lemma \ref{Lem:regionII_IV} and therefore omitted here.
 
 \vspace{5pt}
\textit{Proof of (ii):} Following the proof of Lemma \ref{Lem:regionI_III}, we argue that if $(\alpha_{t},\sigma^2_{t})\in\mathcal{R}_1\cup\mathcal{R}_2$ then the following holds
 \BE\label{Eqn:B1_B2_noisy}
\alpha_{t+1}\ge B_1(\alpha_t,\sigma^2_t)\quad\text{and}\quad  \sigma^2_{t+1}\ge B_2(\alpha_t,\sigma^2_t),
 \EE
 where $B_1(\alpha_t,\sigma^2_t)= \min\left\{ \alpha_t,F_1(\sigma^2_t) \right\}$ and $B_2(\alpha_t,\sigma^2_t)= \max\left\{ \sigma^2_t,F_1^{-1}(\alpha_t) \right\}$. Then, it is easy to show that $(\alpha_{t+1},\sigma^2_{t+1})\in\mathcal{R}_1\cup\mathcal{R}_2\cup\mathcal{R}_3$. Applying this recursively, we see that  $(\alpha,\sigma^2)$ either moves to $\mathcal{R}_3$ at a certain time or stays in $\mathcal{R}_1\cup\mathcal{R}_2$. We next prove that the latter case cannot happen. Suppose that $(\alpha_t,\sigma^2_t)\in\mathcal{R}_1\cup\mathcal{R}_2$ for $t\ge t_0$. If this is the case, then it can be shown that
\BE\label{Eqn:B1_B2_noisy2}
B_1(\alpha_t,\sigma^2_t) \le B_1(\alpha_{t+1},\sigma^2_{t+1}) \quad\text{and}\quad B_2(\alpha_t,\sigma^2_t) \ge B_2(\alpha_{t+1},\sigma^2_{t+1}),\quad\forall t>t_0.
\EE
On the other hand, since we assume $(\alpha_t,\sigma^2_t)\in\mathcal{R}_1\cup\mathcal{R}_2$ for $t\ge t_0$, $B_1$ is upper bounded by $\alpha_{\ast}$ and $B_2$ lower bounded by $F_1^{-1}(\alpha_{\ast})$. Hence, this means the sequences $B_1$ and $B_2$ converges to $\alpha_{\ast}$ and $F_1^{-1}(\alpha_{\ast})$, respectively. This cannot happen since there is no fixed point in $\mathcal{R}_1\cup\mathcal{R}_2$.

The proof for \eqref{Eqn:B1_B2_noisy} and \eqref{Eqn:B1_B2_noisy2} are basically the same as those for the noiseless counterparts and hence skipped here. Please refer to the proof of Lemma \ref{Lem:regionI_III}. We only need to show that some of the key inequalities used in the proof of Lemma \ref{Lem:regionI_III} still hold in the noisy case, which have been listed in Lemma \ref{Lem:noisy} (e) and (f).

\vspace{5pt}
\textit{Proof of (iii):} Lemma \ref{Lem:noisy}-(c), (d) and (g) imply that $F_2<F_1^{-1}(\alpha_{\ast})$ for all $\alpha\in[\alpha_{\ast},1]$. Then, based on the strong global attractiveness of $F_1$ and $F_2$, it is easy to show that if $(\alpha_{t_0},\sigma^2_{t_0})\in\mathcal{R}_3$ then $(\alpha_{t},\sigma^2_{t})\in\mathcal{R}_3$ for all $t\ge t_0$. We have proved in Lemma \ref{Lem:noisy}-(d) that $F_2$ is a decreasing function of $\alpha$ on $[0,\hat{\alpha}]$ and increasing on $[\hat{\alpha},1]$, where $\alpha_{\star}<\hat{\alpha}<1$. Then, the maximum of $F_2$ on $[\alpha_{\star},1]$ can only happen at either $\alpha_{\star}$ or $1$. We assume that the latter case happens; it will be clear that our proof for the former case is a special case of the proof for the latter one.  See the right panel of Fig.~\ref{Fig:regions_noisy}.

As discussed above, we assume that $F_2(1;\delta,\sigma^2_w)>F_2(\alpha_{\star};\delta,\sigma^2_w)$. Hence, by Lemma \ref{Lem:noisy}-(d), there exists a unique number $\alpha_{\diamond}\in(\alpha_{\star},1)$ such that $F_2(\alpha_{\diamond};\delta,\sigma^2_w)=F_2(\alpha_{\star};\delta,\sigma^2_w)$. See the plot in the right panel of Fig.~\ref{Fig:regions_noisy}. We further divide $\mathcal{R}_3$ into four regions:
\[
\begin{split}
\mathcal{R}_{3a} &\Mydef \left\{ (\alpha,\sigma^2)|\alpha_{\ast}\le \alpha\le\alpha_{\diamond},\  F_1^{-1}(\alpha_{\diamond})<\sigma^2 \le F_1^{-1}(\alpha_{\ast}) \right\},\\
\mathcal{R}_{3b} &\Mydef \left\{ (\alpha,\sigma^2)|\alpha_{\ast}\le \alpha\le1,\  0\le\sigma^2 < F_1^{-1}(\alpha_{\diamond}) \right\},\\
\mathcal{R}_{3c} &\Mydef \left\{ (\alpha,\sigma^2)|\alpha_{\diamond}<\alpha\le1, \  F_1^{-1}(\alpha_{\diamond})\le\sigma^2 < F_1^{-1}(\alpha_{*})\right\}.
\end{split}
\]
Based on the strong global attractiveness of $F_1$ and $F_2$ (and similar to the proof of part (i) of this lemma), we can show the following:
\begin{itemize}
\item if $(\alpha_{t_0},\sigma^2_{t_0})\in\mathcal{R}_{3a}$, then $(\alpha_{t_0+1},\sigma^2_{t_0+1})$ can only be in $\mathcal{R}_{3a}$;

\item if $(\alpha_{t_0},\sigma^2_{t_0})\in\mathcal{R}_{3b}$, then $(\alpha_{t_0+1},\sigma^2_{t_0+1})$ can be in $\mathcal{R}_{3a}$, $\mathcal{R}_{3b}$ or $\mathcal{R}_{3c}$;

\item if $(\alpha_{t_0},\sigma^2_{t_0})\in\mathcal{R}_{3c}$, then $(\alpha_{t_0+1},\sigma^2_{t_0+1})$ can be in $\mathcal{R}_{3c}$ or $\mathcal{R}_{3a}$.

\end{itemize}
Putting things together, and similar to the treatment of $\mathcal{R}_0$, it can be shown that there exists a finite $T_3$ such that $(\alpha_{t},\sigma^2_{t})\in\mathcal{R}_{3a}$ for all $t\ge t_0+T_3$.  

It only remains to prove that if $(\alpha_{t'},\sigma^2_{t'})\in\mathcal{R}_{3a}$ at a certain $t'\ge0$, then $\{(\alpha_t,\sigma^2_t)\}_{t\ge t'}$ converges to $(\alpha_{\star},\sigma^2_{\star})$. To this end, define
\[
\begin{split}
B_1^{\mr{low}}(\alpha,\sigma^2) &\Mydef\min\left\{ \alpha_{\star},\alpha, F_1(\sigma^2)\right\},\\
B_1^{\mr{up}}(\alpha,\sigma^2) &\Mydef\max\left\{ \alpha_{\star},\alpha, F_1(\sigma^2)\right\},\\
B_2^{\mr{low}}(\alpha,\sigma^2) &\Mydef\min\left\{ \sigma^2_{\star},\sigma^2, F_1^{-1}(\alpha)\right\}=F_1^{-1}\left(B_1^{\mr{up}}(\alpha,\sigma^2)\right),\\
B_2^{\mr{up}}(\alpha,\sigma^2) &\Mydef\max\left\{ \sigma^2_{\star},\sigma^2, F_1^{-1}(\alpha)\right\}=F_1^{-1}\left(B_1^{\mr{low}}(\alpha,\sigma^2)\right).
\end{split}
\]
See examples depicted in Fig.~\ref{Fig:regions_noisy}. Using the strong global attractiveness of $F_1$ and $F_2$ and noting that $F_1^{-1}(\alpha)>F_2(\alpha)>\sigma^2_{\star}$ for $\alpha\in[\alpha_{\ast},\alpha_{\star})$ and $F_1^{-1}(\alpha)<F_2(\alpha)<\sigma^2_{\star}$ for $\alpha\in(\alpha_{\star},\alpha_{\diamond})$, it can be proved that
\[
\begin{split}
B_1^{\mr{low}}(\alpha_t,\sigma^2_t)& \le \alpha_{t+1} \le B_1^{\mr{up}}(\alpha_t,\sigma^2_t),\\
B_2^{\mr{low}}(\alpha_t,\sigma^2_t)& \le \sigma^2_{t+1} \le B_2^{\mr{up}}(\alpha_t,\sigma^2_t).
\end{split}
\]
Further, the sequences $\{B_1^{\mr{low}}(\alpha_t,\sigma^2_t)\}_{t\ge t'}$ and $\{B_2^{\mr{low}}(\alpha_t,\sigma^2_t)\}_{t\ge t'}$ are monotonically non-decreasing and $\{B_1^{\mr{up}}(\alpha_t,\sigma^2_t)\}_{t\ge t'}$ and $\{B_2^{\mr{up}}(\alpha_t,\sigma^2_t)\}_{t\ge t'}$ are monotonically non-increasing. Also, $B_1^{\mr{low}}$ and $B_2^{\mr{low}}$ are upper bounded by $\alpha_{\star}$ and $\sigma^2_{\star}$, and $B_1^{\mr{up}}$ and $B_2^{\mr{up}}$ are lowered bounded by $\alpha_{\star}$ and $\sigma^2_{\star}$. Together with some arguments about the strict monotonicity of $\{B_1^{\mr{low}}(\alpha_t,\sigma^2_t)\}_{t\ge t'}$ and $\{B_2^{\mr{low}}(\alpha_t,\sigma^2_t)\}_{t\ge t'}$ (see discussions below \eqref{Eqn:lem_regionI_III_2.2}), we have 
\[
\begin{split}
\lim_{t\to\infty}B_1^{\mr{low}}(\alpha_t,\sigma^2_t) &=\lim_{t\to\infty}B_1^{\mr{up}}(\alpha_t,\sigma^2_t)=\alpha_{\star},\\
\lim_{t\to\infty}B_2^{\mr{low}}(\alpha_t,\sigma^2_t) &=\lim_{t\to\infty}B_2^{\mr{up}}(\alpha_t,\sigma^2_t)=\sigma^2_{\star},
\end{split}
\]
which implies that $\lim_{t\to\infty}\alpha_{t+1}=\alpha_{\star}$ and $\lim_{t\to\infty}\sigma^2_{t+1}=\sigma^2_{\star}$.
We skip the proofs for the above statements since similar arguments have been repeatedly used in this paper.
\end{proof}

\subsubsection{Proof of Theorem \ref{thm:noisesens_comp}}
According to Lemma \ref{Lem:SE_converge_noisy}, we know that $(\alpha_t, \sigma_t^2)$ converges to the unique fixed point of the state evolution equation. We now analyze the location of this fixed point and further derive the noise sensitivity. Applying a variable change $s\Mydef\sigma/\alpha$, 
we obtain the following equations for this unique fixed point:
\BS\label{Eqn:SE_fixed_noisy_complex2}
\begin{align}
\alpha &=\phi_1(s), \label{Eqn:SE_fixed_noisy_complex2a}\\
\sigma^2&=  \frac{4}{\delta}\left\{\alpha^2+\sigma^2+1- \alpha \left[\phi_1\left(s\right)+\phi_3\left(s\right)\right]\right\}+4\sigma^2_w,\label{Eqn:SE_fixed_noisy_complex2b}
\end{align}
\ES
where $\phi_1$ and $\phi_3$ are defined in \eqref{Eqn:phi1_phi3}.
Using \eqref{Eqn:SE_fixed_noisy_complex2a} and $\sigma^2 = \alpha^2 s^2=\phi_1^2(s) s^2$, and after some algebra, we can write \eqref{Eqn:SE_fixed_noisy_complex2b} as
\BE\label{eq:TeqFPs}
T(s^2, \sigma_w^2 ) \triangleq \left(1-\frac{4}{\delta}\right)\phi^2_1(s)s^2 + \frac{4}{\delta}\phi_1(s)\phi_3(s)-\left( \frac{4}{\delta}+4\sigma^2_w \right)=0.
\EE
Differentiating with respect to $s^2$ yields
\BE
\frac{\partial T(s^2, \sigma_w^2 )}{\partial s^2} =  \left(1-\frac{4}{\delta}\right)  \left(\phi_1^2(s) + 2\phi_1(s)\frac{\mr{d}\phi_1(s)}{\mr{d}s^2}s^2 \right) + \frac{4}{\delta} \frac{\mr{d}\phi_1(s)\phi_3(s)}{\mr{d}s^2}.
\EE
Using the identities listed in \eqref{Eqn:noisy_identities}, we have
\[
\left. \frac{\partial T(s^2, \sigma_w^2 )}{\partial s^2}\right|_{s=0}=1-\frac{2}{\delta}.
\]
Also, it is straightforward to see that $\frac{\partial T(s^2,\sigma^2_w)}{\partial \sigma^2_w}=-4$.
Note that we have an implicit relation between $s^2$ and $\sigma_w^2$, and by the implicit function theorem we have
\[
\begin{split}
 \lim_{\sigma_w^2 \rightarrow 0}  \frac{\mr{d}s^2}{\mr{d}\sigma_w^2}&=-\lim_{s^2 \rightarrow 0} \left( \frac{\partial T(s^2,\sigma^2_w)}{\partial s^2} \right)^{-1}\frac{\partial T(s^2,\sigma^2_w)}{\partial \sigma^2_w}=\frac{4}{1-\frac{2}{\delta}}.
\end{split}
\]
Further, $s$ is a continuously differentiable function of $\sigma_w^2$. Hence, by the mean value theorem we know that 
\[
 \frac{s^2}{\sigma_w^2} = \left. \frac{\mr{d}s^2}{\mr{d}\sigma_w^2} \right|_{\tilde{\sigma}_w^2},
\]
where $0 \leq \tilde{\sigma}_w \leq \sigma_w$. By taking $\lim_{\sigma_w^2 \rightarrow 0}$ from both sides of the above equality we have
\[
\lim_{\sigma_w^2 \rightarrow 0} \frac{s^2}{\sigma_w^2} = \lim_{\tilde{\sigma}_w \rightarrow 0} \left. \frac{\mr{d}s^2}{\mr{d}\sigma_w^2} \right|_{\tilde{\sigma}_w^2} = -\lim_{s^2 \rightarrow 0} \left( \frac{\partial T(s^2,\sigma^2_w)}{\partial s^2} \right)^{-1}\frac{\partial T(s^2,\sigma^2_w)}{\partial \sigma^2_w}=\frac{4}{1-\frac{2}{\delta}}.
\]

To derive the noise sensitivity, we notice that
\[
\begin{split}
\mr{AMSE}(\sigma^2_w,\delta) &= (\alpha-1)^2 + \sigma^2\\
&=\left[\phi_1(s)-1\right]^2 + s^2 \phi_1^2(s).
\end{split}
\]
As shown in \eqref{Eqn:noisy_identities}, $\phi_1(s)$ can be expressed using elliptic integrals as:
\[
\begin{split}
\phi_1(s) &= \sqrt{1+s^2}E\left(\frac{1}{1+s^2}\right) - \frac{s^2}{\sqrt{1+s^2}}K\left(\frac{1}{1+s^2}\right).
\end{split}
\]
From Lemma \ref{Lem:elliptic}-(i), $E(1-\epsilon)=1+O(\epsilon\log\epsilon^{-1})$, hence $\sqrt{1+s^2}E\left(\frac{1}{1+s^2}\right)=1+O(s^2\log s^{-1})$. Further, since $K(1-\epsilon)=O(\log\epsilon^{-1})$, we have $\frac{s^2}{\sqrt{1+s^2}}K\left(\frac{1}{1+s^2}\right)=O(s^2\log s^{-1})$. Therefore, $\phi_1(s)-1=O(s^2\log s^{-1} )$.
Hence, $\lim_{s^2\to0}\frac{\left[\phi_1(s)-1\right]^2}{s^2}=0$ and so
\[
\lim_{s^2\to0} \frac{\mr{AMSE}(\sigma^2_w,\delta)}{s^2}=\lim_{s^2\to0} \frac{\left[\phi_1(s)-1\right]^2}{s^2}+\phi_1^2(s) = 1.
\]
Finally,
\[
\begin{split}
\lim_{\sigma^2_w\to0}\frac{\mr{AMSE}(\sigma^2_w,\delta)}{\sigma^2_w} &= \lim_{s^2\to0}\frac{\mr{AMSE}(\sigma^2_w,\delta)}{s^2} \cdot \lim_{\sigma_w^2 \rightarrow 0} \frac{s^2}{\sigma_w^2} \\
&=\frac{4}{1 -\frac{2}{\delta}}.
\end{split}
\]

\subsection{Proof of Theorem \ref{thm:noisesens_real}}
In the noisy setting, the state evolution of AMP.A becomes
\BS\label{Eqn:SE_fixed_noisy_real}
\begin{align}
\psi_1(\alpha,\sigma^2) &= \frac{2}{\pi}\mr{arctan}\left(\frac{\alpha}{\sigma}\right), \\
\psi_2(\alpha,\sigma^2;\delta,\sigma^2_w) &= \frac{1}{\delta}\left[\alpha^2+\sigma^2 + 1 -\frac{4\sigma}{\pi} -\frac{4\alpha}{\pi}\mr{arctan}\left(\frac{\alpha}{\sigma}\right)\right] +\sigma^2_w.
\end{align}
\ES
Similar to the complex-valued case, the SE of real-valued AMP.A still converges to the nonzero fixed point, as stated in Lemma \ref{Lem:SE_converge_noisy_real} below. We skip the proof since it is very similar to the proof of Lemma \ref{Lem:SE_converge_noisy}.
\begin{lemma}\label{Lem:SE_converge_noisy_real}
Let $\{\alpha_t\}_{t\ge1}$ and $\{\sigma^2_t\}_{t\ge1}$ be two state sequences generated according to \eqref{Eqn:SE_real} from $\alpha_0>0$ and $\sigma^2_0<\infty$. Then, for any $\delta>\deltaAMP$ the following holds for sufficiently small $\sigma^2_w$:
\[
\lim_{t\to\infty} \alpha_t=\alpha_{\star}(\delta,\sigma^2_w)\quad\text{and}\quad \lim_{t\to\infty} \sigma^2_t=\sigma^2_{\star}(\delta,\sigma^2_w),
\]
where $\alpha_{\star}(\delta,\sigma^2_w)$ is the unique positive solution to $F_1^{-1}(\alpha)=F_2(\alpha;\delta,\sigma^2_w)$ and $\sigma^2_{\star}(\delta,\sigma^2_w)=F_1^{-1}(\alpha_{\star}(\delta,\sigma^2_w))$.
\end{lemma}

Now we can prove Theorem \ref{thm:noisesens_real}. First note that $\mr{AMSE}(\sigma^2_w,\delta)=(\alpha-1)^2+\sigma^2$, where with slight abuse of notation $\alpha$ and $\sigma^2$ denote the solution of \eqref{Eqn:SE_fixed_noisy_real} (which are also functions of $\sigma^2_w$ and $\delta$), i.e.,
\BS\label{Eqn:SE_fixed_noisy_real2}
\begin{align}
\alpha&= \frac{2}{\pi}\mr{arctan}\left(\frac{\alpha}{\sigma}\right), \label{Eqn:SE_fixed_noisy_real2_a}\\
\sigma^2&= \frac{1}{\delta}\left[\alpha^2+\sigma^2 + 1 -\frac{4\sigma}{\pi} -\frac{4\alpha}{\pi}\mr{arctan}\left(\frac{\alpha}{\sigma}\right)\right] +\sigma^2_w. \label{Eqn:SE_fixed_noisy_real2_b}
\end{align}
\ES
Using \eqref{Eqn:SE_fixed_noisy_real2_a} and with simple manipulations we can rewrite \eqref{Eqn:SE_fixed_noisy_real2_b} as 
\BE\label{Eqn:SE_fixed_noisy_real3}
(\delta-1)\sigma^2 +\alpha^2+\frac{4\sigma}{\pi}-1-\delta\sigma^2_w=0.
\EE
We make the following variable change:
\[
s\Mydef \frac{\sigma}{\alpha}.
\]
From \eqref{Eqn:SE_fixed_noisy_real2_a} and the definition of $s$, we have
\BE\label{Eqn:SE_fixed_noisy_real4}
\alpha = \frac{2}{\pi}\mr{arctan}(s^{-1})\quad\text{and}\quad\sigma = \frac{2}{\pi}\mr{arctan}(s^{-1})\cdot s.
\EE
Substituting \eqref{Eqn:SE_fixed_noisy_real4} into \eqref{Eqn:SE_fixed_noisy_real3} yields
\BE\label{Eqn:SE_fixed_noisy_real5}
T(s^2, \sigma_w^2 ) \triangleq \left[(\delta-1)s^2+1\right]\cdot \mr{arctan}^2(s^{-1}) + 2 \cdot s\cdot \mr{arctan}(s^{-1})-\frac{\pi^2}{4}(1+\delta\sigma^2_w)=0.
\EE
We have
\[
\begin{split}
\frac{\partial T(s^2, \sigma_w^2)}{\partial s^2} &=\frac{1}{2s}  \left(2s (\delta-1)  \mr{arctan}^2\left(s^{-1}\right) -2 \left[(\delta-1)s^2+1\right] \frac{\mr{arctan}\left(s^{-1}\right)}{1+s^2} +  2 \mr{arctan}\left(s^{-1}\right) - \frac{2s}{1+s^2} \right)  \\
&=(\delta-1) \cdot \mr{arctan}^2\left(s^{-1}\right) - \mr{arctan}\left(s^{-1}\right) \frac{(\delta-1)s}{1+s^2} + \mr{arctan}\left(s^{-1}\right)\frac{s}{1+ s^2}  -\frac{1}{1+s^2},\\
\frac{\partial T(s^2, \sigma_w^2)}{\partial \sigma_w^2}&=-\frac{\pi^2}{4}\delta.
\end{split}
\]
Note that we have an implicit relation between $s^2$ and $\sigma_w^2$. By the implicit function theorem we have
\[
\begin{split}
\frac{\mr{d}s^2}{\mr{d}\sigma_w^2}&=-\frac{\partial T(s^2,\sigma^2_w)}{\partial \sigma_w^2}\left( \frac{\partial T(s^2,\sigma^2_w)}{\partial s^2} \right)^{-1}\\
&=\frac{\frac{\pi^2}{4}\delta}{(\delta-1) \cdot \mr{arctan}^2\left(s^{-1}\right) - \mr{arctan}\left(s^{-1}\right) \frac{(\delta-1)s}{1+s^2} + \mr{arctan}\left(s^{-1}\right)\frac{s}{1+ s^2}  -\frac{1}{1+s^2}}.
\end{split}
\]
Furthermore, from \eqref{Eqn:SE_fixed_noisy_real5}, we see that $s^2=0$ when $\sigma^2_w=0$ and hence
\[
 \left. \frac{\mr{d}s^2}{\mr{d}\sigma_w^2}\right|_{\sigma^2_w=0}=\frac{\frac{ \pi^2 }{4}\delta}{ \frac{\pi^2}{4} (\delta-1) -1}=\frac{\delta}{\delta -\left(1+\frac{4}{\pi^2}\right)},
\]
where we defined $\mr{arctan} (s^{-1}) =\pi/2$ at $s=0$. Now it is straightforward to use the mean value theorem to prove that
\[
\lim_{\sigma_w^2 \rightarrow 0} \frac{s^2}{\sigma_w^2} =  \left. \frac{\mr{d}s^2}{\mr{d}\sigma_w^2}\right|_{\sigma^2_w=0}=\frac{\delta}{\delta -\left(1+\frac{4}{\pi^2}\right)}.
\]

Further, notice that
\[
\begin{split}
\mr{AMSE}(\sigma^2_w,\delta) &= (\alpha-1)^2 + \sigma^2\\
&= \left[\frac{2}{\pi}\mr{arctan}(s^{-1})-1\right]^2 + \left[\frac{2}{\pi}\mr{arctan}(s^{-1})\cdot s\right]^2,
\end{split}
\]
and it is straightforward to show that
\[
\lim_{s^2\to0}\frac{\mr{AMSE}(\sigma^2_w,\delta)}{s^2} = 1+\frac{4}{\pi^2}.
\]
Hence,
\[
\begin{split}
\lim_{\sigma^2_w\to0}\frac{\mr{AMSE}(\sigma^2_w,\delta)}{\sigma^2_w} &= \lim_{s^2\to0}\frac{\mr{AMSE}(\sigma^2_w,\delta)}{s^2} \cdot \lim_{\sigma_w^2 \rightarrow 0} \frac{s^2}{\sigma_w^2} \\
&=\left( 1+\frac{4}{\pi^2}\right)\cdot\frac{\delta}{\delta -\left(1+\frac{4}{\pi^2}\right)},
\end{split}
\]
which proves Theorem \ref{thm:noisesens_real} by noting that $\deltaGlobal=1+4/\pi^2$.

\section{Spectral initialization}\label{sec:spectral}

\subsection{Initialization}\label{Sec:init_a}

As shown in Section \ref{Sec:asym_framework}, to achieve successful reconstruction, the initial estimate $\bm{x}^{0}$ cannot be orthogonal to the true signal $\bm{x}_{\ast}$, namely,
\BE\label{Eqn:nonorthogonal}
\alpha_0=\lim_{n\rightarrow \infty} \frac{1}{n} \bm{x}_{\ast}^{\UH}\bm{x}^{0} \neq 0. 
\EE
In many important
applications (e.g., astronomic imaging
and crystallography \cite{millane1990phase}), the signal is known to be real and nonnegative. In such cases, the following initialization of $\ampa$ meets the non-orthogonality requirement:  
$$\bm{x}^0=\rho\mathbf{1}, \quad \rho\neq0.$$ 
(At the same time, we set $g(\bm{p}^{-1},\bm{y})=\mathbf{0}$.)

However, note that finding initializations that  satisfy \eqref{Eqn:nonorthogonal} is not straightforward in general settings. For instance, the above initialization may not work for generic complex-valued signals. Also, random initialization does not necessarily work either, since asymptotically speaking a random vector will be orthogoanl to $ \bm{x}_{\ast}$. 
One promising direction to alleviate this issue is the spectral initialization method that was introduced in \cite{Eetrapalli2013} for phase retrieval and subsequently studied in  \cite{CaLiSo15,ChenCandes17,Wang2016,Lu17,Mondelli2017}.  Specifically, the ``direction'' of the signal is estimated by the principal eigenvector $\bm{v}$ ($\|\bm{v}\|^2=n$) \footnote{For the spectral method proposed in \cite{Mondelli2017}, the eigenvalues can be negative and the  eigenvector associated with the largest eigenvalue (not the largest eigenvalue in magnitude) is picked. } of the following matrix:
\BE \label{Eqn:data_matrix}
\bm{D}\Mydef \bm{A}^\UH\mr{diag}\{\mathcal{T}(y_1),\ldots,\mathcal{T}(y_m)\}\bm{A},
\EE
where $\mathcal{T}:\mathbb{R}_{+}\to(-\infty,\tau_{\mr{\max}}]$ is a nonlinear processing function, and $\mr{diag}\{a_1,\ldots,a_m\}$ denotes a diagonal matrix with diagonal elements given by $\{a_1,\ldots,a_m\}$. The exact asymptotic performance of the spectral method was characterized in \cite{Lu17} under some regularity assumptions on $\mathcal{T}$. (In particular, the support of $\mathcal{T}$) The analysis in \cite{Lu17} reveals a phase transition phenomenon: the spectral estimate is not orthogonal to the signal vector $\bm{x}_*$ (i.e., \eqref{Eqn:nonorthogonal} holds) if and only if $\delta$ is larger than a threshold $\delta_{\mr{weak}}$. Later, \cite{Mondelli2017} derived the optimal nonlinear processing function $\mathcal{T}$ (in the sense of minimizing $\delta_{\mr{weak}}$) and showed that the minimum weak threshold is $\delta_{\mr{weak}}=1$ for the complex-valued model. 

The above discussions suggest that the spectral method can provide the required non-orthogonal initialization for $\ampa$. However, the naive combination of the spectral estimate with $\ampa$ will not work: performance of the $\ampa$ that is initialized with the spectral method will not follow the state evolution. This is due to the fact that $\bm{x}^0$ is heavily dependent on the matrix $\bm{A}$ and violates the assumptions of SE.  A trivial remedy is data splitting, i.e, we generate initialization and apply $\ampa$ on two separate sets of measurements \cite{Eetrapalli2013}. However, this simple solution is sub-optimal in terms of sample complexity. To avoid such loss, we propose the following modification to the spectral initialization method, that we call decoupled spectral initialization:

\textbf{Decoupled spectral initialization:} Let $\delta>2$. Set $\bm{v}$ to be the eigenvector of $\bm{D}$ corresponding to the largest eigenvalue defined in \eqref{Eqn:data_matrix}. Let $\bm{x}^0=\rho \cdot \bm{v}$, where $\rho$ is a fixed number which will be discussed later. Define
      \BE \label{Eqn:spectral_p0}
      \bm{p}^0=\left(1-2\tau\mathcal{T}(\bm{y})\right)\circ \bm{Ax}^{0},
      \EE
      where $\circ$ denotes entry-wise product and $\tau$ is the unique solution of \footnote{The uniqueness of solution in \eqref{Eqn:fixed_point_tau_spectral} and \eqref{Eqn:PCA_solution1_b_rewrite} is guaranteed for our choice of $\mathcal{T}(y)$ in \eqref{Eqn:trimming_opt}\cite{Lu17,Mondelli2017}. For the noisy case, we assume that the variance of the noise is known so that  \eqref{Eqn:fixed_point_tau_spectral} and \eqref{Eqn:PCA_solution1_b_rewrite} can be calculated offline. }
 \BE \label{Eqn:fixed_point_tau_spectral}
\varphi_1(\delta,\tau)=\frac{1}{\delta},\quad \tau\in(0,\tau^{\star}),
      \EE
and $\tau^{\star}$ is the unique solution of
\BE\label{Eqn:PCA_solution1_b_rewrite}
\varphi_2(\delta,\tau^{\star})=\frac{1}{\delta},\quad \tau^{\star}\in(0,\tau_{\max}),
\EE
where
\BS\label{Eqn:PCA_auxiliary}
\begin{align}
\varphi_1(\delta,\tau) &\Mydef \mathbb{E}\left[(\delta\, |Z|^2-1)\frac{2\tau\mathcal{T}(Y)}{1-2\tau\mathcal{T}(Y)}\right] ,\\
\varphi_2(\delta,\tau) &\Mydef \mathbb{E}\left[\left( \frac{2\tau\mathcal{T}(Y)}{1-2\tau\mathcal{T}(Y)} \right)^2\right].
\end{align}
\ES

The expectations above are over $Z\sim\mathcal{CN}(0,1/\delta)$ and $Y=|Z|+W$, where $W\sim\mathcal{CN}(0,\sigma^2_w)$ is independent of $Z$.
\\[3pt]
Now we use $\bm{x}^0$ and $\bm{p}^0$ as the initialization for $\ampa$. So far, we have not discussed how we can set $\rho$ and $\mathcal{T}$. In this paper, we use the following $\mathcal{T}(y)$ derived by \cite{Mondelli2017}:
\BE\label{Eqn:trimming_opt}
\mathcal{T}(y)\Mydef \frac{\delta{y}^2-1}{\delta{y}^2+\sqrt{\delta}-1}.
\EE

Note that our initial estimate is given by $\bm{x^0}=\rho\cdot\bm{v}$ (where $\|\bm{v}\|=\sqrt{n}$). Recall from Theorem \ref{Theo:PhaseTransition_complex} that we require $0<|\alpha_0|<1$ and $0\le\sigma^2_0<1$ for $\delta>\deltaAMP$. To satisfy this condition, we can simply set $\rho=\|\bm{y}\|/\sqrt{n}$, which is an accurate estimate of $\|\bm{x}_{\ast}\|/\sqrt{n}$ in the noiseless setting \cite{Lu17}\footnote{Or one can always choose $\rho$ to be small enough. However, this might slow down the convergence rate.}. Under this choice, we have $|\alpha_0|^2+\sigma^2_0=\rho^2=1$. Hence, as long as $\alpha_0\neq0$, we have $0<|\alpha_0|<1$ and $0\le\sigma^2_0<1$.

In summary, our initialization in \eqref{Eqn:spectral_p0} intuitively satisfies ``enough independency'' requirement such that the SE for $\ampa$ still holds. We have clarified this intuition in Section \ref{Sec:spectral_intuition}. Our numerical experiments (see below) suggest that the intuition is correct. Our empirical finding is summarized below.
\begin{finding}\label{Find:1}
Let $\bm{x}^{0}$ and $\bm{p}^{0}$ be generated according to \eqref{Eqn:spectral_p0}, and $\{\bm{x}^t\}_{t\ge1}$ and $\{\bm{p}^t\}_{t\ge1}$ generated by the $\ampa$ algorithm as described in \eqref{Eqn:AMP_complex}. The AMSE converges to  
\[ 
\lim_{n \rightarrow \infty} \frac{1}{n}\|\bm{x}^t -e^{\mr{i}\theta_t} \bm{x}_*\|_2^2 = \left( 1-|\alpha_t|\right)^2 + \sigma^2_t,
\]
where $\theta_{t}=\angle (\bm{x}_*^\UH,\bm{x}_t)$, $\{|\alpha_t|\}_{t\ge1}$ and $\{\sigma^2_t\}_{t\ge1}$ are generated according to \eqref{Eqn:SE_complex} and
\BE
\begin{split}
|\alpha_0|^2 \ =\ \frac{1-\delta\varphi_2(\delta,\tau)}{1+\delta\varphi_3(\delta,\tau)} \quad \qand \quad \sigma^2_0\ =\ 1-|\alpha_0|^2,
\end{split}
\EE
where $\tau$ is the solution to \eqref{Eqn:spectral_p0} and $\varphi_3$ are defined as ($\varphi_2$ is defined in \eqref{Eqn:PCA_auxiliary})
\BE\label{Eqn:PCA_auxiliary2}
\begin{split}
\varphi_3(\delta,\tau) & \Mydef \mathbb{E}\left[(\delta |Z|^2-1)\left(\frac{2\tau\mathcal{T}(Y)}{1-2\tau\mathcal{T}(Y)}\right)^2\right],
\end{split}
\EE
where $Y=|Z|+W$.
 \end{finding}
 
\begin{figure}
\begin{center}
\includegraphics[width=.45\textwidth]{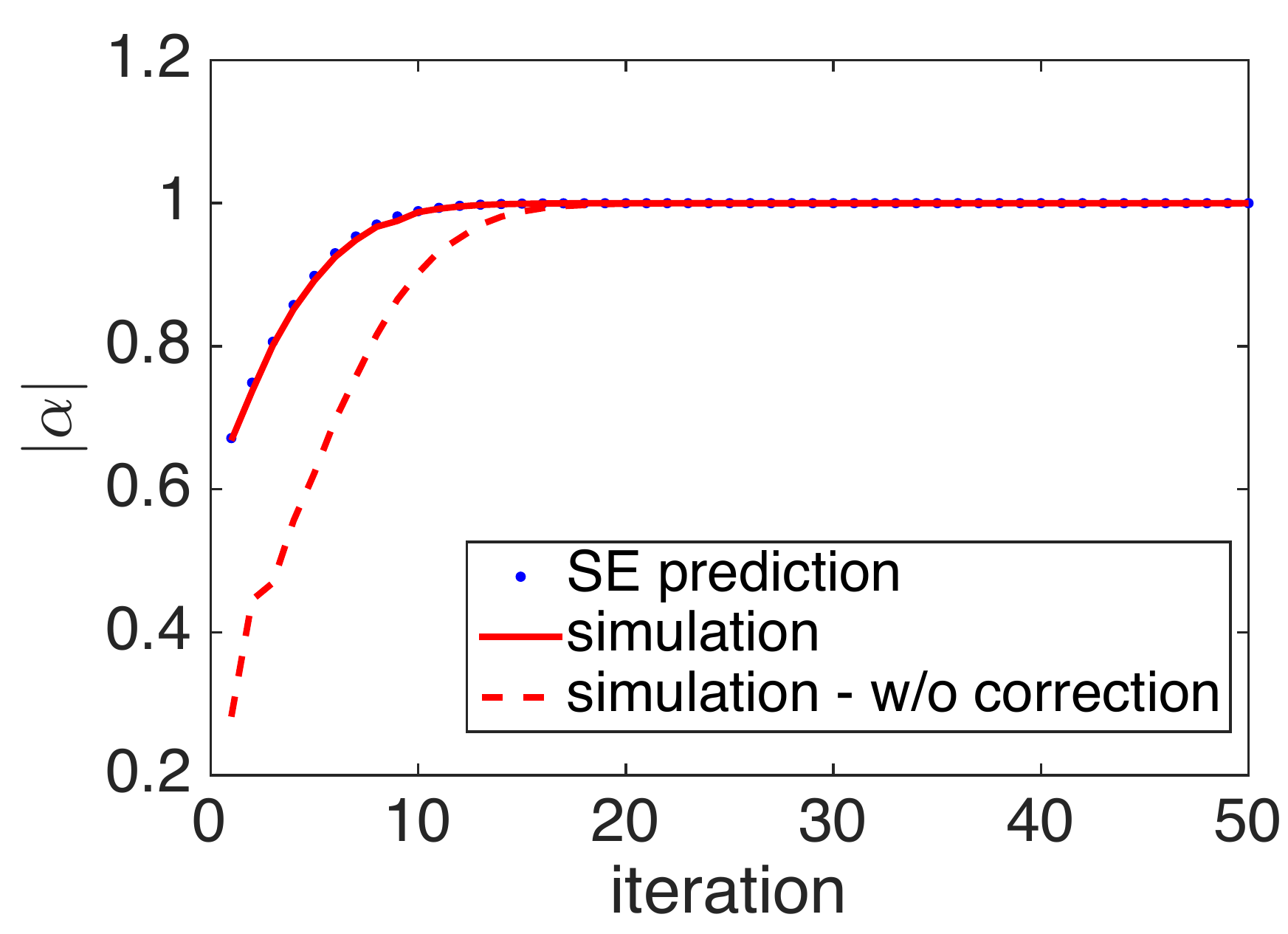}
\includegraphics[width=.45\textwidth]{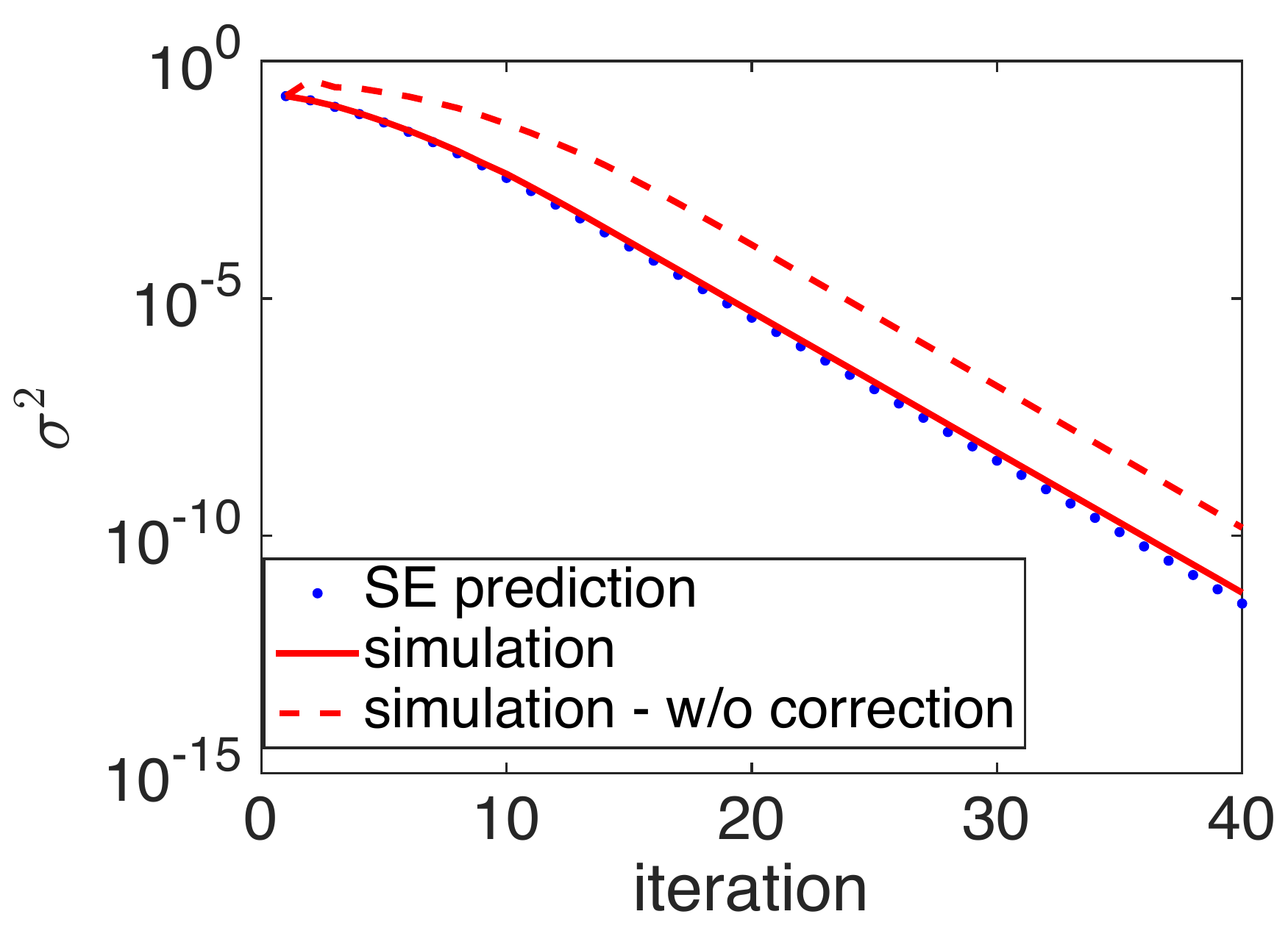}
\caption{State evolution prediction for $\ampa$ with spectral initialization in the noiseless setting. \textbf{Left:} predicted and simulated results of $|\alpha|$. \textbf{Right:} predicted and simulated results of $\sigma^2$. The solid curves show the simulation results for the proposed initialization, and the dashed curves show the results for a naive approach without the proposed correction (namely, we set $\bm{p}^0=\bm{Ax}^0$). In these experiments, $n=5000$ and $m=20000$. The optimal $\mathcal{T}$ in \eqref{Eqn:trimming_opt} is employed.}\label{Fig:SE_accuracy}
\end{center}
\end{figure}

Fig.~\ref{Fig:SE_accuracy} shows a numerical example. The true signal is generated as $\bm{x}_*\sim\mathcal{CN}(\mathbf{0},\bm{I})$. We measure the following two quantities (averaged over 10 runs):
\[
\hat{\alpha}_t =  \frac{\bm{x}_*^\UH \bm{x}^t}{\|\bm{x}_*\|^2}\quad \text{and} \quad \hat{\sigma}^2_t = \frac{\|\bm{x}^t -\hat{\alpha}_t\bm{x}_* \|^2}{\|\bm{x}_*\|^2}.
\] 
We expect $\hat{\alpha}_t$ and $\hat{\sigma}^2_t$ to converge to their deterministic counterparts $\alpha_t$ and $\sigma_t^2$ (as described in Finding \ref{Find:1}). Indeed, Fig.~\ref{Fig:SE_accuracy} shows that the match between the simulated $\hat{\alpha}_t$ and $\hat{\sigma}^2_t$ (solid curves) and the SE predictions (dotted curves) is precise. For reference, we also include the simulation results for the ``blind approach'' where the spectral initialization is incorporated into $\ampa$ without applying the proposed correction (i.e., we use $\bm{p}^0=\bm{Ax}^0$ instead of \eqref{Eqn:spectral_p0}). From Fig.~\ref{Fig:SE_accuracy}, we see that this blind approach deviates significantly from the SE predictions. Note that the blind approach still recovers the signal correctly for the current experiment. However, we found that (results are not shown here) the blind approach can perform rather poorly for other popular choices of $\mathcal{T}$ (such as the orthogonality-promoting method proposed in \cite{Wang2016}).

\subsection{Intuition of our initialization} \label{Sec:spectral_intuition}
Note that in conventional $\ampa$, we set initial $g(\bm{p}^{-1},\bm{y})=\mathbf{0}$ and therefore $\bm{p}^0=\bm{Ax}^0$. Hence, our modification in \eqref{Eqn:spectral_p0} appears to be a rescaling procedure of $\bm{p}^0$. Note that solving the principle eigenvector of $\bm{D}$ in \eqref{Eqn:data_matrix} is equivalent to the following optimization problem:
\BE\label{Eqn:rewrite_goal}
\bm{v}\ =\ \underset{\|\bm{x}\|=\sqrt{n}}{\mr{argmin}}\ -\sum_{a=1}^m \mathcal{T}(y_a)\cdot \big| (\bm{Ax})_a \big|^2.
\EE
Following the derivations proposed in \cite{Rangan11}, we obtain the following approximate message passing algorithm for spectral method (denote as $\amps$): 
\BS\label{Eqn:AMP.S}
\begin{align}
\hat{\tau}^{t} &=\frac{1}{\delta}\frac{1}{ \mr{div}_p(h_{t-1}) } \cdot \frac{\sqrt{n}}{\|\hat{\bm{r}}_{t-1}\|},\label{Eqn:limiteqn1}\\
\hat{\bm{p}}^t &=\bm{A}\hat{\bm{x}}^t  -\frac{1}{\delta} \frac{h\left( \hat{\bm{p}}^{t-1},\bm{y},\hat{\tau}^{t-1}\right) }{\mr{div}_p(h_{t-1}) }\cdot \frac{\sqrt{n}}{\|\hat{\bm{r}}_{t-1}\|} ,\label{Eqn:limiteqn2}\\
\hat{\bm{r}}^t &= \hat{\bm{x}}^t - \frac{\bm{A}^\UH h\left(\hat{ \bm{p} }^{t},\bm{y}, \hat{\tau}^{t}\right) }{\mr{div}_p(h_{t-1})},\\
\hat{\bm{x}}^{t+1} &= -\frac{\sqrt{n}}{\|\hat{\bm{r}}_t\|}\cdot \hat{\bm{r}}^t ,
\end{align}
\ES
where we defined:
\BE
\begin{split}
h(\hat{p},y,\hat{\tau}) &\Mydef \frac{2\mathcal{T}(y)}{1-2\hat{\tau} \mathcal{T}(y)}\cdot \hat{p}.\nonumber
\end{split}
\EE
The optimizer $\bm{v}$ of \eqref{Eqn:rewrite_goal} can be regarded as the limit of the estimate $\hat{\bm{x}}^t$ under correct initialization of $\amps$. Note that $\amps$ acts as a proxy and we do not intend to use it for the eigenvector calculations. (There are standard numerical recipes for that purpose.) But, the correction term used in \eqref{Eqn:spectral_p0} is suggested by the Onsager correction term in AMP.S. To see that let $\hat{\bm{p}}^{\infty} $, $\hat{\bm{x}}^{\infty} $, $\tauh^{\infty}$ represent the limits of $\hat{\bm{p}}^{t} $, $\hat{\bm{x}}^{t} $, $\tauh^{t}$ respectively. Then, from \eqref{Eqn:limiteqn1} and \eqref{Eqn:limiteqn2}, we obtain the following equation
\BE\label{Eqn:AMP.S_tau_fix}
\begin{split}
\hat{\bm{p}}^{\infty}& \overset{(a)}{=}\bm{A}\hat{\bm{x}}^{\infty}  - \tauh^{\infty} h\left( \hat{\bm{p}}^{\infty},\bm{y},\hat{\tau}^{\infty}\right),\\
&\overset{(b)}{=}\bm{A}\hat{\bm{x}}^{\infty}  - \underbrace{\tauh^{\infty}\frac{2\mathcal{T}(\bm{y})}{1-2 \tauh^{\infty} \mathcal{T}(\bm{y})}\circ \hat{\bm{p}}^{\infty}}_{\text{Onsager term}}
\end{split}
\EE
By solving \eqref{Eqn:AMP.S_tau_fix}, we obtain \eqref{Eqn:spectral_p0} with rescaling of $\frac{\|\bm{y}\|}{\sqrt{n}}$ (since $\hat{\bm{x}}^{\infty}=\sqrt{n}\bm{v}$ and $\bm{x}^0=\|\bm{y}\|\bm{v}$). Further, \eqref{Eqn:fixed_point_tau_spectral} and \eqref{Eqn:PCA_solution1_b_rewrite} that determine the value of $\tauh^{\infty}$ can be simplified through solving the fix point of the following state evolution of $\amps$: 
\BS \label{Eqn:PCA_SE_fix2}
\begin{align}
\alphah &=\frac{\alphah\, \varphi_1(\delta,\tauh) }{\sqrt{  \alphah^2\,\varphi_1^2(\delta,\tauh) +\frac{1}{\delta}\varphi_2(\delta,\tauh) + \frac{\alphah^2}{\delta}\varphi_3(\delta,\tauh)}  },  \label{Eqn:PCA_fix_alpha} \\
1 &=  \frac{1}{\delta} \frac{1}{  \sqrt{  \alphah^2\,\varphi_1^2(\delta,\tauh) +\frac{1}{\delta}\varphi_2(\delta,\tauh)+ \frac{\alphah^2}{\delta}\varphi_3(\delta,\tauh) }},
\end{align}
\ES
where $\varphi_1,\varphi_2$ are defined in \eqref{Eqn:PCA_auxiliary} and $\varphi_3$ is defined in \eqref{Eqn:PCA_auxiliary2}.

\end{document}